\documentclass[a4paper,12pt,times,numbered,print,index,custommargin]{PhDThesisPSnPDF}

\usepackage{nomencl} 
\makenomenclature

\usepackage{amsfonts,slashed}
\usepackage{url}
\usepackage{latexsym}
\usepackage{amsfonts}
\usepackage{epsfig}
\usepackage{latexsym,amssymb}
\usepackage{amsmath,amssymb,amsthm}
\usepackage{enumitem}

\def\Gslat{\relax{\slash\kern-.58em G}}
\def\Dslat{\relax{\slash\kern-.68em D}}
\def\Fslat{\relax{\slash\kern-.68em F}}
\def\Phislat{\relax{\slash\kern-.65em \Phi}}

\def\tilde{\widetilde}

\def\1bar{1\hskip -.275cm -}
\def\2bar{2\hskip -.275cm -}
\def\3bar{3\hskip -.275cm -}

\newtheorem{theorem}{Theorem} %teoremi
\newtheorem{definition}{Definition} %definizioni

\def\ii{{\rm i}}

\def\bfone{\relax{\rm 1\kern-.35em 1}}
\def\bfzero{\relax{\rm 0\kern -.45 em 0}}

\newcommand{\so}{\mathfrak{ {}so}}
\newcommand{\osp}{\mathfrak{ {}osp}}

\usepackage{url}
\usepackage{tikz}
\usetikzlibrary{trees}
\usetikzlibrary{positioning,shadows,arrows,decorations.pathmorphing,decorations.markings}
\tikzstyle{block}=[draw opacity=0.7,line width=1.4cm]

\makeatletter
\def\cleardoublepage{\clearpage\if@twoside
\ifodd\c@page
\else\hbox{}\thispagestyle{empty}\newpage
\if@twocolumn\hbox{}\newpage\fi\fi\fi}
\makeatother

\ifsetCustomMargin
  \RequirePackage[left=27mm,right=20mm,top=35mm,bottom=25mm,twoside]{geometry}
  \setFancyHdr % To apply fancy header after geometry package is loaded
\fi

\usepackage{enumitem}
\setlist[enumerate,itemize,description]{topsep=0em}

\ifsetCustomFont
\fi

\RequirePackage[small]{caption}

\RequirePackage[labelsep=space,tableposition=top]{caption}
\usepackage{float}
\usepackage{booktabs} % For professional looking tables
\usepackage{multirow}

\usepackage{siunitx} % use this package module for SI units

\ifuseCustomBib
   \RequirePackage[square, sort, numbers, authoryear]{natbib} % CustomBib

\fi

\renewcommand{\contentsname}{Contents}

\degreetitle{Physics}

\cycle{$30^{th}$ }

\title{Group Theoretical Hidden Structure \\ of Supergravity Theories \\ in Higher Dimensions}
\author{Lucrezia Ravera}
\label{author}

\supervisor{Prof. L. Andrianopoli, \textit{Supervisor} \\
	Prof. R. D'Auria, \textit{Co-Supervisor}}%\\
\committee{Prof. Antoine Van Proeyen, \textit{Referee}, KU Leuven \\
	Prof. Dietmar Klemm, \textit{Referee}, Universit\`{a} degli Studi di Milano, INFN - Sezione di Milano \\
	Prof. Giovanni Barbero, \textit{Internal Member}, Politecnico di Torino \\ 
	Prof. Leonardo Castellani, \textit{External Member}, Universit\`{a} del Piemonte Orientale, INFN - Sezione di Torino, Arnold-Regge Center \\ 
	Prof. Vittorio Penna, \textit{Internal Member}, Politecnico di Torino}

\subject{LaTeX} 

\ifdefineAbstract
\fi

\ifdefineChapter
\fi
\begin{document}

\frontmatter

\maketitle

% ******************************* Thesis Declaration ***************************

\begin{declaration}

I hereby declare that the contents and organization of this dissertation constitute my own original work and does not compromise in any way the rights of third parties, including those relating to the security of personal data.

% Author and date will be inserted automatically from thesis.tex \author \degreedate

\end{declaration}

% ******************************* Thesis Dedidcation ********************************

\begin{dedication} 
\begin{flushright}
\textit{Ai miei genitori, \\ che hanno reso tutto questo possibile.} 
\end{flushright}
\end{dedication}

% ************************** Thesis Acknowledgements **************************

\begin{acknowledgements}      

This thesis is the result of three years of hard work and dedication, and all this would not have been possible without the help and support of many people in so many different ways.

In particular, I am deeply grateful to my supervisor, Prof. Laura Andrianopoli, for helping me with kindness and irreplaceable encouragements to achieve my goals, and to Prof. Riccardo D'Auria, for his guidance, patience, and constant support. I would also thank Prof. Mario Trigiante for the enlightening discussions and fruitful suggestions. I am extremely thankful to all of them for introducing me to many interesting topics.

I am eternally grateful to my parents, Margherita and Walter, for their love, support, and for helping me to
accomplish my dreams. To my brother, Leonardo, for constantly reminding me how beautiful and important it is to be just ourselves.

I wish to say ``\textit{Muchas gracias}'' to my colleagues and friends from Chile, Evelyn Karina Rodr\'{\i}guez, Patrick Keissy Concha Aguilera, and Diego Molina Pe\~{n}afiel, not only for the works done together but also for all the wonderful moments spent together in Italy. Thank you very much to Serena Fazzini and Paolo Giaccone for their friendship and all laughter done together. Thanks a lot to all the PhD students of ``\textit{Sala Dottorandi Giovanni Rana Secondo Piano}''. A special thanks goes to Fabio Lingua, for his invaluable friendship, for his way of facing life and looking at things, and for the questions and doubts to which we have sought (and still try) to find an answer.

I thank all my PhD colleagues and Professors for the path spent together during these three years. It was really a pleasure to meet and know each of them.

Last but not least, I wish to thank Luca Bergesio heartily, for supporting and suggesting me in daily life.

\end{acknowledgements}

% List of publications
\cleardoublepage
\setsinglecolumn
\chapter*{\centering \Large Publications}
\thispagestyle{empty}
\begin{enumerate}

\item P.~Fré, P.~A.~Grassi, L.~Ravera and M.~Trigiante,
  ``Minimal $D=7$ Supergravity and the supersymmetry of Arnold-Beltrami Flux branes,''
  JHEP {\bf 1606} (2016) 018, \\
  doi:10.1007/JHEP06(2016)018
  [arXiv:1511.06245 [hep-th]].

\item L.~Andrianopoli, R.~D'Auria and L.~Ravera,
  ``Hidden Gauge Structure of Supersymmetric Free Differential Algebras,''
  JHEP {\bf 1608} (2016) 095, \\
  doi:10.1007/JHEP08(2016)095
  [arXiv:1606.07328 [hep-th]].
  
\item M.~C.~Ipinza, P.~K.~Concha, L.~Ravera and E.~K.~Rodríguez,
  ``On the Supersymmetric Extension of Gauss-Bonnet like Gravity,''
  JHEP {\bf 1609} (2016) 007, \\
  doi:10.1007/JHEP09(2016)007
  [arXiv:1607.00373 [hep-th]].

\item M.~C.~Ipinza, F.~Lingua, D.~M.~Peñafiel and L.~Ravera,
  ``An Analytic Method for $S$-Expansion involving Resonance and Reduction,''
  Fortsch.\ Phys.\  {\bf 64} (2016) no.11-12,  854, \\
  doi:10.1002/prop.201600094
  [arXiv:1609.05042 [hep-th]].

\item D.~M.~Peñafiel and L.~Ravera,
  ``Infinite S-Expansion with Ideal Subtraction and Some Applications,''
  J.\ Math.\ Phys.\  {\bf 58} (2017) no.8,  081701, \\
  doi:10.1063/1.4991378
  [arXiv:1611.05812 [hep-th]].
  
\item D.~M.~Peñafiel and L.~Ravera,
  ``On the Hidden Maxwell Superalgebra underlying D=4 Supergravity,''
  Fortsch.\ Phys.\  {\bf 65} (2017) no.9,  1700005, \\
  doi:10.1002/prop.201700005
  [arXiv:1701.04234 [hep-th]].
  
\item L.~Andrianopoli, R.~D'Auria and L.~Ravera,
  ``More on the Hidden Symmetries of 11D Supergravity,''
  Phys.\ Lett.\ B {\bf 772} (2017) 578, \\
  doi:10.1016/j.physletb.2017.07.016
  [arXiv:1705.06251 [hep-th]].

\end{enumerate}

%
%\newpage
%
%\begin{center}
%\textbf{\textit{My PhD Research Group}}
%\end{center}
%
%The Supergravity Research Group (SUGRA group) of the Polytechinic of Turin – DISAT (Department of Applied Science and Technology) mainly focuses its activity on the development and the study of supergravity models able to describe our Universe at very small scales, allowing us to understand its origin.
%
%\noindent
%For further information about the SUGRA group and its members, please visit the homepage of the supergravity section of the group ``Physics of Fundamental Interaction'' at the following link:
%
%\begin{center}
%http://areeweb.polito.it/ricerca/sugragroup/index.html
%\end{center}

% ************************** Thesis Abstract *****************************
% Use `abstract' as an option in the document class to print only the titlepage and the abstract.
% The maximum number of characters in abstract must be 4000 (defined by PhD school of Politecnico di Torino). 
\begin{abstract}

The purpose of my PhD thesis is to investigate different group theoretical and geometrical aspects of supergravity theories. 
To this aim, several research topics are explored: On one side, 
the construction of supergravity models in diverse space-time dimensions, including the study of boundary contributions, and the disclosure of the hidden gauge structure of these theories; on the other side, the analysis of the algebraic links among different superalgebras related to supergravity theories.

In the first three chapters, we give a general introduction and furnish the theoretical background necessary for a clearer understanding of the thesis. In particular, we recall the rheonomic (also called geometric) approach to supergravity theories, where the field curvatures are expressed in a basis of superspace. This includes the Free Differential Algebras framework (an extension of the Maurer-Cartan equations to involve higher-degree differential forms), since supergravity theories in $D \geq 4$ space-time dimensions contain gauge potentials described by $p$-forms, of various $p>1$, associated to $p$-index antisymmetric tensors. Considering $D=11$ supergravity in this set up, we also review how the supersymmetric Free Differential Algebra describing the theory can be traded for an ordinary superalgebra of $1$-forms, which was introduced for the first time in the literature in the `80s. This hidden superalgebra underlying $D=11$ supergravity (which we will refer to as the DF-algebra) includes the so called $M$-algebra being, in particular, a spinor central extension of it.

We then move to the original results of my PhD research activity: We start from the development of the so called $AdS$-Lorentz supergravity in $D=4$ by adopting the rheonomic approach and discuss on boundary contributions to the theory. Subsequently, we focus on the analysis of the hidden gauge structure of supersymmetric Free Differential Algebras. More precisely, we concentrate on the hidden superalgebras underlying $D=11$ and $D=7$ supergravities, exploring the symmetries hidden in the theories and the physical role of the nilpotent fermionic generators naturally appearing in the aforementioned superalgebras. After that, we move to the pure algebraic and group theoretical description of (super)algebras, focusing on new analytic formulations of the so called $S$-expansion method. The final chapter contains the summary of the results of my doctoral studies presented in the thesis and possible future developments. In the Appendices, we collect notation, useful formulas, and detailed calculations.

\end{abstract}

% *********************** Adding TOC and List of Figures and tables*************
% By activating (commenting out) these comands the table of contents, list of  
% figures and list of tables auotomatically apear.

\tableofcontents

\listoffigures

%\listoftables

% ********************************** Nomenclature ******************************

\chapter*{Nomenclature}
\thispagestyle{empty}

\textbf{Acronyms / Abbreviations} \\
$\mathcal{N}$ \quad Amount of supersymmetry charges \\
$AdS$ \quad Anti-de Sitter \\
$D$ \quad Space-time dimensions \\
CE-cohomology \quad Chevalley-Eilenberg Lie algebras cohomology \\
CIS \quad Cartan Integrable Systems \\
FDA \quad Free Differential Algebras \\
IW contraction \quad In\"{o}n\"{u}-Wigner contraction

%\printnomenclature

% ******************************** Main Matter *********************************
\mainmatter

%*******************************************************************************
%*********************************** First Chapter *****************************
%*******************************************************************************

\chapter{Introduction}  %Title of the First Chapter
\label{chapter 1}
\ifpdf
    \graphicspath{{Chapter1/Figs/}{Chapter1/Figs/PDF/}{Chapter1/Figs/}}
\else
    \graphicspath{{Chapter1/Figs/Vector/}{Chapter1/Figs/}}
\fi

\nomenclature[Z]{$D$}{Space-time dimensions}

\nomenclature[Z]{$\mathcal{N}$}{Amount of supersymmetry charges}

\begin{flushright}
``\textit{Il pi\`{u} nobile dei piaceri \`{e} la gioia della conoscenza}.''
\\
Leonardo da Vinci
\end{flushright}

In the following chapter, I discuss the state of the art and give some motivations to supersymmetry and supergravity, the latter being the supersymmetric extension of Einstein's General Relativity.
Then, I also furnish a general introduction to the research activity I have done during my PhD.

\section{State of the art}

Three of the four fundamental forces of Nature (strong nuclear interaction, weak nuclear interaction, and electromagnetic interaction) are successfully described by the \textit{Standard Model} of particle Physics, a remarkably successful and predictive physical theory.
These forces are related to gauge symmetries, allowing renormalizability and ensuring a viable quantum theory. On the other hand, gravity is described by General Relativity, and there is not yet a consistent quantum description of gravity which would allow a possible unification with the other interactions. 

In order to reach a unified theory, it is necessary to unify the internal symmetries with the space-time symmetries. A good candidate for this purpose is \textit{supersymmetry} (we will give a theoretical background on supersymmetry in Chapter \ref{chapter 2}; a general introduction to supersymmetry can be found, for example, in Ref. \cite{sohnius}). 
Supersymmetric theories ``put together'' fermions and bosons into multiplets (which are called \textit{supermultiples}).
One of the phenomenological advantages of suspersymmetry is that it allows to cancel quadratic divergences in quantum corrections to the Higgs mass, helping to solve the so called \textit{hierarchy problem} of the Standard Model.

A new algebraic structure, known as \textit{Lie superalgebra}, is necessary in order to describe a supersymmetric theory. This requires to generalize the Poincar\'{e} algebra, introducing, besides the bosonic generators, also fermionic ones (that is, it involves, besides $c$-numbers, also Grassmann variables). In particular, a Lie superalgebras has both commutation and anticommutation relations.
The \textit{simplest supersymmetric extension of gravity} corresponds to \textit{minimal Poincar\'{e} supergravity}, and it can be viewed as the ``gauge'' theory of the Poincar\'{e} superalgebra (more details on supergravity will be furnished in Chapter \ref{chapter 2}; for an exhaustive review, see, for example, \cite{VanNieuwenhuizen:1981ae}).

There is a particular interest in superalgebras going beyond the super-Poincar\'{e} one, which allow to study richer supergravity theories. 
Furthermore, there are several physical models depending on the amount of supersymmetry charges, $\mathcal{N}$, and on the choice of space-time dimensions, $D$. 
The larger $\mathcal{N}$ and the larger $D$, more constraints are present in the theory. 
The maximally extended supergravity theory in four space-time dimensions has $\mathcal{N} = 8$ supersymmetries ($32$ supercharges), while the maximal space-time dimensions in which supersymmetry can be realized is $D=11$. 
Moreover, the inclusion of matter in supergravity theories leads to a vast variety of supergravity models, with diverse physical implications. 

The purpose of my PhD thesis is to investigate different supergravity theories, using geometrical and group theoretical formulations.
The results obtained during my PhD, with national and international collaborators, are presented in \cite{Minimal, Hidden, Gauss, Analytic, GenIW, SM4, Malg}.\footnote{In my thesis I have also corrected some typos that were still present in the aforementioned papers, and also better clarified and contextualized the analyzes we have done.}

\section{Why supergravity? Some motivations}

We would like to discuss here why physicists have been interested
in studying supergravity theories.

An important goal of Theoretical Physics is the understanding of the laws of Physics inside a single, unifying theory.

A first step in this direction has been the unification of electricity with magnetism in the Maxwell laws, and subsequently the formulation of the Standard Model, which unifies the theory of strong interactions with the electroweak one. In the Standard Model of particle Physics, through the \textit{Higgs mechanism} the gauge group
\begin{equation}
SU(3)_C \times SU(2)_L \times U(1)_Y
\end{equation}
breaks down to 
\begin{equation}
SU(3)_C \times U(1)_Q
\end{equation}
(the color, indicated by the sub-index $C$ in $SU(3)_C$, and the charge symmetry, indicated by the sub-index $Q$ in $U(1)_Q$, are still preserved).

A further step has then been that of trying to introduce a \textit{Grand Unified Theory} (GUT): The gauge theory of some simple group
\begin{equation}
G_{GUT} \supset SU(3)_C \times SU(2)_L \times U(1)_Y ,
\end{equation}
allowing, through a double, step-wise Higgs mechanism
\begin{equation}
G_{GUT} \rightarrow _{M_X} \;\;\; SU(3)_C \times SU(2)_L \times U(1)_Y \rightarrow _{M_W} \;\;\; SU(3)_C \times U(1)_Q ,
\end{equation}
an understanding of the Standard Model and of the strong interactions from a unifying theory (with a unified coupling constant), unbroken at a higher energy $M_X \sim 10^{16}$ GeV.

However, in this context, it becomes difficult to justify the deeply different energy scales $M_X \sim 10^{16}$ GeV and $M_W \sim 2$ GeV of the GUT and electroweak breaking respectively, which give particles with very
different masses (in particular, the two Higgs scalars). This is know as the \textit{hierarchy problem}.\footnote{Moreover, GUTs predict that the proton will eventually decay (while it is generally supposed to be stable), even with a very long life-time ($\tau_p \sim 10^{30}$ years for the minimal $SU(5)$ GUT model), which has not yet been experimentally observed.}

As we have already mentioned, the hierarchy problem already exists at the level of the Standard Model, since the Higgs mass is $M_H \sim 125$ GeV, whereas the gravitational scale is of the order of the Plank mass $M_P \sim 10^{19}$ GeV, and $\frac{M_H}{M_P} \sim 10^{-17} << 1$. We might expect that, in a fundamental theory, they should have the same order of magnitude. The Standard Model is considered to be, in a certain sense, ``unnatural'', the loop corrections to the Higgs mass being much larger than the Higgs mass.

In this scenario, \textit{global supersymmetric theories} (with ``rigid'' supersymmetry) are attractive, because they have better renormalization properties than non-supersymmetric ones (for example,
boson and fermion loop corrections to the masses of scalars have opposite sign and cancel each other out). Moreover, the degeneracy in quantum numbers among bosonic and fermionic (super)partners can justify some particular values taken by the quantum numbers of the fields.
The hierarchy between the electroweak scale and the Planck scale is achieved in a natural way, without fine-tuning, as it would be, instead, in the case of the Standard Model, where it is possible to adjust the loop corrections in such a way to keep the Higgs light (requiring cancellations between apparently unrelated tree-level and loop contributions).

For these reasons, supersymmetry is helpful in solving the hierarchy problem of the Standard Model when it
is extended to some GUT, at least if one supposes that it is unbroken up to the scale $M_W$ of breaking of the electroweak symmetry.
Furthermore, supersymmetry implements a unification, since it puts on the same footing bosons (among which the gauge fields that carry the interactions) and fermions
(namely the matter charged under the gauge group).

However, supersymmetry also introduces some phenomenological problem, mainly related to the fact that it must be a somehow broken symmetry, since Nature does not appear to be supersymmetric. 

The supersymmetric extension of the Standard Model gives a quite satisfying understanding of quantum field theory, that is of all quantum interactions \textit{apart from gravitation}. The latter appears to be hardly treated as a quantum field theory, since it is not renormalizable. However, \textit{local supersymmetry} automatically includes \textit{gravity}.

Thus, due to the fact that global supersymmetric theories have better renormalization properties than non-supersymmetric ones and local supersymmetry automatically
includes gravity, \textit{supergravity} (the supersymmetric theory of gravitation) was thought to be, when it was first formulated in the $`70$s, a suitable bridge between quantum gravity and unification. 
Morevoer, supergravity naturally solves the problems related to the breaking of supersymmetry (even if it is non-renormalizable, inheriting this from General Relativity).

\subsection{Supergravity as an effective theory}

The hope was that, even if non-renormalizable, supergravity could be finite, due to a loop-by-loop cancellation of graphs between bosonic and fermionic degrees of freedom.

However, this turned out not to be the case: Even if the divergences in supergravity are softened with respect to non-supersymmetric gravity, supergravity, in general, does not seem
to be a finite theory, and, therefore, it has to be understood as an \textit{effective theory}: It describes the interactions of the light degrees of freedom of some more fundamental
underlying quantum theory.
The natural candidate for such an underlying theory is \textit{superstring theory}: A finite, anomaly free, theory (as general references on superstring, see, for example, Refs. \cite{String, String2}). Actually, it is expected to lead to the
fundamental theory of Nature, not only describing the structure of elementary particles,
but also providing a natural explanation for all interactions in Nature, and even
for the underlying structure of space-time itself. 

Until $1994$, superstring theory was
only known in its perturbative formulation. Five different consistent theories were found: Type IIA, Type IIB, Type I, Heterotic $E_8 \times E_8$, Heterotic $SO(32)$. A big effort was spent in the study of the phenomenological aspects of these theories, in order to understand which
was the one giving rise to our physical world. The spectrum of each theory contains a finite number of massless states and an infinite tower of massive excitations, with
mass scale of the order of the Plank mass ($M_P \sim 10^{19}$ GeV). A feature that the five superstring theories have in common is that their massless modes are described, at low energies (much lower than $M_P$),
by effective supersymmetric field theories, and, in particular, by supergravity theories in ten
space-time dimensions. 
These theories can be suitable for the description of our four-dimensional physical world if the ten-dimensional space-time is thought to be partially compact, with only four non-compact space-time directions.
Indeed, if superstring theory has to provide an explanation of the interactions in our real world which, at low energies, looks four-dimensional, then the vacuum configuration for space-time has to be thought not as a ten-dimensional Minkowski space, but, instead, it should present the form $\mathcal{M}_{(1,3)} \times \mathcal{M}_6$, where $\mathcal{M}_{(1,3)}$ is the $4$-dimensional space-time, while $\mathcal{M}_6$ is a six-dimensional compact manifold, so small that it cannot be observed at the length-scales experimented in our low-energy world.

The main problem in introducing superstring theory as a unifying theory is that, when going down at low energies, one encounters an enormous degeneracy of vacua for string theory. In this sense, we do not gain any predictive power on the quantities characterizing our world. However, we obtain the very important conceptual achievement of unifying gravity with the other interactions and of giving a natural understanding of the origin of all the parameters involved, which are completely arbitrary in the Standard Model.

The supergravity actions which contain fields up to two
derivatives correspond to the effective actions of superstring theory at the lowest order in the string-length parameter $\alpha'$. 

Nowadays, in the context of superstring, supergravity has taken a rather prominent role. Indeed, the understanding, in $1995$, of \textit{D-branes} (extended objects that are included in ``modern'' superstring theory) as non-perturbative objects of string theory has opened the way for the discovery of a web of \textit{dualities} relating all the five superstring theories and supergravity. The current understanding is that the five superstring theories are actually different vacua of a single underlying theory, called
\textit{$M$-theory}, whose low-energy limit is the supergravity theory in eleven dimensions ($D=11$ supergravity, in the following).
In this new perspective, supergravity plays therefore a central role: Properties of $D=11$ supergravity can shed light on string theories in ten dimensions; moreover, D-branes also emerge in supergravity as solitonic objects (as black-holes or domain-walls),
which are solutions to the supergravity equations of motion. 

\section{Overview on my PhD research activity}

During my 1st PhD year, I concentrated my research mainly on the study of supergravity in $D=7$ dimensions, adopting the so called \textit{rheonomic} (or geometric) approach.\footnote{Also known as (super)group-manifold approach.} In this approach to supergravity, the duality between a superalgebra and the Maurer-Cartan equations is used for writing the curvatures in \textit{superspace}, whose basis is given by the so called \textit{vielbein} and \textit{gravitino} $1$-forms (for a theoretical background on this approach, see Chapter \ref{chapter 2}).
In particular, the work \cite{Minimal} I have done with Professors P. Fr\'{e}, P.A. Grassi, and M. Trigiante (in which we have studied some properties of the Arnold-Beltrami flux-brane solutions to the minimal $\mathcal{N}=2$, $D = 7$ supergravity), has been my first opportunity to deal with rheonomy and to understand how to build up supergravity theories within this approach. Indeed, my main contribution to this paper has actually been the rheonomic construction of the minimal $\mathcal{N}=2$, $D=7$ supergravity theory. This turned out to be a necessary step in order to study particular vacuum configurations of the theory (Arnold-Beltrami fluxes) and their supersymmetry breaking pattern. I will not concentrate on this topic in this thesis, since the part I have worked on just involves a lot of cumbersome, heavy calculations. The interested reader can find the complete rheonomic construction of the minimal $\mathcal{N}=2$, $D=7$ supergravity theory in \cite{Minimal}.

In my thesis I will focus, instead, on what I have done in the works \cite{Hidden, Gauss, Analytic, GenIW, SM4, Malg} during the second and third PhD years. The aim is to go beyond the concepts presented above, exploring the group theoretical hidden structure of supergravity theories in diverse dimensions. 
Let me mention, before introducing the main works I will collect in this thesis, that during the PhD I had two great opportunities: The first was to work with my supervisors, L. Andrianopoli and R. D'Auria, to whom I really owe everything. They introduced me to the world of supergravity and to research topics that I really enjoyed and which I hope the reader will appreciate in this thesis. 

The second opportunity was to collaborate with Chilean colleagues, who introduced me (and my PhD colleague F. Lingua) to the $S$-expansion method and to its powerful features, such as that of disclosing the relations among different superalgebras related to supergravity theories. Thanks to our fortuitous meeting and to willpower, we produced some papers together, just among us, PhD colleagues and, first of all, friends. In particular, our aim was to link the pure algebraic aspect of $S$-expansion and algebras that can be obtained or related with this method, to supergravity theories, analyzing the details at the algebraic level and, in a particular case, also the dynamics.
I thanks them all a lot for the good and fruitful job done together.

After a reading, one could say that the ``key word'' of this thesis is ``algebra''... And would be right. Indeed, what I worked on is strongly based on a theoretical study at the algebraic and group level, which, if successful, allows a profound knowledge of the land in which a physical theory has its roots. Thus, I would say that, in this sense, \textit{``algebra''} is not just a key word, but a true \textit{``key'' to open the ``doors'' of the physical world}.

My thesis, in which I collect, reorganize (also correcting some misprints), and clarify the main results of my PhD research activity in details, is organized as follows:
\begin{itemize}
\item In Chapter \ref{chapter 2} and \ref{chapter 3} I furnish some theoretical background on supersymmetry and supergravity, focusing on concepts and frameworks that are necessary for a clearer understanding of the thesis. I also give a review of the $S$-expansion method, recalling, in particular, definitions and useful theorems.
Then, I move to the original results of my PhD research activity (Chapters \ref{chapter 4} to \ref{chapter 6}).
\item Chapter \ref{chapter 4} is devoted to the study of the so called $AdS$-Lorentz supergravity in $D=4$, developed by adopting the rheonomic (geometric) approach, and to the analysis of the theory in a space-time endowed with a non-trivial boundary. In the presence of a (non-trivial) boundary, the fields do not asymptotically vanish, and this has some consequences on the invariances of the theory; in particular, we will concentrate on the supersymmetry invariance of the $AdS$-Lorentz supergravity theory in $D=4$.
\item Chapter \ref{chapter 5} contains the core of my research activity, that is an analysis of the hidden gauge structure of some supersymmetric Free Differential Algebras. In particular, I will concentrate on the $D=11$ and $D=7$ cases and further present a deeper discussion on the symmetries of $D=11$ supergravity; the aim is a clearer understanding of the relations among the hidden superalgebra underlying the eleven-dimensional supergravity theory and other meaningful superalgebras in this context.
\item In Chapter \ref{chapter 6}, moving to the pure algebraic description of (super)algebras, I focus on (new) analytic formulations of the $S$-expansion method developed with my PhD colleagues.
\item Finally, Chapter \ref{chapter 7} contains the conclusions and some possible future developments. In the Appendices, I collect the notation, useful formulas, and some detailed calculations.
\end{itemize}
This is the outline. At the beginning of each chapter, I will provide an introduction which gives an overview of the content and of the main results obtained.

%******************************************************************************
%****************************** Second Chapter *********************************
%*******************************************************************************

\chapter{Theoretical background on supersymmetry and supergravity}
\label{chapter 2}
\ifpdf
    \graphicspath{{Chapter2/Figs/}{Chapter2/Figs/PDF/}{Chapter2/Figs/}}
\else
    \graphicspath{{Chapter2/Figs/Vector/}{Chapter2/Figs/}}
\fi

\nomenclature[Z]{$AdS$}{Anti-de Sitter}

\nomenclature[Z]{CIS}{Cartan Integrable Systems}

\nomenclature[Z]{CE-cohomology}{Chevalley-Eilenberg Lie algebras cohomology}

\nomenclature[Z]{FDA}{Free Differential Algebra}

In this chapter, we first recall the main aspects of supersymmetry (following the lines of Ref. \cite{sohnius}). Then, we move to supergravity, reviewing, in particular, its formulation on superspace and the rheonomic (geometric) approach to supergravity (on the lines of \cite{VanNieuwenhuizen:1981ae, Libro1, Libro2}). We also explain how to study $D=4$ pure supergravity theories in the presence of a boundary (and of a cosmological constant) in the geometric approach, following \cite{bdy}. Finally, we introduce the Free Differential Algebras framework (see, for example, Ref. \cite{Libro2}), and, considering $D=11$ supergravity in this set up, we also review how the supersymmetric Free Differential Algebra can be traded for an ordinary superalgebra of $1$-forms (the hidden superalgebra underlying $D=11$ supergravity was disclosed in 1982 by R. D'Auria and P. Fré in \cite{D'AuriaFre}). This will be useful for a clearer understanding of Chapter \ref{chapter 5}.

\section{Supersymmetry and supergravity in some detail}\label{susyandsugra}

Before proceeding to discuss supersymmetry (and supergravity) in some detail, we should first say something about the Fermi-Bose, matter-force dichotomy (following the discussion presented in \cite{sohnius}). Indeed, the wave-particle duality of Quantum Mechanics, together with the subsequent concept of the ``exchange particle'' in perturbative Quantum Field Theory, seemed to have removed that distinction.
However, forces are mediated by gauge potentials, namely by spin-$1$ vector fields, whereas matter is made of quarks and leptons, that is to say, from spin-$1/2$ fermions.\footnote{Besides integer-spin mesons.} The Higgs particles, mediators of the needed spontaneous breakdown of some of the gauge invariances, play in some sense an intermediate role and must have zero spin (they are bosons), but they are not directly related to any of the forces. 
Supersymmetric theories ``put together'' fermions and bosons into multiplets (which go under the name of \textit{supermultiples}), and the distinction between forces and matter becomes phenomenological: Bosons manifest themselves as forces because they can build up coherent classical fields; on the other hand, fermions are seen as matter because no two identical ones can occupy the same point in space (\textit{Pauli exclusion principle}).

As recalled in \cite{sohnius}, there were several attempts to find 
a unifying symmetry which would directly relate multiplets with different spins,\footnote{Here and in the following, we use the term ``spin'' while actually meaning the \textit{helicity}.} but the failure of attempts to make those ``spin symmetries'' relativistically covariant led to the formulation of a series of ``no-go'' theorems, among which, in particular, the ``no-go'' theorem of Coleman and Mandula ($1967$) \cite{Coleman:1967ad}: 
They proved the impossibility of combining space-time and internal symmetries in any but a trivial way. In particular, they showed that a ``unifying'' group must necessarily be locally isomorphic to the direct product of an internal symmetry group and the Poincaré group.

However, one of the assumptions made in the proof presented in \cite{Coleman:1967ad} turned out to be unnecessary: Coleman and Mandula had admitted only those symmetry transformations which form Lie groups with real parameters, whose generators obey well defined \textit{commutation relations}. 

It was subsequently shown that different spins in the same multiplet are allowed if one includes symmetry operations whose generators obey \textit{anticommutation relations}.
This was first proposed in \cite{Golfand:1971iw} and followed up by \cite{Volkov:1972jx}, where the authors gave what we now call a \textit{non-linear realization of supersymmetry}; their model was non-renormalizable.

Subsequently, Wess and Zumino disclosed field theoretical models with an unusual type of symmetry (that was originally named ``supergauge symmetry'' and is now known as ``supersymmetry''), which connects bosonic and fermionic fields and is generated by charges transforming like spinors under the Lorentz group \cite{Wess:1973kz, Wess:1974tw}. These spinorial charges, which may be considered as generators of a continuous group whose parameters are elements of a Grassmann algebra, give rise to a closed system of commutation-anticommutation relations. It turned out that the energy-momentum operators appear among the elements of this system, so that, in some sense, a (non-trivial) fusion between internal and rigid space-time geometric symmetries occurs \cite{Wess:1973kz, Wess:1974tw, Salam:1974yz}.
 
In particular, in $1973$, Wess and Zumino presented a renormalizable field theoretical model of a spin-$1/2$
particle in interaction with two spin-$0$ particles, in which the particles are related by symmetry transformations and therefore ``sit'' in the same multiplet, which is called in many ways: \textit{Chiral multiplet}, \textit{scalar multiplet} (that is the name which was given by Wess and Zumino in their first paper \cite{Wess:1973kz}), and \textit{Wess-Zumino multiplet}.
The limitations imposed by the Coleman-Mandula ``no-go'' theorem were thus circumvented by introducing a fermionic symmetry operator of spin-$1/2$. Such operators obey anticommutation relations with each other and do not
generate Lie groups; therefore, they are not ruled out by the Coleman-Mandula ``no-go'' theorem.

Consequently to this discovery, in $1975$ the authors of \cite{HLS} extended the results of Coleman and
Mandula to include symmetry operations which obey Fermi statistics. They proved that, in the context of
relativistic field theory, the only models which can lead to a solution of the unification problems are \textit{supersymmetric theories}, and they classified all supersymmetry algebras which can play a role in field theory. 

Supersymmetry transformations are generated by quantum \textit{spinor} operators $Q$ (that are called the \textit{supersymmetry charges}) which change fermionic states into bosonic ones and vice versa. Heuristically:
\begin{equation}
Q \vert fermion \rangle = \vert boson \rangle , \;\;\; Q \vert boson \rangle = \vert fermion \rangle .
\end{equation}
Which particular bosons and fermions are related to each other and how many $Q$'s there are depends on the supersymmetric model under analysis.
However, there are some properties which are common to the $Q$'s in any supersymmetric model, such as (see Ref. \cite{sohnius} for details):
\begin{itemize}
\item The $Q$'s combine space-time with internal symmetries.
%By definition, the $Q$'s change the statistics and hence the spin of the states. Spin is related to the behavior under spatial rotations. Thus, supersymmetry is, in some sense, a space-time symmetry. On the other hand, normally, the $Q$'s also affect some of the internal quantum numbers of the states. It is this property of combining space-time with internal symmetries that makes supersymmetric
%field theories interesting in the attempt to unify all fundamental interactions.
\item The $Q$'s behaves like spinors under Lorentz transformations.
%The $Q$'s transform like tensor operators of spin $1/2$ and, in particular, they do not commute with Lorentz transformations. The result of a Lorentz transformation followed by a supersymmetry transformation is different from that when the order of the transformations is reversed.
\item The $Q$'s are invariant under translations.
%It does not matter whether we translate the coordinate system before or after we
%perform a supersymmetry transformation. In technical terms, this means that we have a vanishing commutator of $Q$ with the energy and momentum operators $E$ and $\mathbf{P}$, which generate space-time
%translations, namely
%\begin{equation}\label{bella}
%[Q,E]=[Q,\mathbf{P}]=0.
%\end{equation}
\item The anticommutator of two $Q$'s is a symmetry generator and, in particular, a Hermitian operator with positive definite eigenvalues.
%The structure of a set of symmetry operations is determined by the result of two subsequent operations. One can show that the anticommutator of two $Q$'s (not the commutator) is again a symmetry generator, albeit one of bosonic nature.\footnote{Note also that the anticommutation rules for fermions reflect the Pauli's exclusion principle.} 
%As spinor components, the $Q$'s are in general not Hermitian, but $\lbrace Q, Q^\dagger \rbrace \equiv Q Q^\dagger + Q^\dagger Q$, where $Q^\dagger$ is the Hermitian adjoint of $Q$, is a Hermitian operator with positive definite eigenvalues:
%\begin{equation}\label{eigen}
%\langle \ldots \vert Q Q^\dagger \vert \ldots \rangle + \langle \ldots \vert Q^\dagger Q \vert \ldots \rangle = \mid Q^\dagger \vert \ldots \rangle \mid ^2 +\mid  Q \vert \ldots \rangle \mid ^2 \geq 0.
%\end{equation}
\item The subsequent operation of two finite supersymmetry transformations induces a translation in space and time of the
states on which they operate.
%Indeed, one can prove that
%\begin{equation}\label{bellabella}
%\lbrace Q, Q^\dagger \rbrace = a E + \mathbf{b} \mathbf{P}.
%\end{equation}
%\item \textit{Sign of the proportionality factor $a$ in the anticommutator of two $Q$'s}: When summing equation (\ref{bellabella}) over
%all supersymmetry generators, we find that the $\mathbf{b} \mathbf{P}$ terms cancel, while the $a E$ terms add up, so that we have
%\begin{equation}\label{ener}
%\sum_{\text{all $Q$}} \lbrace Q, Q^\dagger \rbrace \propto E.
%\end{equation}
%Depending on the sign of the proportionality factor, the spectrum for the energy would have to be either $\geq0$ or $\leq0$. For a physical theory with energies bounded
%from below, but not from above, the proportionality factor will therefore be positive.
\end{itemize}
Many of the most important features of supersymmetric theories can be derived from these crucial properties of the supersymmetry generators by chain.
In particular (see Ref. \cite{sohnius} for further details):
\begin{itemize}
\item The spectrum of the energy operator (the Hamiltonian) in a supersymmetric theory contains no
negative eigenvalues.
\item Each supermultiplet must contain at least one boson and one fermion whose spins differ by $1/2$.
\item All states in a multiplet of unbroken supersymmetry have the same mass.
\item Supersymmetry is spontaneously broken if and only if the energy of the lowest lying state (the \textit{vacuum}) is not exactly zero.
\end{itemize} 
Indeed, referring to the latter feature, since our (low-energy) world does \textit{not} appear to be supersymmetric (experiments do not show elementary particles to be accompanied by \textit{superpartners} with different spin but identical mass), if supersymmetry exists and is fundamental to Nature, it can only be realized as a \textit{spontaneously broken symmetry}: The interaction potentials and the basic dynamics are symmetric, but the state with lowest energy (\textit{ground state} or \textit{vacuum}) is not. If a supersymmetry generator acts on the vacuum, the result will not be zero. Due to the fact that the dynamics retain the essential symmetry of the theory, states with very high energy tend to ``lose the memory'' of the asymmetry of the ground state and the ``spontaneously broken (super)symmetry'' ``gets re-established''.

The number of superpartners of a particle state depend on how many generators $Q$'s are present, as conserved charges, in a supersymmetric model.
As we have already already said, the $Q$'s are spinor operators, and a spinor in $D=4$ space-time dimensions must have at least four real components. Therefore, the total number of $Q$'s must be a
multiple of four. A theory with \textit{minimal supersymmetry}, which is called a theory with $\mathcal{N}=1$ supersymmetry (being $\mathcal{N}$ the number of supersymmetry charges), would be invariant under the transformations generated by just the four independent components of a single spinor operator $Q_\alpha$, with $\alpha=1, \ldots, 4$, and will thus give rise to a single superpartner for each particle state. On the other hand, if there is more supersymmetry, there will be several spinor generators with four components each, namely $Q_{\alpha A}$, $A=1,\ldots, \mathcal{N}$; in this case, we talk about a theory with $\mathcal{N}$-extended supersymmetry (giving rise to $\mathcal{N}$ superpartners for each particle state). 
%In this case, the relationship (\ref{bellabella}) between the generators of
%supersymmetry is replaced by
%\begin{equation}
%\lbrace Q_A , (Q_B)^\dagger \rbrace = \delta_{AB}(a E + \mathbf{b}\mathbf{P}).
%\end{equation}

In supersymmetric theories, the superpartners carry a new
quantum number, which goes under the name of \textit{$R$-charge}.
Then, the so called \textit{$R$-symmetry} is a global symmetry that transforms (rotates) the supercharges into each other (these rotations form an internal symmetry group, in a certain sense like isospin).\footnote{Typically, it is $U(1)$ for $\mathcal{N}=1$ supersymmetric theories, while it becomes non-abelian in $\mathcal{N}$-extended supersymmetry.} Most models with extended supersymmetry are naturally invariant under $R$-symmetry. Let us mention that, in the case of
supergravity, this invariance can be ``gauged'' (made local), and one arrives at a natural link between space-time symmetries (general coordinate invariance and supersymmetry) and gauge interactions.
This speaks very much in favor of extended supergravities.

On the other hand, an important
argument against extended supersymmetry is that it does not allow for chiral fermions as they are observed in Nature (neutrinos) (see Ref. \cite{sohnius} for details on this topic). 
This and other arguments of this type hold strictly only in the absence of gravity. In the context
of supergravity, it is possible to overcome such difficulties. For example, in Kaluza-Klein supergravities \cite{Duff:1986hr}, which are characterized by additional spatial dimensions in
which the space is very highly curved (radii in the region of the Planck length), deviations from the
phenomenology of flat space are particularly large, and many ``no-go theorems'' can be overcome.

Furthermore, from the experimental point of view, referring, in particular, to experimental set up for detecting elementary particles such as the LHC (\textit{Large Hadron Collider}) at CERN, due to the fact that at the energy level currently reached there has been no evidence for supersymmetry, it seems that proper supersymmetric models allowing to describe Nature should be $\mathcal{N}$-extended ones, that means more complicated theories with respect to the $\mathcal{N}=1$ case.

Any multiplet of $\mathcal{N}$-extended supersymmetry contains particles with spins (helicities) at least as large as $\frac{1}{4}\mathcal{N}$ (see \cite{sohnius} for details).

In particular, the following limits arise (in four dimensions):
\begin{itemize}
\item $\mathcal{N}_{\text{max}}=4$ for flat-space renormalizable field theories (super-Yang-Mills);
\item $\mathcal{N}_{\text{max}}=8$ for supergravity.
\end{itemize}
More precisely, if the maximum helicity $\lambda_{\text{MAX}}$ of a \textit{massless} multiplet is 
\begin{itemize}
\item $\vert \lambda_{\text{MAX}}\vert \leq 2$, then $\mathcal{N}\leq 8$;
\item $\vert \lambda_{\text{MAX}}\vert \leq 1$, then $\mathcal{N}\leq 4$;
\item $\vert \lambda_{\text{MAX}}\vert \leq 1/2$, then $\mathcal{N}\leq 2$.
\end{itemize}
As we are going to show in some detail, considering \textit{local} supersymmetry implies to include, together with what we will call the \textit{gravitino} (with helicity $\lambda = 3/2$), also the so called \textit{graviton} (with helicity $\lambda =2$) (plus, for extended supergravity, lower helicity states). Then, in order to have a supermultiplet with maximal helicity $\lambda=2$, the maximally extended theory in four dimensions has $\mathcal{N} = 8$ supersymmetries ($32$ supercharges). 

Actually, one could in principle try to couple supergravity with higher helicity states, by considering $\mathcal{N} > 8$ supergravity; however, no consistent interacting field theory can be constructed for spins higher than two, unless they appear in an infinite number (as it happens for the complete spectrum of superstring theory).

Let us mention here (without going deep in details) that, concerning supergravity, a peculiar feature which distinguishes extended supergravities from the minimal $\mathcal{N}=1$ theory is the fact that for $\mathcal{N}\geq 2$ the vector multiplets include scalars, which can be interpreted, at least locally, as the coordinates of an appropriate Riemannian
manifold (called the \textit{scalar-manifold}). 
Let $U$ be the group of isometries (if any) of the scalar metric defined on the scalar-manifold. The elements
of $U$ correspond to global symmetries of the $\sigma$-model Lagrangian describing the scalar kinetic term.
In $1981$, Gaillard and Zumino discovered that the scalar-manifold isometries $U$ act as duality rotations, interchanging electric with magnetic field-strengths \cite{gaillardzumino}. This fact gives a strong constraint on the geometry of the scalar-manifolds. In particular, the isometry group has to be a subgroup, for all $\mathcal{N}$-extended theories in $D=4$, of the symplectic group $Sp(2n)$ (\textit{symplectic embedding}), where $n$ is the number of vectors in the theory. 

Due to this fact, in $\mathcal{N}$-extended (supergravity) theories, the existence of so called \textit{`t Hooft-Polyakov monopoles}\footnote{A `t Hooft-Polyakov monopole is a \textit{topological soliton} similar to the Dirac monopole, but without any singularities; `t Hooft-Polyakov monopoles are non-singular, solitonic monopole-like solutions appearing in non-abelian gauge theories with the key request that the gauge fields are interacting with scalar fields in the adjoint representation of the gauge group.} \cite{tHooft:1974kcl, Polyakov:1974ek} is a concept implemented in a natural way, and monopoles of such type are always present: Indeed, a crucial point for the existence of `t Hooft-Polyakov monopole like solutions in non-abelian gauge theories (and therefore for having electric-magnetic duality) is the presence in the theory of Higgs fields (scalars) transforming in the adjoint representation of the gauge group $U$.

One can then \textit{switch on charges} (``do the gauging'') with respect to a gauge group. A global symmetry of the action is promoted to be a gauge symmetry gauged by (some of) the vectors of the theory.
In a supersymmetric theory, the global symmetries which are present are the isometries of the scalar-manifold. In doing
the gauging, the interplay between fields of different spin (in particular, vectors and
scalars) is always at work.
Strictly speaking, what happens is that one chooses a subgroup of the isometries of the scalar-manifold that wishes to treat as gauge symmetry, and requires that (some of) the vector fields present in the spectrum of the theory are considered as gauge fields, in the
adjoint representation of the selected group of isometries; then, the interactions with the corresponding gauge fields are turned on.
When this is performed, the theory results to be modified. In particular, it is no more supersymmetric invariant, and the composite connections and vielbein on the scalar-manifold get modified, so that the theory needs further modifications in order to recover supersymmetry invariance.

For the case of $\mathcal{N} = 2$ supergravity, for example, the fermions transformation laws get modified and, in order to restore the invariance under local supersymmetry, the supersymmetry variations of the spin-$3/2$ and $1/2$ fermions acquire a \textit{shift term} (the so called \textit{fermionic shift}). Also the Lagrangian acquires extra terms. In particular, it gets a \textit{scalar potential}, which appears as a \textit{scalar-dependent cosmological constant}. Then,
a \textit{gauged supergravity} with general background configurations for the scalar fields has a vacuum with non-zero cosmological constant. We are not going to explain these aspects in details, since it would require a rather long discussion and it would risky to go astray from the guidelines of the thesis.
The interested reader can find more details on these topics in Ref. \cite{Trigiante2}, where dual gauged supergravities are formulated in a particular fruitful framework which goes under the name of the \textit{embedding tensor formalism}.

Let us now move to the algebraic structure of supersymmetric theories, recalling some technical aspects of \textit{Lie superalgebras}. 

\subsection{Lie superalgebras}

A \textit{Lie superalgebra} (also called \textit{graded Lie algebra}) $\mathfrak{g}$ presents both commutation and anticommutation relations and can be decomposed in subspaces as
\begin{equation}
\mathfrak{g} = \mathfrak{g}_0 \oplus \mathfrak{g}_1,
\end{equation}
where we have denoted by $\mathfrak{g}_0$ the subspace generated by the bosonic generators and by $\mathfrak{g}_1$ the subspace generated by the fermionic ones (associated to Grassmann variables). 

Then, the product $\circ$ defined by
\begin{equation}
\circ : \; \mathfrak{g}\times \mathfrak{g} \rightarrow \mathfrak{g}
\end{equation} 
satisfies the following properties \cite{MullerKirsten:1986cw}:
\begin{itemize}
\item Grading: $\forall \; x_i \in \mathfrak{g}_i$, $i=0,1$,
\begin{equation}
x_i \circ x_j \in \mathfrak{g}_{i+j \; \text{mod(2)}} ,
\end{equation}
namely $\mathfrak{g}$ is a graded Lie algebra.
\item (Anti)commutation properties: $\forall \; x_i \in \mathfrak{g}_i$, $\forall \; x_j \in \mathfrak{g}_j$, $i,j=0,1$,
\begin{equation}
x_i \circ x_j = -(-1)^{ij}x_j \circ x_i = (-1)^{1+ij} x_j \circ x_i .
\end{equation}
\item Generalized Jacobi identities: $\forall \; x_k \in \mathfrak{g}_k$, $\forall \; x_m \in \mathfrak{g}_m $, $\forall \; x_l \in \mathfrak{g}_l$, $k,l,m \in \lbrace{ 0,1\rbrace}$, 
\begin{equation}
x_k \circ (x_l \circ x_m)(-1)^{km} + x_l \circ (x_m \circ x_k)(-1)^{lk}+ x_m \circ (x_k \circ x_l)(-1)^{ml}=0.
\end{equation}
\end{itemize}

Thus, the generators of a Lie superalgebra are closed under (anti)commutation relations of the (schematic) type
\begin{equation}\label{commanticommrel}
[B,B]=B , \;\;\; [B,F]=F, \;\;\; \lbrace{ F,F\rbrace}=B,
\end{equation}
where with $B$ we have denoted the bosonic generators, while $F$ denotes the fermionic ones.

\subsubsection{Super-Poincar\'{e} algebra}

One of the simplest supersymmetry algebras corresponds to the \textit{Poincar\'{e} superalgebra} (or \textit{super-Poincar\'{e} algebra}). In particular, the four-dimensional Poincar\'{e} superalgebra is given by the Lorentz transformations $J_{\mu \nu}=-J_{\nu \mu}$, the space-time translations $P_{\mu}$, with $\mu,\nu, \ldots = 0, 1, 2, 3$ ($J_{\mu \nu}$ and $P_{\mu}$ are the generators of the Poincar\'{e} algebra), and the $4$-component Majorana spinor charge $Q _\alpha$ (in the following, we neglect the spinor index $\alpha=1,2,3,4$, for simplicity) satisfying
\begin{equation}
\bar{Q} \equiv Q^\dagger \gamma_0 = Q^T C.
\end{equation}
The super-Poincar\'{e} (anti)commutation relations read as follows:
\begin{align}
& [J_{\mu \nu}, J_{\rho \sigma}] = 2 \eta_{\rho [\nu}\delta^{\tau \lambda}_{\mu]\sigma}J_{\tau \lambda} , \label{bossuperalge} \\ 
& [J_{\mu \nu}, P_\rho] = \eta_{\rho[\nu}P_{\mu]}, \\
& [J_{\mu \nu},Q]= \frac{1}{2}\gamma_{\mu \nu}Q, \\
& [P_\mu, P_\nu]= 0, \\
& [P_\mu, Q] = 0, \\
& \lbrace Q, \bar{Q} \rbrace = \ii C \gamma^\mu P_\mu ,  \label{ferm}
\end{align} 
where $\eta_{\mu \nu}$ is the Minkowski space-time metric in $D=4$, $\gamma_\mu$ are Dirac gamma matrices in $D=4$ satisfying the Clifford algebra
\begin{equation}
\lbrace \gamma_\mu , \gamma_\nu \rbrace = 2 \eta_{\mu \nu}, \;\;\;\;\;
\gamma_{\mu \nu} \equiv \frac{1}{2}[\gamma_\mu, \gamma_\nu],
\end{equation}
and $C$ is the charge conjugation matrix satisfying
\begin{equation}
C \gamma_\mu C^{-1} = - \gamma_\mu ^T .
\end{equation}
As we can see, the structure of the super-Poincar\'{e} algebra implies that the combination of two supersymmetry transformations gives the generator of a space-time translation, namely $P_\mu$.
On the other hand, the commutativity of
the fermionic generator $Q$ with the bosonic $P_\mu$'s implies that the supermultiplets contain one-particle states with the same mass but different spins.

\subsection{Local supersymmetry and supergravity}

Now, considering supersymmetry, carried by the supercharge $Q$, as a \textit{local} symmetry, implies considering also the translations $P_\mu$ as generators of local transformations, which can be strictly related to a general coordinate transformation.\footnote{When the torsion is zero. In the following we will call these local transformations ``gauge'' transformations.}
In this sense, one can say that local supersymmetry somehow involves gravity. 

We are thus facing \textit{supergravity}, which conciliates supersymmetry with General Relativity, being the supersymmetric extension of the latter. The first publications on supergravity date back to $1976$ and correspond to \cite{Freedman:1976xh, Deser:1976eh, Freedman:1976py}.

Actually, \textit{local supersymmetry needs gravity}, and we can also say that ``local supersymmetry and gravity imply each other'' \cite{VanNieuwenhuizen:1981ae}.
The aforementioned interplay can be explicitly seen by looking at an example given in \cite{VanNieuwenhuizen:1981ae}, in which the author considered the simplest model of \textit{global}
supersymmetry, namely the \textit{Wess-Zumino model} \cite{Wess:1973kz, Wess:1974tw}, describing the propagation of a massless supermultiplet, consisting of a scalar, a pseudo-scalar, and a spin-$1/2$ field, in four-dimensional space-time. The action is left invariant (up to total derivatives) by \textit{global} supersymmetry transformations; then, if one considers \textit{local} supersymmetry transformations (that is the spinor parameter involved in the supersymmetry transformations is considered as a space-time dependent parameter), one can show that, in order to recover the invariance of the action, we have to introduce the interaction with the corresponding ``gauge'' field (that is a vectorial spinor, or, if preferred, spinorial vector), the so called \textit{gravitino} $\psi_\mu$, which carries spin $3/2$; but its effect is consistently included only by introducing the interaction with \textit{gravity} as well, through a new extra tensor field $g_{\mu \nu}$, which can be then identified with the metric tensor of space-time. One then finds that the spin-$2$ field $h_{\mu \nu} \propto g_{\mu \nu}-\eta_{\mu \nu}$ (where $\eta_{\mu \nu}$ is the Minkowski metric) is the quantum gravitational field, called the \textit{graviton}. 

In this sense, supergravity is the ``gauge'' theory of supersymmetry: It describes systems which are left invariant by the action on space-time of local supersymmetry transformations.
The Lagrangian one ends up with is precisely the contribution of a complex scalar and a Majorana spinor to the Lagrangian of General Relativity (actually, \textit{plus extra terms}, but no new fields have to be introduced).

The simplest supergravity action consists of the coupling of a field with spin (helicity) $3/2$ (called the \textit{gravitino field}) to gravity. This can be done by considering the so called Einstein-Hilbert term plus a further term, named the Rarita-Schwinger term \cite{Freedman:1976xh, Deser:1976eh, Freedman:1976py}.

Let us mention here that the fields of a supersymmetric theory form a representation of the Poincaré superalgebra given in (\ref{bossuperalge})-(\ref{ferm}). When this representation is restricted to a specific value of the mass operator $P^\mu P_\mu=m^2$, the representation is called an \textit{on-shell representation multiplet}. On-shell representations are characterized by the equality of the number of bosonic and fermionic states. When trying to construct a supersymmetric Lagrangian based on the fields from the on-shell representation multiplets, one observes that the algebra of the super-Poincaré Noether charges closes only for field configurations satisfying the \textit{equations of motion}.\footnote{The field equations constrain the fields of different spins in different ways, and the pairing of bosonic and fermionic degrees of freedom is therefore no more realized in the off-shell theory.}
For this reason, such actions are called \textit{on-shell actions} and we say that the \textit{supersymmetry algebra} is an \textit{on-shell symmetry}; then, the supersymmetry transformations close on-shell, on the equations of motion.  
The consequence is that the supersymmetry algebra is an ``\textit{open algebra}'': When it is realized as an algebra of transformations on the fields, the ``structure constants'' are not, in fact, constant, but functions of the point, and the superalgebra closes only when the equations of motion are satisfied; then, the ``Jacobi identities'' are not identities anymore, but they are, instead, equations containing the information about the field equations and becoming identically zero on-shell. 

However, with the inclusion of \textit{extra auxiliary fields}, that is to say, by introducing in the Lagrangian \textit{non-dynamical
degrees of freedom} (whose equations of motion do not describe
propagation in space-time) which are then fixed, by their field equations, as functions of the physical fields \cite{Nieuwenhuizen:1978, Stelle:1978}, one can then write a theory which is \textit{off-shell invariant} under local supersymmetry and where supersymmetry is linearly realized.
In other words, the auxiliary fields can be eliminated from the Lagrangian and from the equations of motion by use of their own field equations. The result of their elimination gives the on-shell Lagrangian.

\section{The group-manifold approach}\label{gma}

Let us now move to the theoretical formulation of \textit{(super)gravity theories}.

One would need a framework for formulating (super)gravity theories in a general and basis-independent way, exploiting in some way the power of the symmetries involved in these theories. 

This is the case of the so called \textit{(super)group-manifold approach to (super)gravity theories} \cite{VanNieuwenhuizen:1981ae, Libro1, Libro2, Neeman:1978zvv, DAdda:1980axn}, where the theory is formulated only in terms of external derivatives among differential forms and wedge products among them, in a frame that is completely coordinate-independent.  

Before moving to the case of supergravity theories in the aforementioned geometric approach, it is better to first review the basic features of the \textit{group-manifold approach}, and, in particular, the \textit{geometric, (soft) group-manifold formulation of General Relativity}, fixing conventions and definitions.

Previous knowledge of a bit of group theory and of (Euclidean and) Riemannian geometry in the \textit{vielbein basis} is required.\footnote{Let us mention that the main geometric difference between the linear spaces (Euclidean geometry) and the Riemannian manifolds (Riemannian geometry) is that for linear spaces we have the vanishing of the torsion and curvature $2$-forms, while in Riemannian geometry the torsion and the curvature $2$-forms, in general, do not vanish (even if one can consistently set the torsion to zero, in which case the Christoffel symbol of the natural frame $\lbrace \partial_\mu \rbrace$ results to be symmetric in its lower indexes).} 
The reader can find a review of the \textit{geometry of linear spaces and Riemannian manifolds in the vielbein basis} in Appendix \ref{rgv}, on the same the lines of \cite{Libro1}.

\subsection{Group-manifolds and Maurer-Cartan equations}

We will now start by showing how the concept of group-manifold leads to discuss the Lie algebras associated to Lie groups and to the dual concept of \textit{Maurer-Cartan equations} (the presentation we give strictly follows the lines of Ref. \cite{Libro1}).

Lie groups have a natural manifold structure associated with them, and one can describe Lie groups under a differential geometric point of view. In this sense, the terms \textit{Lie group} and \textit{group-manifold} are kind of synonyms, and the left and right translations of a fixed element $a$ of a Lie group $G$ are \textit{diffeomorphisms} (strictly speaking, general coordinate transformations on Riemannian manifolds).

A peculiar property of group-manifolds is the existence of \textit{left- and right-invariant vector fields} or, in the dual vector space language, \textit{left- and right-invariant $1$-forms}.

Since the left and right translations are diffeomorphisms, by taking into account the fact that the Lie bracket operation is invariant under diffeomorphisms (see Ref. \cite{Libro1}), the subset of left- (right-) invariant vector fields results to be closed under the Lie bracket operation. Hence, the left- (right-) invariant vector fields on $G$ form the \textit{Lie algebra $\mathfrak{g}$ of the group $G$}. According with the convention of \cite{Libro1}, in the following we refer to the left-invariant vector fields.

Since any left-invariant vector field is uniquely determined by its value at $e$ (the identity element of $G$), $\mathfrak{g}$ can be identified with the tangent space at the identity, $T_e(G)$.

Let us now introduce a basis $\lbrace{T_A\rbrace}$ ($A=1, \ldots, n \;=\; \dim (G)$) on $T_e(G)$. The generators $\lbrace{ T_A \rbrace}$ close the Lie algebra
\begin{equation}\label{strutturaaaaaaaaa}
[T_A,T_B]= C^C_{\;\; AB}T_C,
\end{equation}
where $C^C_{\;\;AB}$ are constants called the \textit{structure constants} of the Lie algebra $\mathfrak{g}$ of the group $G$. 
The closure of the algebra is encoded in the Jacobi identities
\begin{equation}\label{CCcond}
C^C_{\;\; A[B}C^A_{\;\; LM]}=0.
\end{equation}

The Lie algebra of $G$ can also be expressed in the \textit{dual vector space of left-invariant $1$-forms}.

In particular, considering the basis $\lbrace{ \sigma^A\rbrace}$ ($A=1, \ldots, n \;=\; \dim(G)$) of left-invariant $1$-forms at $T_e ^{\;\;\star}(G)$ (cotangent space at the identity $e$), we can expand $d\sigma^A$ in the complete basis of $2$-forms at $e$, obtaining the so called \textit{Maurer-Cartan equations} for the left-invariant $1$-forms $\sigma^A$:
\begin{equation}\label{mmc}
d \sigma^A + \frac{1}{2}C^A_{\;\; BC}\sigma^B \wedge \sigma^C =0,
\end{equation}
where ``$\wedge$'' is the wedge product between differential forms and where the $C^A_{\;\; BC}$ functions, being left-invariant, are actually constants.

The content of the Maurer-Cartan equations (\ref{mmc}) is completely equivalent to that of equations (\ref{strutturaaaaaaaaa}). We say that equations (\ref{mmc}) give the \textit{dual formulation of the Lie algebra of $G$}.
This can be shown by introducing the basis of left-invariant vectors $T_A^{(R)}$ dual to the cotangent basis $\lbrace{ \sigma^A \rbrace}$ of the left-invariant $1$-forms:
\begin{equation}\label{dualellammmm}
\sigma^A(T^{(R)}_B)= \delta^A_{\;\;B}.
\end{equation}
The label $R$ is a reminder that the vectors $T_A^{(R)}$ generate right translations on $G$; for notational simplicity, in the sequel we will omit the label $R$.
Now, evaluating both sides of (\ref{mmc}) on the vectors $T_M$ and $T_N$, we get
\begin{equation}
d \sigma^A (T_M,T_N)= - \frac{1}{2}C^A_{\;\;BC} \sigma^B \wedge \sigma^C (T_M,T_N) .
\end{equation}
Then, using the following identity (which gives the link between the exterior derivative on forms and the bracket operation on vector fields):
\begin{equation}
d \omega \left( \overrightarrow{X}, \overrightarrow{Y} \right)=\frac{1}{2} \Big \lbrace \overrightarrow{X} \left(\omega \left( \overrightarrow{Y}\right) \right) - \overrightarrow{Y} \left(\omega \left( \overrightarrow{X}\right) \right)- \omega \left( \left[ \overrightarrow{X},\overrightarrow{Y} \right] \right) \Big \rbrace ,
\end{equation}
we can write
\begin{equation}
d \sigma^A (T_M,T_N)= \frac{1}{2}\left[T_M \sigma^A \left( T_N\right)  - T_N \sigma^A \left( T_M \right) - \sigma^A \left( [T_M,T_N ] \right) \right] = - \frac{1}{2}C^A_{\;\;BC} \sigma^B \wedge \sigma^C (T_M,T_N) .
\end{equation}
Then, since $T_M \sigma^A (T_N)=T_N \sigma^A (T_M)=0$ because of (\ref{dualellammmm}), we have
\begin{equation}
\sigma^A \left( [T_M ,T_N ] \right)= C^A_{\;\;MN},
\end{equation}
and, therefore,
\begin{equation}
[T_M,T_N]= C^A_{\;\;MN}T_A.
\end{equation}
Note that the constants entering the Maurer-Cartan equations are the structure constants defined by the Lie algebra.

In this formulation, the closure of the algebra is encoded into the following identity (that is the intergrability condition $d^2=0$ of the Maurer-Cartan equations):
\begin{equation}
d^2 \sigma^A =0,
\end{equation}
since it gives
\begin{equation}
C^C_{\;\;AB}C^A_{\;\;LM}\sigma^L \wedge \sigma^M \wedge \sigma^B =0,
\end{equation}
which is satisfied when (\ref{CCcond}) holds.

Now, a set of independent $1$-forms (namely a cotangent basis on $G$) can be obtained in terms of the group element $g$. Let us consider the $1$-form
\begin{equation}\label{astero}
\sigma = g^{-1} d g.
\end{equation}
One can show that 
\begin{equation}\label{pallona}
d \sigma + \sigma \wedge \sigma =0,
\end{equation}
that is the $1$-form $\sigma$ is left-invariant. Since (\ref{astero}) is a Lie algebra valued matrix of $1$-forms, it can be expanded along the set of generators $T_A$ (in their matrix representation):
\begin{equation}\label{stellona}
\sigma = \sigma^A T_A .
\end{equation}
Introducing (\ref{stellona}) in (\ref{pallona}), and using (\ref{strutturaaaaaaaaa}), one obtains again the Maurer-Cartan equations (\ref{mmc}). In a matrix representation of $G$, equation (\ref{pallona}) is a matrix equation for a set of $	\dim(G)$ linearly independent $1$-forms, and it can be used to explicitly compute the structure constants of $G$ (see \cite{Libro1} for an example in which the Maurer-Cartan equations and the commutation relations for the Poincaré group in $D$ dimensions, that is the group of rigid motion in $D$ dimensions, are derived).

Let us mention that one can introduce a metric on $G$ which is biinvariant, namely is both left- and right-invariant. This is the so called \textit{Killing metric} (actually, \textit{Killing form},\footnote{The Killing form is bilinear and symmetric, and therefore defines a metric on $T_e(G)$; moreover, one can prove that it is also biinvariant.} if one refers to the Lie algebra $\mathfrak{g}$ of $G$), which we denote by $h_{AB}$.
One can then show that (see Ref. \cite{Libro1} for details):
\begin{equation}
h_{AB} = C^L_{\;\;BM}C^M_{\;\;AL}.
\end{equation}
If the Killing metric (Killing form) is non-degenerate, the Lie group (Lie algebra) is said to be semisimple. For compact groups, one can prove that the Killing metric $h_{AB}$ is negative definite.
One can also show that the biinvariance of $h_{AB}$ implies 
\begin{equation}
C^L_{\;\; AB}h_{LC} + C^L_{\;\;AC}h_{BL}=0.
\end{equation}
Therefore, defining
\begin{equation}\label{killbassi}
C_{ABC}=h_{AL}C^L_{\;\; BC},
\end{equation}
one obtains
\begin{equation}\label{bassi}
C_{ABC}+C_{ACB}=0.
\end{equation}
Taking into account the antisymmetry of $C^C_{\;\;AB}$ in the indexes $A$ and $B$, equation (\ref{bassi}) implies complete antisymmetry of the lowered structure constants (\ref{killbassi}). For semisimple groups, the Killing metric can be used to lower or raise the indexes of the Lie algebra.\footnote{In particular, the adjoint and coadjoint representations of the algebra are equivalent, as shown in Ref. \cite{Libro1}.}

\subsubsection{Soft group-manifolds}

Since the left- (right-) invariant vector fields and $1$-forms have, in a given chart, a fixed coordinate dependence and, moreover, one can show that the Riemannian geometry of a group-manifold $G$ is (locally) fixed in terms of its structure constants (see Ref. \cite{Libro1} for details), we say that group-manifolds $G$ have a ``\textit{rigid}'' structure.
As such, group-manifolds cannot be used as domains of definition of fields describing in a dynamical way the structure of space-time.

Nevertheless, a group-manifold $G$ can be identified with the vacuum configuration of a gravitational theory. We are thus led to consider \textit{soft group-manifolds} $\tilde{G}$, according with the notation of \cite{Libro1}, in which the rigid metric structure of $G$ has been ``\textit{softened}'' in order to describe non-trivial physical configurations. Soft group-manifolds $\tilde{G}$ are \textit{locally diffeomorphic} to group-manifolds $G$.

An example of soft group-manifold is the non-rigid four-dimensional space-time itself (namely, the space-time considered as a Riemannian manifold or, in other words, the space-time of General Relativity), which, being diffeomorphic to $\mathbb{R}^4$, can be thought of as the soft group-manifold of the local four-dimensional translations.

A further example is given by the \textit{soft Poincaré group-manifold}. 
Let us first consider a \textit{flat Minkoskian space-time} in four-dimensions, $\mathcal{M}_4$, whose geometry can described in terms of the \textit{vielbein}\footnote{In German, the term ``vielbein'' literally means ``many legs'' (and covers all dimensions), referring to its property of connecting the natural frame and the moving frame, having indexes (``legs'') of both types. Quite commonly in the literature, in four dimensions the more specific term ``vierbein'' (``four legs'') is adopted. The vierbeins are sometimes also called the \textit{tetrads}.} $V^a$ and a \textit{spin connection} $\omega^{ab}$ fulfilling the following equations:
\begin{align}
& R^a \equiv d V^a - \omega^ a_{\;\; b} \wedge V^b =0, \label{eqzerellator} \\
& R^{a}_{\;\; b}\equiv d\omega^a_{\;\;b}- \omega^{a}_{\;\;c}\wedge \omega^c_{\;\; b} =0 , \label{eqzerellacurv}
\end{align}
where $R^a$ is called the \textit{torsion} (sometimes also denoted by $T^a$) and $R^{ab}$ is called the \textit{curvature}.\footnote{From now on, we use Greek indexes to denote the so called coordinate indexes, while the Latin indexes $a,b,c,\ldots$ will label the vielbein basis of $1$-forms $\lbrace V^a \rbrace$ (see Appendix \ref{rgv}).} They are $2$-forms and we will also refer to both of them together as the curvatures.
In a particular Lorentz gauge the solution to the above equations is
\begin{align}
& V^a(x)=dx^a ,\\
& \omega^{ab}(x)\equiv 0,
\end{align}
while in a general Lorentz gauge the solution reads
\begin{align}
& V^a(x,\eta)= \left( \Lambda^{-1}(\eta)dx \right)^a , \label{generella1} \\
& \omega^{ab}(x,\eta) = \left( \Lambda^{-1}(\eta) d\Lambda (\eta)\right)^{ab}, \label{generella2}
\end{align}
being $\eta^{ab}$ the Lorentz parameters. One can prove that (\ref{generella1}) and (\ref{generella2}) correspond to the left-invariant $1$-forms of the \textit{Poincaré group} in four dimensions, $ISO(1,3)$ (indeed, we can identify the $x^a$'s and the $\eta^{ab}$'s with the parameters associated to translations and Lorentz rotations, respectively) and, therefore, (\ref{generella1}) and (\ref{generella2}) satisfy the Maurer-Cartan equations associated. 
Moreover, since $ISO(1,3)$ is locally isomorphic to $\mathcal{M}_4 \times SO(1,3)$, it can also be considered as a (trivial) \textit{principal bundle}, $P(\mathcal{M}_4,\; SO(1,3))$, with \textit{base space} given by
\begin{equation}
\mathcal{M}_4 \equiv ISO(1,3)/SO(1,3)
\end{equation}
and $SO(1,3)$ as \textit{fiber}.

Let us now suppose that the space-time $\mathcal{M}_4$ is \textit{not flat}. In this case, the fields $V^a$ and $\omega^{ab}$, subject to the gauge transformation laws
\begin{align}
& V'^a = \left( \Lambda^{-1} \right)^a_{\;\; b} V^b, \\
& \omega ^a_{\;\;b}=\left( \Lambda^{-1}\right)^a_{\;\;c}\omega^c_{\;\;d} \Lambda^d_{\;\; b} - \left( \Lambda^{-1} \right)^a_{\;\;c}(d\Lambda)^c_{\;\;b},
\end{align}
respectively (see Appendix \ref{rgv}), are defined on a fiber bundle $P(\tilde{\mathcal{M}}_4, \;SO(1,3))$ that is \textit{not isomorphic}, but just \textit{locally diffeomorphic} to $G/H=ISO(1,3)/SO(1,3)$, due to the diffeomorphism $\tilde{\mathcal{M}}_4 \sim \mathcal{M}_4$. We can say that we have ``softened'' the rigid structure of the base space,
\begin{equation}
\mathcal{M}_4 \rightarrow \tilde{\mathcal{M}}_4 ,
\end{equation}
maintaining the structural group $SO(1,3)$, which guarantees Lorentz covariance.

Observe that the $2$-form curvatures $R^a$ and $R^{ab}$ associated to the $1$-forms $V^a$ and $\omega^{ab}$ are defined on the bundle through the gauge transformations
\begin{align}
& R'^a = \left( \Lambda^{-1} \right)^a_{\;\;b}R^b , \\
& R'^a_{\;\;b}= \left( \Lambda^{-1}\right)^a_{\;\;c} R^c_{\;\;d}\Lambda^d_{\;\; b}
\end{align}
(they transform in the vector and in the adjoint representation of $SO(1,3)$, respectively). These, in turn, imply ``\textit{horizontality}'': The $2$-forms $R^a$ and $R^{ab}$ do \textit{not} contain the differential $d\eta^{ab}$. This is expressed by the following equations:
\begin{equation}
\imath_{J_{ab}} R^{ab} =\imath_ {J_{ab}} R^{a} =0,
\end{equation}
where $J_{ab}$ is the left-invariant vector field associated to the fiber $SO(1,3)$ and where we have denoted by $\imath_{J_{ab}} R^{ab}$ and $\imath_ {J_{ab}} R^{a}$ the contraction of the vector $J_{ab}$ on the curvatures $R^{ab}$ and $R^a$, respectively.

A simpler way to obtain this is to start directly with $V^a$ and $\omega^{ab}$ defined on the principal bundle $P(\tilde{\mathcal{M}}_4,\;SO(1,3))$. In this thesis, in particular, we will adopt this point of view. Let us mention that in Ref. \cite{Libro1} the interested reader can also find a description of the way in which the fiber bundle structure can also be obtained from the variational principle, starting with an action defined on the soft group-manifold.

In Section \ref{sugrasuperspacerheo} we will see that in \textit{supergravity theories} one does not factorize all the coordinates which are not associated with the translations: Starting from the \textit{super-Poincaré group}, \textit{only the Lorentz gauge transformations will be factorized; the gauge transformation of supersymmetry will not}. The resulting theory will be described on a principal fiber bundle $P(\mathcal{M}_{4\vert4}, \; SO(1,3))$ (in four dimensions), whose base space is called the \textit{superspace} $\mathcal{M}_{4 \vert 4}$, where the first ``$4$'' in $\mathcal{M}_{4\vert4}$ refers to the bosonic dimensions, while the second ``$4$'' refers to the Grassmannian dimensions (as we will specify in Section \ref{sugrasuperspacerheo}).

With this in mind, let us now turn to a formal description of soft group-manifolds and, in particular, of the Cartan geometric formulation of General Relativity (where the group-manifold $G$ is the Poincaré group $ISO(1,3)$ in four dimensions), reaching the geometric Einstein Lagrangian for General Relativity. Again, we will strictly follow the lines of \cite{Libro1}, where the interested reader can find more details on this formulation.

\subsection{Cartan geometric formulation of General Relativity}

We start with a rigid group $G$, that will soon be identified with the Poincaré group. As we have already seen, the group-manifold structure can be described in terms of the set of left-invariant $1$-forms $\sigma^A$ ($A=1,\ldots, \dim(G)$) satisfying the Maurer-Cartan equations (\ref{mmc}).

Then, let us ``soften'' $G$ to the locally diffeomorphic soft group-manifold $\tilde{G}$, by introducing new Lie algebra valued $1$-forms
\begin{equation}
\mu=\mu^A T_A.
\end{equation}
The ``\textit{soft}'' (non left-invariant) $1$-forms $\mu^A$ ($A=1,\ldots, \dim(\tilde{G})$, $\dim(\tilde{G}) = \dim(G)$) do \textit{not} satisfy the Maurer-Cartan equations, while developing a non-vanishing right-hand side. We can thus write the geometry of $\tilde{G}$ in terms of a \textit{curvature}:
\begin{equation}\label{curvatureGtilde}
R^A(x)\equiv d\mu^A + \frac{1}{2}C^A_{\;\; BC}\mu^B \wedge \mu^C .
\end{equation}

The $\mu^A$'s span a basis of the cotangent plane of $\tilde{G}$ and they are, in fact, vielbeins on $\tilde{G}$ (see Figure \ref{figuresoft}, reproduced from \cite{Libro1}, for a graphic representation). We can then define covariant derivatives for a general covariant $p$-form $\eta^A$ by
\begin{equation}\label{covderGtilde}
\nabla  \eta^A \equiv d\eta^A + C^A_{\;\;BC}\mu^B \wedge \eta^C ,
\end{equation}
and for a general contravariant $p$-form $\eta_A$ by
\begin{equation}\label{contravderGtilde}
\nabla \eta_A \equiv d \eta_A - C^C_{\;\;AB}\mu^B \wedge \eta_C ,
\end{equation}
where we have introduced a \textit{covariant derivative operator} $\nabla$. 

\begin{figure}
\centering
\pgfdeclareimage[height=6cm]{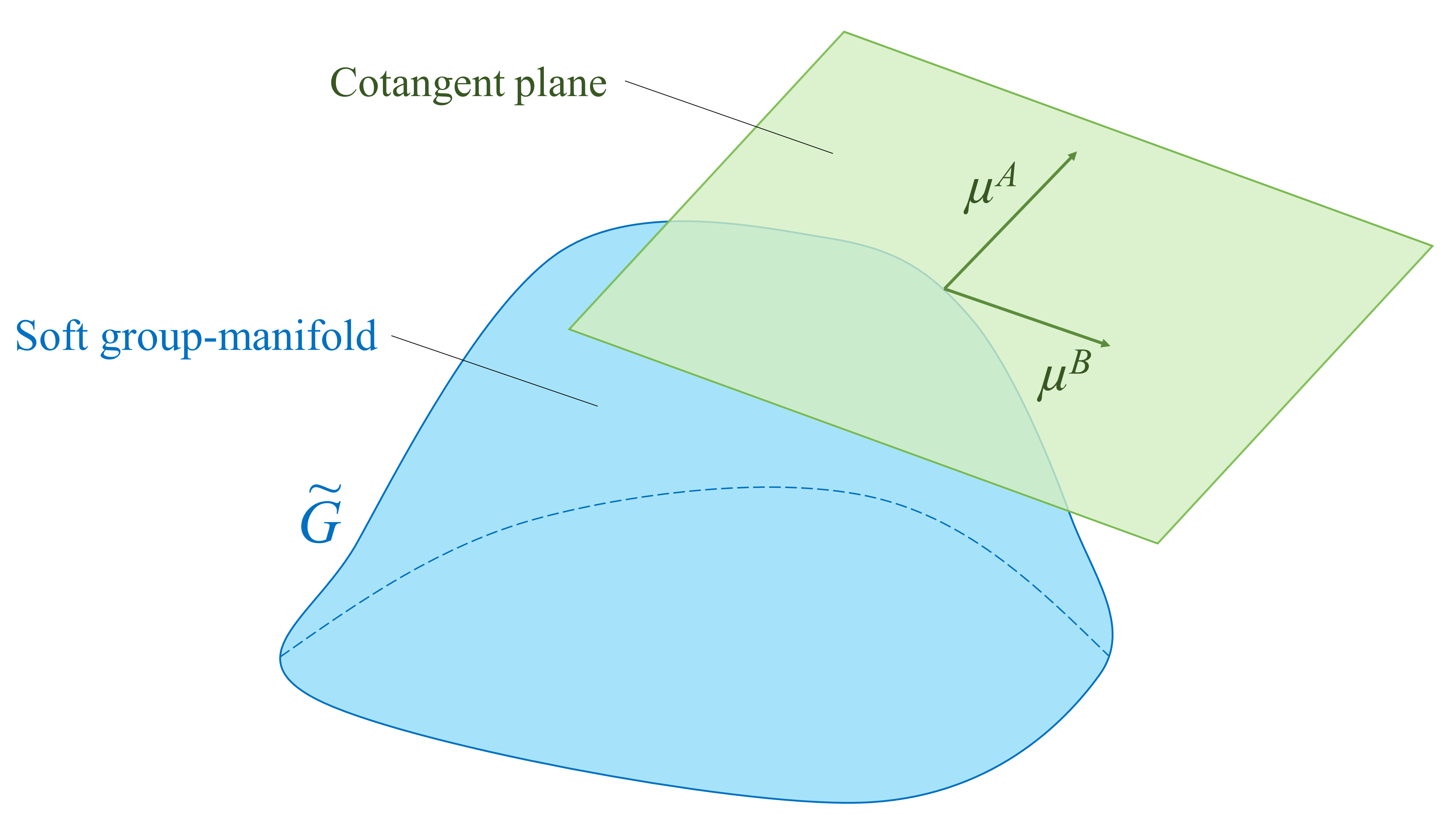}{soft}
\pgfuseimage{soft}
\caption[Soft group-manifold and cotangent space]{\textit{Soft group-manifold and cotangent space.} The soft $1$-forms $\mu^A$'s span a basis of the cotangent plane of $\tilde{G}$ and are, in fact, vielbeins on $\tilde{G}$.} \label{figuresoft}
\end{figure}

Now, taking the exterior derivative of both sides of equation (\ref{curvatureGtilde}), from the request of closure of the algebra ($d^2=0$) we get the so called \textit{Bianchi identity}
\begin{equation}
dR^A + C^A_{\;\;BC} \mu^B \wedge R^C =0,
\end{equation}
which can also be rewritten as
\begin{equation}
\nabla R^A =0 .
\end{equation}

As we have already said, the set of $1$-forms $\lbrace{\mu^A \rbrace}$ forms a basis for the cotangent space to $\tilde{G}$. Thus, the $2$-form $R^A$ can be expanded along the intrinsic basis $\mu^A \wedge \mu^B$: 
\begin{equation}
R^A = R^A_{\;\; BC} \mu^B \wedge \mu^C.
\end{equation}
Then, equation (\ref{curvatureGtilde}) can be rewritten as
\begin{equation}
d\mu^A + \frac{1}{2}\left( C^A_{\;\;BC}-2R^{A}_{\;\;BC}\right)\mu^B \wedge \mu^C =0.
\end{equation}
Therefore, one can derive the commutation relations between the vector fields $\tilde{T}_A$ dual to the $\mu^A$'s 
\begin{equation}
\mu^A (\tilde{T}_{B})=\delta^A_{\;\;B},
\end{equation}
obtaining
\begin{equation}\label{diffalgebrella}
\left[ \tilde{T}_A , \tilde{T}_B \right] = \left( C^C_{\;\;AB} -2 R^C_{\;\;AB} \right) \tilde{T}_C .
\end{equation}
Observe that here the \textit{structure functions} (that are \textit{not} constant) are given in terms of the curvature intrinsic components $R^C_{\;\;AB}$.

For any transformation $\mu^A \rightarrow \mu^A + \delta \mu^A$, the curvature transforms as
\begin{equation}
\delta R^A = \nabla (\delta \mu^A) .
\end{equation}

Let us now consider a \textit{gauge transformation} on $G$. It acts on $\mu^A T_A$ as (in a matrix notation):
\begin{equation}
\mu' = U^{-1}dU + U^{-1}\mu U, \;\;\; U \in G ,
\end{equation}
or, for an infinitesimal gauge transformation generated by $U=\mathbf{1} + \epsilon^A T_A$ (where $\epsilon^A$ is the infinitesimal parameter associated to the gauge transformation), as:
\begin{equation}\label{gaugeonmu}
\delta_\epsilon \mu^A = \left(\nabla \epsilon \right) ^A.
\end{equation}
On the other hand, under a \textit{general coordinate transformation} $x^M \rightarrow x^M + \xi^M(x)$ on the manifold $\tilde{G}$ (with $M=1, \ldots, \dim (\tilde{G})$), that is, if we prefer, under a \textit{generic infinitesimal diffeomorphism} on $\mu^A$ generated by
\begin{equation}
t = \xi^A \tilde{T}_A,
\end{equation}
where $\xi^A = \delta x^A$ is the infinitesimal parameter associated to the shift $x^A \rightarrow x^A + \delta x^A$, the $\mu^A$'s transform with the \textit{Lie derivative} $\ell_\xi \mu^A$, since
\begin{equation}\label{liederiv}
\delta_\xi \mu^A = \delta_\xi (\mu^A_{\;\;N} dx^N)= (\partial_P \mu^A_{\;\;N})\xi^P dx^N + \mu^A_{\;\;N} \partial_P \xi^N dx^P,
\end{equation}
which can then be rewritten as
\begin{equation}\label{veralie}
\delta_\xi \mu^A = \ell _\xi \mu^A \equiv \imath_\xi d\mu^A + d \left( \imath_\xi \mu^A \right) = \left(\nabla \xi \right) ^A + \imath_\xi R^A.
\end{equation}

Observe that the first term $\left(\nabla \xi \right) ^A$ in (\ref{veralie}) corresponds to an \textit{infinitesimal gauge transformation} on $G$.
Hence, we can say that ``\textit{an infinitesimal diffeomorphism on the soft manifold $\tilde{G}$ is a $G$-gauge transformation plus curvature correction terms}'' \cite{Libro1}.

In (\ref{veralie}) we have introduced the contraction $\imath$ of the vector field $\xi = \xi^M \partial_M$ on the curvature $R^A$.
It is defined by $\imath_\xi dx^M = \xi^M$, so that $\imath_\xi \mu^A = \xi^M \mu^A_{\;\;M} \equiv \xi^A$, and it gives
\begin{equation}\label{novella}
\imath_\xi R^A = 2 \xi^B R^A_{\;\; BC}\mu ^c.
\end{equation}
In particular, if the curvature $R^A$ has a vanishing projection along the tangent vector, that is to say if
\begin{equation}
\xi^B R^A_{\;\; BC}=0 ,
\end{equation}
then \textit{the action of the Lie derivative $\ell_\xi$ coincides with a gauge transformation}.

Let us observe that the general coordinate transformation on the $\mu^A$'s can still be written as some sort of ``covariant derivative''
\begin{equation}\label{covderiv}
\delta_\epsilon \mu^A = d\xi^A + (C^A_{\;\;BC}-2 R^A_{\;\; BC})\mu^B \xi^C ,
\end{equation}
in terms, however, of \textit{structure functions} $C^A_{\;\; BC}$ and $R^A_{\;\;BC}$.

The algebra generated by general coordinate transformations on $\tilde{G}$ (diffeomorphisms) is closed:
\begin{equation}
[\delta_{\xi_1},\delta_{\xi_2}] = \delta_{[\xi_1,\xi_2]}
\end{equation}
(which, indeed, is also one of the properties of the Lie derivatives), with the closure condition on the exterior derivative $d^2=0$ provided that the curvatures satisfy the Bianchi identities $\nabla R^A =0$.\footnote{The Lie derivatives close an algebra $[\ell_{\epsilon_1}, \ell_{\epsilon_2}]=\ell_{[\epsilon_1,\epsilon_2]}$ provided that they are consistently defined, namely provided that the operator used in their definitions is a true exterior derivative satisfying the integrability condition $d^2=0$. Then, the same is inherited by diffeomorphisms.}

\subsubsection{Case of the Poincaré group}

Let us carry on our discussion considering the case in which the group $G$ is a semidirect product:
\begin{equation}
G=H \ltimes K ,
\end{equation}
with $H\subset G$ a subgroup. In particular, we consider the case of the \textit{Poincar\'{e} group} $ISO(1,3)$, where $H=SO(1,3)$ and where $K=G/H$ is generated by the translations $P_a$ ($a=0,1,2,3$ in four dimensions). Then, for the Poincar\'{e} algebra we have (schematically):
\begin{align}
& [T_H, T_H]= C^H_{\;\;HH}T_H, \\
& [T_H,T_K]=C^K_{\;\;HK}T_K,
\end{align}
otherwise zero, where $H=1,\ldots, \dim(H)$ and $K=1,\ldots, \dim(K)$.

It is then possible to perform the following decomposition on the soft group-manifold $\tilde{G}$:
\begin{equation}
\mu^A \rightarrow (\mu^{(H)}, \mu^{(K)}),
\end{equation}
such that
\begin{equation}
\mu^{(H)}(T_K)=0, \;\;\;\;\; \mu^{(K)}(T_H)=0.
\end{equation}
We call $\mu^{(H)}=\omega^{ab}=-\omega^{ba}$, and $\mu^{(K)}=V^a$.

We can now ask for \textit{factorization}, that is to say we can ask the theory to be invariant under the Lorentz group $H=SO(1,3)$. As we have already mentioned, factorization means that we are considering $\tilde{G}$ as a \textit{principal bundle} with \textit{base space} $\tilde{G}/H$ and \textit{fiber} $H$. 

The $\mu^A$'s on the principal bundle become the spin connection $\omega^{ab}$ and the vielbein $V^a$, whose curvatures are given by:
\begin{align}
& R^{ab} = d \omega^{ab}-\omega^a_{\;\;c}\wedge \omega^{cb}, \label{c1miau} \\
& R^a = dV^a - \omega^a_{\;\;b}\wedge V^b \equiv DV^a . \label{c2miau}
\end{align}
The associated Bianchi identities are
\begin{align}
& DR^{ab}=0 , \label{b1miau} \\
& DR^a + R^{ab}\wedge V_b =0  . \label{b2miau}
\end{align}

Factorization implies that the general coordinate transformations with parameter $\xi = \xi^{ab}\partial_{ab}$, are, indeed, gauge
transformations, and, by comparison of (\ref{veralie}) with (\ref{covderiv}), this implies the following
(gauge) constraint on the curvature components:
\begin{equation}\label{constrcurv}
R^A_{\;\; (H)\; B}=R^A_{\;\; ab\; B}=0. 
\end{equation}
This means that, as a consequence of the gauge invariance (namely, as a consequence of the constraint (\ref{constrcurv})), the curvature $2$-forms can be expanded on the basis $\lbrace{ V^a\rbrace}$ of the vielbeins, without the inclusion of the spin connection directions.

Let us now consider the \textit{Einstein-Cartan action} (in $D=4$ space-time dimensions) written in the vielbein frame, which reads:
\begin{equation}\label{eccciiuact}
\mathcal{S}= \int_{\mathcal{M}_4} \mathcal{L} = \int_{\mathcal{M}_4} R^{ab}\wedge V^c \wedge V^d \; \epsilon_{abcd} ,
\end{equation}
where $\mathcal{M}_4 = \tilde{G} /H $.

The variation of the action (\ref{eccciiuact}) with respect to the fields $\omega^{ab}$ and $V^a$ gives, respectively, the following $3$-form equations of motion:\footnote{Here we can also use the formula $\delta R^A = \nabla (\delta \mu^A)$ for computing the variations.}
\begin{align}
& \frac{\delta \mathcal{L}}{\delta \omega^{ab}}=0 \;\; \Rightarrow \;\; R^c \wedge V^d =0 \;\; \Rightarrow \;\; \text{Zero torsion}, \label{zerotpuregr} \\
& \frac{\delta \mathcal{L}}{\delta V^{a}}=0 \;\; \Rightarrow \;\; R^{ab} \wedge V^c \; \epsilon_{abcd} =0 \;\; \Rightarrow \;\; R^\mu_{\;\;\nu}-\frac{1}{2}\delta^\mu _{\;\;\nu}R=0,
\end{align}
where the last implications can be proven by expanding along the vielbeins, with a few calculations.
The above equations are the usual Einstein's equations of gravity in the \textit{first order formalism} (in which the vierbein and the spin connection are treated as independent fields in the Lagrangian). 

The Lagrangian 
\begin{equation}\label{lagrangellabellabella}
\mathcal{L}=R^{ab}\wedge V^c \wedge V^d \; \epsilon_{abcd}
\end{equation}
appearing in (\ref{eccciiuact}) can be uniquely determined by using a set of ``\textit{building rules}'' (different from the ones used in the derivation of the Einstein action in the theory of gravitation). The formal nature of this principles, which can be found in Ref. \cite{Libro1}, is useful for finding generalizations of gravity Lagrangians to supergravity ones. We will explore the aforementioned rules in some detail when moving to the geometric approach to supergravity.

The Lagrangian (\ref{lagrangellabellabella}) is exactly the Einstein's Lagrangian for General Relativity. Indeed, expanding $R^{ab}$ on the complete $2$-form basis $V^i \wedge V^j$, we get:
\begin{equation}
\mathcal{L}= R^{ab}_{\;\;\;\;\;\;ij}  V^i_\mu V^j_\nu V^c_\rho V^d_\sigma \; \epsilon_{abcd} dx^\mu \wedge dx^\nu \wedge dx^\rho \wedge dx^\sigma = - 4 R \sqrt{-g}d^4 x,
\end{equation}
where we have used
\begin{equation}
dx^\mu \wedge dx^\nu \wedge dx^\rho \wedge dx^\sigma =  \epsilon ^{\mu \nu \rho \sigma} d^4x , 
\end{equation}
\begin{equation}
V^i_\mu V^j_\nu V^c_\rho V^d_\sigma \epsilon^{\mu \nu \rho \sigma} = \epsilon^{ijcd} \text{det} (V)  ,
\end{equation}
and the definition $R^{ij}_{\;\;\;\;\;\; ij}\equiv R=R^{\mu \nu}_{\;\;\;\;\;\;\mu \nu}$ of the scalar curvature; $\text{det} (V) = \sqrt{-g}$ is the square root of the metric determinant ($g= \text{det} (g_{\mu \nu})$).

Thus, the Einstein's Lagrangian for General Relativity immediately appears to be \textit{geometrical}, since it can be put in the form (\ref{lagrangellabellabella}), that is the most general (and simplest) $4$-form written by only using the differential operator ``$d$'' and the wedge product ``$\wedge$'' between differential forms. Then, being $\mathcal{L}$ a $4$-form, it must be integrated on a $4$-dimensional submanifold of the soft group-manifold $\title{G}$ or, if horizontality has been assumed (which will always be our case), to its restriction to $\mathcal{M}_4 \equiv \tilde{G}/SO(1,3)$. We are thus left with the Einstein-Cartan action as written in (\ref{eccciiuact}).

The Lagrangian of gravity is constructed using the fields of the Poincaré group, but it is invariant only under $SO(1,3)$. This is the reason why the action from which we deduce the gravitational field equations is essentially different from the Yang-Mills action utilized in ordinary gauge theories, that are, instead, invariant under the whole symmetry group.

One can also extend the Einstein-Cartan Lagrangian to the case where the $\mu^A$'s are defined on a de Sitter or anti-de Sitter soft group-manifold. The new Lagrangian corresponds, in tensor calculus formalism, to ordinary gravity plus a cosmological term (see \cite{Libro1} for details). 

We now move to the description of the geometric approach to \textit{supergravity} theories.  

\section{Supergravity in superspace and rheonomy}\label{sugrasuperspacerheo}

The construction of supergravity theories from the technical point of view is a non-trivial task. 

In particular, technical complications arise from the fact that this construction involves fermionic representations. Then, in order to show, for example, that the Lagrangian is supersymmetric, one has often to face with \textit{Fierz identities} (which give the
decomposition of products of spinor representations into irreducible factors). This may involve long and cumbersome calculations.
It is therefore particularly useful to find an efficient method to deal with the technical labor in constructing supergravity theories.
 
In this section, we will describe the so called \textit{rheonomic (geometric) approach to supergravity theories in superspace}.
Before moving to superspace, let us quickly recall the $D=4$, $\mathcal{N}=1$ pure supergravity theory in \textit{space-time}.

\subsection{Review of $\mathcal{N}=1$, $D=4$ supergravity in space-time}

We have previously seen that local supersymmetry requires the introduction of a spin-$3/2$ field $\psi_\mu$ dual to the supersymmetry charge $Q$. Hence, the problem of constructing $\mathcal{N}=1$ supergravity (the ``gauge'' action of the $\mathcal{N}=1$ supersymmetry algebra) turns into the problem of coupling the Rarita-Schwinger field to Einstein gravity. The space-time action of the $\mathcal{N}=1$, $D=4$ supergravity theory, describing the coupling of the spin-$2$ (graviton) and spin-$3/2$ (gravitino) fields, can be written (in the vielbein basis) as:
\begin{equation}\label{muuuuuuuu}
\mathcal{S}= \int_{\mathcal{M}_4} R^{ab}\wedge V^c \wedge V^d \epsilon_{abcd} + \alpha \bar{\psi} \wedge \gamma_5 \gamma_a D\psi \wedge V^a,
\end{equation}
where $D$ is the exterior covariant derivative and where we have defined
\begin{align}
& R^{ab}\equiv d\omega ^{ab}- \omega^{ac}\wedge \omega_c^{\;\;b} , \\
& D\psi \equiv d\psi - \frac{1}{4}\omega^{ab}\wedge \gamma_{ab}\psi .
\end{align}
This action contains one more local invariance besides general coordinate and Lorentz gauge invariance, namely \textit{local supersymmetry invariance}. Indeed, one can prove that (see \cite{Libro2} for details), if the coefficient $\alpha$ in (\ref{muuuuuuuu}) is $\alpha=4$, the action (\ref{muuuuuuuu}) is invariant under the \textit{local supersymmetry transformations}\footnote{We have written the local supersymmetry transformations in the \textit{second order formalism}, in which the vielbein and the spin connection are considered as a single entity.}
\begin{align}
& \delta_\epsilon V^a = \ii \bar{\epsilon}\gamma^a \psi, \label{susyv} \\
& \delta_\epsilon \psi = D \epsilon, \label{susyps} \\
& \delta_\epsilon \omega^{ab}_\mu = -\ii V^{a\vert \rho}V^{b\vert \nu} (\bar{\epsilon} \gamma_\mu D_{[\nu} \psi_{\rho]} + \bar{\epsilon} \gamma_\nu D_{[\mu} \psi_{\rho]} - \bar{\epsilon}\gamma_\rho D_{[\mu}\psi_{\nu]} ), \label{susyom}
\end{align}
being $\epsilon = \epsilon(x)$ the local spinorial parameter of the local supersymmetry transformations. The terms $R^{ab}\wedge V^c \wedge V^d \epsilon_{abcd}$ and $4 \bar{\psi} \wedge \gamma_5 \gamma_a D\psi \wedge V^a$ appearing in the pure $D=4$ supergravity action are called the Einstein-Hilbert and the Rarita-Schwinger terms, respectively.

One can then write the equations of motion of $\mathcal{N}=1$, $D=4$ supergravity. Varying the action with respect to the spin connection $\omega^{ab}$, one obtains:
\begin{equation}\label{varomeghella}
2 R^c \wedge V^d \epsilon_{abcd} =0,
\end{equation}
where we have introduced the \textit{supertorsion} $2$-form
\begin{equation}
R^a \equiv DV^a - \frac{\ii}{2}\bar{\psi}\wedge \gamma^a \psi .
\end{equation}
Equation (\ref{varomeghella}) can be manipulated exactly in the same way as equation (\ref{zerotpuregr}) of pure gravity, the only difference relying in the different definition of $R^a$. Thus, similarly to the case of pure gravity, one obtains:
\begin{equation}
R^a  \equiv DV^a - \frac{\ii}{2}\bar{\psi}\wedge \gamma^a \psi =0 \quad \Rightarrow \quad R^a_{mn}=0.
\end{equation}

Let us observe that $\omega^{ab}$ is a non-Riemannian connection, being $DV^a$ different from zero.
After a suitable decomposition of the spin connection, one can show that the (non-Riemannian) spin connection $\omega^{ab}$ is completely determined in terms of the other two fields $V^a_\mu$ and $\psi^\alpha_\mu$ and, consequently, it does not carry any further physical degree of freedom (see \cite{Libro2} for details).

The condition $R^a = 0$ is called the \textit{on-shell condition for the connection} (it arises where the equations of motion hold, namely on-shell).
When we keep $R^a=0$, we are working in the so called \textit{second order formalism}, where the spin connection is torsionless and given in terms of the vielbein of space-time.
When we do not require $R^a_{mn}=0$, the spin connection is an independent field and we are working in the \textit{first order formalism}; in this case, the field equations of $\omega^{ab}_\mu$ fixes it as a function of both the vielbein and the gravitino, $\omega=\omega (V,\psi)$.
In the sequel of this short review, we will adopt the second order formalism.

Varying the vielbein and the $\psi$-field, after some calculations one ends up with
\begin{align}
& 2 R^{ab} \wedge V^c \epsilon_{abcd} =0 \label{primella}, \\
& 8 \gamma_5 \gamma_a D\psi \wedge V^a =0 \label{newella} ,
\end{align}
respectively.
Notice that equation (\ref{primella}) looks formally the same as in pure gravity; however, the connection, in the present case, is different, and one can show that equation (\ref{primella}) produces the expected interaction between the vielbein field and $\psi$. The same remark applies to (\ref{newella}), in which case one also finds a self-interaction of the gravitino field. 

Thus, the Lagrangian in (\ref{muuuuuuuu}) describes a consistent coupling of the Rarita-Schwinger field $\psi$ to gravity. This suggests the existence of an extra symmetry, extending the gauge invariance of the free field spin Lagrangian (Rarita-Schwinger Lagrangian) to the interacting case. This symmetry results to be, indeed, \textit{supersymmetry}.

Now, the Lagrangian in (\ref{muuuuuuuu}) is invariant under local Lorentz transformations and diffeomorphisms.
The next step is to investigate whether one can define suitable supersymmetry transformations leaving (\ref{muuuuuuuu}) invariant and representing the supersymmetry algebra on the \textit{on-shell} states.

Let us recall that $SO(1,3)$ is an off-shell symmetry of the theory, while supersymmetry is an on-shell one (closing only on the equation of motions).\footnote{This not only holds for the Maurer-Cartan equations, but also when one considers the Free Differential Algebras framework (which will be recalled in the sequel).}

Let us now compare the supersymmetry transformations (\ref{susyv})-(\ref{susyom}) with the gauge transformations of the supersymmetry derived from the super-Poincaré algebra, that are given by (see \cite{Libro2} for details):
\begin{align}
& \delta^{\text{(gauge)}}_\epsilon V^a = \ii \bar{\epsilon}\gamma^a \psi , \label{g1} \\
& \delta^{\text{(gauge)}}_\epsilon \psi = D\epsilon, \label{g2}\\
& \delta^{\text{(gauge)}}_\epsilon \omega^{ab}=0. \label{g3}
\end{align}
Comparing (\ref{susyv})-(\ref{susyom}) with (\ref{g1})-(\ref{g3}), one can see that (\ref{susyv}) and (\ref{g1}) coincides. The same holds for (\ref{susyps}) and (\ref{g2}). On the other hand, the gauge and local supersymmetry transformations for $\omega^{ab}$ are different, since (\ref{susyom}) is different from zero. Furthermore, let us mention that equation (\ref{g2}) resembles the gauge transformation of a gauge field, but this is only due to the fact that we are now considering the simplest, minimal $\mathcal{N}=1$, $D=4$, pure supergravity theory (without matter). In more complicated cases, other terms would appear in (\ref{g2}), making the difference between (\ref{g2}) and the transformation of a true gauge field manifest.

Thus, \textit{the local supersymmetry transformations leaving the supergravity action invariant are not gauge supersymmetry transformations}.\footnote{This is strictly analogous to what happens in the case of pure gravity, where one finds that the action is not invariant under gauge translations, while it is invariant against diffeomorphisms (under which the connection transforms with the Lie derivative), the two transformation laws differing on the spin connection.}

Moreover, one can prove that the local supersymmetry transformations close \textit{on-shell} with \textit{structure functions}, rather than with structure constants, as it would be the case for a genuine gauge transformation, and that a gauge translation leaves the field $\psi$ inert.

The supersymmetry algebra (which closes on-shell) can be interpreted, on space-time, in terms of the general algebra of space-time diffeomorphisms supplemented by super-Poincaré gauge transformations with field-dependent parameter \cite{Libro2}.

As we will discuss in a while, we can say that \textit{the on-shell supersymmetry algebra is the algebra of diffeomorphisms in superspace}.

\subsection{The concept of superspace}

Let us summarize what we have learned till now: If our aim is \textit{local}, rather than global, supersymmetry invariance, this requires the introduction of the spin-$3/2$ field $\psi_\mu$ dual to the supersymmetry charge $Q$. In this set up, the $\mathcal{N}=1$, $D=4$ supergravity action describing the coupling of the Rarita-Schwinger field to Einstein's gravity is given by (\ref{muuuuuuuu}) (with $\alpha=4$). Studying the local supersymmetry transformations of the theory, one obtains that these transformations are \textit{not} gauge transformations of the super-Poincaré algebra (except in the case of the linearized theory).

Now, a key point in the formulation of supergravity theories is a more satisfactory understanding of the (local) supersymmetry transformations rule.

To this aim, a formulation of supergravity which appears natural and particularly useful is based on the concept of \textit{superspace} (I have adopted this formulation in the research carried on during the PhD).
Superspace has as coordinates not only the ordinary ones, but, in addition, $4 \mathcal{N}$ spinorial anticommuting coordinates $\theta^{\alpha A}$ ($\alpha=1,\ldots,4$ and $A=1,\ldots,\mathcal{N}$; if $\mathcal{N}=1$, we do not write the index $A$).

There are various approaches to superspace, based on different geometrical ideas, but they all have in common the fact that the notion of Grassmann variables (anticommuting $c$-numbers) as coordinates is essential.
In \textit{rigid} supersymmetry, we have
\begin{align}
& V^a = dx^a - \frac{\ii}{2} \bar{\theta}_A \gamma^a d \theta^A , \\
& \psi = d\theta^A ;
\end{align}
in the case of \textit{supergravity}, these same degrees of freedom are \textit{dynamical}.

The approaches on ordinary space-time are equivalent to the approaches in superspace, but the superspace framework gives a better geometrical insight (see, for example, Refs. \cite{VanNieuwenhuizen:1981ae, Libro1, Libro2} for details on the geometry of superspace). In particular, on superspace we may have an understanding of supergravity analogous to that of General Relativity
on space-time. Indeed, at each point $(x^\mu, \theta^{\alpha A})$ on superspace, we can erect a local tangent frame and consider general coordinate transformations on the base manifold (the superspace), with parameter $\xi^\Lambda$, where $\Lambda = (\mu , \alpha A)$. Then, the $\xi^\mu$'s generate ordinary general coordinate transformations on space-time, while the $\xi^{\alpha A}$'s generate local supersymmetry transformations.

One can extend in an appropriate way the space-time fields $V^a_\mu$, $\psi_\mu$, and $\omega^{ab}_\mu$ to $1$-form fields defined over \textit{superspace}. These $1$-form fields are called \textit{superfields}. In this way, one can reinterpret the supersymmetry transformations as \textit{superspace Lie derivatives}.

All the approaches to supergravity in superspace involve a large symmetry group and a large number of fields, so that one eventually has to impose constraints in order to recover ordinary supergravity on space-time. On the other hand, one can exploit the power of symmetry to construct general theories in a systematic and straightforward way. 

In this scenario, the so called ``\textit{rheonomy principle}'' (see Ref. \cite{Libro2}) makes the extension from space-time to superspace uniquely defined, allowing for a \textit{geometric} interpretation of the supersymmetry rules. The rheonomy principle can be summarized in one sentence as follows: ``We demand the $\theta$-dependence of every superfield to be determined by the $x$-dependence of all the superfields in our stock'' \cite{Libro2}.

Exploiting the principle of rheonomy to get rid of the unwanted degrees of freedom, we can identify \textit{supergravity} with the \textit{geometric theory of superspace} in the same way as \textit{Einstein's gravity} is the \textit{geometric theory of space-time}. Indeed, the pair of $1$-forms $\lbrace V^a ,\psi \rbrace$, once extended to superspace, can be viewed as a single object, called the \textit{supervielbein}: A local contangent frame on the superspace $\mathcal{M}_{4 \vert 4}$. More generally, the $1$-forms $\mu^A \equiv (\omega^{ab}, V^a ,\psi)$ constitute an intrinsic reference frame in the cotangent plane to the soft super-Poincaré group. 

In the following, we review the rheonomy principle, on the same lines of Ref. \cite{Libro2}.

\subsection{Superspace geometry and the rheonomy principle}

As suggested by the name, the key point of the geometric approach to supergravity in superspace is ``geometricity'': The idea is to formulate an extension of General Relativity which is generally covariant over superspace, so that \textit{the diffeomorphisms in the fermionic directions of superspace correspond to supersymmetry transformations on space-time}.

In order to do this, one has to introduce a set of differential forms on space-time, $\mu^A (x)$,
and lift them to forms on superspace, namely to $\mu^A(x, \theta)$. The introduction of the \textit{superfields} $\mu^A(x, \theta)$ leads to extra degrees of freedom (corresponding to the components in the $\theta$-expansion of the superfields) which are spurious. Then, in order to have the same physical content for the theory extended to superspace as for the theory on space-time, one has to impose some \textit{constraints} on the \textit{supercurvatures} (that is on the field-strengths). As we will see, these constraints turn out to be physically equivalent to the \textit{on-shell constraints}, that is to say, to the \textit{equations of motion}.

This is the way in which the on-shell closure of the supersymmetry algebra is implemented within this approach, and we will see that it allows to find field equations,
supersymmetry transformations and, eventually, also the Lagrangian for completely general supergravity theories, by the request of closure of the superalgebra.

Now we have all the ingredients for generalizing the discussion of Section \ref{gma} to the case of a geometrical formulation of supergravity.\footnote{An extended study of different supergravity theories in this geometrical formulation can be found in Ref. \cite{Libro2}.}
The presentation we give strictly follows the lines of Ref. \cite{Libro2}.

In order to introduce the technical aspects of the rheonomic framework, let us start with a \textit{supergroup-manifold}, instead of a group-manifold, whose corresponding superalgebra is given by
\begin{equation}
[T_A, T_B \rbrace = C^C_{\;\;AB}T_C.
\end{equation}

We start from \textit{rigid} superspace and, in particular, we specialize to the case in which $G = \overline{OSp(1 \vert 4)}$, that is to say, to the $\mathcal{N} = 1$ super-Poincar\'{e} group (in $D=4$), whose superalgebra is given in equations (\ref{bossuperalge})-(\ref{ferm}). The aforementioned superalgebra is naturally factorized in $G=H+K$. In particular, $H = SO(1, 3)$ is the subalgebra spanned by the Lorentz generators $J_{\mu \nu}$,\footnote{For $\mathcal{N}$-extended supergravity theories, the $H$ subalgebra is $SO(1, 3) \times H'$, for some $H'$.} and $K = G/H = I + O$ is split into a bosonic (inner) subspace $I$, spanned by the translations $P_\mu$, and a fermionic (outer) subspace $O$, spanned by the supercharge $Q_\alpha$. Then, the superalgebra can be schematically written as follows:
\begin{equation}\label{superalgHIO}
[H,H]\subset H, \;\;\; [H,I]\subset I, \;\;\; [H,O]\subset O, \;\;\; \lbrace O,O \rbrace \subset I, \;\;\; [I,I]=[I,O]=0.
\end{equation}
The structure constants in equations (\ref{bossuperalge})-(\ref{ferm}) obey \textit{graded Jacobi identities}:
\begin{equation}
[T_A , [T_B , T_C \rbrace \rbrace + (-1)^{A(B+C)}[T_B,[T_C, T_A \rbrace \rbrace + (-1)^{B(C+A)}[T_C, [T_A, T_B \rbrace \rbrace =0.
\end{equation}

All the construction described in Section \ref{gma} can now be repeated (see \cite{Libro1, Libro2} for more details):
\begin{itemize}
\item We consider a basis of bosonic and fermionic $1$-forms $\mu^A$ on the deformed, soft supergroup-manifold $\tilde{G}$, with curvatures defined as in (\ref{curvatureGtilde}), and define covariant and contravariant derivatives as in (\ref{covderGtilde}) and (\ref{contravderGtilde}).
We call the $\mu^A$'s as follows:
\begin{equation}
\mu^A \equiv (\omega^{ab},V^a, \psi^\alpha),
\end{equation}
with corresponding \textit{supercurvatures}
\begin{equation}
R^A = (R^{ab},R^a,\rho^\alpha).
\end{equation}
In particular, the gravitino $\psi^\alpha$ is a spinor $1$-form and, correspondingly, its supercurvature $\rho^\alpha$ is a spinorial $2$-form (in the following, for simplifying the notation, we will neglect the spinor index $\alpha$).
The supercurvatures obey the \textit{Bianchi identities}
\begin{equation}
\nabla R^A = dR^A +C^A_{\;\; BC}\mu^B \wedge R^C =0.
\end{equation}
\item Then, since the Lorentz group is a gauge symmetry of the theory, we require \textit{factorization}, which allows to work on the \textit{fiber bundle} with the \textit{Lorentz group} as \textit{fiber} and the \textit{superspace} as \textit{base space}, namely we require the constraint
\begin{equation}
R^A_{\;\;ab \; B}=0
\end{equation} 
on the components of the curvatures as $2$-forms on the supergroup-manifold. 
From now on, fields and curvatures will be functions of the superspace coordinates $(x^\mu, \theta^\alpha)$. In particular, we find:
\begin{align}
& R^{ab}\equiv d\omega^{ab}-\omega^a_{\; c}\wedge \omega^{cb} , \\
& R^a \equiv dV^a - \omega^{ab}\wedge V_b - \frac{\ii}{2}\bar{\psi}\wedge \gamma^a \psi = DV^a - \frac{\ii}{2}\bar{\psi}\wedge \gamma^a \psi, \\
& \rho \equiv d\psi - \frac{1}{4}\omega^{ab}\wedge \gamma_{ab}\psi = D \psi,
\end{align}
where we have introduced the \textit{Lorentz covariant derivative} $D$.
These supercurvatures subject to the following \textit{Bianchi identities}:
\begin{align}
& DR^{ab}=0, \label{bau1}\\
& DR^a + R^{ab}\wedge V_b - \ii \bar{\psi}\wedge \gamma^a \rho =0, \label{bau2} \\
& D \rho + \frac{1}{4}R^{ab}\wedge \gamma_{ab}\psi =0. \label{bau3}
\end{align}
\item Finally, the \textit{peculiar feature of the rheonomic approach} is the following one: The $1$-forms $\mu^A$ described so far are defined on superspace; however, in order to \textit{reproduce the physical content of supergravity on space-time}, we have to \textit{relate the field-strengths along the fermionic vielbein to the curvatures along the bosonic vielbein $V^a$}, getting rid of the extra degrees of freedom arising in the extension to superspace. Therefore, we have to introduce some sort of ``factorization'', as we have done for the $\tilde{G}$ coordinates in the Lorentz directions. The constraints we will introduce relate the components of the supercurvatures along the basis $\psi^\alpha \wedge V^a$ or $\psi^\alpha \wedge \psi^\beta$ to their components along the $V^a \wedge V^b$ basis in algebraic way, actually linearly.

Will see in a while that the so called \textit{rheonomy principle} is equivalent to \textit{supersymmetry on space-time}. As we have already mentioned, the \textit{local supersymmetry transformations} are not gauge transformations of the super-Poincar\'{e} algebra,
but have to be thought, instead, as \textit{diffeomorphisms in the fermionic directions of superspace}.

Let us explicitly see what we mean, by first considering the $G$-gauge transformations of $\mu^A \equiv (\omega^{ab},V^a,\psi)$, where $G=\overline{OSp(1 \vert 4)}$ (rigid superspace).

Recall that a $G$-gauge transformation of the fields $\mu^A$ is given by the $G$-covariant derivative of $\epsilon^A$, where $\epsilon^A \equiv (\epsilon^{ab},\epsilon^a,\epsilon^\alpha)$ is a parameter in the adjoint representation of $G$:
\begin{equation}\label{ggaugenewwwwwww}
\delta^{\text{(gauge)}}_\epsilon \mu^A = \left( \nabla \epsilon \right)^A.
\end{equation}
The Lorentz content of the $\nabla$ derivative, when acting on the adjoint multiplet, can be directly read off from the explicit form of the Bianchi identities (\ref{bau1})-(\ref{bau3}). We obtain:
\begin{align}
& \delta^{\text{(gauge)}}_\epsilon \omega^{ab} = (\nabla \epsilon)^{ab}\equiv D \epsilon^{ab}, \label{gggggom} \\
& \delta^{\text{(gauge)}}_\epsilon V^a = (\nabla \epsilon)^a \equiv D\epsilon ^a + \epsilon^{ab}V_b - \ii \bar{\psi}\gamma^a \epsilon , \label{gggggviel} \\
& \delta^{\text{(gauge)}}_\epsilon \psi = \nabla \epsilon \equiv D \epsilon + \frac{1}{4}\epsilon^{ab}\gamma_{ab}\psi , \label{gggggpsi}
\end{align}
where $\nabla$ and $D$ represents the $\overline{OSp(1\vert 4)}$ and $SO(1,3)$ covariant derivatives, respectively.

In particular, if $\epsilon^A \equiv (0,0,\epsilon^\alpha)$ we get the explicit form of a \textit{gauge supersymmetry transformation}:
\begin{align}
& \delta_\epsilon \omega^{ab}=0 , \label{susygom} \\
& \delta_\epsilon V^a = - \ii \bar{\psi}\gamma^a \epsilon , \label{susygviel} \\
& \delta_\epsilon \psi = D \epsilon . \label{susygpsi}
\end{align}
Setting instead $\epsilon^A = (\epsilon^{ab},0,0)$ yields the form of a \textit{Lorentz gauge transformation}:
\begin{align}
& \delta \omega^{ab} = D \epsilon^{ab}, \label{lorom} \\
& \delta V^a = \epsilon^{ab}V_b , \label{lorviel} \\
& \delta \psi = \frac{1}{4}\epsilon^{ab} \gamma_{ab}\psi , \label{lorpsi}
\end{align}
while setting $\epsilon^A=(0,\epsilon^a,0)$ yields the form of a \textit{translation gauge transformation}:
\begin{align}
& \delta \omega^{ab} = 0, \label{trom} \\
& \delta V^a = D \epsilon^a , \label{trviel} \\
& \delta \psi =0 . \label{trpsi}
\end{align}
Observe that local supersymmetry transformations are \textit{not} gauge transformations of the super-Poincaré algebra.

Let us also write the transformation law of $\mu^A$ under (infinitesimal) \textit{diffeomorphisms}. Indeed, as we have already mentioned, this will be very important in the sequel for the interpretation of supersymmetry.

Let 
\begin{equation}
\epsilon = \frac{1}{2} \epsilon^{ab} \tilde{D}_{ab} + \epsilon^a \tilde{D}_a + \bar{\epsilon}\tilde{D} \equiv \epsilon^A \tilde{D}_A
\end{equation}
be a general tangent vector on $\tilde{G}$, with $\tilde{D}_A$ dual to $\mu^B$:
\begin{equation}
\mu^B (\tilde{D}_A)= \delta^B_{\;\;A}.
\end{equation}
Here and in the following, according with the same notation of \cite{Libro2}, we denote by ($\tilde{D}_A$) $D_A$ the tangent vector on the (soft) group-manifold dual to the (non) left-invariant $1$-forms ($\mu^A$) $\sigma^A$. We reserve the symbol $T_A \equiv (J_{ab},P_a,Q)$ to the abstract Lie algebra generators (when thought as vector fields, they are left-invariant and $D_A\equiv T_A$). 

As we have already seen in Section \ref{gma}, an \textit{infinitesimal diffeomorphism} on $\mu^A$ is given by the \textit{Lie derivative}:
\begin{equation}
\delta^{\text{(diff.)}} _\epsilon \mu^A \equiv \ell _\epsilon \mu^A = ( \imath_\epsilon d + d \imath_\epsilon ) \mu^A .
\end{equation}
Alternatively (see Section \ref{gma}), we may rewrite
\begin{equation}\label{diffnewwwwwwwwww}
\delta^{\text{(diff.)}} _\epsilon \mu^A \equiv \ell_\epsilon \mu^A = (\nabla \epsilon)^A + \imath_\epsilon R^A = (\nabla \epsilon)^A + 2 \epsilon^B R^A_{\;\; BC}\mu^C 
\end{equation}
and, making the Lorentz content explicit, we find:
\begin{align}
& \ell_\epsilon \omega^{ab} = (\nabla \epsilon)^{ab} + \imath_\epsilon R^{ab} , \label{difflieom} \\
& \ell_\epsilon V^a = (\nabla \epsilon)^{a} + \imath_\epsilon R^a , \label{difflieviel} \\
& \ell_\epsilon \psi = \nabla \epsilon + \imath_\epsilon \rho . \label{diffliepsi}
\end{align}
Thus, if we know the on-shell parametrization of the supercurvatures $R^{ab}$, $R^a$, and $\rho$, then we also know the supersymmetry transformations leaving the theory invariant.

Now, if $\epsilon = \epsilon^a \tilde{D}_a + \bar{\epsilon} \tilde{D}$, equations (\ref{difflieom})-(\ref{diffliepsi}) describe a diffeomorphism in superspace $\mathcal{M}_{4 \vert 4}$, which cannot be interpreted as a pure gauge transformation of the super-Poincaré algebra, unless we also impose the further horizontality constraints $\imath _{\tilde{D}_a}R^A = \imath_{\tilde{D}_\alpha} R^A =0$. 

However, if these conditions were to be imposed, the fields $\mu^A$ would have a trivial (factorized) dependence on the superspace coordinates $(x^\mu, \theta^\alpha)$, and the soft (super)-coset $\tilde{G}/SO(1,3)=\mathcal{M}_{4 \vert 4}$ would reduce to the rigid superspace $G/SO(1,3)\equiv \mathbb{R}^{4 \vert 4}$.

Therefore, in the construction of a physical theory, we need non-vanishing curvature-terms in (\ref{difflieom})-(\ref{diffliepsi}). In this way, the fields $\mu^A$ can exhibit a non-trivial (that is, dynamical) dependence on their argument.

Thus, we have seen that we cannot impose factorization and ``gauge away'' the fermionic degrees of freedom, since the fundamental forms $\mu^A$ have a non-trivial, physical dependence on all the coordinates of superspace, which explicitly reads:\footnote{We adopt the same convention of \cite{Libro2}, in which the indexes labeling the spinorial coordinates are denoted by a lower case Greek letter with a bar, while the unbarred Greek indexes will be reserved to describe intrinsic fermionic indexes.}

\begin{align}
& V^a(x,\theta)= V^a_\mu (x,\theta)dx^\mu + V^a _{\bar{\alpha}} (x,\theta)d\theta^{\bar{\alpha}} , \\
& \psi (x,\theta) = \psi_\mu (x,\theta)dx^\mu + \psi_{\bar{\alpha}} (x,\theta)d\theta^{\bar{\alpha}} , \\
& \omega^{ab}(x,\theta)= \omega^{ab}_\mu (x, \theta)dx^\mu + \omega^{ab}_{\bar{\alpha}} (x,\theta)d\theta^{\bar{\alpha}} .
\end{align} 

On the other hand, in order to have a consistent theory on space-time, with fields having the same number of physical degrees of freedom of the space-time fields
\begin{align}
& V^a (x)= V^a_\mu (x)dx^\mu , \\
& \psi (x)= \psi_\mu (x)dx^\mu, \\
& \omega^{ab}(x)= \omega^{ab}_\mu (x)dx^\mu ,
\end{align} 
all the space-time fields in the $\theta$-expansion of the superfield $\mu^A(x,\theta)$, and all its $d\theta$-components, have to be expressed in terms of the space-time restriction $\mu^A(x)=\mu^A_\mu (x,0)dx^\mu$.

Recalling (\ref{diffnewwwwwwwwww}), for an \textit{infinitesimal diffeomorphism in a fermionic direction} with parameter $\epsilon^A=(0,0,\epsilon^{\bar{\alpha}})$ we have:
\begin{equation}\label{fermdiffnewwwwww}
\delta_\epsilon \mu^A (x,\theta) = ( \nabla \epsilon)^A + 2 \bar{\epsilon}^{\bar{\alpha}} R^A_{\;\; \bar{\alpha} L}dZ^L,
\end{equation}
with $dZ^L \equiv (d\theta^{\bar{\alpha}}, dx^\mu)$. 
 
The constraint that we have to impose, from the request that in the projection to space-time we do not loose physical degrees of freedom, is the following one:
\begin{equation}\label{rheonomicconstraint}
R^A_{\;\;\bar{\alpha} L}= C^{A\vert \mu \nu}_{\;\;\bar{\alpha} L\vert B}R^B_{\;\;\mu \nu}, \;\;\;\;\; C^{A\vert \mu \nu}_{\;\;\bar{\alpha} L\vert B} \; \; \text{constant tensors},
\end{equation} 
where, according to our convention, $\mu$ and $\nu$ are indexes labeling the space-time (bosonic) coordinates, $\bar{\alpha}$ is a spinorial index associated to the $\theta^{\bar{\alpha}}$ coordinates, $L\equiv (\bar{\alpha},\mu)$, and $A$ and $B$ are super-Lie algebra indexes. 

The constraints in (\ref{rheonomicconstraint}) which relate the inner $R^A_{\mu \nu}$ and the outer $R^A_{\bar{\alpha}L}$ components of the curvatures $R^A$ are named \textit{rheonomic constraints} (the property expressed by (\ref{rheonomicconstraint}) is referred to as ``\textit{rheonomy}'' and a theory admitting a set of rheonomic constraints is likewise named a \textit{rheonomic} theory).
As we have previously anticipated, through the rheonomic constraints the components along at least one fermionic vielbein $\psi$ are linearly expressed in terms of the components along the bosonic vielbeins. Indeed, these constraints state that the fermionic components of the curvatures (in their decomposition on a basis of $1$-forms $\mu^A$) can be expressed algebraically in terms of the space-time components $R^A_{\;\;\mu \nu}= \partial_{[\mu}\mu^A_{\nu]}+\frac{1}{2}C^A_{\;\;BC}\mu^B_{[\mu}\mu^C_{\nu]}$. In other words, this means that the derivatives of the fields in the $\theta$-directions are expressed, through (\ref{rheonomicconstraint}), as linear combinations of their derivatives in the space-time directions. 

Thus, one can see that, when the constraint (\ref{rheonomicconstraint}) hold, the knowledge of a purely space-time configuration $\lbrace \mu^A_\mu (x,0);\; \partial_{\mu} \mu^A_\mu (x,0) \rbrace$ determines in a complete way the so called \textit{rheonomic extension mapping}: 
\begin{equation}
\text{Rheonomic mapping:} \left\{
\begin{aligned}
& V^a (x) \rightarrow V^a (x,\theta), \\
& \psi (x) \rightarrow \psi (x,\theta), \\
& \omega^{ab}(x) \rightarrow \omega^{ab}(x,\theta).
\end{aligned} 
\right.
\end{equation}
Indeed, inserting (\ref{rheonomicconstraint}) into (\ref{fermdiffnewwwwww}), we find
\begin{equation}\label{rhmappinggggg}
\delta \mu^A (x,\theta) = (\nabla \epsilon)^A + 2 \epsilon^{\bar{\alpha}} C^{A\vert \mu \nu}_{\;\;\bar{\alpha} L\vert B} R^B _{\mu \nu} (x,0)dZ^L .
\end{equation}
Thus, if rheonomy holds, (\ref{fermdiffnewwwwww}) is equivalent to the \textit{passive} point of view for the Lie derivative (a \textit{flow}, through fermionic diffeomporphisms, from an hypersurface to another which is translated by $\delta \theta$).\footnote{In fact, the term ``\textit{rheonomy}'' takes its origin from the Greek words ``$\rho \epsilon \hat{\imath} \nu$'' $\rightarrow$ ``flow'' and ``$\nu $\'{o}$ \mu o \varsigma$'' $\rightarrow$ ``law'', referring to the flow law for moving from one hypersurface to another (through fermionic diffeomorphisms).} 

Therefore, the complete $\theta$-dependence of the superfield $\mu^A (x,\theta)$ can be recovered starting from the initial purely space-time ($\theta=0$) configuration. In other words, $\mu^A_\mu (x,0)$ and the space-time tangent derivatives $\partial_\mu \mu^A_\nu (x,0)$ (or, equivalently, $\mu^A_\mu(x,0)$ and $R^A_{\mu \nu}(x,0)$) constitute a complete set of Cauchy data on $\mathcal{M}_4$ once (\ref{rheonomicconstraint}) is satisfied. Indeed, one can show that, when the constraints (\ref{rheonomicconstraint}) hold, the space-time normal derivatives $\frac{\partial}{\partial \theta} \mu^A(x,0)$ are expressible in terms of $\mu^A(x,0)$ and $\partial_{[\mu} \mu^A_{\nu]}(x,0)$. The rheonomic constraints (\ref{rheonomicconstraint}) are constraints between inner $\left( \frac{\partial}{\partial x^\mu}\right)$ and outer $\left( \frac{\partial}{\partial \theta^\alpha}\right)$, and this is analogous to the Cauchy-Riemann equations for an analytic function:

\begin{align}
& f(x,y)=u(x,y) + \ii v(x,y), \\
& \frac{\partial}{\partial x}u(x,y)= \frac{\partial}{\partial y}v(x,y), \\
& \frac{\partial}{\partial y}u(x,y)= - \frac{\partial}{\partial x}v(x,y).
\end{align}

According to this analogy, we have

\begin{align}
& x \rightarrow x^\mu, \\
& y \rightarrow \theta^\alpha, \\
& f(x,y) \rightarrow \mu^A (x^\mu, \theta^\alpha).
\end{align}
Moreover, just as the analycity of a function allows for its determination in the whole complex plane once its boundary value on any line (say $y=0$) is given, in the same way rheonomy allows to reconstruct the superfield potential $\mu^A (x,\theta)$ from its boundary value (say $\theta=0$) (see Refs. \cite{Libro2, Fre:2013ika} for details on this analogy).

The idea of rheonomy, together with a ``visualization'' of superspace, is graphically summarized in Figure \ref{figurerheonsuper}, previously proposed in Ref. \cite{Fre:2013ika}.

%As recalled in \cite{Fre:2013ika}, the analogy with Cauchy-Riemann equations and
%analyticity immediately suggests one important consequence of rheonomy: The functions $u$ and $v$, as a consequence of the integrability of Cauchy-Riemann equations, are harmonic functions satisfying Laplace equations $\Delta u = \Delta v = 0$; in the same way, we expect
%that the bosonic and fermionic connections (whose curvatures are rheonomic) should obey, as a consequence of integrability of the rheonomy conditions, some differential equations of the second order in the horizontal variables. The appropriate integrability conditions in this context are the Bianchi identities. By writing the most general rheonomic parametrization of the curvatures with arbitrary coefficients and inserting it into the Bianchi identities, we find that all such coefficients are uniquely determined. Furthermore, some algebraic constraints have to be satisfied by the horizontal curvature components (as we will see in the sequel). These constraints, implied by the rheonomic constraints (\ref{rheonomicconstraint}), are differential equations in the space-time coordinates imposed on the connection components and, in our analogy, they correspond
%to the Laplace equations satisfied by $u$ and $v$. The physical interpretation of these constraints is that they are nothing else but the appropriate field equations of supergravity theory.

\begin{figure}
\centering
\pgfdeclareimage[height=8cm]{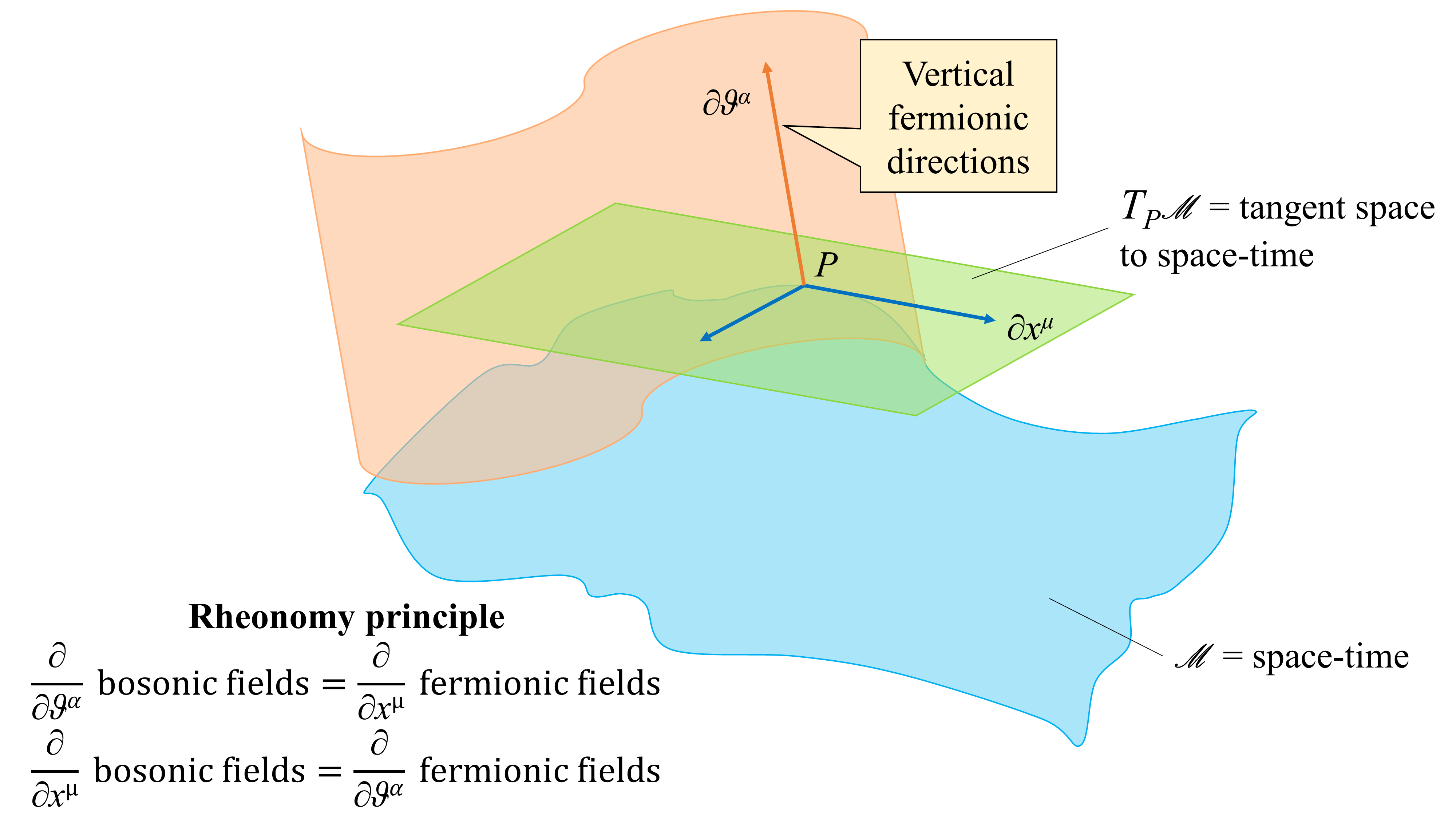}{rheonsuper}
\pgfuseimage{rheonsuper}
\caption[Rheonomy of superspace]{\textit{Rheonomy of superspace.} The principle of rheonomy is reminiscent of the Cauchy-Riemann equations satisfied by the real and the imaginary parts of analytic functions, encoding a sort of analyticity condition for the superconnections that constitute the field content of supergravity theories (see Refs. \cite{Libro2, Fre:2013ika}).} \label{figurerheonsuper}
\end{figure}

What we have shown till now is that \textit{the physical content of a rheonomic theory in superspace is completely determined by means of a purely space-time description}, through the ``flow laws'' that connect the two spaces.

Alternatively, if we regard the Lie derivative as the generator of the functional change of $\mu^A$ at the same coordinate point:
\begin{equation}
\ell_\epsilon \mu^A = \mu^{A'} (x,0)-\mu^A (x,0),
\end{equation}
that is if we consider an \textit{active} point of view for the Lie derivative, sticking to the four-dimensional space-time ($\theta= d\theta=0$), then the rheonomic mapping (\ref{rhmappinggggg}) can be rewritten as follows:
\begin{equation}\label{rhmappello}
\delta \mu^A (x,0)=(\nabla \epsilon)^A + 2 \bar{\epsilon}C^{A \vert \mu \nu}_{\bar{\alpha}L \vert B} R^B_{\mu \nu}(x,0)dZ^L .
\end{equation}
We can thus say that, written in this form, \textit{the rheonomic mapping maps a space-time configuration into a new space-time configuration}. In particular, if the theory described by the fields $\mu^A$ is \textit{invariant under superspace diffeomorphism}, then it can be \textit{restricted to space-time}, and (\ref{rhmappello}) will appear as a \textit{symmetry transformation of the space-time theory}.

Now, since $\bar{\epsilon}^{\bar{\alpha}}$ is a spinorial parameter, \textit{the rheonomic mapping realized on space-time field configurations} will be identified as a \textit{supersymmetry transformation}.

Note that rheonomy does \textit{not} depend on the particular basis chosen for the $1$-forms: The coordinate basis used above, $\lbrace d \theta^{\bar{\alpha}}, d x^\mu \rbrace$, and the anholonomic supervielbein basis, $\lbrace V^a , \psi ^\alpha \rbrace$, are equally viable. In the following, we shall need the expression of the rheonomic constraints using \textit{intrinsic} components of the curvatures.

We can now rewrite the \textit{supersymmetry transformations} as follows:
\begin{equation}\label{susytrasfrheonomy}
\begin{aligned}
\ell_\epsilon \mu^A (x,0) & = (\nabla \epsilon)^A + \imath_\epsilon R^A(x,0) = \\
& = (\nabla \epsilon)^A + 2 \bar{\epsilon}^\alpha R^A_{\alpha C} (x,0) \mu^C = \\
& = (\nabla \epsilon)^A + 2 \bar{\epsilon}^\alpha C^{A\vert mn}_{\alpha C \vert B}R^B_{mn}(x,0)\mu^C , 
\end{aligned}
\end{equation}
where we have used $\epsilon = \bar{\epsilon}\tilde{D}$, $\mu^A (\tilde{D}_\alpha) = \delta^A_{\;\;\alpha}$, and $R^A \equiv R^A_{\;\;BC} \mu^B \wedge \mu^C$.
It is in this \textit{intrinsic form} that the supersymmetry transformations appear in supergravity theories.
\end{itemize}

Summarizing, we have seen that a geometric formulation of supegravity as a theory on the super-Poincar\'{e} group can be done when one imposes factorization in the Lorentz directions and rheonomy (precisely, the rheonomic constraints (\ref{rheonomicconstraint})) on the odd directions. 

Let us now look better inside (\ref{rheonomicconstraint}): In the general discussion of
the group-manifold approach, we have seen that general coordinate transformations on $\tilde{G}$ (diffeomorphisms) close an algebra
\begin{equation}
[\delta_{\epsilon_1}, \delta_{\epsilon_2} ] = \delta_{[\epsilon_1,\epsilon_2]} 
\end{equation}
with the closure condition on the exterior derivative $d^2 = 0$ if the curvatures satisfy the Bianchi identities $\nabla R^A = 0$.

However, the rheonomic constraints (\ref{rheonomicconstraint}) among the holonomic outer and inner components $R^A_{\bar{\alpha}L}$ and $R^A_{\mu \nu}$ imply an analogous relation among the intrinsic components $R^A_{\alpha C}$ and $R^A_{mn}$:
\begin{equation}\label{fuckkkkk}
R^A_{\alpha C} = C'^{A\vert mn }_{\alpha C \vert B}R^B_{mn},
\end{equation}
where $C'$ are constant anholonomic tensors (precisely, the $C$'s appearing in (\ref{rheonomicconstraint}), evaluated in their intrinsic basis), and, in the presence of (\ref{fuckkkkk}), the \textit{Bianchi identities loose their character of identities}, becoming \textit{integrability equations} for the constraints.
Since the rheonomic constraints express each outer component $R^A_{\alpha C}$ in terms of the inner ones $R^A_{mn}$, then the Bianchi integrability equations are (differential) equations among the space-time components of the curvatures which must be valid everywhere in superspace and, in particular, on the restriction to the space-time hypersurface. 

Hence, we reach the conclusion that the supersymmetry transformations (\ref{susytrasfrheonomy}) close an algebra only if the space-time
curvatures $R^A_{\;\;\mu \nu}$ satisfy certain integrability equations encoded in the Bianchi
identities. These equations are the \textit{space-time equations of motion} of the theory,\footnote{In the absence of auxiliary fields, which are indeed introduced to obtain an off-shell closure of the supersymmetry algebra.} and any different equation of motion would be inconsistent with the Bianchi identities.

Summarizing, in a rheonomic theory we expect that the supersymmetry
transformations (\ref{susytrasfrheonomy}) close an algebra only when the field-strengths satisfy the \textit{on-shell} constraints (on-shell configurations of the fields $\mu^A(x,0)$). Therefore, \textit{the requirement of the rheonomy projection
has the same physical content as the request that the equations of motions are satisfied}, in order to end up with a consistently defined supergravity theory.

Thus, we can now say that the constraints (\ref{rheonomicconstraint}) implement the requirement of (on-shell) matching of the bosonic and fermionic degrees of freedom, allowing to produce a supersymmetric theory in a geometric setting. When this is
imposed, we can think of supergravity as a direct extension of General Relativity,
where the manifold which must be described has both bosonic and fermionic coordinates.
For a graphic representation summarizing the scenario we have just described, see Figure \ref{figurerheon} (reproduced from Ref. \cite{Libro2}).

\begin{figure}
\centering
\pgfdeclareimage[height=7cm]{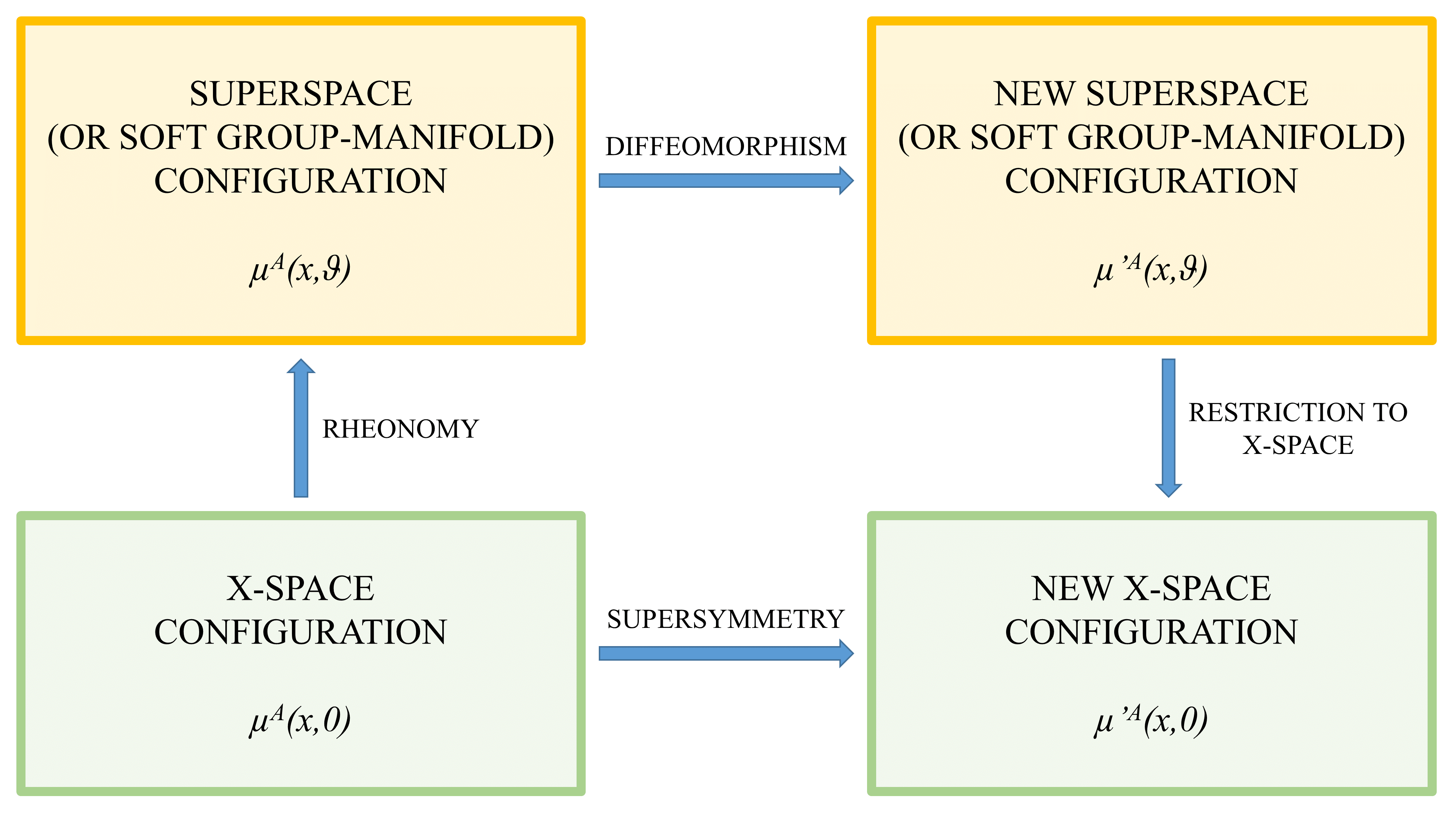}{rheonomy}
\pgfuseimage{rheonomy}
\caption[Relations among $x$-space and superspace configurations]{\textit{Relations among $x$-space and superspace configurations.} In this figure we give a graphic representation of the relations linking $x$-space and superspace configurations.} \label{figurerheon}
\end{figure}

The \textit{Lie derivative} formula
\begin{equation}
\ell_\epsilon \mu^A = (\nabla \epsilon)^A + \imath_\epsilon R^A ,
\end{equation}
with $\epsilon = \epsilon^{ab} \tilde{D}_{ab} + \epsilon^a \tilde{D}_a + \bar{\epsilon}\tilde{D}$, supplemented with the \textit{horizontality} and the \textit{rheonomic constraints}, gives:
\begin{itemize}
\item The \textit{Lorentz gauge transformations} ($\epsilon^a = \epsilon^\alpha =0$):
\begin{equation}\label{mimimi}
\ell_{(\epsilon^{ab} \tilde{D}_{ab})} \mu^A = (\nabla \epsilon)^A .
\end{equation} 
\item The \textit{general diffeomorphisms on space-time}:
\begin{equation}\label{mamama}
\ell_{(\epsilon^a \tilde{D}_a)} \mu^A = (\nabla \epsilon)^A + \epsilon^{a}R^A_{aC}\mu^C .
\end{equation}
\item The \textit{supersymmetry transformations} ($\epsilon^{ab}=\epsilon^a=0$):
\begin{equation}
\ell_{(\bar{\epsilon}\tilde{D})} \mu^A = (\nabla \epsilon)^A + 2 \bar{\epsilon}^\alpha C^{A \vert mn}_{\alpha C \vert B} R^B_{mn}\mu^C .
\end{equation}
\end{itemize}

As we have already pointed out, the closure of the \textit{supersymmetry algebra} requires, in general, further constraints on the space-time components $R^A_{mn}$, and these constraints are identified with the \textit{space-time equations of motion}; on the other hand, the closure of the gauge transformations (\ref{mimimi}) and of the space-time diffeomorphisms (\ref{mamama}) does not give further constraints.
This is the main difference between supersymmetry and all the other symmetries of a physical theory.

In the above discussion, we have considered the $D=4$ super-Poincaré group, which is the basis of the simplest supergravity theory.
The whole procedure can then be generalized to any supergroup in any space-time dimension $D$ (see, for example, \cite{Libro2} and also the work I have done during my first PhD year, \cite{Minimal}, in which minimal $\mathcal{N}=2$, $D=7$ supergravity is constructed within this framework). In $D$ dimensions, the superspace will be given by $\tilde{G}/H$, $\tilde{G}$ being the supergroup and $H$ the factorized subgroup which should always contain the Lorentz group $SO(1,D-1)$. This formulation can thus be used for constructing any supergravity theory. To this aim, we have to find the expansion of the curvatures $R^A$ in terms of the vielbein $\lbrace V^a, \psi \rbrace$ of superspace (supervielbein). This is done by imposing that the Bianchi equations are satisfied. After that, we have at once the supersymmetry transformation laws of the fields, encoded in (\ref{susytrasfrheonomy}), and the field equations.

\subsubsection{Lagrangian formulation in the rheonomic framework}

When working in the rheonomic framework, we have a \textit{geometrical} understanding of the theory, where all the ingredients have a clear (physical and geometrical) meaning. This is the main peculiarity of the rheonomic approach to supergravity theories.
In line of principle, there is no need of a Lagrangian formulation for describing the theory. However, even if the Lagrangian in superspace is not necessary, it is extremely useful in order to determine the rheonomic constraints and the space-time Lagrangian.

In fact, a Lagrangian formulation can however be given to supergravity, and, in particular, it is often
very useful for writing the equations of motion of the fields, since finding them from
the solutions of the Bianchi identities is usually very complicated.

The idea is to find an \textit{extended action principle}, that is a variational principle giving
as variational equations both the space-time equations of motion of the fields and the rheonomic constraints (and the supersymmetry transformation laws).
\textit{The geometric approach discussed so far allows a flexible
action principle, which cannot be given in other approaches}.

Suppose to construct a Lagrangian in terms of differential forms (which are invariant under diffeomorphisms) and by only using diffeomorphisms invariant operators among
them. Then, the four-dimensional space-time surface can be considered as a hypersurface embedded in the appropriate superspace $\mathcal{M}_{4 \vert 4}$, and we can construct an action by integrating the Lagrangian density which, in $D=4$, is a $4$-form, on the $4$-dimensional submanifold $\mathcal{M}_4 \subset \mathcal{M}_{4 \vert 4}$:
\begin{equation}
\mathcal{S}(\mu^A, \mathcal{M}_4)= \int_{\mathcal{M}_4}\mathcal{L}_4(\mu^A).
\end{equation}
When studying the variation of the action, in line of principle one should also vary $\mathcal{M}_4$, but, since $\mathcal{L}$ is \textit{geometrical} (it is constructed only with forms and using the differential operator ``$d$'' and the wedge product ``$\wedge$'', without using the Hodge duality operator), any deformation of $\mathcal{M}_4$ can be compensated by a diffeomorphism on the fields (the whole discussion can be extended to $\mathcal{M}_D$ and $\mathcal{M}_{D,\alpha}$ where $D>4$ and $\alpha>4$). Then, the equations of motion
\begin{equation}\label{eom}
\frac{\partial \mathcal{L}_4}{\partial \mu^A}=0
\end{equation}
can be found at fixed submanifold $\mathcal{M}_4$, but recalling that, since they are equations
on exterior forms, they actually hold over all the manifold $\mathcal{M}_{4\vert 4}$. 

For this reason, when constructing a geometric Lagrangian for supergravity, we have to avoid the use of the Hodge duality operator, since it involves the notion of a metric, and does not allow the smooth variation of integration manifold (besides the
fact that it is not clear how to extend its notion to a supermanifold).

In order to construct a \textit{geometric Lagrangian}, one can follow a set of ``\textit{building rules}'' allowing to write the most general Lagrangian with the expected good properties.
Let us list the general building rules for a supergravity Lagrangian in $D=4$ (see Ref. \cite{Libro2} for more details on these rules):
\begin{enumerate}
\item \textit{Geometricity}: The Lagrangian should be a $4$-form, constructed with the soft $1$-forms $\mu^A$ (and, when scalars and spin-$1/2$ fields are present, also with the corresponding $0$-forms) using only the diffeomorphic invariant operators $d$ and $\wedge$ (excluding the Hodge duality operator), that is \textit{coordinate invariance} is required. 

Explicitly, the Lagrangian will be written as a polynomial in the curvatures $R^A$, namely (in $D=4$) as
\begin{equation}
\mathcal{L}= \Lambda^{(4)}+R^A \wedge \Lambda^{(2)}_A + R^A \wedge R^B \Lambda^{(0)}_{AB},
\end{equation}
where the $n$-forms $\Lambda^{(n)}$ have the general expression $\Lambda^{(n)}_{AB \cdots} = C_{A_1 \cdots A_n A B \cdots} \; \mu^{A_1} \wedge \cdots \wedge \mu^{A_n}$, with the quantities $ C_{A_1 \cdots A_n A B \cdots} $ possibly depending on the scalars and spin-$1/2$ fields, when they are present. The degree in $R^A$ is at most $2$ for supergravity theories, also when dealing with higher dimensions. This comes from the request of having a Lagrangian with at most $2$ derivatives (in order to have field equations up to order $2$).
\item \textit{$H$-gauge invariance}: The Lagrangian must be $H$-invariant, with $H$ given by $H = SO(1, 3) \times H'$ (where, in $\mathcal{N}=1$ supergravities, $H'= \mathbf{1}$).\footnote{In order to implement this principle, each term in the Lagrangian must clearly be an \textit{$H$-scalar}.}
\item \textit{Homogeneous scaling law}: All fields must scale in such a way to leave invariant the curvatures and the Bianchi identities.
Precisely, the equations defining the curvatures $R^A$ are left invariant when the $1$-forms $\mu^A$ are rescaled according to
\begin{equation}\label{scaling1}
\omega^{ab} \rightarrow \omega^{ab}, \;\;\; V^a \rightarrow \omega V^a, \;\;\; \psi \rightarrow \omega^{1/2}\psi ,
\end{equation}
and the corresponding (super)curvatures as
\begin{equation}\label{scaling2curv}
R^{ab} \rightarrow R^{ab}, \;\;\; R^a \rightarrow \omega R^a, \;\;\; \rho \rightarrow \omega^{1/2}\rho.
\end{equation}
Then, the Bianchi identities are independent on $\omega$, and, as a consequence, also the field equations have to be independent on $\omega$. Therefore, each term in the Lagrangian must scale homogeneously under the above scaling law. 

In particular, in $D$ dimensions, each term must scale as $[\omega^{D-2}]$, which is the scale weight of the Einstein term.
\item \textit{Existence of the vacuum}: The defining equations for the curvatures $R^A$ always admit the solution $R^A = 0$
(\textit{vacuum solution}), in which case they reduce to the Maurer-Cartan equations of the super group-manifold $G$ for the left-invariant forms $\sigma^A$.
We then ask that also the field equations should admit the vacuum solution $R^A = 0$, where we recover a \textit{flat superspace}. Therefore, the field equations have to be at least linear in $R^A$.
\item \textit{Rheonomy}: We assume that in the field equations the parametrizations of the curvatures obey the constraints
\begin{equation}
R^A_{(O)B}= K^{A\vert (I)(I')}_{\;\; (O)B\vert C}R^C_{(I)(I')} \;\;\;\;\; \text{Rheonomy on outer components}.
\end{equation}
\end{enumerate}
Note that all these axioms are required for finding a \textit{locally supersymmetric Lagrangian}.

In order to obtain the space-time Lagrangian from the rheonomic one defined on superspace, one has to \textit{restrict} all the terms to the $\theta=0$, $d\theta=0$ hypersurface $\mathcal{M}_4$. In practice, we restrict all the superfields to their lowest ($\theta^\alpha=0$) component and to the space-time bosonic vielbein or differentials. This gives the Lagrangian $4$-form \textit{on space-time} (that is, the Lagrangian restricted from superspace to space-time).

Let us finally recall that, in some supergravity theories\footnote{In particular, auxiliary fields can be introduced in $\mathcal{N}=1$ and $\mathcal{N}=2$, $D\leq 5$ supergravity theories.}, one can also add \textit{auxiliary fields}, which allow the matching of the number of bosonic and fermionic degrees of freedom \textit{off-shell}. When this happens, the Bianchi identities are really identities, they do not imply the equations of motion, and one can construct a Lagrangian in which supersymmetry transformations close off-shell.

\subsection{Geometrical approach for the description of $D=4$ pure supergravity theories on a manifold with boundary}

Let us now introduce $D=4$ supergravity theories in the presence of a (non-trivial) space-time boundary in the geometric approach discussed above. We will strictly follow the lines of \cite{bdy}, and this short review will be useful for a clear understanding of the new (original) results we will present in Chapter \ref{chapter 4} of this thesis.

Before moving to the technical aspects of this formulation, let us introduce the scenario and give some motivations to the study of supergravity theories in the presence of a space-time boundary.

The presence of a boundary in (super)gravity theories has been studied with great interest from the `70s. In particular, in Refs. \cite{York:1972sj, Gibbons:1976ue, Brown:1992br} the authors pointed out the necessity of adding a boundary term to the gravity
action in such a way to implement Dirichlet boundary conditions for the metric field, in attempts to study the quantization of gravity with a path integral approach, in order to have an action depending only on the first derivatives of the metric.
Subsequently, the addition of boundary terms was considered in \cite{Horava:1996ma} by Horava and Witten, to cancel gauge and gravitational anomalies in eleven-dimensional supergravity. 

The inclusion of boundary terms has proved to be fundamental for the study of the so called $AdS$/CFT duality, a duality between string theory on asymptotically $AdS$ space-time (times a compact manifold) and a (conformal) quantum field theory living on the
boundary (see, for example, \cite{Maldacena:1997re, Gubser:1998bc, Witten:1998qj, Aharony:1999ti, DHoker:2002nbb} and references therein).
In the supergravity limit of string theory (that is, in the low-energy limit of the latter), the aforementioned duality implies a one-to-one correspondence between quantum operators in the conformal field theory (CFT) living on the boundary and the fields of the supergravity theory living in the bulk. In this scenario, the duality requires to supplement the supergravity action functional with appropriate boundary conditions for the supergravity fields, the latter acting as sources for the CFT operators.
In particular, the divergences presented by the bulk metric near the boundary can be eliminated through the so called ``holographic renormalization'' (see, for example, Ref. \cite{Skenderis:2002wp} for a review on this topic), with the inclusion of appropriate counterterms at the boundary.

The inclusion of boundary terms and counterterms to the \textit{bosonic} sector of $AdS$ supergravity has been studied in many different contexts. Of particular relevance are the works \cite{Aros:1999id, Aros:1999kt, Mora:2004kb, Olea:2005gb, Jatkar:2014npa}, in which it was shown that the addition of the topological Euler-Gauss-Bonnet term to the Einstein action of $D=4$ $AdS$ gravity leads to a background-independent definition of Noether charges, without the necessity of imposing Dirichlet boundary conditions on the fields. The Euler-Gauss-Bonnet boundary term regularizes the action and the related (background-independent) conserved charges.

In the context of full supergravity, boundary contributions were considered from several authors, adopting different approaches. In particular, in Refs. \cite{vanNieuwenhuizen:2005kg, Belyaev:2005rt, Belyaev:2007bg, Belyaev:2008ex, Grumiller:2009dx, Belyaev:2010as} it was pointed out that the supergravity action should
be invariant under local supersymmetry without imposing Dirichlet boundary conditions on the fields, in contrast to the Gibbons-Hawking prescription \cite{Gibbons:1976ue}.

From the above results, one can conclude that, in order to restore all the invariances of a (super)gravity Lagrangian
with cosmological constant in the presence of a non-trivial space-time boundary, one needs to add topological (boundary) contributions, also providing the counterterms necessary for regularizing the action
and the conserved charges.

More recently, in \cite{bdy} the authors worked out the construction of the $\mathcal{N} = 1$ and $\mathcal{N} = 2$, $D = 4$ supergravity theories with negative cosmological constant in the presence a non-trivial boundary (generalizing, in this way, to $D=4$ extended supergravity the results of \cite{Aros:1999id, Aros:1999kt, Mora:2004kb, Olea:2005gb, Jatkar:2014npa} and \cite{vanNieuwenhuizen:2005kg, Belyaev:2005rt, Belyaev:2007bg, Belyaev:2008ex, Grumiller:2009dx, Belyaev:2010as}), using a different approach with respect to that of \cite{vanNieuwenhuizen:2005kg, Belyaev:2005rt, Belyaev:2007bg, Belyaev:2008ex, Grumiller:2009dx, Belyaev:2010as} and extending to superspace the geometric approach of \cite{Aros:1999id, Aros:1999kt, Mora:2004kb, Olea:2005gb, Jatkar:2014npa}: Precisely, they introduced in a geometric way (generalizing the rheonomic approach to
supergravity we have introduced so far to the case in which a non-trivial space-time boundary is present) appropriate boundary terms to the Lagrangian in such a way to end up with an action (including the boundary contributions) invariant under supersymmetry transformations. 

We now recall, on the same lines of Ref. \cite{bdy}, what happens in the geometric approach when considering $D=4$ simple supergravity theories in the presence of a (non-trivial) space-time boundary, in view of a clearer understanding of the analysis we will perform in Chapter \ref{chapter 4}.

Let $V^a$ ($a=0,1,2,3$) and $\psi_A^\alpha$ ($A=1,\ldots,N$, $\alpha=1,\ldots,4$) be the bosonic and fermionic vielbein $1$-forms \textit{in superspace}, respectively. The index $A$ is the $U(N)$ $R$-symmetry index, while $\alpha$ is a four-dimensional spinor index.
In the $\mathcal{N}=1$ case (which is the one we will consider in Chapter \ref{chapter 4}), we have just $\psi^\alpha$, being $A=1$.

In any supergravity theory, the Lagrangian $\mathcal{L}$ must be invariant under supersymmetry transformations. As we have previously discussed, in the rheonomic (geometric) set up, supersymmetry transformations in space-time are interpreted as diffeomorphisms in the fermionic directions of superspace; they are generated by Lie derivatives with fermionic parameter $\epsilon^\alpha_A$. In other words, the rheonomy principle is equivalent to the requirement of space-time supersymmetry.
It follows that \textit{the supersymmetry invariance of the Lagrangian is accounted for by requiring that the Lie derivative $\ell_\epsilon$ of the Lagrangian vanishes for infinitesimal diffeomorphisms in the fermionic directions of superspace}:
\begin{equation}\label{susy}
\delta_\epsilon \mathcal{L}\equiv \ell_\epsilon \mathcal{L}= \imath_\epsilon d\mathcal{L} + d(\imath_\epsilon \mathcal{L})=0,
\end{equation}  
where $\epsilon_A (x,\theta)$ is the fermionic parameter along the tangent vector $D^A$ dual to the gravitino $\psi_A$, $\bar{\psi}_A^\alpha (D^B_\beta)=\delta^\alpha_\beta \delta^B_A$. In particular, we have $\imath_\epsilon (\psi_A)=\epsilon_A$ and $\imath_\epsilon (V^a)=0$, where $\imath$ denotes, as usual, the contraction operator.

Now, since $d\mathcal{L}$ is a $5$-form \textit{in superspace}, the first contribution, that is $\imath_\epsilon d\mathcal{L}$, which would be identically zero in space-time, is not trivial here. The second contribution, namely $d(\imath_\epsilon \mathcal{L})$, is a boundary term and does not affect the bulk result. Then, a \textit{necessary condition} for a supergravity Lagrangian is
\begin{equation}\label{bulkinv}
\imath_\epsilon d \mathcal{L}=0,
\end{equation}
which corresponds to require supersymmetry invariance in the bulk.
We will assume in the sequel that the condition (\ref{bulkinv}) always holds. Under the condition (\ref{bulkinv}), the supersymmetry transformation of the action reduces to
\begin{equation}
\delta_\epsilon \mathcal{S} = \int_{\mathcal{M}_4} d (\imath _\epsilon \mathcal{L}) = \int_{\partial \mathcal{M}_4} \imath_\epsilon \mathcal{L}. 
\end{equation}
When considering a supergravity theory on Minkowski background or, generally, on a space-time with boundary at infinity, the fields asymptotically vanish, so that

\begin{equation}\label{contrbdy}
\imath_\epsilon \mathcal{L}|_{\partial \mathcal{M}}=0,
\end{equation}
and then
\begin{equation}
\delta _\epsilon \mathcal{S}=0.
\end{equation}
In this case, equation (\ref{bulkinv}) is also a \textit{sufficient condition} for the supersymmetry invariance of the Lagrangian.

On the other hand, when the background space-time has a \textit{non-trivial boundary}, the condition (\ref{contrbdy}), modulo an exact differential, becomes non-trivial, and it is necessary to check it in an explicit way in order to get supersymmetry invariance of the action.

Let us mention that in the cases considered by the authors of \cite{bdy} (that is to say, the $\mathcal{N}=1$ and the $\mathcal{N}=2$ pure supergravity theories in $D=4$ with negative cosmological constant),
the bulk Lagrangian $\mathcal{L}_{bulk}$ is not supersymmetric when a non-trivial boundary of space-time is present. The authors of \cite{bdy} showed that, in this case, supersymmetry invariance is recovered by adding topological (boundary) contributions $\mathcal{L}_{bdy}$ to the bulk Lagrangian: Even if these contributions do not affect the bulk, they restore the supersymmetry invariance of the total Lagrangian (bulk and boundary), besides modifying the boundary dynamics. They found that the boundary values of the superspace curvatures are dynamically fixed by the field equations of the full Lagrangian, and that the introduction of a supersymmetric extension of the Gauss-Bonnet term allows to recover supersymmetry invariance. 

As the Gauss-Bonnet term in pure gravity allows to
recover invariance of the theory under all the bosonic symmetries (lost in the presence of a boundary), and further regularizes the action \cite{Aros:1999id, Aros:1999kt, Mora:2004kb, Olea:2005gb, Jatkar:2014npa}, the authors of \cite{bdy} argued that the same mechanism should also take place in the $D=4$ supersymmetric case.

The authors of \cite{bdy} also showed that the total Lagrangian $\mathcal{L}_{full}=\mathcal{L}_{bulk} + \mathcal{L}_{bdy}$ they obtained can be rewritten in a suggestive way as a sum of quadratic terms in $OSp(\mathcal{N}\vert 4)$-covariant super field-strengths (the same structure should appear also for higher $\mathcal{N}$ theories). In particular, for the $\mathcal{N}=1$ case the result presented in \cite{bdy} reproduce the MacDowell-Mansouri action \cite{MM}.

\section{Free Differential Algebras and Lie algebras cohomology}\label{revfda1}

Let us briefly introduce, in this section, the concept of \textit{Free Differential Algebra} (FDA in the following), since it will be a key concept in this thesis (mainly in Chapter \ref{chapter 5}). The presentation we give here strictly follows the lines of \cite{Hidden}.

The concept of FDA was introduced by Sullivan in \cite{Sullivan}.
Subsequently, the FDA framework was applied to the study of supergravity theories by R. D'Auria and P. Fr\'{e}, in particular in \cite{D'AuriaFre}, in which the FDA was referred to as \textit{Cartan Integrable System} (CIS), since the authors were unaware of the previous work by Sullivan \cite{Sullivan}. Actually, FDA and CIS are equivalent concepts \cite{DAuria:1982ada}.
The latter is also known as the \textit{Chevalley-Eilenberg Lie algebras cohomology framework in supergravity} (CE-cohomolgy in the following).

FDAs, which accomodate forms of degree higher than two, extending the concept of Lie algebras, emerged as underlying symmetries of field
theories containing antisymmetric tensors, that is to say, theories such as supergravity and superstring.\footnote{In particular, antisymmetric tensors are naturally contained in supergravity theories in $4 \leq D \leq 11$ space-time dimensions.} Indeed, FDAs extend the Maurer-Cartan equations of ordinary Lie (super)algebras by incorporating $p$-form potentials, with $p > 1$, that are associated to $p$-index antisymmetric tensors.

We now shortly recall the standard procedure for the construction of a \textit{minimal} FDA (a minimal FDA is one where the differential of any $p$-form does not contain forms of degree greater than $p$), starting from an ordinary Lie algebra (see, for example, Ref. \cite{Libro2} for more details on FDAs). 

Let us thus start by considering the Maurer-Cartan $1$-forms $\sigma^A$ of a Lie algebra, and let us construct the so called \textit{$(p+1)$-cochains} (\textit{Chevalley cochains}) $\Omega^{i\vert(p+1)}$ in some representation $D^i_j$ of the Lie group, that is to say, $(p+1)$-forms of the type
\begin{equation}\label{coch}
   \Omega^{i|(p+1)}=\Omega^i_{A_1\dots A_{p+1}}\sigma^{A_1}\wedge \dots \wedge \sigma^{A_{p+1}},
\end{equation}
where $ \Omega^i_{A_1\dots A_{p+1}}$ is a constant tensor. 

If the above cochains are closed:
\begin{equation}
d\Omega^{i|(p+1)}=0 ,
\end{equation}
they are called \textit{cocycles}.
If a cochain is exact, it is called a \textit{coboundary}.

Of particular interest are those \textit{cocycles that are not coboundaries}, which are elements of the CE-cohomology.\footnote{If the closed cocycles are also coboundaries (exact cochains), then the cohomology class is trivial.}
In the case in which this happens, we can introduce a $p$-form $A^{i\vert (p)}$ and write the following closed equation:
\begin{equation}\label{enlarge}
 d\,A^{i\vert(p)}+\Omega^{i\vert ({p+1})}=0 ,
\end{equation}
which, together with the Maurer-Cartan equations of the Lie algebra, is the first germ of a FDA, containing, besides the Maurer-Cartan $1$-forms $\sigma^A$, also the new $p$-form $A^{i\vert(p)}$.

This procedure can be now \textit{iterated} taking as basis of new cochains $\Omega^{j\vert (p'+1)}$ the full set of forms, namely $\sigma^A$ and $\,A^{i \vert (p)}$, and looking again for cocycles.  
If a new cocycle $\Omega^{j\vert (p'+1)}$ exists, then we can add again to the FDA a new equation
\begin{equation}\label{enlarge1}
 d\,A^{j \vert (p')}+\Omega^{j\vert (p'+1)}=0\,.
\end{equation}

The procedure can be iterated again and again, till
no more cocycles can be found. In this way, we obtain the largest FDA associated with the initial Lie algebra.

Of particular relevance (at least for a clearer understanding of this thesis) is the following Chevalley-Eilenberg theorem (see \cite{Libro2} for more details on the Chevalley-Eilenberg theorems):
\begin{theorem}\label{CEteobello}
If a Lie algebra $\mathfrak{g}$ is \textit{semisimple} and $D$ is the (trivial) identity representation, then there are no non-trivial $1$-form and $2$-form cohomology classes. 

There is, however, always a non-trivial $3$-form cohomology class, namely:
\begin{equation}
\Omega^{(3)}= C_{ABC}\sigma^A \wedge \sigma^B \wedge \sigma^C ,
\end{equation}
where $C_{ABC}$ are the structure constants with all the indexes lowered.
\end{theorem}
This means that for $\mathfrak{g}$ semisimple every closed $1$-form or $2$-form is also exact.

\subsection{Extension to supersymmetric theories}

The extension of this method to Lie superalgebras is  straightforward. Actually, in the supersymmetric case a set of non-trivial cocycles is generally present in superspace due to the existence of Fierz identities obeyed by the wedge products of gravitino $1$-forms.

In the case of supersymmetric theories, the $1$-form fields of the superalgebra we start from are the vielbein $V^a$, the gravitino $\Psi$, the spin connection $\omega^{ab}$, and, possibly, a set of gauge fields.
However, we should further impose the physical request that the FDA should be described in terms of fields living in \textit{ordinary superspace}, whose cotangent space is spanned by the supervielbein $\{V^a,\Psi\}$, dual to supertranslations. 

This corresponds to the physical request that the Lie superalgebra has a \textit{fiber bundle structure}, whose base space is spanned by the supervielbein, the rest of the fields spanning a fiber $H$. 
This fact implies an \textit{horizontality condition} on the FDA, corresponding to \textit{gauge invariance}: The gauge fields and the Lorentz spin connection belonging to $H$ must be excluded from the construction of the cochains. 

Under a geometrical point of view, this corresponds to require the CE-cohomology to be restricted to the so called \textit{$H$-relative CE-cohomology}.

\section{$D=11$ supergravity and its hidden superalgebra}\label{11D}

We now have the ingredients for moving to a review (on the same lines of Section $2$ of Ref. \cite{Hidden}) of the work \cite{D'AuriaFre} concerning the hidden superalgebra underlying $D=11$ supergravity in the FDAs framework. This is necessary for a clearer understanding of some of the original results we will present in this thesis (see Chapter \ref{chapter 5}). To this aim, let us first introduce the physical context. Then we will recall the FDA construction of $D=11$ supergravity. The presentation given here strictly follows the lines of \cite{Hidden}.

In supergravity theories in $4\leq D\leq 11$ space-time dimensions, the bosonic field content is given by the metric, a set of $1$-form gauge potentials, and $(p+1)$-form gauge potentials of various $p\leq 9$. Therefore, these theories are appropriately discussed in the context of FDAs. 

The action of $D=11$ supergravity was first constructed in \cite{Cremmer}. The theory has a bosonic field content given by the metric $g_{\mu\nu}$ and a $3$-index antisymmetric tensor $A_{\mu\nu\rho}$ (where $\mu,\nu,\rho,\ldots =0,1,\ldots ,D-1$); the theory is also endowed with a single Majorana gravitino $\Psi_\mu$ in the fermionic sector. By dimensional reduction, the $D=11$ theory yields $\mathcal{N}=8$ supergravity in four dimensions, which is considered as a possible unifying theory of all interactions. 

An important task to accomplish in the context of $D=11$ supergravity was the identification of the supergroup underlying the theory.
The authors of \cite{Cremmer} proposed $\mathfrak{osp}(1 \vert 32)$ as the most likely candidate. However, the field $A_{\mu \nu \rho}$ ($3$-index photon) of the Cremmer-Julia-Scherk theory is a $3$-form rather than a $1$-form, and therefore it cannot be interpreted as the potential of a generator in a supergroup.

The structure of the $D=11$ Cremmer-Julia-Scherk theory was then reconsidered in \cite{D'AuriaFre}, in the (supersymmetric) FDAs framework, using the superspace geometric approach (namely, in its dual Maurer-Cartan
formulation, introducing the notion of Cartan Integrable Systems).\footnote{As we have already said, in the original paper \cite{D'AuriaFre} the FDA was referred to as Cartan Integrable System, since the authors were unaware of the previous work by Sullivan \cite{Sullivan}, who introduced the mathematical concept of FDAs to which the CIS are equivalent.}

In this scenario, its bosonic sector includes, besides the supervielbein $\lbrace{ V^a,\Psi \rbrace}$, a $3$-form potential $A^{(3)}$ (whose pull-back on space-time is $A_{\mu \nu \rho}$), with field-strength $F^{(4)}= dA^{(3)}$ (modulo fermionic bilinears in terms of the gravitino $1$-form), together with its Hodge dual $F^{(7)}$, defined in such a way that its space-time components are related to the ones of the $4$-form by $F_{\mu_1    \ldots       \mu_7}= \frac 1{84} \epsilon_{\mu_1   \ldots       \mu_7\nu_1   \ldots       \nu_4} F^{\nu_1   \ldots       \nu_4}$. This amounts to say that it is associated with a $6$-form potential $B^{(6)}$ in superspace. The on-shell closure of the supersymmetric theory relies on Fierz identities involving three gravitinos, and requires $F^{(7)}= dB^{(6)}-15 A^{(3)}\wedge F^{(4)}$ (modulo fermionic currents).

In \cite{D'AuriaFre}, the supersymmetric $D=11$ FDA was introduced and investigated in order to see whether the FDA formulation could be interpreted in terms of an ordinary Lie superalgebra (in its dual Maurer-Cartan formulation).
Interestingly, this was proven to be true: The existence of a superalgebra underlying the $D=11$ supergravity theory was presented for the first time
(the authors of \cite{D'AuriaFre} got a dichotomic solution, consisting in two different supergroups, whose $1$-form potentials can be alternatively used to parametrize the $3$-form).

The superalgebra found in \cite{D'AuriaFre} includes, as a subalgebra, the super-Poincar\'e algebra of the eleven-dimensional theory, but it also contains two extra bosonic generators, called $Z_{ab}$ and $Z_{a_1    \ldots       a_5}$ (with $a,b,   \ldots       = 0,1,   \ldots       10$), which commute with the $4$-momentum $P_a$, while having appropriate commutators with the eleven-dimensional Lorentz generators $J_{ab}$. Generators that commute with all the superalgebra but the Lorentz generators will be named ``almost-central''.
Moreover, in \cite{D'AuriaFre} the authors showed that, in order to have a superalgebra that reproduce the FDA, an extra, nilpotent, fermionic generator, named $Q'$, must be included.

Indeed, besides the standard Poincar\'{e} Lie algebra, the superalgebra of \cite{D'AuriaFre} presents the following structure of (anti)commutators:
\begin{align}
	\{Q,Q\}&= -\ii C\Gamma^a P_a - \frac{ 1}{2} C\Gamma^{ab} Z_{ab} - \frac{ \ii}{5!} C\Gamma^{a_1    \ldots       a_5} Z_{a_1    \ldots       a_5}\,,\label{democracy}\\
    \left[Q,P_a \right]&\propto  \Gamma_a Q' \label{newdemo1}\,,\\
   \left[Q,Z_{ab}\right]&\propto    \Gamma_{ab} Q' \label{newdemo2}\,,\\
   \left[Q,Z_{a_1    \ldots       a_5}\right]&\propto  \Gamma_{a_1    \ldots       a_5} Q' \label{newdemo3}\,,\\
   \{Q',Q'\}&=0\,,\label{nilpot}
\end{align}
together with the (Lorentz) commutation relations involving $J_{ab}$, the other (anti)commutation relations being zero. The structure of the full superalgebra hidden in the superymmetric $D=11$ FDA also requires, for being equivalent to the FDA in superspace, the presence of a \textit{nilpotent fermionic charge}, which has been named $Q'$ in Ref. \cite{D'AuriaFre} and is dual to a spinor $1$-form $\eta$.\footnote{Actually, as we will explicitly show in Chapter \ref{chapter 5}, the extra spinor $1$-form $\eta$ (dual to the nilpotent fermionic generator $Q'$) can be parted into two different spinors, whose integrability conditions ($d^2=0$) close separately.}

The consistency of the $D=11$ theory fully relies on $3$-fermions Fierz identities obeyed by the gravitino $1$-forms.

The anticommutation relation (\ref{democracy}) generalizes to almost-central charges the central extension of the supersymmetry algebra \cite{HLS}, which, in \cite{Witten:1978mh}, was shown to be associated with topologically non-trivial configurations of the bosonic fields. The possible extension (\ref{democracy}) of the supersymmetry algebra for supergravity theories in $D\geq4$ dimensions was later widely considered (see, in particular, Refs. \cite{vanHolten:1982mx, Achucarro:1987nc, deAzcarraga:1989mza, Abraham:1990nz, Bandos:2004xw, Bandos:2005}). 
After the discovery of D$p$-branes as sources for the Ramond-Ramond gauge potentials \cite{Polchinski:1995mt} and the subsequent understanding of the duality relation occurring between $D=11$ supergravity and the Type IIA theory in $D=10$, the (extra) bosonic generators $Z_{ab}$ and $Z_{a_1    \ldots       a_5}$ were understood as $p$-brane charges, sources of the dual potentials $A^{(3)}$ and $B^{(6)}$, respectively \cite{Hull:1994ys, Townsend:1995gp}. Equation (\ref{democracy}) was then interpreted as the natural generalization of the supersymmetry algebra in higher dimensions, in the presence of non-trivial topological extended sources (black $p$-branes).

The role played by the nilpotent fermionic generator $Q'$ and its group-theoretical and physical meaning was much less investigated with respect to that of the almost-central bosonic charges. The most relevant contributions were given in \cite{vanHolten:1982mx, Bandos:2004xw, Bandos:2005}, where the results obtained in \cite{D'AuriaFre} were further analyzed and generalized. However, the physical meaning of $Q'$ remained obscure. In Chapter \ref{chapter 5} of this thesis, following the discussion we have presented in the work \cite{Hidden}, we will shed some light on this topic. 

Let us mention that the Lie superalgebra (\ref{democracy}) was rediscovered some years after the publication of \cite{D'AuriaFre} and named $M$-algebra \cite{deAzcarraga:1989mza, Sezgin:1996cj, Townsend:1997wg, Hassaine:2003vq, Hassaine:2004pp}. It is commonly considered as the Lie superalgebra underlying $M$-theory \cite{Schwarz:1995jq, Duff:1996aw, Townsend:1996xj} in its low-energy limit, corresponding to supergravity in eleven dimensions in the presence of non-trivial $M$-brane sources \cite{Achucarro:1987nc, Townsend:1995gp, Bergshoeff:1987cm, Duff:1987bx, Bergshoeff:1987qx, Townsend:1995kk}. The superalgebra disclosed in \cite{D'AuriaFre} can thus be viewed as a (Lorentz-valued) central extension of the $M$-algebra including a nilpotent fermionic generator, $Q'$. 

Here and in the following, we refer to a superalgebra descending from a given FDA as a \textit{hidden superalgebra}. The set of generators $\{Z_{ab}, \; Z_{a_1 \ldots a_5}, \; Q'\}$ span an abelian ideal of the hidden superalgebra written above (that is, the hidden superalgebra is \textit{non-(semi)simple}). The generators $\{Z_{ab}, \; Z_{a_1 \ldots a_5}, \; Q'\}$ will also be referred to as \textit{hidden generators}.

\subsection{Review of the hidden superalgebra in $D=11$}

We will now review in detail (on the same lines of \cite{Hidden}) the complete disclosure of the hidden superalgebra found in \cite{D'AuriaFre} (namely, the hidden superalgebra underlying $D=11$ supergravity).

In the approach adopted in \cite{D'AuriaFre}, the vielbein $V^a$ ($a=0,1,   \ldots       , 10$) and the gravitino $\Psi$ span a basis of the cotangent superspace $K\equiv\{V^a,\Psi\}$, where also the superspace $3$-form $A^{(3)}$ is defined. Actually, as stressed in \cite{D'AuriaFre}, one can fully extend the FDA to include also a (magnetic) $6$-form potential $B^{(6)}$, related to $A^{(3)}$ by Hodge duality of the corresponding field-strengths.
Then, the supersymmetric FDA defining the ground state of the $D=11$ theory is given by the vanishing of the following supercurvatures:

\begin{eqnarray}
R^{ab}&\equiv& d\omega^{ab} -  \omega^{ac}\wedge \omega_c^{\; b}=0\,,\label{FDA11omega}\\
T^a&\equiv& D V^a - \frac{\ii}{2}\bar{\Psi}\wedge \Gamma^a \Psi =0\,,\label{FDA11v} \\
\rho&\equiv&D \Psi=0\,,\label{FDA11psi}\\
F^{(4)} &\equiv&dA^{(3)} - \frac{1}{2}\bar{\Psi}\wedge \Gamma_{ab}\Psi \wedge V^a \wedge V^b =0\,,\label{FDA11a3} \\
F^{(7)}&\equiv&dB^{(6)} - 15 A^{(3)}\wedge  dA^{(3)} -\frac{\ii}{2}\bar{\Psi}\wedge \Gamma_{a_1   \ldots       a_5}\Psi \wedge V^{a_1} \wedge    \ldots    \wedge   V^{a_5}=0\,,\label{FDA11b6}
\end{eqnarray}
where $D$ ($D=d-\omega$, according with the convention of \cite{Hidden, D'AuriaFre}) denotes the Lorentz-covariant derivative in eleven dimensions. The closure ($d^2=0$) of this FDA is a consequence of $3$-gravitinos Fierz identities in $D=11$ (see Section \ref{fierz} of Appendix \ref{apphidden}).

As mentioned above, the authors of \cite{D'AuriaFre} found that one can trade the FDA structure on which the theory is based with an ordinary Lie superalgebra, written in its dual Maurer-Cartan formulation, namely in terms of $1$-form gauge fields valued in non-trivial tensor representations of Lorentz group $SO(1,10)$, allowing the disclosure of the fully extended superalgebra hidden in the supersymmetric FDA.

In particular, the authors of \cite{D'AuriaFre} reached this result in the following way: First of all, they associated to the forms $A^{(3)}$ and $B^{(6)}$ the bosonic $1$-forms $B^{ab}$ and $B^{a_1    \ldots       a_5}$ (in the antisymmetric representations of $SO(1,10)$), respectively. Their corresponding Maurer-Cartan equations read
\begin{equation}\label{also}
\begin{aligned}
& D B^{a_1a_2} = \frac{1}{2}\bar{\Psi}\wedge \Gamma^{a_1a_2}\Psi , \\
& D B^{a_1   \ldots       a_5} = \frac{\ii}{2} \bar{\Psi}\wedge \Gamma^{a_1    \ldots       a5}\Psi\,,
\end{aligned}
\end{equation}
where $D$ is the Lorentz-covariant derivative.
Then, they presented a general decomposition of the $3$-form $A^{(3)}$ in terms of the $1$-forms $B^{ab}$ and $B^{a_1    \ldots       a_5}$ (and of the supervielbein), by requiring the Bianchi identities of the $3$-form ($d^2 A^{(3)}=0$) to be satisfied also when $A^{(3)}$ is decomposed in terms of $1$-forms. They showed that this can be accomplished if and only if one also introduces an extra spinor $1$-form $\eta$ satisfying
\begin{equation}
 D \eta = \ii E_1 \Gamma_a \Psi \wedge V^a + E_2 \Gamma_{ab}\Psi \wedge B^{ab}+ \ii E_3 \Gamma_{a_1    \ldots       a_5}\Psi \wedge B^{a_1    \ldots       a_5}\,.\label{Deta}
\end{equation}

The consistency of the theory requires the $d^2$-closure of $B^{ab}$, $B^{a_1    \ldots         a_5}$, and $\eta$. For the two bosonic $1$-form fields, the $d^2$-closure is trivial in the ground state, due to the vanishing of the curvatures $R^{ab}$ and $\rho$, while on $\eta$ it requires the following condition:

\begin{equation}\label{integrability11}
E_1+10E_2-720E_3=0\,.
\end{equation}

Then, the authors of \cite{D'AuriaFre} found that the most general ansatz for the $3$-form $A^{(3)}$ (written in terms of $1$-forms) satisfying all the above requirements is the following one:\footnote{Here and in the following, with $B_{a_1    \ldots       a_{p-1}}^{\;\;\;\;\;\;\;\;\;\;\;\;\;\;\;\;\;b}$ we mean $B_{a_1    \ldots       a_p}\eta^{b a_p}$, where $\eta_{ab}=(+,-,   \cdots      ,-)$ denotes the Minkowski metric.}
\begin{eqnarray}
A^{(3)} & = & T_0 B_{ab} \wedge V^a \wedge V^b + T_1 B_{a b}\wedge B^{b} _{\;c}\wedge B^{c a}+  T_2 B_{b_1 a_1    \ldots       a_4}\wedge B^{b_1}_{\; b_2}\wedge B^{b_2 a_1    \ldots       a_4}+ \nonumber \\
& + & T_3 \epsilon_{a_1    \ldots       a_5 b_1    \ldots       b_5 m}B^{a_1   \ldots       a_5}\wedge B^{b_1    \ldots       b_5}\wedge V^m + \nonumber \\
& + & T_4 \epsilon_{m_1   \ldots      m_6 n_1    \ldots        n_5}B^{m_1m_2m_3p_1p_2}\wedge B^{m_4m_5m_6p_1p_2}\wedge B^{n_1    \ldots        n_5} + \nonumber\\
& + & \ii S_1 \bar{\Psi}\wedge \Gamma_a \eta \wedge V^a + S_2 \bar{\Psi}\wedge \Gamma_{ab}\eta \wedge B^{ab}+  \ii S_3 \bar{\Psi}\wedge \Gamma_{a_1    \ldots        a_5}\eta \wedge B^{a_1    \ldots        a_5}\,.\label{a3par}
\end{eqnarray}
The requirement that $A^{(3)}$ in (\ref{a3par}) satisfies equation (\ref{FDA11a3}) fixes the constants $T_i$ and $S_j$ in terms of the structure constants $E_1$, $E_2$, and $E_3$.

The final result, obtained in \cite{D'AuriaFre} by also taking into account condition (\ref{integrability11}), reads as follows:
\begin{eqnarray}
T_0 &=& \frac{120 {E_3}^2}{({E_2}-60{E_3})^2}+\frac{1}{6} , \;\;\; T_1 \; = \; -\frac{{E_2} ({E_2}-120 {E_3})}{90 ({E_2}-60 {E_3})^2} , \;\;\; T_2 \;=\; -\frac{5 {E_3}^2}{({E_2}-60 {E_3})^2},  \nonumber\\
T_3 &=& \frac{{E_3}^2}{120 ({E_2}-60 {E_3})^2}, \;\;\; T_4 \;=\; -\frac{{E_3}^2}{216 ({E_2}-60 {E_3})^2} , \nonumber \\
S_1 &=& \frac{{E_2}-48 {E_3}}{24({E_2}-60 {E_3})^2}, \;\;\;
S_2 \; = \; -\frac{{E_2}-120 {E_3}}{240 ({E_2}-60 {E_3})^2}, \;\;\; S_3\;=\; \frac{{E_3}}{240 ({E_2}-60{E_3})^2}, \nonumber\\
E_1 &=& -10 ({E_2}-72{E_3}) ,\label{11dsol}
\end{eqnarray}
where the constants $E_1$, $E_2$, and $E_3$ now define the new structure constants of the hidden superalgebra. The reader can find some details concerning this calculation in Section \ref{coeff11D} of Appendix \ref{apphidden}.

Let us mention that in \cite{D'AuriaFre} the first coefficient $T_0$ was arbitrarily fixed to $T_0=1$, leading, in this way, only to two possible solutions for the set of parameters $\{T_i,S_j,E_k\}$. As it was later pointed out in \cite{Bandos:2004xw}, this restriction (due to the particular choice $T_0=1$ on the coefficient $T_0$) can be relaxed, thus giving the general solution (\ref{11dsol}). 
Indeed, one of the $E_i$'s can be reabsorbed in the normalization of $\eta$, so that, owing to the relation (\ref{Deta}), we are left with \textit{one} free parameter, which can be written, for example, as $E_3/E_2$.\footnote{In Ref. \cite{Bandos:2004xw}, the free parameter $s$ is related to $E_3/E_2=\rho$ by $\frac{120\rho -1}{90\left(60\rho-1\right)^2} =\frac{2(3+s)}{15s^2}$.} 

The full Maurer-Cartan equations of the hidden superalgebra (in its dual formulation) are then:
\begin{eqnarray}
R^{ab }&=&  d\omega^{ab} - \frac 12 \omega^{ac}\wedge \omega _c^{\;b}\; =\; 0 , \\
 D V^a &=& \frac{\ii}{2}\bar{\Psi}\wedge \Gamma^a \Psi, \\
D \Psi&=&0, \\
 D B^{a_1a_2} & = & \frac{1}{2}\bar{\Psi}\wedge \Gamma^{a_1a_2}\Psi , \\
D B^{a_1    \ldots        a_5}& = & \frac{\ii}{2} \bar{\Psi}\wedge \Gamma^{a_1    \ldots        a_5}\Psi,\\
 D \eta & = & \ii E_1 \Gamma_a \Psi \wedge V^a + E_2 \Gamma_{ab}\Psi \wedge B^{ab}+ \ii E_3 \Gamma_{a_1    \ldots        a_5}\Psi \wedge B^{a_1    \ldots        a_5}\,. \label{deta}
\end{eqnarray}

We can finally write the hidden superalgebra in terms of generators closing a set of commutation (and anticommutation) relations. For a generic set of $1$-forms $\sigma^\Lambda$ satisfying the Maurer-Cartan equations
\begin{equation}
d\sigma^\Lambda = -\frac 12 C^\Lambda_{\ \ \Sigma\Gamma}\sigma^\Sigma \wedge\sigma^\Gamma \,,
\end{equation}
in terms of structure constants $C^\Lambda_{\ \ \Sigma\Gamma}$, this is performed by introducing a set of dual generators $T_\Lambda$ satisfying
\begin{equation}
\sigma^\Lambda(T_\Sigma) = \delta^\Lambda_\Sigma\,, \qquad
d\sigma^\Lambda (T_\Sigma,T_\Gamma)=  C^\Lambda_{\ \ \Sigma\Gamma} ,
\end{equation}
so that the $T_\Lambda$'s close the algebra $[T_\Sigma, T_\Gamma]=  C^\Lambda_{\ \ \Sigma\Gamma} T_\Lambda$.
In the case under analysis, the $1$-forms $\sigma^\Lambda$ are
\begin{equation}
\sigma^\Lambda\equiv\{V^a, \Psi, \omega^{ab}, B^{ab}, B^{a_1    \ldots        a_5}, \eta\}\,.
 \label{sigma11d}
\end{equation}
In order to recover the superalgebra in terms of (anti)commutators of the dual Lie superalgebra generators
\begin{equation}T_\Lambda \equiv\{P_a, Q, J_{ab},  Z_{ab}, Z_{a_1    \ldots        a_5}, Q'\}\,,\label{t11d}
\end{equation}
we use the duality between $1$-forms and generators, which is defined by the conditions:
\begin{eqnarray}
&V^a(P_b)= \delta^a_b\,,\quad \Psi(Q) =\mathbf{1} \,,\quad \omega^{ab}(J_{cd})=  {2}\delta^{ab}_{cd}\,,&\nonumber\\
& B^{ab}(Z_{cd})= {2}\delta^{ab}_{cd}\,,\quad B^{a_1    \ldots        a_5}(Z_{b_1    \ldots        b_5})=  {5!}\delta^{a_1    \ldots        a_5}_{b_1    \ldots        b_5}\,,\quad \eta(Q') =\mathbf{1} , &
\end{eqnarray}
where $\mathbf{1}$ denotes the unity in the spinor representation.

Then, the $D=11$ FDA corresponds to the following hidden contributions to the superalgebra (besides the Poincar\'e algebra):

\begin{eqnarray}
\lbrace Q,\bar Q \rbrace &=& -\left(\ii \Gamma^a P_a + \frac 12 \Gamma^{ab}Z_{ab}+ \frac {\ii}{5!} \Gamma^{a_1    \ldots        a_5}Z_{a_1    \ldots        a_5}\right)\,, \label{qq11new}\\ \nonumber
\lbrace Q',\bar Q' \rbrace &=& 0\,,\\ \nonumber
[Q, P_a] &=& -2 \ii E_1 \Gamma_a Q'\,,\\ \nonumber
[Q, Z_{ab}] &=&-4 E_2 \Gamma_{ab}Q' \,, \\ \nonumber
[Q, Z_{a_1    \ldots        a_5}] &=&- 2 \,(5!) \ii E_3 \Gamma_{a_1    \ldots        a_5}Q'\,, \\ \nonumber
[J_{ab}, Z_{cd}]&=&-8 \delta^{[c}_{[a}Z_{b]}^{\ d]}\,,\\ \nonumber
[J_{ab}, Z_{c_1\dots c_5}]&=&- 20 \delta^{[c_1}_{[a}Z^{c_2\dots c_5]}_{b]}\,,\\ \nonumber
[J_{ab}, Q]&=&- \Gamma_{ab} Q\,,\\ \nonumber
[J_{ab}, Q']&=&- \Gamma_{ab} Q'\,.
\end{eqnarray}
All the other (anti)commutators (beyond the Poincar\'e part) vanish.
As said before, the $E_i$'s satisfy equation (\ref{integrability11}) and one of them can be reabsorbed in the normalization of the spinor $1$-form $\eta$. The closure of the superalgebra under super-Jacobi identities is a consequence of the $d^2$-closure of the Maurer-Cartan $1$-forms equations.

In the following, we will refer to the hidden $D=11$ superalgebra disclosed in \cite{D'AuriaFre} as the ``DF-algebra'' (the acronym ``DF'' stands for ``D'Auria-Fré'').
Let us mention that the DF-algebra has recently raised a certain interest in the Mathematical-Physicists community, due to the fact that it can be reformulated in terms of $\mathcal{L}_n \subset \mathcal{L}_\infty$ algebras, or ``\textit{strong homotopy Lie algebras}'' (a comprehensive reference to this approach can be found in Refs. \cite{Sati:2015yda, stronghom}).

Note that the procedure introduced in \cite{D'AuriaFre} can be thought of as the reverse of the construction of a FDA from a given Lie superalgebra: Indeed, in the set up of \cite{D'AuriaFre}, one starts from the physical FDA as it was given \textit{a priori}, and then tries to reconstruct the hidden Lie superalgebra that could have originated it, using the algorithm of the CE-cohomology we have previously recalled.

\chapter{Algebraic background on $S$-expansion} \label{chapter 3}

% **************************** Define Graphics Path **************************
\ifpdf
    \graphicspath{{Chapter3/Figs/}{Chapter3/Figs/PDF/}{Chapter3/Figs/}}
\else
    \graphicspath{{Chapter3/Figs/Vector/}{Chapter3/Figs/}}
\fi

\nomenclature[Z]{IW contraction}{In\"{o}n\"{u}-Wigner contraction}

In Mathematics as well as in Physics, there is a great interest in studying the relations among different Lie (super)algebras related to the symmetries of different physical theories, since this can disclose connections among these theories. Furthermore, finding a new Lie (super)algebra from an already known one also means that a new physical theory could emerge.
There are many different methods for obtaining new Lie (super)algebras from given ones, for example deformations, extensions, expansions, and contractions (for short reviews on these topics see, for example, \cite{Azca1, Azca2, Azca3}).

Referring to the latter, of particular relevance is the so called \textit{In\"on\"u-Wigner contraction} \cite{IW1} (for short, IW contraction).
It has a lot of applications in Mathematics and in Physics, among which, for example, the well known case of the Poincar\'e algebra as an In\"on\"u-Wigner contraction the Anti-de Sitter algebra. 

On the other hand, in $2006$, a new expansion approach, which goes under the name of \textit{semigroup expansion} (\textit{$S$-expansion}, for short), was developed \cite{Iza1} and subsequently further enhanced, for example in \cite{Iza2, Iza3, Caroca:2010ax}. The $S$-expansion method is based on combining the structure constants of an initial Lie (super)algebra $\mathfrak{g}$ with the inner multiplication law of a discrete set $S$, endowed with the structure of a semigroup, in such a way to define the Lie bracket of a new, larger, expanded (super)algebra; the new Lie algebra obtained through this procedure is called \textit{$S$-expanded (super)algebra}, and it is commonly written as $\mathfrak{g}_S= S \times \mathfrak{g}$. In other words, the $S$-expansion method replicates through the elements of a semigroup the structure of the original Lie (super)algebra into a new one. 

From the physical point of view, several (super)gravity theories have been extensively studied and analyzed in the context of expansions and contractions, enabling numerous results over recent years (among which, for example, those presented in Refs. \cite{Azca1, Azca2, Azca3, Iza4, Fierro2, Salgado, Fierro1, Concha:2013uhq, Concha1, Concha2, CR2, Concha:2014zsa, CRSnew, Concha:2016hbt, Concha:2016tms}).

The $S$-expansion procedure turns out to be especially suitable for the construction of Chern-Simons Lagrangians for the expanded (super)algebras. The reason is that, for Chern-Simons forms, the key ingredient in the construction is the invariant tensor, and in the $S$-expansion set up general theorems have been developed, allowing for non-trivial invariant tensors for the $S$-expanded (super)algebras to be systematically constructed (see \cite{Iza1} for details).

In this chapter, we first give a brief review of In\"on\"u-Wigner contractions of Lie (super)algebras; then, we furnish the group theoretical background on $S$-expansion, since it will be useful in the last part of this thesis, where we will present some new (original) results regarding analytic formulations of this expansion method.

\section{In\"{o}n\"{u}-Wigner contractions of Lie (super)algebras}

The In\"on\"u-Wigner contraction \cite{IW1} of a Lie (super)algebra $\mathfrak{g}$ with respect to a subalgebra $\mathfrak{h}_0 \subset \mathfrak{g}$ is performed by rescaling the generators of the coset $\mathfrak{g}/\mathfrak{h}_0$, and by subsequently taking a singular limit for the rescaling parameter. 
The generators in $\mathfrak{g}/\mathfrak{h}_0$ become
abelian in the contracted algebra; the contracted algebra has a semidirect structure and the abelian generators determine an ideal of it. The contracted algebra has the same dimension as $\mathfrak{g}$.
This procedure is also referred to as \textit{standard In\"on\"u-Wigner contraction}.

The concept of standard IW contraction can then be extended to the so called \textit{generalized In\"on\"u-Wigner contraction} (\textit{i.e.} a contraction that rescales the algebra generators through different powers of the contraction parameter), in the sense intended in \cite{WW1, WW2}, by Evelyn Weimar-Woods.\footnote{Any contraction is equivalent to a generalized In\"on\"u-Wigner contraction with integer exponents \cite{WW1, WW2}.} 

More technically, the generalized IW contractions are defined when the Lie (super)algebra $\mathfrak{g}$ can be decomposed in a direct sum of $n+1$ vector subspaces
\begin{equation}
\mathfrak{g}= V_0 \oplus V_1 \oplus \cdots \oplus V_n = \bigoplus_{p=0}^n V_p ,
\end{equation}
being $V_0$ the vector space of the subalgebra $\mathfrak{h}_0$ of $\mathfrak{g}$ and $p=0,1, \ldots, n$, such that the following (Weimar-Woods) conditions are satisfied:
\begin{equation}\label{conditionsWW}
[V_p , V_q]\subset \bigoplus_{s\leq p + q} V_s,
\end{equation}
$p,q=0,1,\ldots,n$, or, in other words,
\begin{equation}
c^{k_s}_{i_p j_q}=0 \quad \text{if} \quad s > p + q,
\end{equation}
where $i_p$ labels the generators $T_{i_p}$ of $\mathfrak{g}$ in $V_p$ and $c^i_{jk}$ are the structure constants of $\mathfrak{g}$. Then, the Weimar-Woods contracted algebra \cite{WW1, WW2} is obtained by rescaling the generators of $\mathfrak{g}$ as 
\begin{equation}
T_{i_p} \rightarrow \epsilon^p T_{i_p} , \; p=0,1,\ldots, n,
\end{equation}
and by subsequently taking a singular limit for $\epsilon$. The case $n=1$ corresponds to the standard IW contraction.

\section{$S$-expansion for an arbitrary semigroup $S$}\label{sss}

As we have already mentioned, the $S$-expansion consists in combining the structure constants of a Lie (super)algebra $\mathfrak{g}$ with the inner multiplication law of an abelian semigroup $S$, in such a way to define the Lie bracket of a new, $S$-expanded (super)algebra $\mathfrak{g}_S= S \times \mathfrak{g}$. Let us now reformulate this statement more technically through the following definition (form Ref. \cite{Iza1}):
 
\begin{definition}\label{defSexp} 
Let $S= \lbrace \lambda_\alpha \rbrace$, with $\alpha=1,...,N$, be a finite, abelian semigroup with $2$-selector $K_{\alpha \beta}^{\;\;\;\; \gamma}$ defined by
\begin{equation}\label{kseldef}
K_{\alpha \beta}^{\;\;\;\; \gamma} = \left\{ \begin{aligned} &
1 , \;\;\;\;\; \text{when} \; \lambda_\alpha \lambda_\beta = \lambda_\gamma,
\\ & 0 , \;\;\;\;\; \text{otherwise}. \end{aligned} 
\right.
\end{equation}
Let $\mathfrak{g}$ be a Lie (super)algebra with basis $\lbrace T_A \rbrace$ and structure constants $C_{AB}^{\;\;\;\;C}$, defined by the commutation relations 
\begin{equation}
\left[T_A, T_B \right]= C_{AB}^{\;\;\;\;C}\; T_C .
\end{equation}
Denote a basis element of the direct product $S\times \mathfrak{g}$ by $T_{(A,\alpha)}= \lambda_\alpha T_A$, and consider the induced commutation relations
\begin{equation}
\left[ T_{(A,\alpha)},T_{(B,\beta)}\right] \equiv \lambda_\alpha \lambda_\beta \left[T_A,T_B \right].
\end{equation}
Then, the direct product
\begin{equation}\label{prodsexp}
\mathfrak{g}_S= S \times \mathfrak{g}
\end{equation}
corresponds to the Lie (super)algebra given by
\begin{equation}\label{expandedone}
\left[T_{(A,\alpha)},T_{(B,\beta)}\right]= K_{\alpha \beta}^{\;\;\;\; \gamma} C_{AB}^{\;\;\;\;C}\; T_{(C,\gamma)},
\end{equation}
whose structure constants can be written as
\begin{equation}\label{strconstant}
C_{(A,\alpha)(B,\beta)}^{\;\;\;\;\;\;\;\;\;\;\;\;\;\;\;\;(C,\gamma)}= K_{\alpha \beta}^{\;\;\;\gamma}C_{AB}^{\;\;\;\;C}.
\end{equation}
Thus, for every abelian semigroup $S$ and Lie (super)algebra $\mathfrak{g}$, the algebra $\mathfrak{g}_S$ obtained through the product (\ref{prodsexp}) is also a Lie (super)algebra, with a Lie bracket given by (\ref{expandedone}).
The new, larger Lie (super)algebra obtained in this way is called \textit{$S$-expanded (super)algebra}, and it is commonly written as $\mathfrak{g}_S= S \times \mathfrak{g}$. 
\end{definition}

Imposing extra conditions, relevant (sub)algebras can be systematically extracted from $S\times\mathfrak{g}$. In particular, let us describe in some detail the cases of reduced algebras and resonant subalgebras, because of their particular relevance for the research on this topic presented in this thesis.

\section{Reduced algebras}

We recall the following definition from Ref. \cite{Iza1}:

\begin{definition}\label{defreduced}
Let us consider a Lie (super)algebra $\mathfrak{g}$ of the form $\mathfrak{g}=V_0 \oplus V_1$, where $V_0$ and $V_1$ are two subspaces given by  $V_0=\lbrace T_{a_0} \rbrace$ and $V_1=\lbrace T_{a_1} \rbrace$, respectively. When $\left[V_0,V_1 \right]\subset V_1$, namely when the commutation relations between generators present the following form:
\begin{align}
&\left[T_{a_0},T_{b_0}\right]= C_{a_0 b_0}^{\;\;\;\;\;c_0} T_{c_0} + C_{a_0b_0}^{\;\;\;\;\;c_1}T_{c_1}, \label{commreduced1}\\ 
&\left[T_{a_0},T_{b_1}\right]= C_{a_0 b_1}^{\;\;\;\;\;c_1} T_{c_1}, \label{commreduced2}\\ 
&\left[T_{a_1},T_{b_1}\right]= C_{a_1b_1}^{\;\;\;\;\;c_0}T_{c_0}+C_{a_1b_1}^{\;\;\;\;\;c_1}T_{c_1}, \label{commreduced3}
\end{align}
the structure constants $C_{a_0b_0}^{\;\;\;\;\;c_0}$ satisfy the Jacobi identities. Therefore,
\begin{equation}
\left[T_{a_0},T_{b_0}\right]= C_{a_0b_0}^{\;\;\;\;\;c_0}T_{c_0} 
\end{equation}
itself corresponds to a Lie (super)algebra, which is called a \textit{reduced algebra} of $\mathfrak{g}$ and is commonly symbolized as $\vert V_0 \vert$.
\end{definition}

Let us observe that, in general, a reduced algebra does \textit{not} correspond to a subalgebra of $\mathfrak{g}$.

\subsection{$0_S$-reduction of $S$-expanded algebras}

The so called \textit{$0_S$-reduction} \cite{Iza1} consists in the extraction of a smaller (super)algebra from an $S$-expanded Lie (super)algebra $\mathfrak{g}_S$, when certain conditions are met. 

Let us consider a Lie (super)algebra $\mathfrak{g}$, an abelian semigroup $S$, and the $S$-expanded (super)algebra $\mathfrak{g}_S= S \times \mathfrak{g}$. The abelian semigroup $S$ can also be provided with a \textit{unique zero element} $\lambda_{0_S} \in S$ (also indicated with the symbol $0_S$ in the literature), defined as one for which
\begin{equation}
\lambda_{0_S} \lambda_\alpha = \lambda_\alpha \lambda_{0_S} = \lambda_{0_S},
\end{equation}
for each $\lambda_\alpha \in S$.

If the semigroup $S$ has a zero element $\lambda_{0_S} \in S$, then this element plays a peculiar role in the $S$-expanded (super)algebra.
Let us see what we mean, following Ref. \cite{Iza1}.

We can split the semigroup $S$ into non-zero elements $\lambda_{i}$, $i=0,...,N$, and a zero element $\lambda_{N+1}= \lambda_{0_S}$. Correspondingly, we can write 
\begin{equation}
S = \lbrace \lambda_{i}\rbrace \cup \lbrace\lambda_{N+1}=\lambda_{0_S}\rbrace ,
\end{equation}
with $i = 1, ... ,N$. Then, the $2$-selector of $S$ satisfies the relations
\begin{equation}
\left\{
\begin{aligned}
& K_{i,N+1}^{\;\;\;\;\;\;\;\;\;\; j} = K_{N+1,i}^{\;\;\;\;\;\;\;\;\;\; j}=0, \\
& K_{i,N+1}^{\;\;\;\;\;\;\;\;\;\; N+1} = K_{N+1,i}^{\;\;\;\;\;\;\;\;\;\; N+1}=1, \\
& K_{N+1,N+1}^{\;\;\;\;\;\;\;\;\;\;\;\;\;\;\;\;j} = 0, \\
& K_{N+1,N+1}^{\;\;\;\;\;\;\;\;\;\;\;\;\;\;\;\;N+1} = 1,
\end{aligned}
\right.
\end{equation}
which mean, when written in terms of multiplication rules,
\begin{align}
& \lambda_{N+1}\lambda_i = \lambda_{N+1}, \\
& \lambda_{N+1}\lambda_{N+1} = \lambda_{N+1}.
\end{align}
Therefore, for the (super)algebra $\mathfrak{g}_S=S \times \mathfrak{g}$ we can write the following commutation relations:
\begin{align}
&\left[T_{(A,i)},T_{(B,j)}\right]=K_{ij}^{\;\;\;k}C_{AB}^{\;\;\;\;C}T_{(C,k)} + K_{ij}^{\;\;\; N+1}C_{AB}^{\;\;\;\;C}T_{(C,N+1)}, \\
&\left[T_{(A,N+1)},T_{(B,j)}\right]=C_{AB}^{\;\;\;\;C}T_{(C,N+1)}, \\
&\left[T_{(A,N+1)},T_{(B,N+1)}\right]=C_{AB}^{\;\;\;\;C}T_{(C,N+1)}.
\end{align}
If we now compare these commutation relations with (\ref{commreduced1}), (\ref{commreduced2}), and (\ref{commreduced3}), we can see that
\begin{equation}\label{eqredalg}
\left[T_{(A,i)},T_{(B,j)}\right]= K_{ij}^{\;\;\;k}C_{AB}^{\;\;\;\;C}T_{(C,k)}
\end{equation}
are those of a \textit{reduced} Lie algebra of $\mathfrak{g}_S$ generated by $\lbrace T_{(A,i)} \rbrace$, whose structure constants are given by $K_{ij}^{\;\;\;\;k}C_{AB}^{\;\;\;\;C}$. 

Now, let us observe that the reduction procedure, in this particular case, results to be tantamount to impose the condition
\begin{equation}
T_{(A,N+1)}=\lambda_{0_S} T_A  = \mathbf{0}, \quad \forall T_A \in \mathfrak{g}.
\end{equation}
Note that, in this case, the reduction abelianizes large sectors of the (super)algebra, and, for each $i, \,j, \, k$ satisfying $K_{ij}^{\;\;\;k}=0$, we have 
\begin{equation}
\left[ T_{(A,i)},T_{(B,j)}\right]=\mathbf{0}
\end{equation}
in the reduced algebra of $\mathfrak{g}_S$.

The above considerations led the authors of Ref. \cite{Iza1} to the formulation of the following definition:

\begin{definition}\label{def0Sreduced}
Let $S$ be an abelian semigroup with a zero element $\lambda_{0_S} \in S$ and $\mathfrak{g}_S = S \times \mathfrak{g}$ be an $S$-expanded algebra. Then, the algebra obtained by imposing the condition
\begin{equation}\label{zero}
\lambda_{0_S} T_A =\mathbf{0}
\end{equation}
on $\mathfrak{g}_S$ (or on a subalgebra of it) is called the \textit{$0_S$-reduced algebra} of $\mathfrak{g}_S$ (or of the subalgebra).
\end{definition}

When a $0_S$-reduced (super)algebra presents a decomposition into subspaces which is \textit{resonant} with respect to the partition of the semigroup involved in the $S$-expansion process (we will define the concept of resonant subalgebra in a while), the whole procedure goes under the name of \textit{$0_S$-resonant-reduction}.

\section{Resonant subalgebras}

Another way for obtaining smaller algebras (in this case, subalgebras) from $S$-expanded ones, is described in the definitions below (again from Ref. \cite{Iza1}).

\begin{definition}\label{defresonant}
Let $\mathfrak{g}=\bigoplus_{p\in I}V_p$ be a decomposition of $\mathfrak{g}$ into subspaces $V_p$, where $I$ is a set of indexes. For each $p, q \in I$, it is always possible to define the subsets $i_{(p,q)} \subset I$ such that 
\begin{equation}\label{decomposition}
\left[V_p,V_q\right]\subset \bigoplus_{r\in i_{(p,q)}} V_r,
\end{equation}
where the subsets $i_{(p,q)}$ store the information on the subspace structure of $\mathfrak{g}$.

Now, let $S=\bigcup_{p\in I} S_p$ be a subset decomposition of the abelian semigroup $S$, such that
\begin{equation}\label{groupdecomposition}
S_p\cdot S_q \subset \bigcap_{r \in i_{(p,q)}} S_r,
\end{equation}
where the product $S_p \cdot S_q$ is defined as
\begin{equation}\label{prod}
S_p \cdot S_q = \lbrace \lambda_\gamma \mid \lambda_\gamma= \lambda_{\alpha_p}\lambda_{\alpha_q}, \; \text{with} \; \lambda_{\alpha_p}\in S_p, \lambda_{\alpha_q}\in S_q \rbrace \subset S.
\end{equation} 
When such a subset decomposition $S =\bigcup_{p\in I} S_p$ exists, with the same $p, \, q, \, r$ of (\ref{decomposition}), it is said to be \textit{in resonance} with the decomposition of $\mathfrak{g}$ into subspaces, that is with $\mathfrak{g}= \bigoplus_{p\in I}V_p$. 
\end{definition}

The resonant subset decomposition is essential in order to systematically extract subalgebras from $S$-expanded algebras, as it was enunciated and proven in Ref. \cite{Iza1} with the following theorem (which corresponds to Theorem IV.2 of \cite{Iza1}):
\begin{theorem}\label{Tres} 
Let $\mathfrak{g} = \bigcup_{p\in I} V_p$ be a subspace decomposition of $\mathfrak{g}$, with a structure as the one described by equation (\ref{decomposition}). Let $S=\bigcup_{p\in I} S_p$ be a resonant subset decomposition of the abelian semigroup $S$, with the structure given in equation (\ref{groupdecomposition}). Define the subspaces of the $S$-expanded algebra $\mathfrak{g}_S = S \times \mathfrak{g}$ as
\begin{equation}
W_p=S_p \times V_p, \;\;\; p\in I.
\end{equation}
Then, 
\begin{equation}
\mathfrak{g}_R=\bigoplus_{p\in I}W_p
\end{equation}
is a subalgebra of $\mathfrak{g}_S=S \times \mathfrak{g}$, called resonant subalgebra of $\mathfrak{g}_S$.
\end{theorem}

\section{Reduction of resonant subalgebras}

$S$-expanded (super)algebras have, in general, larger dimensions than the original ones. However, the $S$-expansion method can reproduce the In\"on\"u-Wigner contractions when certain conditions are met.
In particular, the standard IW contraction can be reproduced by performing a (finite) $S$-expansion involving resonance and $0_S$-reduction; on the other hand, the generalized IW contraction (in the sense intended in \cite{WW1, WW2}) fits within the scheme described in \cite{Iza1} when it is possible to extract reduced algebras from resonant subalgebras of (finite) $S$-expanded algebras. Then, the generalized In\"{o}n\"{u}-Wigner contraction does \textit{not} correspond to a resonant subalgebra, but to its reduction \cite{Iza1}. 

Let us report in the following Theorem VII.1 of Ref. \cite{Iza1}, which provides necessary conditions under which a reduced algebra can be extracted from a resonant subalgebra:
\begin{theorem}\label{teoiza}
Let $\mathfrak{g}_R = \bigoplus_{p \in I} S_p \times V_p$ be a resonant subalgebra of $\mathfrak{g}_S=S \times \mathfrak{g}$. Let $S_p = \hat{S}_p \cup \check{S}_p$ be a partition of the subset $S_p \subset S$ such that
\begin{equation}\label{intzero}
\hat{S}_p \cap \check{S}_p = \emptyset ,
\end{equation}
\begin{equation}\label{ressemigroup}
\check{S}_p \cdot \hat{S}_q \subset \bigcap_{r \in i_{(p,q)}}\hat{S}_r.
\end{equation}
Conditions (\ref{intzero}) and (\ref{ressemigroup}) induce the decomposition $\mathfrak{g}_R = \check{\mathfrak{g}}_R \oplus \hat{\mathfrak{g}}_R$ on the resonant subalgebra, where
\begin{equation}
\check{\mathfrak{g}}_R = \bigoplus_{p \in I} \check{S}_p \times V_p ,
\end{equation}
\begin{equation}
\hat{\mathfrak{g}}_R = \bigoplus_{p \in I}\hat{S}_p \times V_p .
\end{equation}
When the conditions (\ref{intzero}) and (\ref{ressemigroup}) hold, then
\begin{equation}
\left[ \check{\mathfrak{g}}_R, \hat{\mathfrak{g}}_R \right] \subset \hat{\mathfrak{g}}_R,
\end{equation}
and therefore $\vert \check{\mathfrak{g}}_R \vert$ corresponds to a reduced algebra of $\mathfrak{g}_R$.
\end{theorem}

As shown in \cite{Iza1}, from the structure constants for the resonant subalgebra it is then possible to write the structure constants for the reduced algebra $\vert \check{\mathfrak{g}}_R \vert$.

Observe that, when every $S_p \subset S$ of a resonant subalgebra includes the zero element $\lambda_{0_S}$, the choice $\hat{S}_p = \lbrace \lambda_{0_S} \rbrace$ automatically satisfies the conditions (\ref{intzero}) and (\ref{ressemigroup}). As a consequence of this, the $0_S$-reduction can be regarded as a particular case of Theorem \ref{teoiza}.

Theorem \ref{teoiza} will be useful in the last part of this thesis, in particular when we will present a new prescription for $S$-expansion, involving an \textit{infinite} abelian semigroup and the subtraction of an infinite ideal subalgebra from an infinite resonant subalgebra of the infinitely $S$-expanded (super)algebra.

\chapter{$AdS$-Lorentz supergravity in the presence of a non-trivial boundary} \label{chapter 4}

% **************************** Define Graphics Path **************************
\ifpdf
    \graphicspath{{Chapter4/Figs/}{Chapter4/Figs/PDF/}{Chapter4/Figs/}}
\else
    \graphicspath{{Chapter4/Figs/Vector/}{Chapter4/Figs/}}
\fi

In this chapter, our aim is to explore the supersymmetry invariance of a particular supergravity theory in the presence of a \textit{non-trivial boundary} (namely, when the boundary is not thought of as to be set at infinity). The discussion will be based on the work \cite{Gauss} that I have done in collaboration with M. C. Ipinza, P. K. Concha, and E. K. Rodríguez. 
Some motivations to our study can be found in Chapter \ref{chapter 2}, where he have spent a few words on the context in which such an analysis can be located. We have also recalled the geometrical approach for the description of $D=4$ pure supergravity on a manifold with boundary, on the same lines of \cite{bdy}.

In particular, in \cite{Gauss}, using the rheonomic (geometric) approach reviewed in Chapter \ref{chapter 2} of this thesis, we have
explored the boundary terms needed in order to restore, in the presence of a non-trivial boundary, a particular enlarged supersymmetry, known as ``$AdS$-Lorentz''.
We have first of all performed the explicit geometric construction of a bulk Lagrangian based on this enlarged superalgebra ($AdS$-Lorentz superalgebra) and then we have shown that the supersymmetric extension of a Gauss-Bonnet like term is required in order to restore the supersymmetry invariance of the theory in the presence of a non-trivial boundary.

The $AdS$-Lorentz (super)algebra was obtained as a tensorial semisimple extension of the (super)Poincaré algebra \cite{Sorokas}, and it can be alternatively derived through an $S$-expansion (see Chapter \ref{chapter 3} for a review of $S$-expansion) of the $AdS$ (super)algebra \cite{Fierro2, Salgado, Fierro1} (see also \cite{CRSnew, Concha:2016hbt} and references therein). 
The super $AdS$-Lorentz algebra can also be viewed as a deformation of the Maxwell (super)symmetries \cite{Durka:2011nf}.

Here, let us just open a small parenthesis, spending a few words on the Maxwell (super)algebras and on their interest in Physics, before proceeding with our discussion on the $AdS$-Lorentz (super)algebra. 

The Maxwell algebra is a non-central extension of Poincaré algebra. In particular, it is obtained by replacing the
commutator $[P_a, P_b] = 0$ of the Poincaré algebra with $[P_a,P_b]=Z_{ab}$, where $Z_{ab}=-Z_{ba}$ are \textit{abelian} generators commuting with translations and behaving like a tensor
with respect to Lorentz transformations. This extension of the Poincaré algebra arises when considering symmetries of systems evolving in \textit{flat Minkowski space} filled in by \textit{constant electromagnetic background} \cite{bacry, schrader}. 
Indeed, in order to interpret the Maxwell algebra and the corresponding Maxwell group, a Maxwell group-invariant particle model was studied on an extended space-time with coordinates $(x^\mu, \phi^{\mu \nu})$, where the translations of $\phi^{\mu \nu}$ are generated by $Z_{\mu \nu}$ \cite{galalg, gomis0, gomis1, Bonanos:2008ez, gibbons}. The interaction term described by
a Maxwell-invariant $1$-form introduces new tensor degrees of freedom, momenta conjugate to $\phi^{\mu \nu}$, and, in the equations of motion, they play the role of a background electromagnetic field which is constant on-shell and leads to a closed, Maxwell-invariant $2$-form.

Subsequently, the Maxwell algebra attracted some attention due to the fact that its supersymmetrization leads to a new form of $\mathcal{N} = 1$, $D = 4$ superalgebra, containing the super-Poincaré algebra \cite{gomis2}. The so called \textit{super-Maxwell algebra} introduced in \cite{gomis2} (and, subsequently, further discussed and deformed in \cite{gomis3}) is a minimal super-extension of the Maxwell algebra and can be considered as an enlargement of the so called Green algebra \cite{green}. In particular, the $\mathcal{N}=1$, $D=4$ super-Maxwell algebra describes the supersymmetries of a generalized $\mathcal{N}=1$, $D=4$ superspace in the presence of a constant, abelian, supersymmetric field-strength background.
Further generalizations of Maxwell (super)algebras where then derived and studied in the context of expansion of Lie (super)algebras \cite{deAzcarraga:2012zv}.
Lately, in \cite{CR2} the authors presented the construction of the $D = 4$ pure supergravity action (plus boundary terms) starting from a minimal Maxwell superalgebra (which can be derived from $\mathfrak{osp}(1 \vert 4)$ by applying the $S$-expansion procedure), showing, in particular, that the $\mathcal{N} = 1$, $D = 4$ pure supergravity theory can be alternatively obtained as the MacDowell-Mansouri like action built from the curvatures of this minimal Maxwell superalgebra.
Remarkably, in this context the Maxwell-like fields do \textit{not} contribute to the \textit{dynamics} of the theory, appearing only in the \textit{boundary terms}. 

Coming back to the non-supersymmetric case, in \cite{deAzcarraga:2010sw}, driven by the fact that it is often thought that the cosmological constant problem
may require an alternative approach to gravity, the authors presented a geometric framework based on the $D=4$ gauged Maxwell algebra, involving six new gauge fields associated with their abelian generators, and described its application as source of an additional contribution to the cosmological term in Einstein gravity, namely as a generalization of the cosmological term.
Subsequently, in \cite{Durka:2011nf} the authors deformed the $AdS$ algebra by adding extra \textit{non-abelian} $Z_{ab}$ generators, forming, in this way, the negative cosmological constant counterpart of the Maxwell algebra. Then, they gauged this algebra and constructed a dynamical model; in the resulting theory, the gauge fields associated with the Maxwell-like generators $Z_{ab}$ appear only in \textit{topological terms} that \textit{do not influence dynamical field equations}.

Let us stress that a good candidate for describing the dark energy corresponds to the cosmological constant \cite{Frieman:2008sn, padme}. It is well known that we can introduce a cosmological term in a gravitational theory in four space-time dimensions by considering the $AdS$ algebra. In particular, one can obtain the supersymmetric extension of gravity including a cosmological term in a geometric formulation. In this framework, the supergravity theory is built from the curvatures of the $\mathfrak{osp}(1\vert 4)$ superalgebra, and the resulting action is known as the \textit{MacDowell-Mansouri action} \cite{MM}.

In Refs. \cite{Salgado, Concha:2016tms}, the authors showed that it is possible to introduce a generalized cosmological constant term in a Born-Infeld like gravity action when the so called \textit{$AdS$-Lorentz} algebra is considered.
Lately, in \cite{CRSnew} the authors showed that, analogously, the supersymmetric extension of the $AdS$-Lorentz algebra (namely, the \textit{$AdS$-Lorentz superalgebra} we are going to consider) allows to introduce a generalized supersymmetric cosmological constant term in a geometric four-dimensional supergravity theory. In particular, in \cite{CRSnew} the $\mathcal{N}=1$, $D=4$ supergravity action is built only from the curvatures of the $AdS$-Lorentz superalgebra, and it corresponds to a MacDowell-Mansouri like action \cite{MM}.

In the following, recalling what we have done in \cite{Gauss}, we will first introduce the so called \textit{$AdS$-Lorentz superalgebra}.
Then, we shall present the explicit construction of the \textit{bulk} Lagrangian in the \textit{rheonomic
framework} (see Chapter \ref{chapter 2} for a review on this geometric approach).
The rheonomic approach to the construction of $D=4$ supergravity theories was generalized to the case of theories with (non-trivial) boundaries in \cite{bdy}. 

Subsequently, we will study the \textit{supersymmetry invariance} of the Lagrangian in the presence of a \textit{non-trivial boundary}. In particular, we will show that the supersymmetric
extension of a Gauss-Bonnet like term is required in order to restore the supersymmetry invariance
of the full Lagrangian (bulk plus boundary). The supergravity action finally obtained can be written as a MacDowell-Mansouri like action \cite{MM}.

\section{The $AdS$-Lorentz superalgebra}

The $D=4$ $AdS$-Lorentz superalgebra is generated by $\left\{
J_{ab},P_{a},Z_{ab},Q_{\alpha }\right\} $, it is semisimple, and its
(anti)commutation relations read as follows:
\begin{align}
\left[ J_{ab},J_{cd}\right] & =\eta _{bc}J_{ad}-\eta _{ac}J_{bd}-\eta
_{bd}J_{ac}+\eta _{ad}J_{bc}\,,  \label{ADSL401} \\
\left[ J_{ab},P_{c}\right] & =\eta _{bc}P_{a}-\eta _{ac}P_{b}\,, \\
\left[ J_{ab},Z_{cd}\right] & =\eta _{bc}Z_{ad}-\eta _{ac}Z_{bd}-\eta
_{bd}Z_{ac}+\eta _{ad}Z_{bc}\,, \\
\left[ J_{ab},Q_{\alpha }\right] & =-\frac{1}{2}\left( \gamma _{ab}Q\right)
_{\alpha }\,, \\
\left[ P_{a},P_{b}\right] &= Z_{ab}\,, \\
\left[ Z_{ab},P_{c}\right] & =\eta _{bc}P_{a}-\eta _{ac}P_{b}\,, \\
\left[ P_{a},Q_{\alpha }\right] & =-\frac{1}{2}%
\left( \gamma _{a}Q\right) _{\alpha }\,, \\
\left[ Z_{ab},Z_{cd}\right] & =\eta _{bc}Z_{ad}-\eta _{ac}Z_{bd}-\eta
_{bd}Z_{ac}+\eta _{ad}Z_{bc}\,,  \label{ADSL403} \\
\left[ Z_{ab},Q_{\alpha }\right] & =-\frac{1}{2}\left( \gamma _{ab}Q\right)
_{\alpha }\,, \\
\left\{ Q_{\alpha },Q_{\beta }\right\} & =-\frac{1}{2}\left[ \left( \gamma
^{ab}C\right) _{\alpha \beta }Z_{ab}-2\left( \gamma ^{a}C\right) _{\alpha
\beta }P_{a}\right] \,,  \label{ADSL408}
\end{align}%
where $C$ is the charge conjugation matrix and $\gamma _{a}$, $\gamma_{ab}$ are Dirac gamma matrices in four dimensions; $J_{ab}$ and $P_a$ are the Lorentz and translations generators, respectively, $Q_\alpha$ ($\alpha=1,\ldots, 4$) is the supersymmetry charge, and $Z_{ab}$ are non-abelian Lorentz-like generators. 

The presence of $Z_{ab}$ implies the introduction of its dual, new, bosonic $1$-form field $k^{ab}$ (which modifies the definition of the curvatures) in the supergravity theory based on the $AdS$-Lorentz superalgebra.

Notice that the Lorentz-type algebra $\mathcal{L}=\left\{ J_{ab},Z_{ab}\right\} $ is a subalgebra of the above superalgebra and that the generators $\lbrace P_a, Z_{ab}, Q_\alpha \rbrace$ span a \textit{non-abelian} ideal of the $AdS$-Lorentz superalgebra.

The $AdS$-Lorentz superalgebra written above and its extensions to higher dimensions have been useful to derive General
Relativity from Born-Infeld gravity theories \cite{Concha:2013uhq, Concha1, Concha:2014zsa}. Further generalizations of the $AdS$-Lorentz superalgebra containing more than one spinor charge can be found in Ref. \cite{CRSnew} and they can be seen as deformations of the minimal Maxwell superalgebras \cite{Concha2, gomis2, deAzcarraga:2012zv}. Let us finally mention that the following redefinition of the generators: $J_{ab} \rightarrow J_{ab}$, $Z_{ab} \rightarrow \frac{1}{\sigma^2}Z_{ab}$, $P_a \rightarrow \frac{1}{\sigma}P_a$, $Q_\alpha \rightarrow \frac{1}{\sigma}Q_\alpha$ provides the so called non-standard Maxwell superalgebra (see, for example, \cite{GenIW, Lukierski:2010dy}) in the limit $\sigma\rightarrow 0$.

\section{$AdS$-Lorentz supergravity and rheonomy}

In $\mathcal{N}=1$, $D=4$ supergravity in superspace, the bosonic $1$-forms $V^ a$ ($a = 0, 1, 2, 3$) and the fermionic ones $\psi^\alpha$ ($\alpha = 1, \ldots, 4$) define the supervielbein basis in superspace.

As we have seen in Chapter \ref{chapter 2}, in the rheonomic (geometric) approach to supergravity in superspace the supersymmetry transformations in space-time are interpreted as diffeomorphisms in the fermionic directions of superspace, and they are generated by Lie derivatives with fermionic parameter $\epsilon^\alpha$ (in the sequel, we will neglect the spinor index $\alpha $, for simplicity). In this framework, the supersymmetry invariance of the theory is satisfied requiring that the Lie derivative of the Lagrangian vanishes for diffeomorphisms
in the fermionic directions of superspace, that is equation (\ref{susy}) of Chapter \ref{chapter 2}, which we also report 
here, for the sake of completeness:
\begin{equation}\label{susymiaunewbella}
\delta_\epsilon \mathcal{L}\equiv \ell_\epsilon \mathcal{L}= \imath_\epsilon d\mathcal{L} + d(\imath_\epsilon \mathcal{L})=0.
\end{equation}

As we have already discussed in Chapter \ref{chapter 2}, the second contribution $d(\imath_\epsilon \mathcal{L})$ in (\ref{susymiaunewbella}) is a boundary term and does not affect the bulk result, and a necessary condition for a supergravity Lagrangian turns out to be the following one:
\begin{equation}\label{blublublu}
\imath_\epsilon d \mathcal{L}=0,
\end{equation}
which corresponds to require supersymmetry invariance \textit{in the bulk}. Under the condition (\ref{blublublu}), the supersymmetry transformation of the action reduces to $\delta_\epsilon \mathcal{S} = \int_{\mathcal{M}_4} d (\imath _\epsilon \mathcal{L})= \int_{\partial \mathcal{M}_4} \imath_\epsilon \mathcal{L}$ and, when we consider a supergravity theory on a space-time with boundary at infinity, the fields are asymptotically vanishing, so that we have $\imath_\epsilon \mathcal{L} |_{\partial \mathcal{M}_4}=0$ and, then, $\delta _\epsilon \mathcal{S}=0$.

When the background space-time has, instead, a \textit{non-trivial boundary}, in order to get the supersymmetry invariance of the action we need to check explicitly the condition
\begin{equation}\label{cheduepalle}
\imath_\epsilon \mathcal{L} |_{\partial \mathcal{M}_4}=0
\end{equation}
(modulo an exact differential).  

Before analyzing the $\mathcal{N} = 1$, $D = 4$ $AdS$-Lorentz supergravity theory in the presence of a non-trivial
boundary, we will study the geometric construction of the bulk Lagrangian and the corresponding supersymmetry transformation laws. In the sequel, we first apply the rheonomic approach to derive the parametrization of Lorentz-like curvatures involving the extra bosonic $1$-form $k^{ab}$ by studying the
different sectors of the \textit{on-shell} Bianchi identities.
Then, we will show that by constructing in a geometric way the $\mathcal{N} = 1$, $D = 4$ $AdS$-Lorentz supergravity bulk Lagrangian, we end up with a Lagrangian written in terms of the aforementioned Lorentz-like curvatures: This is an alternative way to introduce a generalized supersymmetric cosmological
term to supergravity. 
After that, we will write the supersymmetry transformation laws and subsequently move to the analysis of the theory in the presence of a non-trivial boundary.

\subsection{Curvatures parametrization}

Let us consider the following Lorentz-type curvatures in superspace:\footnote{Observe that the Lorentz-type curvatures (\ref{curlor1})-(\ref{curlor4}) can also be viewed as a ``torsion-deformed'' version of the super-Poincar\'{e} curvatures in $D=4$, in the sense intended in Chapter \ref{chapter 5} of this thesis, where we will analyze some superalgebras related to $D=11$ supergravity.}
\begin{align}
& R^{ab} \equiv d\omega ^{ab}+\omega ^{ab} \wedge \omega _{c}^{\; b}\,,
\label{curlor1} \\
& R^{a} \equiv D_\omega V^{a}+k_{\;b}^{a} \wedge V^{b}-\frac{1}{2}\bar{\psi }\wedge \gamma
^{a}\psi \,, \label{curlor2} \\
& F^{ab} \equiv D_\omega k^{ab}+k_{\;c}^{a} \wedge k^{cb}\,,  \label{curlor3} \\
& \rho  \equiv D_\omega \psi +\frac{1}{4}k^{ab} \wedge \gamma _{ab}\psi \,,
\label{curlor4}
\end{align}
where, according with the convention we have adopted in \cite{Gauss}, $D_\omega = d + \omega$ is the Lorentz covariant exterior derivative.
Let us observe that they can be viewed as an extension involving the $1$-form $k^{ab}$ (and the corresponding super field-strength $F^{ab}$) of the Lorentz-covariant field-strengths in superspace considered, for example, in \cite{bdy}.

The supercurvatures (\ref{curlor1})-(\ref{curlor4}) satisfy the Bianchi identities ($\nabla R^A=0$, where $\nabla$ is the gauge covariant derivative):
\begin{align}
& D_\omega R^{ab} =0\,,  \label{eqn:drab} \\
& D_\omega R^{a} =R_{\;b}^{a}\wedge V^{b}+F_{\;b}^{a}\wedge V^{b}+R^{b}\wedge k_{b}^{\;a}+\bar{\psi }\wedge \gamma ^{a}\rho \,,
\label{eqn:dra} \\
& D_\omega F^{ab} =R_{\;c}^{a}\wedge k^{cb}-R_{\;c}^{b}\wedge k^{ca}+F_{\;c}^{a}\wedge k^{cb}-F_{\;c}^{b}\wedge k^{ca}\,,
\label{eqn:dfab} \\
& D_\omega \rho =\frac{1}{4}R_{ab}\wedge \gamma ^{ab}\psi +\frac{1}{4}F_{ab}\wedge \gamma ^{ab}\psi -\frac{1}{4}k_{ab}\wedge \gamma ^{ab}\rho \,.
\label{eqn:drho}
\end{align}

The most general ansatz for the Lorentz-type curvatures in the
supervielbein basis $\lbrace V^a, \psi \rbrace$ of superspace is given by:
\begin{align}
& R^{ab} =\mathcal{R}_{\;\;cd}^{ab}V^{c}\wedge V^{d}+\bar{\Theta }
_{\;\;c}^{ab}\psi \wedge V^{c}+\alpha e\; \bar{\psi }\wedge \gamma ^{ab}\psi \,,
\label{eqn:rab} \\
& R^{a} =R_{\;cd}^{a}V^{c}\wedge V^{d}+\bar{\Theta }_{\;b}^{a} \psi \wedge V^{b}+\xi  \; \bar{\psi }\wedge \gamma ^{a}\psi \,,  \label{eqn:ra} \\
& F^{ab} =\mathcal{F}_{\;\;cd}^{ab}V^{c}\wedge V^{d}+\bar{\Lambda }%
_{\;\;c}^{ab}\psi \wedge V^{c}+\beta e \; \bar{\psi } \wedge \gamma ^{ab}\psi \,,
\label{eqn:fab} \\
& \rho =\rho _{ab}V^{a}\wedge V^{b}+\delta e\; \gamma _{a}\psi \wedge V^{a}+\Omega
_{\alpha \beta }e^{1/2}\psi ^{\alpha } \wedge \psi ^{\beta }\, , \label{eqn:rho}
\end{align}%
where $e$ (that has the dimension of a mass) is the rescaling parameter. Setting $R^{a}=0$ (\textit{on-shell} condition), we can
withdraw some terms appearing in the above ansatz for the curvatures through the study of the
scaling constraints. On the other hand, the coefficients $\alpha$, $\beta$, $\xi $, and $\delta $ appearing in the ansatz can be determined from the study of the various sectors of the Bianchi identities in superspace (\ref{eqn:drab})-(\ref{eqn:drho}).

One can show that the Bianchi identities in superspace (\ref{eqn:drab})-(\ref{eqn:drho}) are solved by parametrizing (on-shell) the full set of field-strengths in the following way:
\begin{align}
& R^{ab} =\mathcal{R}_{\;\;cd}^{ab}V^{c} \wedge V^{d}+\bar{\Theta }%
_{\;\;c}^{ab}\psi \wedge V^{c}\,,  \label{eqn:parRab} \\
& R^{a} =0\,,  \label{eqn:parra} \\
& F^{ab} =\mathcal{F}_{\;\;cd}^{ab}V^{c}\wedge V^{d}+\bar{\Lambda }%
_{\;\;c}^{ab}\psi \wedge V^{c}+e \; \bar{\psi }\wedge \gamma ^{ab}\psi \,,
\label{eqn:parFab} \\
& \rho =\rho _{ab}V^{a}\wedge V^{b}-e \; \gamma _{a}\psi \wedge V^{a}\,,
\label{eqn:parrho}
\end{align}
where $\bar{\Theta }_{\;\;c}^{ab}=\bar{\Lambda }%
_{\;\;c}^{ab}=\epsilon ^{abde}\left( \bar{\rho}_{cd}\gamma _{e}\gamma
_{5}+\bar{\rho} _{ec}\gamma _{d}\gamma _{5}-\bar{\rho} _{de}\gamma _{c}\gamma
_{5}\right) $. 

In this way, we have found the
parametrization of the Lorentz-type curvatures (\ref{curlor1})-(\ref{curlor4}) on a basis of superspace. We can now move to the geometric construction of
the bulk $AdS$-Lorentz Lagrangian.

\subsection{Rheonomic construction of the Lagrangian}

We can write the most general ansatz for the $AdS$-Lorentz Lagrangian
as follows:
\begin{equation}
\mathcal{L}=\nu ^{(4)}+F^{A}\wedge \nu _{A}^{(2)}+F^{A} \wedge F^{B}\nu _{AB}^{(0)}\,,
\label{eqn:L}
\end{equation}
where the upper index $(p)$ denotes $p$-forms; the $F^{A}$'s are the super $AdS$-Lorentz Lie algebra valued curvatures defined by:\footnote{The super $AdS$-Lorentz Lie algebra valued curvatures can also be viewed as a ``torsion-deformed'' version of the $\mathfrak{osp}(1\vert 4)$ curvatures, again in the sense intended in Chapter \ref{chapter 5} of this thesis.}
\begin{align}
& R^{ab} \equiv d\omega ^{ab}+\omega ^{ac} \wedge \omega _c^{\;b}, \label{curva1} \\
& R^{a} \equiv D_\omega V^{a}+k _{\;b}^{a} \wedge V^{b}-\frac{1}{2}\bar{\psi} \wedge \gamma^{a} \psi , \\
& \mathcal{F}^{ab} \equiv D_\omega k^{ab}+k _{\;c}^{a}\wedge k^{cb}+4 e^{2}\; V^{a} \wedge V^{b}+ e \; \bar{\psi} \wedge \gamma ^{ab} \psi , \\
& \Psi  \equiv D_\omega \psi +\frac{1}{4} k^{ab} \wedge \gamma _{ab}\psi -e\; \gamma_{a} \psi \wedge V^{a},  \label{curva4}
\end{align}
where $k^{ab}$ is the bosonic $1$-form dual to the Lorentz-like generator $Z_{ab}$ appearing in the $AdS$-superalgebra, while $\omega^{ab}$, $V^a$, $\psi$ are, as usual, the $1$-forms dual to the generators $J_{ab}$, $P_a$, and $Q$, respectively (see Chapter \ref{chapter 2} for details on this dual formulation involving differential forms).
Here we have used the symbols $\mathcal{F}^{ab}$ and $\Psi$ to denote the $AdS$-Lorentz super field-strengths, in order to avoid confusion with the Lorentz-type ones (\ref{curlor3}) and (\ref{curlor4}), denoted by $F^{ab}$ and $\rho$, respectively. 

The terms appearing in (\ref{eqn:L}) explicitly read
\begin{align}
& \nu ^{(4)}=\alpha _{1}\epsilon_{abcd}V^{a}\wedge V^{b}\wedge V^{c}\wedge V^{d}+\alpha _{2}%
\bar{\psi}\wedge \gamma ^{ab}\psi \wedge V^{c}\wedge V^{d}\epsilon _{abcd}+ \notag \\
& +\alpha _{3}\bar{\psi}\wedge \gamma _{ab}\psi \wedge V^{a}\wedge V^{b}\,,  \label{eqn:Lambda} \\
&  \notag \\
& F^{A} \wedge \nu _{A}^{(2)}=\gamma _{1}\epsilon _{abcd}R^{ab}\wedge V^{c}\wedge V^{d}+\gamma _{2}\epsilon _{abcd}\mathcal{F}^{ab}\wedge V^{c}\wedge V^{d}+ \notag \\
& +\gamma _{3}\bar{\Psi }\wedge \gamma _{5}\gamma _{a}\psi \wedge V^{a}+\gamma _{4}\bar{\Psi } \wedge \gamma _{a}\psi \wedge V^{a}+  \notag \\
& + \gamma _{5}R^{a}\wedge \bar{\psi }\wedge \gamma _{a}\psi +\gamma _{6}R^{ab}\wedge \bar{\psi }\wedge \gamma _{ab}\psi +\gamma _{7}R^{ab}\wedge V_{a}\wedge V_{b}+\gamma _{8}\epsilon _{abcd}R^{ab}\wedge \bar{\psi }\wedge \gamma ^{cd}\psi +  \notag \\
& +\gamma _{9}\mathcal{F}^{ab}\wedge V_{a}\wedge V_{b}+\gamma _{10}\epsilon _{abcd}\mathcal{F}^{ab}\wedge \bar{\psi }\wedge \gamma ^{cd}\psi +\gamma _{11}\mathcal{F}^{ab}\wedge \bar{\psi }\wedge \gamma _{ab}\psi
\,,  \label{eqn:RAnuA} \\
&  \notag \\
& F^{A}\wedge F^{B}\nu _{AB}^{(0)}=\beta _{1}R^{ab} \wedge R_{ab}+\beta
_{2}\mathcal{F}^{ab}\wedge \mathcal{F}_{ab}+ \notag \\
& + \beta _{3}\epsilon _{abcd}R^{ab} \wedge R^{cd}+\beta _{4}\epsilon _{abcd}R^{ab}\wedge \mathcal{F}^{cd}+  \notag \\
& +\beta _{5}\epsilon _{abcd}\mathcal{F}^{ab}\wedge \mathcal{F}^{cd}+\beta _{6}\bar{\Psi }\wedge \Psi
+\beta _{7}\bar{\Psi }\wedge \gamma _{5}\Psi +\beta _{8}R^{a}\wedge R_{a}\,,
\label{eqn:RARBnuAB}
\end{align}
the $\alpha _{i}$'s, $\beta _{j}$'s, and $\gamma _{k}$'s being constants. 

The curvatures (\ref{curva1})-(\ref{curva4}) are invariant under the rescaling 
\begin{equation}
\omega ^{ab}\rightarrow \omega ^{ab}, \;\;\; k^{ab}\rightarrow k^{ab}, \;\;\; V^{a}\rightarrow \omega V^{a}, \;\;\; \psi \rightarrow \omega^{1/2}\psi , \;\;\; e\rightarrow \omega^{-1}e.
\end{equation}
Additionally, the Lagrangian must scale with $\omega^{2}$, being $\omega^{2}$ the scale-weight of the Einstein term. One can prove
that the term $R^{a}\wedge R_{a}$ in (\ref{eqn:RARBnuAB}) is linear in the curvature. Furthermore, due to the scaling constraints reasons recalled in Chapter \ref{chapter 2}, some
of the terms in (\ref{eqn:RARBnuAB}) disappear. 

Let us remind that
a theory in $AdS$ space includes a cosmological constant and, since the coefficients appearing in the Lagrangian can be dimensional objects and scale with negative powers of $e$, some of the terms in $F^{A}\wedge F^{B}\nu_{AB}^{(0)}$ can survive the scaling and contribute to the Lagrangian as total derivatives. However, since we are now just constructing the \textit{bulk} Lagrangian, we can neglect them and simply set $F^{A}\wedge F^{B}\nu _{AB}^{(0)}=0$. These terms, however, will be fundamental for the construction of the \textit{boundary contributions} needed in order to restore supersymmetry invariance of the full Lagrangian (bulk plus boundary) in the presence of a non-trivial space-time boundary.

Let us now analyze the scaling of the terms in (\ref{eqn:Lambda}), whose coefficients
must be redefined as follows in order to give non-vanishing
contributions to the Lagrangian:
\begin{equation}
\alpha _{1}\equiv e^{2}\alpha _{1}^{\prime }\,,\text{ \ \ }\alpha
_{2}\equiv e\alpha _{2}^{\prime },\text{ \ \ }\alpha _{3}\equiv e%
\alpha _{3}^{\prime }\,.
\end{equation}%
In this way, all the terms in $\nu^{(4)} $ result to scale as $\omega^{2}$. Then, applying the
scaling and the parity conservation law to (\ref{eqn:Lambda}) and (\ref{eqn:RAnuA}), we obtain
\begin{equation}
\alpha _{3}=0\,,\ \text{\ \ \ \ }\gamma _{4}=\gamma _{5}=\gamma _{6}=\gamma
_{7}=\gamma _{8}=\gamma _{9}=\gamma _{10}=\gamma _{11}=0\,.
\end{equation}
Therefore, we are left with
\begin{align}
\mathcal{L}& =\epsilon _{abcd}R^{ab}\wedge V^{c}\wedge V^{d}+\gamma _{3}\bar{\psi }\wedge \gamma _{a}\gamma _{5}\Psi \wedge V^{a}+\gamma _{2}\epsilon_{abcd}\mathcal{F}^{ab}\wedge V^{c}\wedge V^{d} + \notag \\
& +\alpha _{1}^{\prime }e^{2} \; \epsilon _{abcd}V^{a}\wedge V^{b}\wedge V^{c}\wedge V^{d}+\alpha_{2}^{\prime }e\; \epsilon _{abcd}\bar{\psi }\wedge \gamma ^{ab}\psi \wedge V^{c}\wedge V^{d}\,,
\end{align}
where we have also consistently set $\gamma _{1}=1$. Using the definition of the $AdS$-Lorentz curvatures (\ref{curva1})-(\ref{curva4})
and the gamma matrices identities
\begin{align}
& \gamma_{ab}\gamma_5 = \gamma_5 \gamma_{ab} = -\frac{1}{2}\epsilon_{abcd}\gamma^{cd}, \label{gmid1} \\
& \gamma^c \gamma^{ab} = -2 \gamma^{[a}\delta^{b]}_{\; c} - \epsilon^{abcd}\gamma_5 \gamma_d , \label{gmid2}
\end{align}
we can then write
\begin{align}
\mathcal{L}& =\epsilon _{abcd}R^{ab}\wedge V^{c}\wedge V^{d}+\gamma _{3}\bar{\psi }\wedge \gamma _{a}\gamma _{5}D\psi \wedge V^{a}+\frac{\gamma _{3}}{4}\epsilon _{abcd}k^{ab}\wedge \bar{\psi}\wedge \gamma ^{c}\psi \wedge V^{d} + \notag \\
& +\gamma _{2}\epsilon _{abcd}\left( D k^{ab}+k_{\;c}^{a}\wedge k^{cb}\right)\wedge V^{c}\wedge V^{d}+\left( \alpha _{1}^{\prime
}+4\gamma _{2}\right) e^{2}\; \epsilon _{abcd}V^{a}\wedge V^{b}\wedge V^{c}\wedge V^{d} + \notag \\
& +\left( \alpha _{2}^{\prime }+\gamma _{2}+\frac{\gamma _{3}}{2}\right)
e\; \epsilon _{abcd}\bar{\psi }\wedge \gamma ^{ab}\psi \wedge V^{c}\wedge V^{d}\,,
\end{align}
where $D\equiv D_\omega$ (here and in the following, for simplicity, we omit the lower index $\omega$ denoting the Lorentz covariant exterior derivative).

We can now determine the coefficients $\alpha _{1}^{\prime }$, $\alpha
_{2}^{\prime }$, $\gamma _{2}$, and $\gamma _{3}$ through the study of the field equations. In order to find them out, let us compute the variation of the Lagrangian with respect to the different fields. The variation of the
Lagrangian with respect to the spin connection $\omega ^{ab} $ reads
\begin{equation}
\delta _{\omega }\mathcal{L}=2\epsilon _{abcd}\delta \omega ^{ab}\left(
D V^{c}+\gamma _{2}\; k_{\;f}^{c}\wedge V^{f}-\frac{1}{8}\gamma _{3}\bar{\psi }\wedge \gamma ^{c}\psi \right) \wedge V^{d}\,.
\end{equation}%
Here we see that, if $\gamma _{2}=1$ and $\gamma _{3}=4$, then $\delta _{\omega } \mathcal{L}=0$ leads to the field equation for the $AdS$-Lorentz supertorsion:
\begin{equation}
\epsilon _{abcd}R^{c} \wedge V^{d}=0\,.
\end{equation}

The variation of the Lagrangian with respect to $k^{ab}$ gives the same result. This means that the $1$-form field $k^{ab}$ does not add any constraints to the on-shell theory. 

On the other hand, by varying the Lagrangian with respect to the vielbein $V^{a}$ we get
\begin{equation}
2\epsilon _{abcd}(R^{ab}\wedge V^{c}+\mathcal{F}^{ab}\wedge V^{c})+4\bar{\psi }\wedge \gamma
_{d}\gamma _{5}\Psi =0\,,
\end{equation}
where we have exploited
\begin{equation}
\epsilon _{abcd}k^{ab}\wedge \bar{\psi }\wedge \gamma ^{c}\psi =\bar{\psi }\wedge \gamma _{d}\gamma _{5}k^{ab} \wedge \gamma _{ab}\psi \,,
\end{equation}
and where we have set $\alpha _{1}^{\prime }=-2$ and $\alpha _{2}^{\prime}=-1$, in order to recover the $AdS$-Lorentz curvatures. 

Finally, from the variation with respect to the gravitino field $\psi $, we find the field equation
\begin{equation}
8V^{a}\wedge \gamma _{a}\gamma _{5}\Psi +4\gamma _{a}\gamma _{5}\psi \wedge R^{a}=0\,.
\end{equation}

Summarizing, we have found the following values for the coefficients:
\begin{equation}
\alpha _{1}^{\prime }=-2,\text{ \ \ }\alpha _{2}^{\prime }=-1,\text{ \ \ }%
\gamma _{2}=1,\text{ \ \ }\gamma _{3}=4\,.
\end{equation}

We have thus completely determined the bulk Lagrangian $\mathcal{L}_{bulk}$
of the theory, which can
be written in terms of the Lorentz-type curvatures (\ref{curlor1})-(\ref{curlor4}) as follows:
\begin{align}
\mathcal{L}_{bulk}& =\epsilon _{abcd}R^{ab}\wedge V^{c}\wedge V^{d}+\epsilon
_{abcd}F^{ab}\wedge V^{c}\wedge V^{d}+4\bar{\psi }\wedge \gamma _{a}\gamma
_{5}\rho \wedge V^{a}  + \notag \\
& +2e^{2} \; \epsilon _{abcd}V^{a}\wedge V^{b}\wedge V^{c}\wedge V^{d}+2e\; \epsilon _{abcd}
\bar{\psi }\wedge \gamma ^{ab}\psi \wedge V^{c}\wedge V^{d}\,. \label{bulklagrangian}
\end{align}

Note the presence in (\ref{bulklagrangian}) of $e=\frac{1}{2l}$, being $l$ the radius of the asymptotic $AdS$ geometry; the equations of motion of the Lagrangian admit an $AdS$ vacuum solution with cosmological constant (proportional to $e^2$).

Notice that (\ref{bulklagrangian}) has been written as a first-order Lagrangian, and the field equations for the spin connection $\omega^{ab}$ implies (up to boundary terms, which will be considered in a while) the vanishing on-shell of the supertorsion $R^a$ defined in equation (\ref{curlor2}). This is in agreement with the condition $R^a=0$ that we have previously imposed\footnote{When this on-shell condition is imposed, we say that we are working in the second order formalism, where the spin connection $\omega^{ab}$ is torsionless and given in terms of the vielbein of space-time.} in order to find the on-shell curvature parametrizations (\ref{eqn:parRab})-(\ref{eqn:parrho}) by studying the (on-shell) Bianchi identities.

In this way, we have introduced a generalized supersymmetric cosmological constant term to a supergravity theory in an alternative way.

Let us also observe that the bosonic $1$-form field $k^{ab}$ appears, through the Lorentz-type curvatures, in the bulk Lagrangian (\ref{bulklagrangian}) and, as we have previously said, the field equation we obtain from the variation of the bulk Lagrangian with respect to $k^{ab}$ is just a double of the one obtained by varying the Lagrangian with respect to $\omega^{ab}$, namely the field equation for the $AdS$-Lorentz supertorsion: $\epsilon_{abcd}R^c \wedge V^d=0$.

\subsection{Supersymmetry transformation laws}

The parametrizations (\ref{eqn:parRab})-(\ref{eqn:parrho}) that we have previously obtained allow to write the supersymmetry transformation laws in a direct way. Indeed, in the rheonomic formalism, the transformations on space-time
are given by (see Chapter \ref{chapter 2}):
\begin{equation}
\delta \mu ^{A}=\left( \nabla \epsilon \right) ^{A}+ \imath _{\epsilon }F^{A}\,,
\end{equation}
where $\epsilon ^{A}\equiv \left( \epsilon ^{ab},\epsilon ^{a},\varepsilon
^{ab},\epsilon^\alpha \right) $; then, restricting to supersymmetry transformations, we have $\epsilon ^{ab}=\epsilon ^{a}=\varepsilon ^{ab}=0$, and
\begin{align}
& \imath_{\epsilon }R^{ab}= \bar{\Theta }_{\;\;c}^{ab}\epsilon
V^{c}\,,  \notag \\
& \imath_{\epsilon }R^{a}=0\,, \notag \\
& \imath_{\epsilon }F^{ab}= \bar{\Lambda }_{\;\;c}^{ab}\epsilon
V^{c}+2 e \; \bar{\epsilon }\gamma ^{ab}\psi \,,  \notag\\
& \imath_{\epsilon }\rho =-e \; \gamma _{a}\epsilon V^{a}\,, \label{eqn:susy}
\end{align}
which provide the following supersymmetry transformation laws:
\begin{align}
& \delta _{\epsilon }\omega ^{ab} = \bar{\Theta }_{\;\;c}^{ab}\epsilon
V^{c}\,, \\
& \delta _{\epsilon }V^{a} =\bar{\epsilon}\gamma ^{a}\psi \text{\thinspace },
\\
& \delta _{\epsilon }k^{ab} =\bar{\Lambda }_{\;\;c}^{ab}\epsilon V^{c} + 2 e\; \bar{\epsilon }\gamma ^{ab}\psi \,, \\
& \delta _{\epsilon }\psi  = d\epsilon +\frac{1}{4}\omega ^{ab}\gamma
_{ab}\epsilon +\frac{1}{4}k^{ab}\gamma _{ab}\epsilon - e\; \gamma
_{a}\epsilon V^{a}\,.
\end{align}

Under these supersymmetry transformations of the fields on space-time, the space-time Lagrangian previously introduced results to be invariant up to boundary terms. 
As we have already mentioned, in the case in which the space-time background has a \textit{non-trivial boundary} we have to check explicitly the condition (\ref{cheduepalle}).

\section{Including boundary (topological) terms}

In the following, on the same lines of \cite{bdy}, we analyze the supersymmetry invariance of the Lagrangian in the presence of a non-trivial space-time boundary and, in particular, we present the explicit boundary
terms required to recover the supersymmetry invariance of the full Lagrangian (bulk plus boundary).

Let us consider the bulk Lagrangian (\ref{bulklagrangian}), which we report here for completeness:
\begin{align}
\mathcal{L}_{bulk}& =\epsilon _{abcd}R^{ab}\wedge V^{c}\wedge V^{d} + 4\bar{\psi }\wedge V^a \wedge \gamma _{a}\gamma
_{5} \rho + \notag \\
& +\epsilon_{abcd}\left( F^{ab}\wedge V^{c}\wedge V^{d} + 2 e \;V^a \wedge V^b \wedge \bar{\psi}\wedge \gamma^{cd}\psi + 2e^{2} \;V^{a}\wedge V^{b}\wedge V^{c}\wedge V^{d} \right)\,. \label{bulklagrangiannew}
\end{align}
The supersymmetry invariance in the bulk is satisfied \textit{on-shell} (where $R^a=0$).
Nevertheless, for this theory the boundary invariance of the Lagrangian under supersymmetry is \textit{not} trivially satisfied, and the condition (\ref{cheduepalle}) has to be checked in an explicit way (see Chapter \ref{chapter 2} for details). In fact, we find that, if the fields do \textit{not} asymptotically vanish at the boundary, we have
\begin{equation}
\imath _{\epsilon }\mathcal{L}_{bulk} |_{\partial \mathcal{M}_{4}}\neq
0\,.
\end{equation}

In order to restore the supersymmetry invariance of the theory, we must provide a more subtle approach. In particular, it is possible to modify the bulk Lagrangian by \textit{adding boundary (topological) terms}, which do not alter the bulk Lagrangian and only affect the boundary Lagrangian, so that the condition (\ref{susymiaunewbella}), which is the condition for the supersymmetry invariance of the theory in our geometric framework, is still satisfied.

The only possible boundary contributions (topological $4$-forms) that are compatible with parity, Lorentz-like
invariance, and $\mathcal{N}=1$ supersymmetry are the following ones:
\begin{align}
& d\left( \varpi ^{ab}\wedge N^{cd}+\varpi _{f}^{a}\wedge \varpi ^{fb} \wedge \varpi ^{cd}\right) \epsilon _{abcd} =\epsilon _{abcd}N^{ab}\wedge N^{cd}\,, \label{boundarino1} \\
& d\left( \bar{\rho}\wedge \gamma _{5}\psi \right)  =\bar{\rho}\wedge \gamma _{5}\rho +\frac{1}{8}\epsilon _{abcd}N^{ab} \wedge \bar{\psi}\wedge \gamma ^{cd}\psi \,, \label{boundarino2}
\end{align}
where we have defined $\varpi ^{ab}=\omega ^{ab}+k^{ab}$ and $N^{ab}=R^{ab}+F^{ab}$, with $R^{ab}$ and $F^{ab}$ given by equations (\ref{curlor1}) and (\ref{curlor3}),
respectively. Note that $\varpi ^{ab}$ and $N^{ab}$ can be thought as related to a Lorentz-like generator $M_{ab}\equiv J_{ab}+Z_{ab}$ (see equations (\ref{ADSL401})-(\ref{ADSL408})). 

The boundary terms written above correspond to the following boundary Lagrangian:
\begin{align}
\mathcal{L}_{bdy}& =\alpha \epsilon _{abcd}\left( R^{ab}\wedge R^{cd}+2 R^{ab}\wedge F^{cd}+ F^{ab}\wedge F^{cd}\right) +  \notag \\
& +\beta \left( \bar{\rho}\wedge \gamma _{5}\rho +\frac{1}{8}\epsilon _{abcd}R^{ab}\wedge \bar{\psi}\wedge \gamma ^{cd}\psi +\frac{1}{8}\epsilon _{abcd}F^{ab}\wedge \bar{\psi}\wedge \gamma ^{cd}\psi \right) \,.
\end{align}
Observe that the structure of a supersymmetric Gauss-Bonnet like term appears. 

Let us then consider the following ``full'' Lagrangian (bulk plus boundary):
\begin{align}\label{tuttalal}
\mathcal{L}_{full}& =\mathcal{L}_{bulk}+\mathcal{L}_{bdy} = \notag \\
& =\epsilon _{abcd}R^{ab}\wedge V^{c}\wedge V^{d}+4\bar{\psi}\wedge V^{a}\wedge \gamma_{a}\gamma _{5}\rho + \notag \\
& + \epsilon _{abcd}\left( F^{ab}\wedge V^{c}\wedge V^{d}+2e\; V^{a}\wedge V^{b}\wedge \bar{\psi}\wedge \gamma ^{cd}\psi +2e^{2}\; V^{a}\wedge V^{b}\wedge V^{c}\wedge V^{d}
\right) + \notag \\
& +\alpha \epsilon _{abcd}\left(R^{ab}\wedge R^{cd}+2\epsilon
_{abcd}R^{ab}\wedge F^{cd}+\epsilon _{abcd}F^{ab}\wedge F^{cd}\right)   + \notag \\
& +\beta \left( \frac{1}{8}\epsilon _{abcd}R^{ab}\wedge \bar{\psi}\wedge \gamma^{cd}\psi +\frac{1}{8}\epsilon _{abcd}F^{ab}\wedge \bar{\psi}\wedge \gamma^{cd}\psi +\bar{\rho}\wedge \gamma _{5}\rho \right) \,.
\end{align}%
Due to the $e^{\; -2}$-homogeneous scaling of the Lagrangian, we have that the coefficients $\alpha $ and $\beta $ must be proportional to $e^{\; -2}$ and $e^{\; -1}$, respectively.

Let us now study the conditions under which  the full Lagrangian (\ref{tuttalal}) is invariant under supersymmetry.

As we have previously pointed out, the supersymmetry invariance of the full Lagrangian $\mathcal{L}_{full}$ requires
\begin{equation}
\delta _{\epsilon }\mathcal{L}_{full}=\ell _{\epsilon }\mathcal{L}_{full}=\imath
_{\epsilon }d\mathcal{L}_{full}+d\left( \imath _{\epsilon }\mathcal{L}
_{full}\right) =0\,.
\end{equation}
Now, since the boundary terms that we have introduced, namely (\ref{boundarino1}) and (\ref{boundarino2}), are total differentials, the condition for supersymmetry in the bulk, that is $\imath _{\epsilon }d\mathcal{L}_{full}=0$, is trivially satisfied.
Thus, the supersymmetry invariance of the full Lagrangian $\mathcal{L}_{full}$ just requires to verify that, for a suitable choice of $\alpha$ and $\beta$, the condition $\imath _{\epsilon }\left( \mathcal{L}_{full}\right)=0$ (modulo an exact differential) holds \textit{on the boundary}, namely $\imath _{\epsilon }\left( \mathcal{L}_{full}\right)  |_{\partial \mathcal{M}}=0$. We have:
\begin{align}
\imath _{\epsilon }\left( \mathcal{L}_{full} \right) & =\epsilon
_{abcd}\imath _{\epsilon }\left( R^{ab}+F^{ab}\right)\wedge
V^{c}\wedge V^{d}+4\bar{\epsilon}V^{a}\wedge \gamma _{a}\gamma _{5}\rho +4\bar{\psi}\wedge V^{a}\gamma _{a}\gamma _{5}\imath _{\epsilon } \left(\rho \right)  + \notag \\
& +\epsilon _{abcd}4e\; V^{a}\wedge V^{b}\wedge \bar{\epsilon}\gamma ^{cd}\psi + \notag \\
& + 2\imath
_{\epsilon }\left( R^{ab}+ F^{ab}\right) \left\{ \alpha
R^{cd}+\frac{\beta }{16}\bar{\psi}\wedge \gamma ^{cd}\psi +\alpha
F^{cd}\right\} \epsilon _{abcd} + \notag \\
& +\frac{\beta }{4}\epsilon _{abcd}\left( R^{ab}+F^{ab}\right)\wedge \bar{\epsilon}\gamma ^{cd}\psi +2\beta  \imath _{\epsilon
}\left( \bar{\rho} \right) \gamma _{5}\rho \,.
\end{align}
In general, this is not zero, but \textit{its projection on the boundary should be zero}. Indeed,
in the presence of a boundary, the field equations in superspace for the Lagrangian (\ref{tuttalal}) acquire non-trivial boundary contributions, that result in the following constraints (which hold on the boundary):
\begin{equation}\label{constbdy}
\left. \frac{\delta \mathcal{L}_{full}}{\delta \mu ^{A}}\right\vert
_{\partial \mathcal{M}}=0 \;\; \Rightarrow \;\;
\left\{
\begin{aligned}
& \left( R^{ab}+F^{ab}\right) |_{\partial \mathcal{M}} =-\frac{1}{2\alpha }V^{a}\wedge V^{b}-\frac{\beta }{16\alpha }\bar{\psi}\wedge \gamma
^{ab}\psi \,, \\
& \rho |_{\partial \mathcal{M}} =\frac{2}{\beta }V^{a}\wedge \gamma _{a}\psi \, .
\end{aligned}
\right.
\end{equation}
We can thus see that the supercurvatures on the boundary are not dynamical, but fixed to constant values in the anholonomic basis of the bosonic and fermionic vielbein (analogously to what was found in Ref. \cite{bdy}).
Then, upon use of (\ref{constbdy}), on the boundary we have
\begin{align}
\imath _{\epsilon }\left( \mathcal{L}_{full}\right) |_{\partial \mathcal{M}}& =-\frac{\beta }{8\alpha }\epsilon _{abcd}\bar{\epsilon}\gamma ^{ab}\psi \wedge V^{c}\wedge V^{d}+4\bar{\epsilon}V^{a}\wedge \gamma _{a}\gamma _{5}\rho + \notag \\
& + \frac{8}{\beta }\bar{\psi}\wedge V^{a} \wedge \gamma _{a}\gamma _{5}V^{b}\gamma _{b}\epsilon + 4e\; \epsilon _{abcd}V^{a}\wedge V^{b}\wedge \bar{\epsilon}\gamma ^{cd}\psi + \notag \\
& -\left(\frac{\beta }{4\alpha }\bar{\epsilon}\gamma ^{ab}\psi \right)\wedge \left\{ \alpha
R^{cd}+\frac{\beta }{16}\bar{\psi}\wedge \gamma ^{cd}\psi +\alpha
F^{cd}\right\} \epsilon _{abcd} + \notag \\
& +\frac{\beta }{4}\epsilon _{abcd}\left\{ R^{ab}\wedge \bar{\epsilon}
\gamma ^{cd}\psi +F^{ab}\wedge \bar{\epsilon}\gamma ^{cd}\psi \right\} -4\bar{\epsilon}\gamma _{a}V^{a}\wedge \gamma _{5}\rho \,.
\end{align}
Subsequently, using the Fierz identity $\gamma _{ab}\psi \wedge  \bar{\psi} \wedge \gamma ^{ab}\psi =0$,\footnote{The other useful Fierz identity for the study of the $\mathcal{N}=1$, $D=4$ theory is $\gamma _{a}\psi \wedge \bar{\psi} \wedge \gamma ^{a}\psi =0$.} we can write
\begin{equation}
\imath _{\epsilon }\left( \mathcal{L}_{full}\right) |_{\partial \mathcal{M}
}=\left( 4e-\frac{\beta }{8\alpha }\right) \epsilon _{abcd}\bar{%
\epsilon}\gamma ^{ab}\psi \wedge V^{c}\wedge V^{d}+\frac{8}{\beta }\bar{\psi}\wedge V^{a}\wedge \gamma
_{a}\gamma _{5}V^{b}\gamma _{b}\epsilon \,.
\end{equation}
After that, using the gamma matrices identity (\ref{gmid1}), we find that $\imath _{\epsilon }\left( \mathcal{L}_{full}\right) |_{\partial \mathcal{M}}=0$ if the following relation between the coefficients $\alpha $ and $\beta $ holds:
\begin{equation}
\frac{\beta }{4\alpha }+\frac{8}{\beta }=8e\,.
\end{equation}

Solving the above equation for $\beta $, we obtain
\begin{equation}
\beta =16e\alpha \left( 1\pm \sqrt{1-\frac{1}{8e^{2}\alpha }}\right)
\,.
\end{equation}
Interestingly, by setting the square root to zero, which implies
\begin{equation}
\alpha =\frac{1}{8e^{2}} \; \; \Rightarrow \; \; \beta =\frac{2}{e}\,,
\end{equation}
we recover the following $2$-form supercurvatures:
\begin{align}
& \mathcal{N}^{ab} = R^{ab}+F^{ab}+4e^{2}\; V^{a}\wedge V^{b}+e\; \bar{\psi}\wedge \gamma ^{ab}\psi \,,  \label{2f1} \\
& \Psi  = \rho -e\;  \gamma _{a}\psi \wedge V^{a} \,\,,  \label{2f2} \\
& R^{a} = D V^{a}+k_{\;b}^{a}\wedge V^{b}-\frac{1}{2}\bar{\psi}\wedge \gamma
^{a}\psi \,, \label{2f3}
\end{align}
which reproduce the $AdS$-Lorentz curvatures with
\begin{equation}
\mathcal{N}^{ab} =R^{ab}+\mathcal{F}^{ab}\,,
\end{equation}
where
\begin{align}
& R^{ab} =d\omega ^{ab}- \omega^{ac}\wedge \omega_c^{\; b} \,, \\
& \mathcal{F}^{ab}=F^{ab}+4e^{2}\; V^{a}\wedge V^{b}+e\; \bar{\psi}\wedge \gamma
^{ab}\psi \,.
\end{align}

Finally, the full Lagrangian (\ref{tuttalal}), written in terms of the $2$-form supercurvatures (\ref{2f1}) and (\ref{2f2}), can be recast as a MacDowell-Mansouri like form \cite{MM}, that is:
\begin{equation}
\mathcal{L}_{full}=\frac{1}{8e^{2}}\epsilon _{abcd}\mathcal{N}^{ab}\wedge \mathcal{N}^{cd}+\frac{2}{e}\bar{\Psi} \wedge \gamma _{5}\Psi \,,  \label{LFULL}
\end{equation}
whose boundary term corresponds to a supersymmetric Gauss-Bonnet like term:
\begin{align}
\mathcal{L}_{bdy}& =\frac{1}{8 e^2} \epsilon _{abcd}\left( R^{ab}\wedge R^{cd}+2 R^{ab}\wedge F^{cd}+F^{ab}\wedge F^{cd}\right) +  \notag \\
& +\frac{2}{e} \left( \frac{1}{8}\epsilon _{abcd}R^{ab}\wedge \bar{\psi}\wedge \gamma ^{cd}\psi +\frac{1}{8}\epsilon _{abcd}F^{ab}\wedge \bar{\psi}\wedge \gamma ^{cd}\psi + \bar{\rho}\wedge \gamma _{5}\rho \right) \,.
\end{align}

We have thus shown that this topological Gauss-Bonnet like term allows to recover the supersymmetry invariance of the (on-shell) theory in the presence of a non-trivial boundary. 
As we have already mentioned, the same phenomenon occurs in pure gravity, where the Gauss-Bonnet term ensures the invariance of the Lagrangian in the presence of a non-trivial space-time boundary. 

The supersymmetric extension of the Gauss-Bonnet term was also introduced in \cite{bdy} in order to restore the supersymmetry invariance of $\mathcal{N}=1$ and $\mathcal{N}=2$, $OSp(\mathcal{N},4)$ supergravities in the presence of a non-trivial boundary.

Notice that, in terms of the supercurvatures (\ref{2f1}) and (\ref{2f2}), the boundary conditions on the field-strengths (\ref{constbdy}) take the following simple form: $\mathcal{N}^{ab} |_{\partial \mathcal{M}}=0$ and $\Psi  |_{\partial \mathcal{M}}=0$. This means that the linear combination $R^{ab} + \mathcal{F}^{ab}$ of the $AdS$-Lorentz supercurvatures and the $AdS$-Lorentz supercurvature $\Psi$ vanish at the boundary (analogously to what was found in Ref. \cite{bdy} in the case of $\mathcal{N}=1$ and $\mathcal{N}=2$, $OSp(\mathcal{N},4)$ supergravities).

Concerning the bulk Lagrangian, let us observe that it reproduces the generalized supersymmetric cosmological terms presented in \cite{CRSnew} and that it corresponds to a supersymmetric extension of the results found in \cite{Salgado, deAzcarraga:2010sw}.

%Let us also note that a suitable In\"{o}n\"{u}-Wigner contraction of the full Lagrangian (\ref{LFULL}) leads to the Maxwell MacDowell-Mansouri Lagrangian presented in \cite{CR2} (in the context of minimal super-Maxwell algebras), containing, besides other contributions, also the $\mathcal{N} = 1$, $D=4 $ pure supergravity Lagrangian plus boundary terms.

Summarizing our results, in this chapter (following \cite{Gauss}) we have presented the explicit construction of the $\mathcal{N} = 1$, $D = 4$ $AdS$-Lorentz supergravity bulk Lagragian in the rheonomic framework. In particular, we have shown an alternative way to introduce a generalized supersymmetric cosmological term to supergravity. Subsequently, we have studied the supersymmetry invariance of the Lagrangian in the presence of a non-trivial boundary, and we have found that the supersymmetric
extension of a Gauss-Bonnet like term is required in order to restore the supersymmetry invariance of the full Lagrangian (bulk plus boundary). 

As we have already said in Chapter \ref{chapter 2}, the inclusion of boundary terms has proved to be fundamental for the study of the so called $AdS$/CFT duality. In this context, as far as the metric field is concerned, the bulk metric is divergent near the boundary; however, these divergences can be successfully eliminated through a procedure called ``holographic renormalization'' (see, for example, Ref. \cite{Skenderis:2002wp}), with the inclusion of appropriate counterterms at the boundary. 
In this scenario, we argue that the presence in the boundary of the $k^{ab}$ fields appearing in our model could somehow play a role in allowing the regularization of the action in the holographic renormalization language.

\chapter{Hidden gauge structure of Free Differential Algebras} \label{chapter 5}

% **************************** Define Graphics Path **************************
\ifpdf
    \graphicspath{{Chapter5/Figs/}{Chapter5/Figs/PDF/}{Chapter5/Figs/}}
\else
    \graphicspath{{Chapter5/Figs/Vector/}{Chapter5/Figs/}}
\fi

In this chapter, I will discuss the core of my PhD research, basing my discussion on the works \cite{Hidden} and \cite{Malg} that I have done in collaboration with L. Andrianopoli and R. D'Auria. 
An introduction to the physical context and a detailed review of the hidden superalgebra underlying $D=11$ supergravity can be found in Section \ref{11D} of Chapter \ref{chapter 2}.

As shown in Ref. \cite{D'AuriaFre} and recalled in Chapter \ref{chapter 2}, the structure of the full superalgebra hidden in the superymmetric FDA describing $D=11$ supergravity also requires, for being equivalent to the FDA in superspace, the presence of a \textit{fermionic nilpotent charge}, that has been named $Q'$.
This fact is fully general, and a hidden superalgebra underlying the supersymmetric FDA (containing, at least, one nilpotent fermionic generator) can be constructed for any supergravity theory in the presence of antisymmetric tensor fields.

As we have already mentioned in Chapter \ref{chapter 2}, the role of the extra, nilpotent, fermionic generator $Q'$ and
its physical meaning was much less investigated with respect to that of the bosonic almost-central charges.

In the present chapter, following the work \cite{Hidden}, we further analyze the superalgebra hidden in all the supersymmetric FDAs and clarify the role played by its generators (mainly, the extra, nilpotent, fermionic ones).
In particular, we will discuss in detail the gauge structure of the supersymmetric $D=11$ FDA in relation to its hidden gauge superalgebra, focusing on the role played by the nilpotent generator $Q'$. Then, we will consider minimal $\mathcal{N}=2$ supergravity in $D=7$ in order to test the universality of the construction and to investigate possible extensions of the underlying superalgebra of \cite{D'AuriaFre} (as we have already said in Chapter \ref{chapter 2}, we refer to this latter superalgebra as ``DF-algebra'').

The FDA of the minimal $\mathcal{N}=2$, $D=7$ supergravity theory is particularly rich, since it includes, besides a triplet of gauge vectors $A^x$, also a $2$-form $B^{(2)}$, a $3$-form $B^{(3)}$ (related to $B^{(2)}$ by Hodge duality of the corresponding field-strengths), and a triplet of $4$-forms $A^{x|(4)}$ (related to $A^{x}$ by Hodge duality of the corresponding field-strengths). This theory can be obtained by dimensional reduction, on a four-dimensional compact manifold, preserving only half of the supersymmetries, from $D=11$ supergravity.
We will give the parametrization in terms of $1$-forms of the mutually non-local fields $B^{(2)}$ and $B^{(3)}$, obtaining, in this way, the corresponding superalgebra hidden in the supersymmetric $D=7$ FDA. 

We will also consider the dimensional reduction of the $D=11$ FDA to the $D=7$ one on an orbifold $T^4/Z_2$, showing the conditions under which the $D=7$ theory can be obtained by dimensional reduction of the eleven-dimensional one of Section \ref{11D} of Chapter \ref{chapter 2}.

We will see that, in general, one needs more than one nilpotent fermionic generator to construct the fully extended superalgebra hidden in the supersymmetric FDA (this will be the case, for example, of the minimal supersymmetric $D=7$ FDA, where we will find that \textit{two} nilpotent fermionic generators are required for the equivalence of the hidden superalgebra to the FDA).
Actually, as we have subsequently shown in \cite{Malg}, also the $D=11$ case admits the presence of (at least) two extra, nilpotent, fermionic generators, in the sense that the $1$-form $\eta$ dual to the nilpotent fermionic generator $Q'$ can be parted into two contributions, which close separately (that is, which have two different integrability conditions). We will recall in some detail this aspect in Section \ref{furthanal} of the current chapter. 

The main result of \cite{Hidden} has been to disclose the physical interpretation of the fermionic hidden generators: As we will recall and review in detail in the following, they have a cohomological meaning. 
In particular, to clarify the crucial role played by the nilpotent hidden fermionic generator(s), we will consider a singular limit, where the associated spinor $1$-form(s) goes to zero. In this limit, the unphysical degrees of freedom get mixed with the physical directions of the superspace, and all the generators of the hidden superalgebra act as generators of external diffeomorphisms.
On the contrary, in the presence of the spinor $1$-form(s), the hidden supergroup acquires a \textit{principal fiber bundle structure}, allowing to separate in a dynamical way the physical directions of superspace, generated by the supervielbein $\lbrace V^a,\Psi \rbrace$, from the non-physical ones, such that one recovers the gauge invariance of the FDA.

On the other hand, considering the bosonic hidden generators of the hidden algebra (we will call $H_b$ the corresponding tangent space directions of the hidden group manifold), we will show that they are associated with internal diffeomorphisms of the supersymmetric FDA in $D$ dimensions. More precisely, once a $p$-form $A^{(p)}$ of the FDA is parametrized in terms of the hidden $1$-forms, the contraction of $A^{(p)}$ along a generic tangent vector $\vec z\in H_b$ gives a $(p-1)$-form gauge parameter, and the Lie derivative of the FDA along a tangent vector $\vec z$ gives a gauge transformation leaving the FDA invariant.

This construction is not limited to the supergravity theories we have considered in \cite{Hidden}. In particular, the FDA of $D=10$ Type IIA supergravity, which naturally descends from the $D=11$ theory, includes a $2$-form field $B^{(2)}$ (also appearing in all superstring-related supergravities), which has a natural understanding in terms of the antisymmetric $3$-form $A^{(3)}$ of $D=11$ supergravity. The corresponding hidden bosonic $1$-form field, $B^a$, is associated with a charge $Z_a$ which carries a Lorentz index. In the fully extended hidden superalgebra in any space-time dimension $D\leq 10$, $P_a$ and $Z_a$ appear on the same footing, and the action of the hidden superalgebra includes automorphisms interchanging them. When some of the space-time directions are compactified on circles, these automorphisms are associated with T-duality transformations interchanging momentum with winding in the compact directions.

The structure described above appears to be strongly related to the one of the generalized geometry framework \cite{Hitchin:2004ut, Gualtieri:2003dx, Hull:2005hk, Dabholkar:2005ve, Grana:2008yw} and to its extensions to \textit{$M$-theory} \cite{Hull:2007zu, Pacheco:2008ps, Coimbra:2012af}, \textit{Double Field Theory} (DFT) \cite{Hull:2009mi, lust, Hull:2009zb, Hohm:2010jy, Hohm:2010pp, Hull:2014mxa}, and \textit{Exceptional Field Theory} (EFT) \cite{Hohm:2013pua, Hohm:2013uia, Hohm:2014qga}.
Thus, we expect that our approach and formalism could be useful in these contexts.

In Section \ref{furthanal} of this chapter, we will clarify the relations occurring among the $\osp(1|32)$ superalgebra, the $M$-algebra, and the DF-algebra.
In this context, we will also further discuss on the crucial role played by the $1$-form spinor $\eta$ for the $4$-form cohomology of the $D=11$ theory on superspace (basing our discussion on the work \cite{Malg}).

We will limit ourselves to consider FDAs and underlying superalgebras corresponding to the ground state of supergravity theories, that is the ``vacuum'' (defined by the condition that all the supercurvatures vanish). 
We will not consider the full dynamical content of the theories (out of the vacuum); some progress on this topic were obtained in \cite{Bandos:2005}.

Our notations, conventions, and some technical details, can be found in Appendix \ref{apphidden} (as well as in Ref. \cite{Hidden}).

\section{FDA gauge structure and $D=11$ supergravity}\label{CE}

As we have already said in Chapter \ref{chapter 2}, the action of $D=11$ supergravity was first constructed in \cite{Cremmer}. The theory has a field content given by the metric $g_{\mu\nu}$, a $3$-index antisymmetric tensor $A_{\mu\nu\rho}$ ($\mu,\nu,\rho,\ldots =0,1,\ldots ,D-1$), and a single Majorana gravitino $\Psi_\mu$.\footnote{We denote by capital case $\Psi$ the gravitino in eleven dimensions.}

An important task to accomplish in the context of $D=11$ supergravity was the identification of the supergroup underlying the theory.
The authors of \cite{Cremmer} proposed $\mathfrak{osp}(1 \vert 32)$ as the most likely candidate. However, the field $A_{\mu \nu \rho}$ cannot be interpreted as the potential of a generator in a supergroup.

The structure of this same theory was reconsidered in \cite{D'AuriaFre}, in the (supersymmetric) FDAs framework, using the superspace geometric approach (namely, in its dual Maurer-Cartan
formulation), and the existence of a superalgebra underlying the $D=11$ supergravity theory was presented for the first time (for details, see the review in Chapter \ref{chapter 2}).

The aim of this section is to analyze in detail the hidden gauge structure of the FDA describing $D=11$ supergravity, in the case in which the exterior $p$-forms are parametrized in terms of the hidden $1$-forms $B^{ab}$, $B^{a_1 \ldots  a_5}$, and $\eta$ (plus, as obvious, the supervielbein).
In particular, we will investigate the conditions under which the gauge invariance of the FDA is realized once $A^{(3)}$ is expressed in terms of hidden $1$-forms (see the theoretical background in Chapter \ref{chapter 2}).

Let us start by recalling and discussing some aspects of the construction of $D=11$ supergravity in the FDA framework, playing particular attention to its structure once the parametrization in terms of $1$-forms has been implemented.

%\subsection{$D=11$ FDA hidden structure}
In the case of $D=11$ supergravity, the first step of the construction outlined in Section \ref{revfda1} of Chapter \ref{chapter 2} is the introduction of the $H$-relative $4$-cocycle
\begin{equation}
\frac{1}{2}\bar{\Psi}\wedge \Gamma_{ab}\Psi \wedge V^a \wedge V^b,
\end{equation}
which allows to define the 3-form $A^{(3)}$ that appears in the FDA satisfying
\begin{equation}
 dA^{(3)}=\frac 12 \bar{\Psi}\wedge \Gamma_{ab}\Psi \wedge V^a \wedge V^b\,\label{coc1}
\end{equation}
(that is equation (\ref{FDA11a3}) of Chapter \ref{chapter 2}).
Including this $3$-form  $A^{(3)}$ in the basis of the relative cohomology of the supersymmetric FDA, we can perform a second step and construct a new cocycle of order seven (allowing the introduction of the $6$-form $B^{(6)}$), namely 
\begin{equation}
15 A^{(3)}\wedge  dA^{(3)} +\frac{\ii}{2}\bar{\Psi}\wedge \Gamma_{a_1    \ldots         a_5}\Psi \wedge V^{a_1} \wedge     \ldots         V^{a_5}, 
\end{equation}
satisfying
\begin{equation}
 dB^{(6)}=15 A^{(3)}\wedge  dA^{(3)} +\frac{\ii}{2}\bar{\Psi}\wedge \Gamma_{a_1    \ldots         a_5}\Psi \wedge V^{a_1} \wedge     \ldots         V^{a_5}\, \label{coc2}
\end{equation}
(that is equation (\ref{FDA11b6}) of Chapter \ref{chapter 2}).
The fact that the two cochains (\ref{coc1}) and (\ref{coc2}) are indeed cocycles is due to Fierz identities in $D=11$ (see Section \ref{fierz} of Appendix \ref{apphidden}). As we can see, the second step defined above requires to enlarge the CE-relative cohomolgy to include the $3$-form $A^{(3)}$.

We wish to remark that the inclusion of a new $p$-form, which is a gauge potential enjoying a gauge freedom, in the basis of the $H$-relative CE-cohomology of the FDA, is physically meaningful only if the whole of the FDA is \textit{gauge invariant}; this, in particular, requires that the non-physical degrees of freedom in $A^{(3)}$ and $B^{(6)}$ are projected out from the FDA.

Let us now consider the supersymmetric FDA describing $D=11$ supergravity parametrized in terms of $1$-forms (see Section \ref{11D} of Chapter \ref{chapter 2} for details). In this set up, the symmetry structure is based on the hidden supergroup manifold, that we call $G$, which extends the super-Poincar\'e Lie group to include the extra hidden directions associated with the higher $p$-forms.

The hidden supergroup $G$ has the structure of a \textit{principal fiber bundle} $(G/H,H)$, where $G/H$ corresponds to superspace; in this case, the fiber $H$ includes, besides the Lorentz transformations, also the hidden generators. More explicitly, let us now rewrite the hidden Lie superalgebra $\mathfrak{g}$ associated with $G$ as $\mathfrak{g}=H+K$, and decompose $H=H_0 +H_b + H_f$, so that the generators $T_\Lambda \in \mathfrak{g}$ are grouped in the following way: $\{J_{ab}\}\in H_0$, $\{Z_{ab}, \; Z_{a_1    \ldots         a_5}\}\in H_b$, $\{Q'\} \in H_f$, and $\{P_{a}, \; Q \}\in K$.\footnote{Here and in the following, with an abuse of notation, we will use for the cotangent space of the group manifold $G$, spanned by the $1$-forms $\sigma^\Lambda$, the same symbols defined above   for the tangent space of $G$, spanned by the vector fields $T_\Lambda$.} The subalgebra $H_b+H_f$ defines an abelian ideal of $\mathfrak{g}$.

The physical condition under which the CE-cohomology is restricted to the $H$-relative CE-cohomology corresponds now to request the FDA to be described in terms of $1$-form fields living on $G/H$. This implies that the hidden $1$-forms in $H_b$ and $H_f$ (related to the tangent space sectors $H_b$ and $H_f$, respectively), necessary for the parametrization of $A^{(3)}$ in terms of $1$-forms, do \textit{not} appear in $dA^{(3)}$ (see equation (\ref{coc1})). As we will see in a while, the presence of the spinor $1$-form $\eta$ is exactly what makes it possible to express $dA^{(3)}$ in terms of the relative cohomology only (that is to say, just in terms of the supervielbein).

\subsection{Analysis of the gauge transformations} \label{BRS}

Taking into account the hidden $D=11$ superalgebra reviewed in Chapter \ref{chapter 2},
we now consider in detail the relation between the FDA gauge transformations and those of its hidden supergroup $G$.

The $D=11$ supersymmetric FDA given in equations (\ref{FDA11omega})-(\ref{FDA11b6}) of Chapter \ref{chapter 2} is invariant under the gauge transformations:
\begin{equation}\label{gauge11}
\left\{
\begin{array}{l}
\delta A^{(3)}=d\Lambda^{(2)} \,, \\
\delta B^{(6)}=d\Lambda^{(5)} + \frac{15}2 \, \Lambda^{(2)}\wedge  \bar\Psi \wedge \Gamma_{ab}\Psi \wedge V^a\wedge  V^b ,
\end{array}\right.
\end{equation}
generated by the arbitrary forms $\Lambda^{(2)}$ and $\Lambda^{(5)}$ ($2$- and $5$-form, respectively).

On the other hand, the bosonic hidden $1$-forms in $H_b$ are abelian gauge fields, whose gauge transformations read
\begin{equation}\label{gauge}
\left\{
\begin{array}{l}
\delta_{b} B^{ab} = d\Lambda^{ab}\, ,\\
\delta_b B^{a_1    \ldots         a_5} = d \Lambda^{a_1    \ldots         a_5} \,,
\end{array} \right.
\end{equation}
being $\Lambda^{ab}$ and $\Lambda^{a_1 \ldots a_5}$ arbitrary Lorentz-valued scalar functions ($0$-forms).

Now, the requirement that $A^{(3)}$, parametrized in terms of $1$-forms, transforms as written in (\ref{gauge11}) under the gauge transformations (\ref{gauge}) of the hidden $1$-forms, implies the (bosonic) gauge transformation of $\eta$ to be
\begin{equation}\label{deltaeta}
\delta_b \eta = -E_2 \Lambda^{ab} \Gamma_{ab}\psi  - \ii E_3 \Lambda^{a_1    \ldots         a_5} \Gamma_{a_1    \ldots         a_5}\psi\,,
\end{equation}
consistently with the condition $D \delta \eta = \delta D \eta$.

In this case, it turns out that the corresponding $2$-form gauge parameter of $A^{(3)}$ is
\begin{eqnarray}
\Lambda^{(2)} & = & T_0 \Lambda_{ab}   V^a \wedge V^b + 3 T_1 \Lambda_{a b}  B^{b} _{\;c}\wedge B^{c a}+ \nonumber\\ & +& T_2( 2 \Lambda_{b_1 a_1    \ldots        a_4}  B^{b_1}_{\; b_2}\wedge B^{b_2 a_1    \ldots        a_4} -B_{b_1 a_1    \ldots        a_4} \Lambda^{b_1}_{\; b_2}\wedge B^{b_2 a_1    \ldots        a_4})+ \nonumber\\ & +& 2T_3 \epsilon_{a_1    \ldots        a_5 b_1    \ldots        b_5 m}\Lambda^{a_1    \ldots        a_5}\wedge B^{b_1    \ldots        b_5}\wedge V^m +\nonumber \\
& +& 3 T_4 \epsilon_{m_1    \ldots        m_6 n_1    \ldots        n_5}\Lambda^{m_1m_2m_3p_1p_2}\wedge B^{m_4m_5m_6p_1p_2}\wedge B^{n_1    \ldots        n_5} + \nonumber \\
& + &  S_2 \bar{\Psi}\wedge \Gamma_{ab}\eta   \Lambda^{ab}+ \ii S_3 \bar{\Psi}\wedge \Gamma_{a_1    \ldots        a_5}\eta  \Lambda^{a_1    \ldots        a_5}\,.\label{a3contr}
\end{eqnarray}

Then, considering also the gauge transformation of the spinor $1$-form $\eta$ generated by the tangent vector in $H_f$, overall we have
\begin{equation}
\delta \eta = D\epsilon ' +\delta_b \eta ,
\end{equation}
where we have introduced the infinitesimal spinor parameter $\epsilon'$. 

We can then write the $2$-form gauge parameter $\tilde\Lambda^{(2)}$ corresponding to the transformation in $H_f$ as follows:
\begin{equation}\label{a3contr'}
\tilde\Lambda^{(2)} =
-\ii S_1 \bar{\Psi}\wedge \Gamma_{a} \epsilon '  V^a - S_2 \bar{\Psi}\wedge \Gamma_{ab} \epsilon '  B^{ab}- \ii S_3 \bar{\Psi}\wedge \Gamma_{a_1    \ldots        a_5} \epsilon '  B^{a_1    \ldots        a_5}\,.
\end{equation}

In the sequel, we will show that all the diffeomorphisms in the hidden supergroup $G$, generated by Lie derivatives, are invariances of the FDA, the ones in the fiber $H$ directions being associated with a particular form of the gauge parameters of the FDA gauge transformations (\ref{gauge11}).

To this aim, let us first show that equation (\ref{a3contr}) can be rewritten in a rather simple way, using the contraction operator in the hidden Lie superalgebra $\mathfrak{g}$ associated with $G$.
Defining the tangent vector
\begin{equation}
{\vec z} \equiv \Lambda^{ab}Z_{ab}+\Lambda^{a_1    \ldots         a_5} Z_{a_1    \ldots         a_5}\in H_b\,,
\end{equation}
a gauge transformation leaving the $D=11$ FDA invariant is recovered, once $A^{(3)}$ is parametrized in terms of $1$-forms, if
\begin{equation}
\Lambda^{(2)}=  \imath_{\vec z}(A^{(3)})\,,\label{assumpta2}
\end{equation}
where with $\imath$ we have denoted the contraction operator.\footnote{This result is true as a consequence of the set of relations in (\ref{cond11}) of Appendix \ref{apphidden} obeyed by the coefficients of the $1$-forms parametrization of $A^{(3)}$ given in (\ref{a3par}) of Chapter \ref{chapter 2}, that is to say, under the same conditions required by supersymmetry for the consistency of the parametrization (\ref{a3par}).}
Introducing the Lie derivative $\ell_{\vec z}\equiv d \imath_{\vec z} + \imath_{\vec z} d $,
we find that the corresponding gauge transformation of $A^{(3)}$ is given by:
\begin{equation}\label{gauge2}
\delta A^{(3)}= d \left(\imath_{\vec z}(A^{(3)})\right) + \imath_{\vec z}\left(d A^{(3)}\right)= \ell_{\vec z} A^{(3)} = d\Lambda^{(2)}\, 
\end{equation}
(gauge invariance and diffeomorphisms coincides in the group Lie algebra). Then, the equality $d \left(\imath_{\vec z}(A^{(3)})\right) = \ell_{\vec z} A^{(3)}$ follows from the fact that $dA^{(3)}$, as given in equation (\ref{FDA11a3}) of Chapter \ref{chapter 2}, is invariant under transformations generated by $\vec z$ corresponding to the gauge invariance of the supervielbein.\footnote{This is in agreement with the fact that the right-hand side of $dA^{(3)}$ is in the $H$-relative CE-cohomology.}

Concerning the $6$-form $B^{(6)}$, to recover the general gauge transformation of $B^{(6)}$ in terms of the hidden algebra would require the knowledge of its explicit parametrization in terms of $1$-forms, which, at the moment, we ignore (some work is in progress on this topic). 
However, if we assume its behavior under gauge transformations to be analogous to the one of the $3$-form $A^{(3)}$, namely if we require
\begin{equation}
\Lambda^{(5)}=  \imath_{\vec z}(B^{(6)})\, \label{assumpta}
\end{equation}
(where $B^{(6)}$ is intended as parametrized in terms of $1$-forms), then, a simple computation gives
\begin{align}
\delta B^{(6)} & =  \ell_{\vec z} B^{(6)}= d \left(\imath_{\vec z}(B^{(6)})\right)+ \imath_{\vec z}\left(d B^{(6)}\right)=d \Lambda^{(5)}+\imath_{\vec z}\left(15 A^{(3)} \wedge dA^{(3)}\right) = \nonumber\\
& = d \Lambda^{(5)}+\ 15 \Lambda^{(2)} \wedge dA^{(3)}\,,\label{gauge5}
\end{align}
which indeed reproduces equation (\ref{gauge11}). The assumption (\ref{assumpta}) will be corroborated by the analogous computation we will do for the seven-dimensional model we will consider later. Indeed, in that case we can use, together with the explicit parametrization of $B^{(3)}$, also the one of its Hodge dual-related $B^{(2)}$ (appearing in the dimensional reduction of the eleven-dimensional $6$-form $B^{(6)}$), and, as we will see, the assumption (\ref{assumpta}) can be fully justified if we think of $B^{(2)}$ as a remnant of $B^{(6)}$ in the dimensional reduction. \footnote{Notice that the gauge transformations (\ref{gauge2}) and (\ref{gauge5}) are not fully general, since the corresponding gauge parameters are not (they are indeed restricted to the ones satisfying (\ref{assumpta2}) and (\ref{assumpta})).}

We still have to consider the gauge transformations generated by the other elements of $H$. Due to the fact that the Lorentz transformations, belonging to $H_0 \subset H$, are not effective on the FDA (all the higher $p$-forms being Lorentz invariant), this analysis reduces to consider the transformations induced by the tangent vector 
\begin{equation}
 \vec q \equiv \bar \epsilon' Q' \in H_f\, \subset H \,. \label{qvec}
\end{equation}
Thus, for the spinor $1$-form $\eta$ dual to $Q'$, we find
\begin{equation}
\delta_{\vec q}\eta=D\epsilon '=\ell_{\vec q}\eta ,
\end{equation}
and
\begin{align}\label{bellissimaalmondo}
\delta_{\vec q}A^{(3)}
&=  \ii S_1 \bar{\Psi}\wedge \Gamma_{a}D\epsilon' \wedge V^a + S_2\bar{\Psi}\wedge \Gamma_{ab}D\epsilon' \wedge  B^{ab} + \ii S_3\bar{\Psi}\wedge \Gamma_{a_1    \ldots         a_5}D\epsilon'  \wedge B^{a_1    \ldots         a_5} \, =  \nonumber \\
&=  d \left( \imath_{\vec q} ( A^{(3)} ) \right)=\ell_{\vec q}A^{(3)} ,
\end{align}
where in the second line, after integration by parts, we have used the following relation on the parameters $S_i$:
\begin{equation}
 S_1  +10 S_2-720S_3=0,
\end{equation}
which derives from $3$-gravitinos Fierz identities (see Appendix \ref{apphidden} for details). Note that, indeed, equation (\ref{bellissimaalmondo})
reproduces $D \tilde{\Lambda}$, in terms of $\tilde{\Lambda}$ defined in equation (\ref{a3contr'}).

\subsection{Role of the nilpotent fermionic generator $Q'$} \label{q'}

In deriving the gauge transformations leaving invariant the supersymmetric FDA, in terms of hidden $1$-forms, a crucial role is played by the spinor $1$-form $\eta$, dual to the nilpotent generator $Q'\in H_f$. Indeed, besides it is required for the equivalence of the hidden superalgebra to the FDA, it also
guarantees the gauge invariance of the supersymmetric FDA, due to its non-trivial gauge transformation (\ref{deltaeta}).

Actually, we can think of $\eta$ as a spinor $1$-form playing the role of an intertwining field between the base superspace and the fiber $H$ of the principal fiber bundle corresponding to the hidden supergroup manifold $G=(G/H,\, H)$.\footnote{This can also be understood from the covariant differential $D\eta$, given in equation (\ref{Deta}) of Chapter \ref{chapter 2}, which is parametrized not only in terms of the supervielbein (as it happens for all the fields of the FDA and for $DB^{ab}$ and $DB^{a_1 \ldots a_5}$, see the equations in (\ref{also}) of Chapter \ref{chapter 2}), but also in terms of the gauge fields in $H_b$, as we can see in equation (\ref{deta}).} The spinor $1$-form $\eta$ has a \textit{cohomological meaning}.

A clarifying example corresponds to considering the singular limit $\eta \rightarrow 0$, so that its dual generator $Q'$ can be dropped out from the hidden superalgebra. This limit can be obtained, in its simplest form, by redefining the coefficients appearing in the parametrization of $A^{(3)}$ (see Chapter \ref{chapter 2}) as
\begin{equation}
E_2 \to E'_2=\epsilon E_2\,,\quad  E_3 \to E'_3=\epsilon^2 E_3\,,
\end{equation}
and then taking the limit $\epsilon \to 0$. In this way, we find:
\begin{equation}
T_0\to \tilde T_0=\frac 16\,,\quad T_1 \to \tilde T_1= -\frac 1{90}\,,\quad T_2=T_3=T_4\to 0\,,\quad E_1=E_2=E_3\to 0\,,
\end{equation}
while $S_1,\; S_2, \; S_3 \to \infty$ in this limit. 

Now, recalling the parametrization (\ref{a3par}) of $A^{(3)}$, we can see that, setting $\eta=0$, the following finite limit can be obtained for $A^{(3)}$:
\begin{equation}\label{a3lim}
A^{(3)}  \to  A^{(3)}_{lim} = \tilde T_0 B_{ab} \wedge V^a \wedge V^b +\tilde T_1 B_{a b}\wedge B^{b}_{\;c}\wedge B^{c a}\,,
\end{equation}
so that its differential gives
\begin{equation}
dA^{(3)}_{lim} =\tilde T_0\left( \frac 1{2} \bar\Psi \wedge \Gamma_{ab}\Psi \wedge V^a \wedge V^b -   \ii   B_{ab} \wedge \bar\Psi \wedge \Gamma^{a}\Psi \wedge V^b\right) + \frac{3}{2} \tilde T_1 \bar\Psi \wedge \Gamma_{ab} \Psi\wedge B^{b}_{\;c}\wedge B^{c a}\,.
\end{equation}

We can now observe that the parametrization (\ref{a3lim}) does \textit{not} reproduce the sector of the FDA given in (\ref{FDA11a3}) of Chapter \ref{chapter 2}, being in fact obtained by a singular limit; however, this different FDA is based on the same hidden superalgebra $\mathfrak{g}$, where now the cocycles are in the $H_0$-relative CE-cohomology. Indeed, $dA^{(3)}_{lim}$ is now expanded on a basis of the \textit{enlarged superspace} $K_{enlarged}= K+H_b$, which includes, besides the supervielbein, also the bosonic hidden $1$-form $B^{ab}$. 

The case where all the $E_i$'s are proportional to the same power of $\epsilon$ can be done in an analogous way: It again requires $\eta =0$ and leads to an $A^{(3)}_{lim}$ with all $\tilde T_i\neq 0$, $i=0,1,    \ldots 4$. In this case, $dA^{(3)}_{lim}$ is expanded on a basis of $K_{enlarged}$ that also includes the $1$-form $B^{a_1 \ldots  a_5}$.

Let us mention that a singular limit of the parametrization of $A^{(3)}$ was already considered by the authors of \cite{Bandos:2004xw}.\footnote{More precisely, the singular limit considered in Ref. \cite{Bandos:2004xw} is given in terms of a parameter $s\to 0$ and the relation between their parameter and our ones is $\frac{120\rho -1}{90\left(60\rho-1\right)^2} =\frac{2(3+s)}{15s^2}$, being $E_3/E_2=\rho$.} In particular, in \cite{Bandos:2004xw} the authors were studying the description of the hidden superalgebra as an expansion of $OSp(1|32)$. They observed that a singular limit exists (which includes our limit as a special case), such that the authomorphism group of the FDA is enlarged from what we have called $H$ to $Sp(32)$, but where the trivialization of the FDA in terms of an explicit $A^{(3)}$, written in terms of $1$-forms, breaks down. From the analysis we have carried on above, we can see that, at least for the restriction of the limit considered here, what actually breaks down is the trivialization of the FDA on \textit{ordinary superspace}, while a trivialization on the enlarged superspace $K_{enlarged}$ is still possible.

Notice that, however, in this latter case, the gauge invariance of the new FDA requires that $B^{ab}$ (and, analogously, $B^{a_1 \ldots a_5}$) is \textit{not} a gauge field anymore. Correspondingly, $A^{(3)}_{lim}$ does \textit{not} enjoy gauge freedom, all of its degrees of freedom propagating in $K_{enlarged}$. $A^{(3)}_{lim}$ may then be interpreted as a gauge-fixed form of $A^{(3)}$. Indeed, it is precisely the gauge transformation of $\eta$, given in equation (\ref{deltaeta}), that guarantees the gauge transformation of $A^{(3)}$ to be the one written in (\ref{gauge11}); actually, this relies on the fact that $D\eta \in K_{enlarged}$.

Note that the transformation (\ref{gauge}), even if it is not a gauge transformation in this limit case, still generates a diffeomorphism leaving invariant the new FDA (that is indeed based on the same supergroup $G$), since
\begin{equation}
\delta_{\vec z} A^{(3)}_{lim}= \ell_{\vec z} A^{(3)}_{lim}\,. \end{equation}

One can now associate with the transformations generated by the tangent vector $\vec q$ (introduced in (\ref{qvec})), in the particular case $\delta_{\vec q}\eta = D\epsilon '= -\eta$, a gauge transformation bringing $A^{(3)}$ to $A^{(3)}_{lim}$ and, more generally, a gauge transformation such that $\eta' = \eta +\delta \eta =0$.

We can then conclude that the role of the extra, nilpotent, fermionic generator amounts to require the hidden $1$-forms of the Lie superalgebra to be true gauge fields living on the fiber $H$ of the associated principal fiber bundle $(G/H,\, H)$.\footnote{Let us observe that this is equivalent to require the construction of the FDA starting from the Lie algebra associated with the supergroup $G$ to be done using the $H$-relative CE-cohomology of the hidden superalgebra $\mathfrak{g}$.} In this sense, we can say that the behavior of the extra fermionic generator has some analogy with the one of a BRST ghost\footnote{BRST symmetry, from Becchi-Rouet-Stora-Tyutin, may be taken to be the basic invariance of quantum mechanics of a geometrical system. It contains (and extends) the concept of gauge invariance and, in its set up, the ``original fields'' get mixed with (non-physical) fields called ``ghosts''. For an interesting and physically clear introduction to BRST symmetry see, for example, Ref. \cite{BRSTlect}.}, since it guarantees that only the \textit{physical degrees of freedom} of the exterior forms appear in the supersymmetric FDA in a ``dynamical'' way. This amounts to say that, once the superspace is enlarged to $K_{enlarged}$, in the presence of $\eta$ (and, more generally, of a non empty $H_f$) no explicit constraint has to be imposed on the fields, since the non-physical degrees of freedom of the fields in $H_b$ and in $H_f$ transform into each other and do not contribute to the FDA.

Figure \ref{figurefib} offers a figurative tangent space description of the hidden gauge structure of the FDA (in the $D=11$ case), in which we can ``visualize'' what has been stated above about the role of the extra, nilpotent, fermionic generator $Q'$.

\begin{figure}
\centering
\pgfdeclareimage[height=8cm]{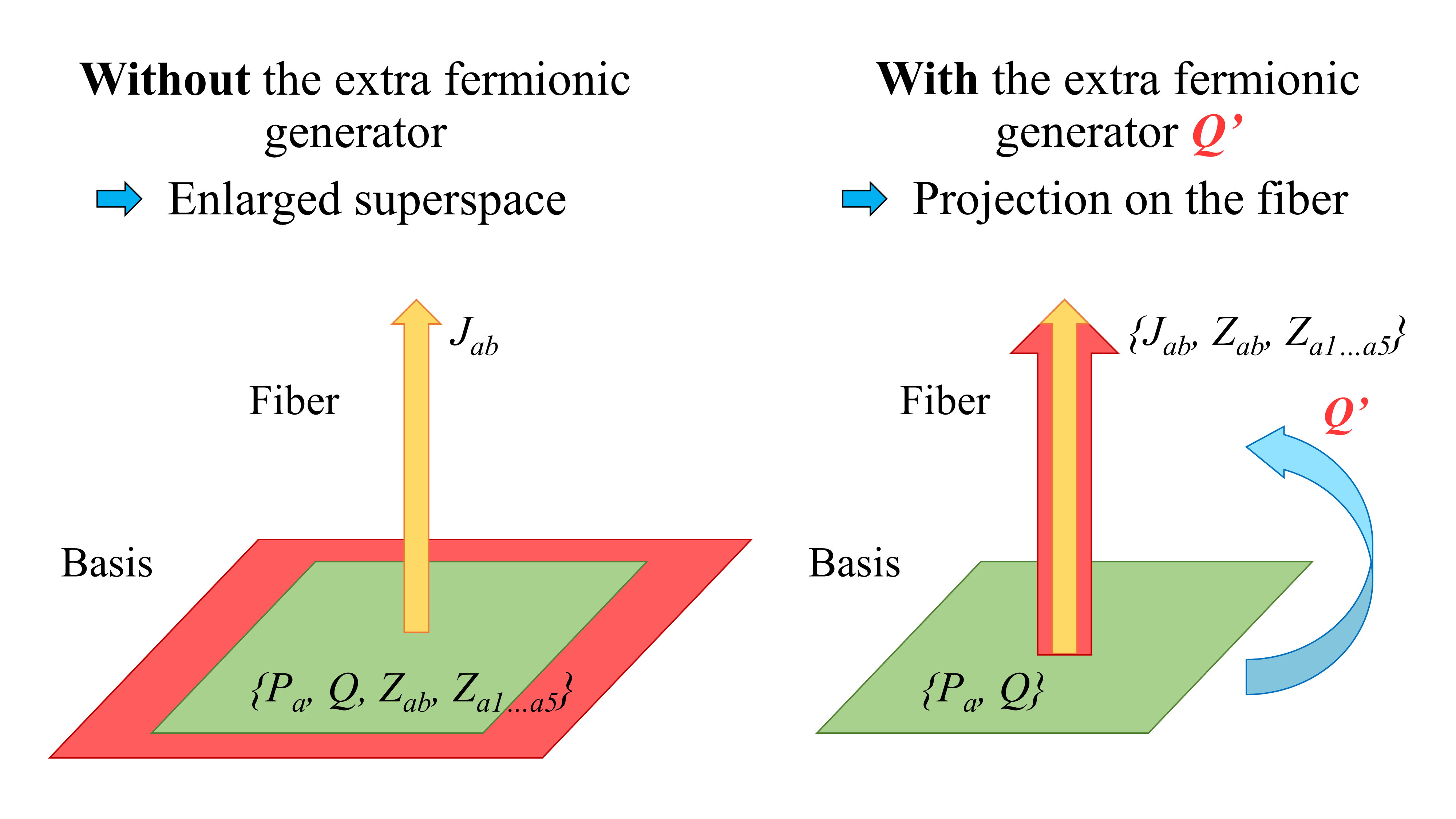}{fiber}
\pgfuseimage{fiber}
\caption[Tangent space description of the $D = 11$ case and role of $Q'$]{\textit{Tangent space description of the $D = 11$ case and role of $Q'$.} The presence of the nilpotent fermionic generator in the hidden superalgebra gives rise to a principal fiber bundle structure, and the physical directions of superspace are separated from the non-physical ones in a dynamical way, allowing to recover the gauge invariance of the FDA on ordinary superspace.}
\label{figurefib}
\end{figure}

\section{Hidden gauge algebra of the $\mathcal{N}=2$, $D=7$ FDA}\label{7D}

The same procedure explained in the eleven-dimensional case can be applied to lower-dimensional supergravity theories containing $p$-forms ($p>1$), in order to associate to any such theory a hidden Lie superalgebra containing, as a subalgebra, the super-Poincar\'e algebra.
In particular, since in the theory in eleven dimensions the closure of the supersymmetric FDA and of the corresponding hidden superalgebra are strictly related to $3$-gravitinos Fierz identities, the same must happen in any lower dimensions.

As an interesting example, in \cite{Hidden} we have considered the minimal $\mathcal{N}=2$, $D=7$ theory (not coupled to matter), in which the hidden structure turned out to be particularly rich since, in its most general form, it includes \textit{two} nilpotent fermionic generators.
In this section, we recall this $D=7$ example in details. 

Let us start by giving the physical content of the $D=7$ supergravity theory: It is given by the vielbein $1$-form $V^a$, a triplet of vectors $1$-forms $A^x$ ($x=1,2,3$), a $2$-form $B^{(2)}$, and a gravitino $1$-form Dirac spinor, which we describe as a couple of $8$-components spin-$3/2$ pseudo-Majorana fields $\psi_{A\mu}$ ($A=1,2$) satisfying the reality condition $\bar\psi^A = \epsilon^{AB} (\psi_B)^T$.\footnote{The charge conjugation matrix in $D=7$ can always be chosen to be $C=\mathbf{1}$.}
Here and in the following, we will denote by lower case $\psi$ the gravitino in seven dimensions, in order to avoid confusion with the capital case $\Psi$ denoting the gravitino field in eleven dimensions.

The interacting minimal $D=7$ theory was studied, at the Lagrangian level, in many works \cite{Minimal, bershoffo1, Townsend:1983kk, SalamSezgin}. In particular, in \cite{Townsend:1983kk} the authors observed that one can trade the $2$-form formulation of the theory by a formulation in terms of a $3$-form, $B^{(3)}$, the two being related by Hodge duality of the corresponding field-strengths on space-time and giving rise to different Lagrangians.
From our point of view, where we consider the FDA rather than a Lagrangian description, both the $2$ and the $3$-form are required for a fully general formulation, together with a triplet of $4$-forms, $A^{x\vert(4)}$, whose field-strengths are Hodge dual to the gauge vectors $A^x$.

One of the main reasons for choosing the minimal $D=7$ model is related to the fact that, in this case, we will be able to find an explicit parametrization in terms of $1$-forms of both $B^{(2)}$ and $B^{(3)}$. In this case, as we will see, a general parametrization requires the presence of \textit{two independent hidden spinor $1$-forms}. Since $B^{(2)}$ in $D=7$ can be obtained through dimensional reduction of $B^{(6)}$ in the $D=11$ FDA, this investigation also sheds some light on the extension of the hidden superalgebra of $D=11$ supergravity when also the parametrization of $B^{(6)}$, still unknown, would be considered (see Section \ref{11d7d} for further details on this aspect).

\subsection{$D=7$ FDA in terms of $1$-forms}

We now introduce the supersymmetric FDA on which the minimal $\mathcal{N}=2$, $D=7$ supergravity theory is based and we then write the parametrization of its $2$ and $3$-form in terms of $1$-forms, with the aim of finding the hidden superalgebra underlying this $D=7$ model \cite{Hidden}.

The $D=7$ FDA (in the vacuum) is the following:
\begin{align}
& R^{ab}\equiv d\omega^{ab}-\omega^{ac}\wedge \omega_{c}^{\; b}=0 \,,\label{fdaomega}\\
& T^a\equiv D V^a -\frac{\ii}{2}\bar{\psi}^A \wedge \Gamma^a \psi_A =0\,,\label{fdav} \\
& \rho _A \equiv D \psi_A =0\,,\label{fdapsi}\\
& F^x \equiv d {A}^x - \frac{\ii}{2} \sigma^{x \vert B}_{\;\;\;\;\;A} \bar{\psi}^A\wedge \psi_B=0\,,\label{fdaa1} \\
& F^{(3)}\equiv  d {B}^{(2)}+d {A}^x\wedge {A}^x - \frac{\ii}{2} \bar{\psi}^A \wedge \Gamma_a \psi_A \wedge V^a =0\, ,\label{fdab2} \\
& G^{(4)}\equiv d {B}^{(3)}-\frac{1}{2} \bar{\psi}^A \wedge \Gamma_{ab}\psi_A \wedge V^a \wedge V^b=0\, , \label{fdab3}\\
& F^{x(4)}\equiv d{A}^{x\vert (4)}  + \frac{1}{2}\left(d{A}^x \wedge  {B}^{(3)} - {A}^x \wedge d {B}^{(3)}\right) + \nonumber \\
& -\frac{1}{6} \sigma^{x \vert B}_{\;\;\;\;\;A} \bar{\psi}^A \wedge \Gamma_{abc}\psi_B \wedge V^a \wedge V^b \wedge V^c =0\,, \label{fdaa4}
\end{align}
where now $D$ denotes the Lorentz-covariant differential in seven dimensions ($D=d-\omega$, according with the convention adopted in \cite{Hidden,Malg,D'AuriaFre}) and $\sigma^{x \vert B}_{\;\;\;\;\;A}$ are the usual Pauli matrices. The $d^2$-closure of this FDA relies on the Fierz identities relating $3$- and $4$-gravitinos currents in $D=7$ (see Appendix \ref{apphidden}).

In order to find the hidden superalgebra underlying the theory, let us introduce the following set of bosonic Lorentz-indexed $1$-forms: $B^a$, which is associated with $B^{(2)}$, $B^{ab}$, associated with $B^{(3)}$, and $A^{x \vert abc}$, that is associated with $A^{x \vert (4)}$; we require their Maurer-Cartan equations to be the following ones:
\begin{align} \label{d1form}
& D B^{ab} = \alpha \bar{\psi}^A \wedge \Gamma^{ab}\psi_A ,  \nonumber\\
& D B^a = \beta \bar{\psi}^A \wedge \Gamma^a \psi_A , \nonumber\\
& D {A}^{x\vert abc}= \gamma \sigma^{x \vert B}_{\;\;\;\;\;A} \bar{\psi}^A \wedge \Gamma^{abc}\psi_B \,,
\end{align}
whose integrability conditions are automatically satisfied since $R^{ab}=0$.
The arbitrary choice of the coefficients in the right-hand side fixes the normalization of the bosonic $1$-forms $B^a$, $B^{ab}$, and ${A}^{x\vert abc}$. We choose $\alpha=\frac 12$, $\beta=\frac{\ii}2$, and $\gamma=\frac{1}{6}$.

The bosonic $2$- and $3$-form ($B^{(2)}$ and $B^{(3)}$, respectively) will be parametrized in terms of the $1$-forms $V^a$, $A^x$ (both of them already present in the FDA), $B^a$, $B^{ab}$, $A^{x\vert abc}$ (new, hidden, bosonic $1$-forms). Furthermore, as we are going to show, the consistency of their parametrizations also requires the presence of \textit{two nilpotent fermionic $1$-forms}, $\eta_A$ in the parametrization of $B^{(2)}$ and $\xi_A$ in that of $B^{(3)}$, whose covariant derivatives respectively satisfy:
\begin{align}
& D \eta_A = l_1 \Gamma_a \psi_A \wedge V^a + l_2  \Gamma_a \psi_A \wedge B^a + l_3 \Gamma_{ab}\psi_A \wedge B^{ab} + \nonumber \\
& + l_4 \psi_B \sigma^{x \vert B}_{\;\;\;\;\;A} \wedge {A}^x + l_5 \Gamma_{abc}\psi_B \sigma^{x \vert B}_{\;\;\;\;\;A} \wedge {A}^{x\vert abc},\label{deta7}\\
& D \xi_A = e_1 \Gamma_a \psi_A \wedge V^a + e_2 \Gamma_a \psi_A \wedge B^a + e_3 \Gamma_{ab}\psi_A \wedge B^{ab} + \nonumber \\
& + e_4 \psi_B \sigma^{x \vert B}_{\;\;\;\;\;A} \wedge {A}^x + e_5 \Gamma_{abc}\psi_B \sigma^{x \vert B}_{\;\;\;\;\;A} \wedge {A}^{x\vert abc}\,,\label{dxi}
\end{align}
where the $l_i$'s and $e_i$'s will be structure constants of the hidden superalgebra and they are constrained to satisfy (from the integrability of $D\eta_A$ and $D\xi_A$ and the use of the Fierz identities in $D=7$) the following equations:
\begin{align}
& - \ii l_1 -  \ii l_2  + 6 l_3  - \ii l_4 -10 l_5  = 0, \label{inteta} \\
& - \ii e_1-\ii e_2 +6 e_3  - \ii e_4 -10 e_5  =0 .\label{intxi}
\end{align}

The consistency of the parametrizations of $B^{(2)}$ and $B^{(3)}$ amounts to require that their differential, as given in equations (\ref{fdab2}) and (\ref{fdab3}), respectively, must be reproduced by the differential of their parametrizations.
This is analogous to what happens in $D=11$; in that case, however, only the parametrization of the $3$-form was considered, and its closure required (besides precise values of the coefficients) the presence of just \textit{one} spinor $1$-form.

Let us now explicitly write the general ansatz for the parametrization of $B^{(2)}$ and $B^{(3)}$ in terms of the set of $1$-forms $\{V^a,\psi_A,B^a,B^{ab},A^x,{A}^{x \vert abc},\xi_A,\eta_A\} $.\footnote{We should, in principle, also consider the parametrization of the $4$-form $A^{x \vert (4)}$. However, this deserves further investigation (work in progress).} It reads:
\begin{align}
 {B}^{(2)} =& \sigma B_a \wedge V^a + \tau \bar{\psi}^A \wedge \eta_A ,\label{B2par}
\\
 {B}^{(3)} = & \; \tau_0 B_{ab}\wedge V^a \wedge V^b + \tau_1  B_{ab}\wedge B^a \wedge V^b + \tau_2  B_{ab}\wedge B^a \wedge B^b +  \tau_3 B_{ab}\wedge B^{bc}\wedge B_c^{\; a}+ \nonumber \\
& +   \epsilon_{ab_1    \ldots        b_3c_1    \ldots        c_3}( \tau_4\,V^a + \tau_5 \,B^a)\wedge {A}^{x \vert b_1    \ldots        b_3} \wedge {A}^x_{\; c_1    \ldots        c_3}+ \nonumber \\
& + \tau_6  B_{ab}\wedge {A}^{x\vert a} _{\;\;\;\;\; cd}\wedge {A}^{x\vert bcd}+ \tau_7  \epsilon_{xyz} {A}^x \wedge A^y_{\; abc}\wedge {A}^{z\vert abc} + \nonumber \\
& + \tau_8 \epsilon_{xyz} {A}^x \wedge A^y \wedge A^z +  \tau_9 \epsilon_{xyz} \epsilon_{abcdlmn}A^{x\vert abc}\wedge A^{y \vert dlp}\wedge A^{z \vert mn}_{\;\;\;\;\;\;\;\;p}+ \nonumber \\
& + \sigma_1 \bar{\psi}^A \wedge \Gamma_a \xi_A \wedge V^a + \sigma_2  \bar{\psi}^A \wedge \Gamma_a \xi_A \wedge B^a +  \sigma_3 \bar{\psi}^A  \wedge \Gamma_{ab}\xi_A \wedge B^{ab}+ \nonumber \\
& +\sigma_4 \bar{\psi}^A \wedge \xi_B \sigma^{x \vert B}_{\;\;\;\;\;A} \wedge {A}^x +  \sigma_5 \bar{\psi}^A \wedge \Gamma_{abc}\xi_B \sigma^{x \vert B}_{\;\;\;\;\;A} \wedge {A}^{x \vert  abc}\,.\label{B3par}
\end{align}

The sets of coefficients $\{\tau_j\}$ and $\{\sigma_i\}$ are determined by requiring that the parametrizations (\ref{B2par}) and (\ref{B3par}) satisfy the equations (\ref{fdab2}) and (\ref{fdab3}) of the FDA. The reader can find their explicit (and rather long) expression in the appendices of Ref. \cite{Hidden}. 
In \cite{Hidden} we have fixed the normalization of the spinor $1$-forms $\xi_A$ and $\eta_A$ in order to obtain a simple expression. In particular, we have chosen the normalization of $\eta_A$ by imposing, in the parametrization of $B^{(2)}$, $\tau=1$.
In this way, we have obtained $\frac{e_2}{\sigma_2}=\frac{e_5}{\sigma_5}\equiv H$, where, with the normalization chosen for the bosonic $1$-forms, we have set
\begin{equation}
H=-2\left(e_1+e_2 -2\ii e_3\right)\left(e_1+e_2 -2\ii e_5\right)\,.
\end{equation}
After that, we have chosen $H=1$, which is a valid normalization in all cases where $H\neq 0$, that is to say, for $e_1+e_2 \neq 2\ii e_3$ or $e_1+e_2 \neq 2\ii e_5$. 
Actually, by looking at the general solution for the parameters given in the appendices of \cite{Hidden}, one can see that to choose $\tau \neq 0$ and $H \neq 0$ are not restrictive assumptions, since the cases $\tau = 0$ and/or $H=0$ would correspond to singular limits where the gauge structure of the supersymmetric FDA breaks down (this is strictly analogous to what we have discussed for the $D=11$ case, as far as the gauge structure of the theory is concerned). With the above normalizations, in \cite{Hidden} we have obtained:
\begin{eqnarray}
&&\sigma = 2 \ii l_2, \;\;\; l_1 = \frac{\ii}{2}\left(-1+2 \ii  l_2 \right), \;\;\; l_4 = \frac{\ii}{2}, \;\;\; l_3=l_5=0, 
\end{eqnarray}
\begin{eqnarray}
&&\tau_0= 2 \left[\ii e_1 (e_3-e_5)+\left(\frac{\ii}{2}e_2+e_3\right)\left(\frac{\ii}{2}e_2+e_5\right)\right], \nonumber\\
&&\tau_1= -4\ii e_2(\ii e_2 + 2 e_5)\,,\quad \tau_2= -2 e_2^2\,,\quad \tau_3= -\frac 83 e_3(e_3-2e_5), \nonumber\\
&& \tau_4 = e_5(\ii e_2 +2 e_5)\,,\quad \tau_5 =-\ii e_2e_5\,,\quad \tau_6= 36 e_5^2, \nonumber\\
&&\tau_7 = -12 e_5^2 \,,\quad \tau_8= \frac 23 e_4[e_1+e_2-6\ii(e_3+e_5)]\,,\quad \tau_9= -3e_5^2, \nonumber\\
&&\sigma_1= -e_1-2e_2+4\ii e_5\,,\quad \sigma_2= e_2\,,\quad \sigma_3= -e_3+2e_5\,, \nonumber \\
&& \sigma_4= -e_4\,, \quad \sigma_5=e_5\label{taus}\,,
\end{eqnarray}
where the $e_i$'s are constrained by equation (\ref{intxi}).

Let us observe that the combination $\tau_4\,V^a + \tau_5 \,B^a\equiv \tilde B^a$ could be used, instead of $B^a$, in the parametrization of $B^{(3)}$; this redefinition  simplifies the expression of $B^{(3)}$ and, in particular, the term $B_{ab}\wedge \tilde B^a\wedge V^b $ vanishes.

\subsection{The hidden superalgebra in $D=7$}\label{hidden}

We now write, analogously to what was done in \cite{D'AuriaFre} in the $D=11$ case, the $D=7$ hidden superalgebra in terms of the generators $T_\Lambda$ dual to the set of $1$-forms $\sigma^\Lambda$ of the $D=7$ theory. In the present case, the set of $1$-forms is given by
\begin{equation}
\sigma^\Lambda=\{V^a,\psi_A,\omega^{ab},A^x, B^a,B^{ab},{A}^{x \vert abc},\xi_A,\eta_A\}, \label{sigma7d}
\end{equation}
and the set of (dual) generators
\begin{equation}
T_\Lambda= \{P_a,Q_A,J_{ab},T^x , Z_a, Z_{ab},T^{x\vert}_{\;\;\;abc},Q'_A,Q''_A\}\,.\label{t7d}
\end{equation}
The mappings between the $1$-forms and the generators can be found in \cite{Hidden}.
The (anti)commutators of the superalgebra (besides those of the Poincar\'e Lie algebra) can then be written as:

\begin{align}
& \lbrace Q^A, \bar{Q}_B \rbrace = - \ii \Gamma^a \left( P_a+Z_a \right) \delta^A_{\; B} - \frac 12 \Gamma^{ab}Z_{ab}\delta^A_{\; B}-\sigma^{x \vert A}_{\;\;\;\;\; B}\left(\ii T^x +\frac 1{18} \Gamma^{abc}T^{x\vert}_{\;\;\; abc}\right)\,, \\
& [Q_A, P_a]= -2 \Gamma_a (e_1   Q'_A+l_1Q''_A)\,, \\
& [Q_A, Z_a]= -2 \Gamma_a (e_2Q'_A+l_2Q''_A)\,, \\
& [Q_A, Z_{ab}]= -4  e_3 \Gamma_{ab} Q'_A\,,\\
& [Q_A, T^x]= - 2 \sigma^{x \vert B}_{\;\;\;\;\;A}(e_4Q'_B+ l_4 Q''_B) \,, \\
& [Q_A, T^{x \vert}_{\;\;\; abc}]=-12 e_5 \Gamma_{abc}\sigma^{x \vert B}_{\;\;\;\;\;A}Q'_B\,,\\
& [J_{ab}, Z_{c}]=-2 \delta^{c}_{[a}Z_{b]}\,,\\
& [J_{ab}, Z_{cd}]=-4\delta^{[c}_{[a}Z^{d]}_{b]}\,,\\
&[J_{ab}, T^{x | c_1c_2 c_3}]=-12 \delta^{[c_1}_{[a}T^{x \vert \;c_2  c_3]}_{\;\;\;\;b]}\,,\\
& [J_{ab}, Q_A]=- \Gamma_{ab} Q_A\,,\\
& [J_{ab}, Q'_A]=- \Gamma_{ab} Q'_A\,,\\
& [J_{ab}, Q''_A]=- \Gamma_{ab} Q''_A\,.
\end{align}
All the other possible (anti)commutators (except, obviously, the Poincar\'e part) vanish.
This hidden superalgebra includes all the $1$-forms associated with the $D=7$ FDA once the latter is extended to include all the couples of Hodge dual field-strengths. 
 
We can see that two independent fermionic generators are necessary if we want to include in the hidden superalgebra description of the FDA involving both $B^{(2)}$ and $B^{(3)}$ also the $1$-forms $B^{ab}$ and $A^{x\vert abc}$. We did not consider in the above description the $1$-form $B^{a_1 \ldots a_5}$ (associated with the non-dynamical volume form $F^{(7)}= dB^{(6)}+  \ldots  $). 

\subsubsection{Lagrangian subalgebras}

Let us now consider and discuss \textit{two relevant subalgebras} of the general hidden superalgebra presented above, where only one nilpotent fermionic generator appears. In \cite{Hidden} we called them ``electric hidden subalgebras'' or ``Lagrangian subalgebras'', because of their role in the construction of the Lagrangian, as we will clarify in the following.

The first subalgebra is the one where $Q'_A=Q''_A \equiv \frac 12 \hat Q_A$. This corresponds to consider a FDA including both $B^{(2)}$ and  $B^{(3)}$. However, in this case, the corresponding spinor $1$-form appears in the parametrizations (\ref{B2par}) and (\ref{B3par}), that is to say $\eta_A=\xi_A$, and the Maurer-Cartan equations (\ref{deta7}) and (\ref{dxi}) coincide, implying $\{e_i\}= \{l_i\}$. In particular, we have $e_3=e_5=0$, since $l_3=l_5=0$. This implies, on the set of $\{\tau_j\}$ given in (\ref{taus}), that all the contributions in $B^{ab}$ and $A^{x \vert abc}$ in the parametrization of $B^{(3)}$ disappear; then, the corresponding generators $Z_{ab}$ and $T^x_{abc}$ decouple and can be set to zero. The resulting subalgebra is given by
 \begin{align}
& \lbrace Q^A, \bar{Q}_B \rbrace = - \ii \Gamma^a \left( P_a+Z_a \right) \delta^A_{\; B} -\ii \sigma^{x \vert A}_{\;\;\;\;\; B}  T^x , \\
& [Q_A, P_a]= -2 e_1\Gamma_a    \hat Q_A, \\
& [Q_A, Z_a]= -2 e_2\Gamma_a \hat Q_A, \\
& [Q_A, T^x]= - 2 e_4\sigma^{x \vert B}_{\;\;\;\;\;A}\hat Q_B \,.
\end{align}
Let us observe that the same subalgebra can be obtained by truncating the hidden superalgebra to the subalgebra where $Q'_A \to 0$ or, equivalently, $\xi_A\to 0$. However, recalling the discussion about the role of the nilpotent fermionic generators for the consistency of the gauge structure of the FDA (referring, in particular, to the singular limit $\eta\to 0$)\footnote{Actually, the discussion we have done concerned the $D=11$ theory. Analogous considerations can be worked out for the $D=7$ case, as we will show in a while.}, from the point of view of the FDA this  corresponds to consider, instead, the ``sub-FDA'' where only $A^x$ and $B^{(2)}$ appear, but not their mutually non-local forms $B^{(3)}$ and $A^{x|(4)}$, respectively. This is the appropriate framework for a Lagrangian description in terms of $B^{(2)}$ (considered, for example, in Refs. \cite{bershoffo1} and \cite{SalamSezgin}).

The other Lagrangian subalgebra is found by setting, instead, $Q''_A\to 0$, that implies the vanishing of the coefficients $\{l_i\}$'s. In this case, the whole parametrization of $B^{(2)}$ drops out. This subalgebra thus corresponds to consider the restricted FDA where $B^{(2)}$ is excluded, together with $A^{x|(4)}$. This is the appropriate framework for the construction of the Lagrangian in terms of $B^{(3)}$ only (see Ref. \cite{Townsend:1983kk}). The $1$-forms $B^a$ and $A^{x \vert abc}$ could still be included in the parametrization of $B^{(3)}$ as kind of trivial deformations, and they can be consistently decoupled by setting $e_2=e_5=0$.

Observe that both Lagrangian subalgebras require the truncation of the superalgebra to only \textit{one} (out of the two) nilpotent fermionic generator.

The complete hidden superalgebra is larger than the one just involving the fields appearing in the Lagrangian in terms of either $B^{(2)}$ or $B^{(3)}$ only. This is reminiscent of an aspect of $D=4$ extended supersymmetric theories, namely of the fact that a central extension of the supersymmetry algebra is associated with electric and magnetic charges \cite{Witten:1978mh}, while the electric subalgebra only involves electric charges whose gauge potentials appear in the Lagrangian description.

\subsection{Including $B^{a_1 \ldots a_5}$ in the $D=7$ theory}

One could ask whether the inclusion of the extra contributions involving the $1$-form $B^{a_1 \ldots a_5}$ in the parametrization of $B^{(2)}$ and $B^{(3)}$ could significantly alter the results we have previously obtained, and if this would require the presence of extra spinorial charges. We discuss this issue in the following, on the same lines of what we have done in \cite{Hidden}.
This analysis will also turn out to be useful once we will relate, in the next section, the $D=7$ theory to the $D=11$ one.

Thus, for completing the analysis of the minimal theory in $D=7$, let us include the (non-dynamical) form $B^{(6)}$ associated with the volume form in $D=7$ in the FDA, and investigate the superalgebra hidden in this extension of the FDA.
It contributes to the FDA as
\begin{equation}
 dB^{(6)} - 15 B^{(3)}\wedge dB^{(3)} = \frac \ii 2 \bar\psi^A \wedge \Gamma_{a_1    \ldots         a_5} \psi_A \wedge V^{a_1}     \ldots         \wedge V^{a_5}\,
\end{equation}
(we will treat the dimensional reduction of the $D=11$ $6$-form in Section \ref{11d7d} and we will see that this is evident).

We require the covariant derivatives of the spinor $1$-forms to be now:
\begin{align}
& D \xi_A = e_1 \Gamma_a \psi_A \wedge V^a + e_2 \Gamma_a \psi_A \wedge B^a + e_3 \Gamma_{ab}\psi_A \wedge B^{ab} + \nonumber \\
& + e_4 \psi_B \sigma^{x \vert B}_{\;\;\;\;\;A} \wedge {A}^x + e_5 \Gamma_{abc}\psi_B \sigma^{x \vert B}_{\;\;\;\;\;A} \wedge {A}^{x\vert abc} + e_6 \Gamma_{a_1    \ldots        a_5}\psi_A B^{a_1    \ldots        a_5}\,,\label{include1}\\
& D \eta_A = l_1 \Gamma_a \psi_A \wedge V^a + l_2 \Gamma_a \psi_A \wedge B^a + l_3 \Gamma_{ab}\psi_A \wedge B^{ab} + \nonumber \\
& +l_4 \psi_B \sigma^{x \vert B}_{\;\;\;\;\;A} \wedge {A}^x + l_5 \Gamma_{abc}\psi_B \sigma^{x \vert B}_{\;\;\;\;\;A} \wedge {A}^{x\vert abc} + l_6 \Gamma_{a_1    \ldots        a_5}\psi_A B^{a_1    \ldots        a_5}\label{include2}\, ,
\end{align}
and, besides the equations in (\ref{d1form}), we also define (in an analogous way):
\begin{align}
& D B^{a_1    \ldots        a_5} = \frac{\ii}{2} \bar{\psi}^A \wedge \Gamma^{a_1    \ldots        a_5} \psi_A, \,.
\end{align}
The integrability conditions of (\ref{include1}) and (\ref{include2}) give the following equations:
\begin{align}
& - \ii l_1 -  \ii l_2  + 6 l_3  - \ii l_4 -10 l_5 -\ii 360 l_6 = 0, \label{inteta1} \\
& - \ii e_1-\ii e_2 +6 e_3  - \ii e_4 -10 e_5 -\ii 360 e_6  =0 .\label{intxi1}
\end{align}
The new parametrizations for the $2$- and $3$-form $B^{(2)}$ and $B^{(3)}$ are
\begin{align}
&{B}^{(2)} = B^{(2)}_{old} + \chi \epsilon_{a_1    \ldots        a_5 a b} B^{a_1    \ldots        a_5}\wedge B^{ab} ,\\
&{B}^{(3)} = \; B^{(3)}_{old} + \tau_{10}B_{a a_1    \ldots        a_4}\wedge B^a_{\; b}\wedge B^{b a_1    \ldots        a_4}  + \tau_{11} \epsilon_{a_1    \ldots        a_5 a b}B^{a_1    \ldots        a_5}\wedge V^a \wedge V^b + \nonumber \\
& + \tau_{12} \epsilon_{a_1    \ldots        a_5 a b}B^{a_1    \ldots        a_5}\wedge B^a \wedge V^b  + \tau_{13} \epsilon_{a_1    \ldots        a_5 a b}B^{a_1    \ldots        a_5}\wedge B^a \wedge B^b + \nonumber \\
& + \tau_{14} \epsilon_{a_1    \ldots        a_5 a b}B^{a_1    \ldots        a_5}\wedge {A}^{x \vert a}_{\;\;\;\;\; cd}\wedge {A}^{x\vert bcd} +
  + \sigma_6 \bar{\psi}^A \wedge \Gamma_{a_1    \ldots        a_5} \xi_A \wedge B^{a_1    \ldots        a_5}\, ,
\end{align}
where $ B^{(2)}_{old}$ and $ B^{(3)}_{old}$ are given by equations (\ref{B2par}) and (\ref{B3par}), respectively.
The values of the new set of coefficients can be found in the appendices of Ref. \cite{Hidden}. In the following, we directly move to the result: The parametrization of the extended forms in terms of $1$-forms is more complicated in this case, but the closure of the hidden superalgebra does \textit{not} require any new spinor $1$-form besides $\xi_A$ and $\eta_A$.

In order to express the superalgebra in the dual form, it is now sufficient to introduce the bosonic generator $Z_{a_1    \ldots         a_5}$, dual to $B^{a_1 \ldots a_5}$ (that is, satisfying $B^{a_1 \ldots a_5} (Z_{b_1 \ldots b_5})= 5! \delta^{a_1 \ldots a_5}_{b_1 \ldots b_5} $). The interested reader can find the complete form of the hidden superalgebra including $Z_{a_1    \ldots         a_5}$ in \cite{Hidden}.

We now move to the analysis of the hidden gauge structure of the $D=7$ theory.

\subsection{Gauge structure of the minimal $D=7$ FDA}\label{BRS7}

Analogously to what we have previously done in the case of the $D=11$ theory, we now analyze the gauge structure of the $D=7$ FDA.
We limit ourselves just to a short discussion of it, since the main results concerning the role of the nilpotent charges (dual to the spinor $1$-forms $\eta_A$ and $\xi_A$) is completely analogous to the one discussed for $\eta$ in the eleven-dimensional case.

The supersymmetric $D=7$ FDA is invariant under the following gauge transformations:
\begin{eqnarray} \label{gaugex}
\left\{ \begin{array}{l}
\delta A^x = d \Lambda^x \,, \\
\delta B^{(2)}= d\Lambda^{(1)}-\Lambda^x dA^x\, , \\
\delta B^{(3)} =d \Lambda^{(2)} \,, \\
\delta A^{x \vert (4)} =d \Lambda^{x \vert (3)} +\frac{1}{2}(\Lambda^x dB^{(3)}+\Lambda^{(2)}\wedge dA^x)\,,\\
\delta  B^{(6)} =d \Lambda^{(5)} +15 \Lambda^{(2)} \wedge dB^{(3)}\,.
\end{array} \right.
 \label{gauge7d}
\end{eqnarray}

Let us stress that, as for the $D=11$ case, the gauge transformations (\ref{gauge7d}) leaving invariant the $D=7$ FDA can be obtained, for particular $(p-1)$-form parameters, through Lie derivatives acting on the hidden symmetry supergroup $G$ underlying the theory.
$G$ has again the fiber bundle structure $G= H + K$, where $K=G/H$ is spanned by the supervielbein $\lbrace V^a,\psi_A \rbrace$. The (tangent space description of the) fiber $H= H_0+ H_b+H_f$ is again generated by the Lorentz generators in $H_0$ and by the gauge and hidden generators in $H_b$ and $H_f$, where now we have that the set $\{T^x,Z_a,Z_{ab},T^{x\vert}_{\;\;\;abc},Z_{a_1    \ldots a_5}\}$ spans $H_b$, while the set $\{\xi_A,\eta_A\}$ spans $H_f$.

Explicitly, we define the tangent vector in $H_b$ as follows:
\begin{equation}
\vec z \equiv \Lambda^x T^x + \Lambda^a Z_a + \Lambda^{ab} Z_{ab} +\Lambda^{x \vert abc} T^{x|}_{\;\;\; abc} + \Lambda^{a_1    \ldots         a_5} Z_{a_1    \ldots         a_5}\in H_b.
\end{equation}

Now, by a straightforward calculation, one gets that the gauge transformations of $A^x$, $B^{(2)}$, and $B^{(3)}$ in (\ref{gauge7d}) can be obtained by requiring
\begin{align}
& \delta A^x = \ell_{\vec z}A^x\,,\\
& \delta B^{(2)}= \ell_{\vec z} B^{(2)}\,,\\
& \delta B^{(3)}= \ell_{\vec z} B^{(3)} ,
\end{align}
for the following choice of the $(p-1)$-form gauge parameters:
\begin{align}
& \Lambda^x = \imath_{\vec z} A^x\,,\\
& \Lambda^{(1)}=\imath_{\vec z} B^{(2)}\,,\\
& \Lambda^{(2)}= \imath_{\vec z} B^{(3)}\,.
\end{align}

We expect that also for the forms $A^{x \vert (4)}$ and $B^{(6)}$, whose parametrizations in terms of $1$-forms remain still unknown (work in progress on this topic), the rest of the gauge transformations in (\ref{gauge7d}) leaving invariant the supersymmetric FDA should be
\begin{align}
& \delta A^{x|(4)}= \ell_{\vec z}A^{x|(4)}\,,\\
& \delta B^{(6)}= \ell_{\vec z} B^{(6)}\,,
\end{align}
for the following choice of  the $(p-1)$-form gauge parameters:
\begin{align}
& \Lambda^{x|(3)}= \imath_{\vec z} A^{x|(4)}\,,\\
& \Lambda^{(5)}= \imath_{\vec z} B^{(6)}\,.
\end{align}
The corresponding gauge transformations of the $1$-forms in $H_b$ read:
\begin{equation}\label{gauge7hb}
\left\{ \begin{array}{l}
\delta A^x = d \Lambda^x \,, \\
\delta B^a= d\Lambda^a\, , \\
\delta B^{ab} =d \Lambda^{ab} \,, \\
\delta A^{x \vert abc} =d \Lambda^{x \vert abc}\,,\\
\delta B^{a_1    \ldots         a_5} =d \Lambda^{a_1    \ldots         a_5}\,,
\end{array} \right.
\end{equation}
and the corresponding gauge transformations of the $1$-forms in $H_f$ are
\begin{equation}\label{gauge7hf}
\left\{ \begin{array}{l}
\delta \xi_A = D\epsilon'_A   - e_2 \Gamma_a \psi_A  \Lambda^a - e_3 \Gamma_{ab}\psi_A \Lambda^{ab} +  \\
 - e_4 \psi_B \sigma^{x \vert B}_{\;\;\;\;\;A} \Lambda^x - e_5 \Gamma_{abc}\psi_B \sigma^{x \vert B}_{\;\;\;\;\;A} \Lambda^{x\vert abc} - e_6 \Gamma_{a_1    \ldots        a_5}\psi_A \Lambda^{a_1    \ldots        a_5}\,,\\
\delta \eta_A = D\epsilon''_A - l_2 \Gamma_a \psi_A \Lambda^a - l_3 \Gamma_{ab}\psi_A \Lambda^{ab} +  \\
 -l_4 \psi_B \sigma^{x \vert B}_{\;\;\;\;\;A} \Lambda^x - l_5 \Gamma_{abc}\psi_B \sigma^{x \vert B}_{\;\;\;\;\;A} \Lambda^{x\vert abc} - l_6 \Gamma_{a_1    \ldots        a_5}\psi_A \Lambda^{a_1    \ldots        a_5} \,.
\end{array} \right.
\end{equation}
The parameters $\Lambda$'s appearing in (\ref{gauge7hb}) are arbitrary Lorentz (and/or $SU(2)$) valued $0$-forms, while $\epsilon'_A$ and $\epsilon''_A$ in (\ref{gauge7hf}) are arbitrary spinor parameters.

\section{Relation between $D=7$ and $D=11$ supergravities}\label{11d7d}

The hidden Lie superalgebra we have presented and discussed in Section \ref{7D} is the most general one for the minimal $\mathcal{N}=2$, $D=7$ supergravity theory.
As we said in \cite{Hidden}, we now expect that, for special choices of the parameters, the whole structure could be retrieved by dimensional reduction of the $D=11$ supergravity theory, in the case where four of the eleven-dimensional space-time directions belong to a four-dimensional compact manifold preserving one-half of the supercharges.

In \cite{Minimal} (the other work in which I dealt with $D=7$ supergravity), we have explicitly performed the dimensional reduction of $D=11$ supergravity on an orbifold $T^4/\mathbb{Z}_2$ to the minimal $D=7$ theory. There, we pointed out that the minimal $D=7$ theory can be obtained as a truncation of the dimensional reduction of $D=11$ supergravity on a torus $T^4 $ (that would give the maximal $D=7$ theory), where the $SO(4)=SO(3)_+\times SO(3)_-$ holonomy on the internal manifold is truncated to $SO(3)_+$, so that only the reduced fields which are $SO(3)_-$-singlets are retained.

From the point of view of the fermionic fields, the truncation selects only $16$ out of the $32$ components of the eleven-dimensional Majorana spinors, which result to be described by pseudo-Majorana spinors valued in the seven-dimensional $SU(2)= SO(3)_+$ $R$-symmetry. In particular, the eleven-dimensional gravitino $1$-form $\Psi$ becomes, in $D=7$,
\begin{equation}
\Psi \rightarrow \psi_A \,, \;\;\; \text{with} \;\; A=1,2\,.
\end{equation}

As far as the bosonic fields are concerned, let us now parametrize the Lie algebra of $SO(4)$ (the holonomy group of the internal manifold) in terms of the four-dimensional `t Hooft matrices $J^{x\,\pm}_{ij}$, where $x=1,2,3$ and $i,j,  \ldots     =1 \ldots         ,4$ (the reader can find their properties in Section \ref{tooooft} of Appendix \ref{apphidden}). 

The above-mentioned truncation corresponds to drop out the contributions that are proportional to $J^{x\,-}_{ij}\in SO(3)_-$ in the decomposition of the eleven-dimensional bosonic forms to $D=7$, so that
\begin{eqnarray}
A^{(3)} &\to& B^{(3)} + A^x \wedge J^{x\,+}_{ij} V^i\wedge V^j , \label{117a3}\\
B^{(6)} &\to& B^{(6)} + A^{x\vert(4)} \wedge J^{x\,+}_{ij} V^i\wedge V^j- 8 B^{(2)}\wedge \Omega^{(4)} , \label{117b6}
\end{eqnarray}
where $V^i$'s are the vielbein of the compact manifold and where $\Omega^{(4)}=\frac 1{4!}V^{i_1}\wedge  \ldots  \wedge V^{i_4}\epsilon_{i_1    \ldots         i_4}$ denotes its volume-form.

Then, let us consider the dimensional reduction of the Lorentz-valued $1$-forms $ \{B^{\hat a\hat b}, \; B^{\hat a_1     \ldots         \hat a_5}\}$ (defining the Lie superalgebra hidden in the $D=11$ FDA) to the minimal $D=7$ theory.
The comparison of the $D=11$ to the $D=7$ theory would require to consider the version of the seven-dimensional theory which also includes the $1$-form $B^{a_1 \ldots a_5}$, that, in seven dimensions, is associated with the (non-dynamical) volume-form $dB^{(6)}$.
Indeed, by a straightforward dimensional reduction, we obtain:
\begin{eqnarray}
B^{\hat a\hat b} &\to& \left\{
                                \begin{array}{ll}
                                 B^{ab} \\
                                  A^x \,J^{x\,+}_{ij}
                                \end{array}
                              \right. ,  \label{babi}\\
 B^{\hat a_1    \ldots         \hat a_5} &\to& \left\{
                                \begin{array}{ll}
                                 B^{a_1    \ldots         a_5} \\
                                 - \frac{3\ii }{2 } A^{x \vert abc} \,J^{x\,+}_{ij}   \\
                                 - B^a \epsilon_{i_1    \ldots         i_4}
                                \end{array}
                              \right. , \label{bab5}
\end{eqnarray}
being $\hat a=0,1,  \ldots  10$, $a=0,1, \ldots  6$, and $i=7, \ldots 10$. Let us observe that neglecting $B^{a_1    \ldots         a_5}$ would imply, for consistency, to drop out also all the other forms in (\ref{bab5}).

We are now going to compare the dimensional reduction of eleven-dimensional fields considering only the fields appearing in the parametrization of the $3$-form (due to the fact that the hidden superalgebra underlying the $D=11$ theory was obtained in \cite{D'AuriaFre} by parametrizing only the $3$-form $A^{(3)}$ in terms of $1$-forms, without considering the one of the $6$-form $B^{(6)}$).
Then, considering the fact that the $D=7$ field $B^{(2)}$ descends from the $D=11$ $6$-form $B^{(6)}$ (as we can see in equation (\ref{117b6})), the comparison of the two theories could shed some light on the parametrization of the form $B^{(6)}$ in the $D=11$ model. 
In particular, the analysis done in Section \ref{7D} shows that the full hidden superalgebra in $D=7$ also includes a second nilpotent fermionic generator dual to a spin-$3/2$ field appearing in the parametrization of $B^{(2)}$ (see equation (\ref{B2par})); since $B^{(2)}$ is a descendent of $B^{(6)}$, this could suggest that considering also the parametrization of $B^{(6)}$ in the $D=11$ case would amount to include one extra, nilpotent, fermionic generator. An explicit verification of this conjecture is left to future investigations.

The set of relations we found in \cite{Hidden} between the $D=7$ and $D=11$ structure constants are the following ones:
\begin{align} \label{rele117}
& e_1 = \ii E_1, \;\;\; e_2 = 120 \ii E_3, \;\;\; e_3 =E_2, \nonumber \\
& e_4 = 4 \ii E_2, \;\;\; e_5 =120E_3 , \;\;\; e_6 = \ii E_3 . 
\end{align}
The corresponding relation between the coefficients in the parametrizations of the $3$-form are:
\begin{align}\label{t117}
& \tau_0 = T_0, \;\;\; \tau_1=0 , \;\;\; \tau_2 =- 24 T_2, \;\;\; \tau_3= T_1 , \;\;\; \tau_4=7200 T_3 , \;\;\;
\tau_5 =-12960 T_4 , \nonumber \\
& \tau_6= - 216 T_2, \;\;\; \tau_7 =144 T_2 , \;\;\; \tau_8 = -4T_1, \;\;\; \tau_9 =216\cdot 180 T_4 , \;\;\; \tau_{10}=T_2 , \nonumber \\
& \tau_{11}=0, \;\;\; \tau_{12}=-240 T_3 , \;\;\; \tau_{13}=0, \;\;\; \tau_{14}=1944 T_4 .
\end{align}
Let us observe that, in particular, in the dimensional-reduced theory we have $\tau_1=0$, $\tau_{11}=0$, and $\tau_{13}=0$.

The dimensional reduction of the $D=11$ theory to the $D=7$ one also entails the relation (\ref{117a3}), implying the condition $T_0=1$ on the set of coefficients of the $D=11$ case; curiously, this selects the particular solution (\ref{t01}) we have recalled in Appendix \ref{apphidden}, originally found in \cite{D'AuriaFre}. 

Let us finally report in the following the (anti)commutation relations of the hidden superalgebra in the $D=7$ case obtained by dimensional reduction from the $D=11$ theory:
\begin{align}
& \lbrace Q^A, \bar{Q}_B \rbrace = - \ii \Gamma^a \left(  P_a+ Z_a \right) \delta^A_{\; B} - \frac{1}{2} \Gamma^{ab}Z_{ab}\delta^A_{\; B}-\sigma^{x \vert A}_{\;\;\;\;\; B}\left(\ii T^x + \frac{1}{18} \Gamma^{abc}T^{x \vert}_{\;\;\; abc}\right), \\
& [Q_A, P_a]= -2 \ii \begin{pmatrix}
5E_2 \\
0
\end{pmatrix} \Gamma_a Q'_A, \\
& [Q_A, Z_a]=-720 \begin{pmatrix}
E_2/48 \\
E_2/72
\end{pmatrix} \Gamma_a  Q'_A, \\
& [Q_A, Z_{ab}]=-4E_2 \Gamma_{ab} Q'_A, \\
& [Q_A, T^x]= - 8 \ii E_2 \sigma^{x \vert B}_{\;\;\;\;\;A} Q'_B , \\
& [Q_A, T^{x \vert}_{\;\;\; abc}]= -1440\begin{pmatrix}
E_2/48 \\
E_2/72
\end{pmatrix}  \Gamma_{abc}\sigma^{x \vert B}_{\;\;\;\;\;A}Q'_B, \\
& [Q_A, Z_{a_1    \ldots        a_5}] = -2(5!)\ii\begin{pmatrix}
E_2/48 \\
E_2/72
\end{pmatrix}  \Gamma_{a_1    \ldots        a_5} Q'_A\,.
\end{align}
Note that here, indeed, two inequivalent solutions exist. In particular, in the second solution the commutator $[Q_A, P_a]$ vanishes in correspondence of the solution $e_1=E_1=0$.
This has a special meaning in the $D=7$ theory, since it can be obtained if we further
require the following identification to hold in the seven-dimensional theory:
\begin{equation}
B^{a_1     \ldots         a_5}= \frac{1}{2} B_{ab}\epsilon^{a_1     \ldots         a_5ab}\label{kfields}\,.
\end{equation}
This identification is possible in $D=7$ due to the degeneration of the Lorentz-index structure for the two $1$-forms in (\ref{kfields}); however, in the corresponding $D=11$ theory, the two $1$-forms that get identified through (\ref{kfields}), namely $B^{ab}$ and $B^{a_1 \ldots a_5}$, are associated with the mutually non-local exterior forms $A^{(3)}$ and $B^{(6)}$, respectively. As we have said in \cite{Hidden}, we speculate that the absence of the coupling between the translation generator and the fermionic generator $Q'$ in this case could possibly be related to the intrinsically topological structure in $D=11$ inherent in the identification (\ref{kfields}).

\section{Further analysis of the symmetries of $D=11$ supergravity}\label{furthanal}

In this section, that is based on the work \cite{Malg} in collaboration with L. Andrianopoli and R. D'Auria, our aim is to clarify the relations occurring among the following superalgebras: The $\osp(1|32)$ algebra, the $M$-algebra and the DF-algebra (that is the hidden superalgebra underlying the FDA of $D=11$ supergravity introduced in \cite{D'AuriaFre}, further analyzed in \cite{Hidden}, and recalled before in this thesis).

The DF-algebra found by R. D'Auria and P. Fré in $1981$ can be seen as a (Lorentz-valued) central extension of the $M$-algebra including the nilpotent fermionic generator $Q'$. 
On the other hand, $\osp(1|32)$ is the most general simple superalgebra involving a fermionic generator with $32$ components, $Q_\alpha$, $\alpha =1,\ldots ,32$. 
This is also the dimension of the fermionic generator of eleven-dimensional supergravity. It is then natural that, already from the first construction of $D=11$ supergravity in \cite{Cremmer}, it was conjectured that $\osp(1|32)$ should somehow underlie, at least in some contracted version, the $D=11$ theory. They are however quite different: As we have already said, $D=11$ supergravity contains, besides the super-Poincar\'{e} fields given by the Lorentz spin connection $\omega^{ab}$ and the supervielbein $\lbrace{ V^a,\Psi^\alpha \rbrace}$ ($a=0,1,\ldots ,10$), also a $3$-form $A^{(3)}$, satisfying, in the superspace vacuum, the equation we have already discussed, namely
\begin{equation}\label{solito}
d A^{(3)}-\frac12 \bar\Psi\wedge\Gamma_{ab}\Psi\wedge V^a\wedge V^b=0.
\end{equation}
This theory, as we have largely discussed, is based on a FDA on the superspace spanned by the supervielbein. 

On the other hand, the fields involved in $\osp(1|32)$ are $1$-forms dual to generators which include, besides the $AdS$ generators $J_{ab}$ and $P_a$ ($a=0,1,\ldots,10$), and the supersymmetry charge $Q_\alpha$, also an extra generator $Z_{a_1\ldots a_5}$ carrying five antisymmetrized Lorentz indexes (its dual is a five-indexed antisymmetric Lorentz $1$-form $B^{a_1\ldots a_5}$).
Thus, in the case of $\osp(1|32)$, the set of $1$-forms is $\sigma^\Lambda\equiv\{\omega^{ab}, V^a, B^{a_1\ldots a_5}, \Psi^\alpha\}$, and these $1$-forms are dual to the $\osp(1|32)$ generators $T_\Lambda\equiv \{J_{ab}, P_a, Z_{a_1\ldots a_5}, Q_\alpha\}$, respectively.

The explicit form of the Maurer-Cartan equations for the $\osp(1|32)$ Lie superalgebra, once decomposed in terms of its subalgebra $\so(1,10)$, is the following:
\begin{equation}
\label{osp32}
\begin{aligned}
& d\omega^{ab} -\omega^{ac}\wedge \omega_{c}^{\ b} + e^2 V^a \wedge V^b+ \frac{e^2}{4 !}B^{a b_1 \ldots b_4}\wedge B^b_{\;\; b_1 \ldots b_4} +\frac{e}{2}\bar{\Psi}\wedge \Gamma^{ab}\Psi=0 , \\
& DV^{a}-  \frac{e}{2\cdot (5!)^2}\epsilon^{a b_1 \ldots b_5 c_1 \ldots c_5}B_{b_1 \ldots b_5}\wedge B_{c_1 \ldots c_5}-   \frac{\rm i}{2}\bar{\Psi} \wedge \Gamma^a \Psi =0,\\
& D B^{a_1 \ldots a_5} - \frac{e}{5!}\epsilon^{a_1 \ldots a_5 b_1 \ldots b_6}B_{b_1 \ldots b_5}\wedge V_{b_6}+ \frac{5 e}{6!}\epsilon^{a_1 \ldots a_5 b_1 \ldots b_6}B^{c_1 c_2}_{\;\;\;\;\;\;\; b_1 b_2 b_3}\wedge B_{c_1 c_2 b_4 b_5 b_6} + \\
& - \frac{\rm i}{2}\bar{\Psi}\wedge \Gamma^{a_1 \ldots a_5}\Psi =0 ,\\
& D \Psi - \frac{\rm i}{2} e \Gamma_a \Psi \wedge V^a - \frac{\rm i}{2 \cdot 5!} e \Gamma_{a_1 \ldots a_5}\Psi \wedge B^{a_1 \ldots a_5}=0,
\end{aligned}
\end{equation}
being $D= D(\omega)$ the Lorentz covariant derivative. Here, $e$ is a dimensionful constant. Indeed, in (\ref{osp32}) we are considering dimensionful $1$-form generators: Precisely, the bosonic $1$-forms $V^a$ and $B^{a_1\ldots a_5}$ carry length dimension $1$, the gravitino $1$-form $\Psi$ has length dimension $1/2$, while the Lorentz spin connection $\omega^{ab}$ is adimensional. As a consequence of this, the parameter $e$ has dimension $-1$ and can be thought of as proportional to the square root of a cosmological constant.

Let us observe that the presence of the bosonic $1$-form $B^{a_1\ldots a_5}$ in the simple superalgebra (\ref{osp32}) does \textit{not} allow to interpret a theory based on such an algebra as a theory on \textit{ordinary superspace}, whose cotangent space is spanned by the supervielbein, with Lorentz spin connection $\omega^{ab}$. To allow an interpretation of this type, the Lorentz covariant derivative of the $1$-form fields should be expressed only in terms of $2$-forms bilinears of the supervielbein.
This is the case of the Lie supergroup manifold $OSp(1|32)$ unless one would enlarge the ordinary notion of superspace by including the $1$-form $B^{a_1\ldots a_5}$ as an \textit{extra bosonic cotangent vector}, playing, in a certain sense, the role of a ``\textit{dual vielbein}''.

In the current section, we investigate the role played by $\osp(1|32)$ on the FDA of $D=11$ supergravity, and clarify the analogies and differences between the two algebraic structures.
Referring to what we have just said, the comparison between $D=11$ supergravity and a theory based on the Lie superalgebra $\osp(1|32)$ could be summarized as follows: On one hand, we have a theory which is well defined on superspace, but which involves a $3$-form, and is therefore based on an algebraic structure which is associated with a FDA rather than a Lie superalgebra; on the other hand, we have an algebraic structure corresponding to a Lie superalgebra, $\osp(1|32)$, which one can however hardly associate to a theory on ordinary superspace (due to the fact that it defines the tangent space to a Lie supergroup manifold corresponding to an \textit{enlarged superspace}).

As we have already discussed, the Lie superalgebra of $1$-forms leaving invariant $D=11$ supergravity and reproducing the FDA \textit{on ordinary eleven-dimensional superspace} (introduced in \cite{D'AuriaFre}) involves the extra bosonic $1$-forms $B^{a_1\ldots a_5}$ and $B^{ab}$, dual to the central generators of a central extension of the supersymmetry algebra including, besides the Poincar\'e algebra, the anticommutator
\begin{equation}\label{bla}
\left\{ Q ,\bar Q \right\} =- \left[\ii  \left(C
\Gamma ^{a} \right)  P_{a}+\frac{1}{2}\left( C \Gamma
^{ab}\right) Z_{ab}
+\frac{\ii}{5!}\left( C \Gamma ^{a_1\ldots a_5} \right)Z_{a_1\ldots a_5}\right]\, .
\end{equation}

The Lie superalgebra (\ref{bla}) was named \textit{$M$-algebra} \cite{deAzcarraga:1989mza, Sezgin:1996cj, Townsend:1997wg, Hassaine:2003vq, Hassaine:2004pp}, and it is commonly considered as the Lie superalgebra underlying $M$-theory \cite{Schwarz:1995jq, Duff:1996aw, Townsend:1996xj} in its low-energy limit, corresponding to eleven-dimensional supergravity in the presence of non trivial $M$-brane sources \cite{Achucarro:1987nc, Townsend:1995gp, Bergshoeff:1987cm, Duff:1987bx, Bergshoeff:1987qx, Townsend:1995kk} (as we have already mentioned in Chapter \ref{chapter 2}).
The algebra (\ref{bla}) generalizes to $D=11$ supergravity (and, by dimensional reduction, to all supergravities in dimensions higher than four) the topological notion of central extension of the supersymmetry algebra introduced in \cite{Witten:1978mh}, as it encodes the on-shell duality symmetries of string and $M$-theory \cite{Hull:1994ys, Duff:1990hn, Witten:1995ex, Duff:1995wd, Becker:1995kb}.

A field theory theory based on the $M$-algebra (\ref{bla}), however, is naturally described on an \textit{enlarged superspace} spanning, besides the gravitino $1$-forms, also the set of bosonic fields $\{V^a, B^{ab}, B^{a_1\ldots a_5} \}$.
If the low energy limit of $M$-theory should be based on the ordinary superspace spanned by the supervielbein $\lbrace{ V^a,\Psi \rbrace}$, as it happens for $D=11$ supergravity, then the $M$-algebra \textit{cannot} be the final answer, due to the fact that its generators are \textit{not} sufficient to reproduce the FDA on which $D=11$ supergravity is based.

This issue was raised already in \cite{D'AuriaFre}, and, as we have recalled, solved by still enlarging the enlarged superspace with the inclusion of an \textit{extra, nilpotent, fermionic generator} $Q'$, whose dual spinor $1$-form $\eta$ satisfies 
\begin{equation}\label{detanew}
  D\eta=\ii E_1 \Gamma_a \Psi \wedge V^a + E_2 \Gamma_{ab}\Psi \wedge B^{ab}+ \ii E_3 \Gamma_{a_1\ldots a_5}\Psi \wedge B^{a_1\ldots a_5}\,.
\end{equation}
In this way, the authors of \cite{D'AuriaFre} disclosed the hidden superalgebra that, in \cite{Malg} as well as in this thesis, we have called DF-algebra, containing the $M$-algebra (\ref{bla}) as a subalgebra, but including also a nilpotent fermionic generator $Q'$ (satisfying $Q'^2=0$), dual to a spinor $1$-form $\eta$, whose contribution to the DF-algebra Maurer-Cartan equations is given by (\ref{detanew}).
The DF-algebra underlies the formulation of the $D=11$ FDA on superspace (and, therefore, the $D=11$ theory on space-time introduced in \cite{Cremmer}) once the $3$-form is expressed in terms of $1$-form generators including, as we have already seen, also the spinor $1$-form $\eta$.

As we have shown in Ref. \cite{Hidden} and previously reviewed, this in turn implies that the group manifold generated by the DF-algebra has a fiber bundle structure whose base space is ordinary superspace, while the fiber is spanned, besides the Lorentz spin connection $\omega^{ab}$, also by the bosonic $1$-form generators $B^{ab}$ and $B^{a_1,\dots a_5}$. In particular, the nilpotent generator $Q'$, dual to the $1$-form $\eta$, allows to consider the extra $1$-forms $B^{ab}$ and $B^{a_1\ldots a_5}$ as gauge fields in ordinary superspace and not as additional vielbeins of an enlarged superspace.

During the years, many attempts have been made to relate $\osp(1|32)$ to the full DF-algebra, or to its $M$-subalgebra, and, thus, to $D=11$ supergravity (see, in particular, Ref. \cite{Castellani}).
Furthermore, in Refs. \cite{Azca1, Azca2, Azca3, Iza1}, the authors discussed the precise relation occurring between the $M$-algebra and $\osp(1|32)$. In particular, in \cite{Azca1} the general theory of expansions of Lie algebras was introduced and applied, showing that the $M$-algebra can be found as an expansion of $\osp(1|32)(2, 1, 2)$ (this was further explained in \cite{Azca2} and considered in the context of the so called $S$-expansion method in \cite{Iza1}). Then, in \cite{Azca3}, the possibility of an ``enlarged superspace variables/fields correspondence principle in $M$-theory'' was discussed.

Important contributions to the relations among $\osp(1|32)$, the full DF-algebra, or its $M$-subalgebra, and $D=11$ supergravity were also given in Refs. \cite{Bandos:2004xw, Bandos:2005, Hassaine:2003vq, Hassaine:2004pp, Iza1, Troncoso:1997va, Horava:1997dd, Troncoso:1998ng, Zanelli:2005sa, Izaurieta:2011fr}, principally in the construction of a Chern-Simons $D=11$ supergravity based on the supergroup $OSp(1|32)$.

Let us mention here that in Chapter \ref{chapter 6} of this thesis we will describe an analytic method for connecting different (super)algebras and, in particular, we will give an example of application in which we will show the way in which $\osp(1|32)$ is linked to the $M$-algebra, reproducing the result obtained in \cite{Iza1}.

In this section, we will show, on the same lines of the work \cite{Malg}, that the DF-algebra (which accounts for the non-trivial $4$-form cohomology of $D=11$ supergravity) \textit{cannot} be (directly) found as a contraction from $\osp(1|32)$.
More precisely, we will focus on the $4$-form cohomology in $D=11$ superspace of the supergravity theory, strictly related to the presence in the theory of the $3$-form $A^{(3)}$. Indeed, once formulated in terms of its hidden superalgebra of $1$-forms, we will find that $A^{(3)}$ can be decomposed into the sum of two parts, having different group-theoretical meaning: One allows to reproduce the FDA of $D=11$ supergravity (due to non-trivial contributions to the $4$-form cohomology in superspace) and explicitly breaks $\osp(1|32)$, while the other does \textit{not} contribute to the $4$-form cohomology (being a closed form in the vacuum); however, this second part defines a one-parameter family of trilinear forms invariant under a symmetry algebra related to $\mathfrak{osp}(1|32)$ by redefining the spin connection and adding a new Maurer-Cartan equation (it is a $3$-cocycle of the FDA enjoying invariance under $OSp(1|32)$).

Moreover, we will further discuss on the crucial role played by the $1$-form spinor $\eta$ (dual to $Q'$) for the $4$-form cohomology of the $D=11$ theory on superspace.

\subsection{Torsion-deformed $\mathfrak{osp}(1|32)$ algebra}\label{unusual}

Let us start our analysis by reformulating, as we have done in \cite{Malg}, $\osp(1|32)$ in such a way to be able to compare it with the DF-algebra and with its $M$-subalgebra.

In particular, as we will see in a while, this reformulation allows to overcome a possible obstruction in relating the two suparalgebras, due to the presence in the $M$-algebra of the bosonic $1$-form generator $B^{ab}$ associated with the central charge $Z_{ab}$, while no such generator appears in the $\osp(1|32)$ Maurer-Cartan equations (\ref{osp32}).
This problem can be indeed easily overcome by exploiting the freedom of redefining the Lorentz spin connection in $\osp(1|32)$ by adding an antisymmetric tensor $1$-form $B^{ab}$ (carrying length dimension $1$) as follows:
\begin{equation}\label{newomega}
  \omega^{ab}\rightarrow \omega^{ab} + eB^{ab}\equiv\hat\omega^{ab}\,,
\end{equation}
where $e$ is a dimensionful parameter, with length dimension $-1$, which can be identified with the one already present in $\osp(1|32)$ as written in (\ref{osp32}). The discussion presented here essentially follows some of the results obtained in Ref. \cite{Castellani} and further analyzed and clarified in \cite{Malg}.

Such a redefinition is always possible and it implies a \textit{change of the torsion $2$-form}.

After this redefinition of the spin connection, renaming $\hat \omega \rightarrow \omega$, equations (\ref{osp32}) take the following form:
\begin{equation}\label{osp32'}
\begin{aligned}
& d\omega^{ab} -\omega^{ac}\wedge \omega_{c}^{\ b} - eDB^{ab} - e^2 B^{ac}\wedge B_{c}^{\ b}  + e^2 V^a \wedge V^b+ \\
& + \frac{e^2}{4 !}B^{a b_1 \ldots b_4}\wedge B^b_{\;\; b_1 \ldots b_4} +\frac{e}{2}\bar{\Psi}\wedge \Gamma^{ab}\Psi=0 \, , \\
& DV^{a} +e B^{ab}\wedge V_b -  \frac{e}{2\cdot (5!)^2}\epsilon^{a b_1 \ldots b_5 c_1 \ldots c_5}B_{b_1 \ldots b_5}\wedge B_{c_1 \ldots c_5}-   \frac{\rm i}{2}\bar{\Psi} \wedge \Gamma^a \Psi =0\, ,\\
& D B^{a_1 \ldots a_5} - 5 e B^{m [a_1}\wedge B^{a_2 \ldots a_5 ]}_{\; \; \; \; \; \; \; \;\;\;\; \; m}- \frac{e}{5!}\epsilon^{a_1 \ldots a_5 b_1 \ldots b_6}B_{b_1 \ldots b_5}\wedge V_{b_6}+ \\
&+\frac{5 e}{6!}\epsilon^{a_1 \ldots a_5 b_1 \ldots b_6}B^{c_1 c_2}_{\;\;\;\;\;\;\; b_1 b_2 b_3}\wedge B_{c_1 c_2 b_4 b_5 b_6}- \frac{\rm i}{2}\bar{\Psi}\wedge \Gamma^{a_1 \ldots a_5}\Psi =0\, ,\\
& D \Psi - \frac{\rm i}{2} e \Gamma_a \Psi \wedge V^a - \frac{1}{4}e \Gamma_{ab}\Psi \wedge B^{ab}- \frac{\rm i}{2 \cdot 5!} e \Gamma_{a_1 \ldots a_5}\Psi \wedge B^{a_1 \ldots a_5}=0\,.
\end{aligned}
\end{equation}

Now, if one requires, as an extra condition, that the Lorentz $\so(1,10)$ spin connection $\omega^{ab}$ satisfies
\begin{equation}
\mathcal{R}^{ab}=d\omega^{ab} -\omega^{ac}\wedge \omega_c^{\ b}=0 , \label{min}
\end{equation}
corresponding to a Minkowski background $D^2=0$,
then, the first equation in (\ref{osp32'}) (which corresponds to the Maurer-Cartan equation for the $\osp(1|32)$ connection) splits into two equations: Equation (\ref{min}) plus the condition
\begin{equation}\label{bab}
DB^{ab} + e B^{ac}\wedge B_{c}^{\ b}  = e V^a \wedge V^b+ \frac{e}{4 !}B^{a b_1 \ldots b_4}\wedge B^b_{\;\; b_1 \ldots b_4} +\frac{1}{2}\bar{\Psi}\wedge \Gamma^{ab}\Psi\,,
\end{equation}
which defines the Maurer-Cartan equation for the new tensor field $B^{ab}$.

The algebra obtained from $\osp(1|32)$ after having performed the above procedure is \textit{not} isomorphic to $\osp(1|32)$, because of the extra constraint (\ref{min}), which implies (\ref{bab}).
A generalization of this superalgebra was introduced in the literature in 1982, in Ref. \cite{Castellani}.\footnote{The paper \cite{Castellani} appeared soon after \cite{D'AuriaFre}, as a possible semisimple extension of the DF-algebra.}

Actually, the algebra introduced in \cite{Castellani} generalizes (\ref{osp32'}) with the constraint (\ref{bab}), since it contains
an extra Maurer-Cartan equation for an extra spinor $1$-form of length dimension $3/2$. In \cite{Malg}, we have called it $\eta_{SB}$ (to avoid confusion with the extra spinor $1$-form $\eta$ of the DF-algebra). 
The explicit form of the Maurer-Cartan equations for the algebra presented in \cite{Castellani} reads as follows:
\begin{equation}\label{castella}
\begin{aligned}
& \mathcal{R}^{ab}\equiv d\omega^{ab} -\omega^{ac}\wedge \omega_{c}^{\ b} = 0 , \\
& DV^{a}=- e B^{ab}\wedge V_b +  \frac{e}{2\cdot (5!)^2}\epsilon^{a b_1 \ldots b_5 c_1 \ldots c_5}B_{b_1 \ldots b_5}\wedge B_{c_1 \ldots c_5}+   \frac{\rm i}{2}\bar{\Psi} \wedge \Gamma^a \Psi , \\
& D B^{ab} =e V^a \wedge V^b- e B^{ac}\wedge B_{c}^{\; b}+ \frac{e}{24}B^{a b_1 \ldots b_4}\wedge B^b_{\;\; b_1 \ldots b_4}+ \frac{1}{2}\bar{\Psi}\wedge \Gamma^{ab}\Psi, \\
& D B^{a_1 \ldots a_5} =5 e B^{m [a_1}\wedge B^{a_2 \ldots a_5 ]}_{\; \; \; \; \; \; \; \;\;\;\; \; m} + \frac{e}{5!}\epsilon^{a_1 \ldots a_5 b_1 \ldots b_6}B_{b_1 \ldots b_5}\wedge V_{b_6}+ \\
&-\frac{5 e}{6!}\epsilon^{a_1 \ldots a_5 b_1 \ldots b_6}B^{c_1 c_2}_{\;\;\;\;\;\;\; b_1 b_2 b_3}\wedge B_{c_1 c_2 b_4 b_5 b_6}+ \frac{\rm i}{2}\bar{\Psi}\wedge \Gamma^{a_1 \ldots a_5}\Psi , \\
& D \Psi = \frac{\rm i}{2} e \Gamma_a \Psi \wedge V^a + \frac{1}{4}e \Gamma_{ab}\Psi \wedge B^{ab}+ \frac{\rm i}{2 \cdot 5!} e \Gamma_{a_1 \ldots a_5}\Psi \wedge B^{a_1 \ldots a_5}, \\
& D \eta_{SB} = \frac{\rm i}{2}   \Gamma_a \Psi \wedge V^a + \frac{1}{4}  \Gamma_{ab}\Psi \wedge B^{ab}+ \frac{\rm i}{2 \cdot 5!}  \Gamma_{a_1 \ldots a_5}\Psi \wedge B^{a_1 \ldots a_5}= \frac 1 e D\Psi\, ,
\end{aligned}
\end{equation}
where $D$, as usual, denotes the Lorentz covariant derivative.

Actually, (\ref{castella}) is a (Lorentz-valued) central extension of (\ref{osp32'}) after having imposed (\ref{min}) and, consequently, (\ref{bab}), since the dual of $\eta_{SB}$ is a nilpotent generator commuting with all the generators but the Lorentz ones (the rational reason of its introduction being that of \textit{trying} to reproduce the DF-algebra with the In\"{o}n\"{u}-Wigner contraction $e \rightarrow 0$). In the following, we will refer to the algebra (\ref{castella}) as the SB-algebra, and to its semisimple subalgebra given by (\ref{osp32'}), (\ref{min}), and (\ref{bab}) as the \textit{restricted} SB-algebra (for short, RSB-algebra in the sequel).\footnote{The acronym SB(-algebra) stands for ``Stony Brook''(-algebra). Indeed, in \cite{Malg}, having observed that the authors of \cite{Castellani} were all affiliated to Stony Brook University, we decided of adopting the acronym SB.}

Let us mention that the algebra (\ref{castella}) is actually closed under differentiation even if one deletes the last equation containing the covariant differential $D\eta_{SB}$ (namely, when considering what corresponds to its subalgebra, the RSB-algebra); this Maurer-Cartan equation is, in fact, a double of the gravitino Maurer-Cartan equation, rescaled with the parameter $e$.\footnote{Let us mention that, as said in \cite{Castellani}, the group associated with the RSB-algebra is $O(10,1) \times OSp(1 \vert 32)$.}

Furthermore, the Maurer-Cartan equation for the $1$-form $\eta_{SB}$ does \textit{not} depend on any free parameter, meaning that it cannot be identified with the 1-form $\eta$ of the DF-algebra (see equation (\ref{detanew})).

We can thus conclude that, at the price of introducing the (torsion) field $B^{ab}$ satisfying (\ref{bab}), $\osp(1|32)$ can be mapped into the RSB-algebra, whereby the spin connection $\omega^{ab}$ is identified with the Lorentz connection of a $D=11$ Minkowski space-time with vanishing Lorentz curvature (albeit with a modification of the (super)torsion, which is non vanishing in both cases).
We refer to the RSB-algebra also as ``torsion-deformed $\osp(1|32)$ algebra''.

The RSB-algebra can be easily compared to the $M$-algebra. Indeed, the Maurer-Cartan equations of the RSB-algebra exactly reproduce the $M$-algebra (but not the full DF-algebra) by the In\"{o}n\"{u}-Wigner contraction $e \rightarrow 0$. 

For the RSB-algebra given by (\ref{osp32'}), (\ref{min}), and (\ref{bab}), analogously to what happens for $\osp(1|32)$ in the standard formulation (\ref{osp32}), an interpretation in terms of ordinary superspace spanned by the supervielbein is \textit{not} possible, because of the presence of two kinds of extra ``vielbeins'', $B^{ab}$ and $B^{a_1\cdots a_5}$, whose dual generators are not (Lorentz-valued) central charges, in this case: Indeed, the bosonic $1$-forms $B^{ab}$ and $B^{a_1\cdots a_5}$ are elements of a \textit{semisimple} bosonic subalgebra and, for this reason, independently of their super-extension, they cannot be related to central charges.
The same observation also holds for the SB-algebra, since it shares the same bosonic subalgebra with the RSB-algebra.

On the other hand, the DF-algebra is \textit{non-semisimple} and it enjoys a fiber bundle structure over ordinary superspace, where the fiber includes, besides the Lorentz connection, also the $1$-forms $B^{ab}$ and $B^{a_1\cdots a_5}$; in this theory, they are dual to Lorentz-valued central charges and can therefore be interpreted as abelian gauge fields on superspace (as we have shown in \cite{Hidden} and reviewed in Section \ref{CE} of this chapter).
  
At the dynamical level, the space-time components $B^{ab}_{\;\;\;\;|c}$ and $B^{a_1\cdots a_5}_{\;\;\;\;\;\;\;\;\;\;\;|c}$ of the $1$-form gauge fields $B^{ab}$ and $B^{a_1\cdots a_5}$ (we are using rigid Lorentz indexes), present extra degrees of freedom with respect to the component fields $A_{[abc]}$ and  $B_{[a_1\cdots a_6]}$, respectively, appearing in the FDA on which $D=11$ supergravity is based.\footnote{The possible interpretation of the field $A_{\mu\nu\rho}$ of $D=11$ supergravity in terms of the totally antisymmetric part of the contorsion tensor in $\osp(1|32)$ was already considered in Ref. \cite{Troncoso:1997va}. $B_{[a_1\cdots a_6]}$ are the components of the $6$-form $B^{(6)}$, related to $A^{(3)}$ by Hodge duality of their field-strengths.} As we will clarify in the following, the extra degrees of freedom are dynamically decoupled from the physical spectrum in the DF-algebra (contrary to what happens in the case of the $M$-algebra) because of the presence of the spinor $1$-form $\eta$, dual to the nilpotent fermionic generator $Q'$, which thus behaves as a BRST ghost, guaranteeing the equivalence of the hidden algebra with the supersymmetric FDA. This mechanism does \textit{not} work for the semisimple RSB-algebra, since, in that case, the extra components in $B^{ab}_{\;\;\;\;|c}$ and $B^{a_1\cdots a_5}_{\;\;\;\;\;\;\;\;\;\;\;|c}$ besides the fully antisymmetrized ones  are \textit{not} decoupled from the physical spectrum.

The detailed relation between the full SB-algebra and the DF-algebra (including the relations and differences between the extra spinors $\eta_{SB}$ and $\eta$ of the two algebras) is more subtle, and will be analyzed in the following, on the same lines of what we have done in the paper \cite{Malg}.

Before moving to that topic, let us analyze some properties of the RSB-algebra related to its feature of being a semisimple superalgebra.

\subsubsection{Properties of the RSB-algebra}

For any semisimple Lie algebra, as it is well known from the Lie algebras CE-cohomology (see Chapter \ref{chapter 2}), and as already pointed out in \cite{Castellani}, it is always possible to define a non-trivial $3$-cocycle $H^{(3)}$ (satisfying $dH^{(3)}=0$) given by
\begin{equation}\label{3cocycle}
 H^{(3)}= C_{ABC}\, \sigma^A\wedge\sigma^B\wedge\sigma^C = -2 h_{AB}\, \sigma^A \wedge d \sigma^B\, ,
\end{equation}
where $C_{ABC}=h_{AL}\, C^L_{\ BC}$ are the structure constants of the algebra, with an index lowered with the Killing metric $h_{AB}$.

The closure of $H^{(3)}$ can be easily proven by using the Maurer-Cartan equations:
\begin{equation}\label{MC}
  d\sigma^A+\frac12 C^A_{\ BC} \, \sigma^B\wedge\sigma^C=0 ,
\end{equation}
where the $\sigma^A$'s $1$-forms are in the coadjoint representation of the Lie (super)algebra. Indeed, we can write
\begin{equation}\label{ji}
  dH^{(3)}=- \frac 32 C_{ABC}\, C^C_{\ LM} \, \sigma^A\wedge\sigma^B\wedge\sigma^L\wedge \sigma^M=0 \,,
\end{equation}
where the vanishing of this expression is due to (super-)Jacobi identities.

For the semisimple RSB-algebra, the set of $1$-forms is $\sigma^A=\{\omega^{ab},V^a,\Psi^\alpha,\- B^{ab},B^{a_1\ldots a_5}\}$. However, the Lorentz quotient of the RSB-group admits the Lorentz-covariant Maurer-Cartan equations
\begin{equation}\label{MC2}
  D\sigma^\Lambda+\frac12 C^\Lambda_{\ \Sigma\Gamma}\, \sigma^\Sigma\wedge\sigma^\Gamma=0
\end{equation} 
for the restricted set of $1$-forms $\sigma^\Lambda=\{V^a,\Psi^\alpha,B^{ab},B^{a_1\ldots a_5}\}$, allowing to rewrite
\begin{equation}\label{h3h}
  H^{(3)}=-2\, \sigma^\Lambda\wedge  D\sigma^\Sigma \, h_{\Lambda\Sigma}
\end{equation}
satisfying $dH^{(3)}=0$ (see the CE-theorem \ref{CEteobello} in Chapter \ref{chapter 2}).
From a direct calculation, one can find, up to an overall normalization, that the cocycle $H^{(3)}$ can be written as follows:
\begin{eqnarray}\label{3cocycle2}
 H^{(3)}&=& V^a\wedge DV_a +\frac 12 B^{ab}\wedge DB_{ab}+\frac 1{5!} B^{a_1 \ldots a_5}\wedge DB_{a_1 \ldots a_5}-  \frac 1 e \bar\Psi \wedge D\Psi = \label{h1}\\
 &=&
 e\Bigl( B_{ab} \wedge V^a \wedge V^b + \frac 13 B_{a b}\wedge B^{b} _{\;c}\wedge B^{c a}+\frac 1{4!}B^{b_1 b_2} \wedge  B_{b_1 a_1...a_4}\wedge B_{b_2}^{\; \ a_1...a_4}+ \nonumber\\
& & + \frac 1{(5!)^2}  \epsilon_{a_1...a_5 b_1...b_5 m}B^{a_1...a_5}\wedge B^{b_1...b_5}\wedge V^m + \nonumber\\
& &- \frac 13 \frac 1{[2!\cdot (3!)^2 \cdot 5!]}  \epsilon_{m_1...m_6 n_1...n_5}B^{m_1m_2m_3p_1p_2}\wedge B^{m_4m_5m_6p_1p_2}\wedge B^{n_1...n_5}\Bigr)\,.\label{h2}
\end{eqnarray}
Let us observe that $H^{(3)}$ is actually a bosonic $3$-form (see equation (\ref{h2}), the same expression holding for the $3$-cocycle of its bosonic subalgebra).
An analogous result can be obtained for $\osp(1|32)$ by setting $B^{ab}= 0$ in (\ref{h1}) and (\ref{h2}).

Let us remark that the $e\to 0$ limit of $H^{(3)}$ is a \textit{singular} limit: Indeed, in this limit we have that $H^{(3)}\to 0$, but $\frac 1e H^{(3)}$ is finite if one considers the second expression (\ref{h2}), while $\frac 1e d H^{(3)}\neq 0$ in the limit, corresponding to the fact that the Killing metric of the contracted superalgebra at $e\to 0$ is degenerate.

For $e\neq 0$, instead, $ H^{(3)}$ is a $3$-cocycle of the superalgebra and (following the general Sullivan construction of FDAs \cite{Sullivan}) it could be trivialized in terms of a $2$-form $Q^{(2)}$, writing
\begin{equation}\label{dq}
  dQ^{(2)}+H^{(3)}=0.
\end{equation}
In this way, a new FDA in the semisimple case is realized.

It could be interesting to investigate about a hidden superalgebra of (\ref{dq}), which would allow to parametrize $Q^{(2)}$ in terms of an appropriate set of $1$-forms; however, to ascertain if one can associate a hidden Lie superalgebra to the FDA (\ref{dq}), one has to introduce extra fields besides the set of generators $\{\sigma^\Lambda\}$ of the SB-algebra.
This is left to future investigations.

\subsection{Relating $\mathfrak{osp}(1|32)$ to the DF-algebra}\label{relating}

In the following, we clarify the relation between the DF-algebra and the SB-algebra. 

The complete Maurer-Cartan equations for the DF-algebra can be found in Section \ref{11D} of Chapter \ref{chapter 2}.

For our purpose, it is convenient to rewrite the real parameters $\{ T_i,S_i, E_i\}$ appearing in the parametrization of the $3$-form $A^{(3)}$ in terms of $1$-forms (see the result reported in (\ref{11dsol}) of Chapter \ref{chapter 2}), namely in $A^{(3)}(\sigma)$, as follows:\footnote{Here we have defined, using the notations of \cite{Hidden}, $C\equiv E_2-60 E_3$ and $\alpha \equiv 5!\frac{E_3^2}{C^2}$.}

\begin{equation}\label{deco}
\begin{aligned}
& \left\{\begin{array}{ccl}
T_0 &=& \frac{1}{6}+ \alpha, \\
T_1 &=&-\frac{1}{90} + \frac{1}{3} \alpha, \\
T_2 &=&-\frac 1{4!} \alpha, \\
T_3 &=& \frac{1}{(5!)^2}\alpha , \\
T_4 &=& - \frac{1}{3[2!\cdot (3!)^2 \cdot 5!]} \alpha ,
\end{array}\right.
\left\{\begin{array}{ccl}
S_1 &=& \frac{1}{4!C} + \frac{1}{2\cdot (5!) E_3} \alpha, \\
S_2 &=& -\frac{1}{10\cdot (4!)C} + \frac{1}{4 \cdot(5!)E_3} \alpha, \\
S_3 &=& \frac{1}{2 \cdot (5!)^2 E_3} \alpha ,
\end{array}\right. \\
& \left\{\begin{array}{ccl}
E_1 &=& -10 C + \frac{C^2}{E_3} \alpha, \\
E_2 &=& C + \frac{C^2}{2E_3} \alpha, \\
E_3 &=& \frac{C^2}{ 5! E_3} \alpha .
\end{array}\right.
\end{aligned}
\end{equation}

Given the above expressions, it is then useful to decompose the spinor $1$-form $\eta$ of the DF-algebra as follows:
\begin{equation}
\eta= -10 C(\xi + \alpha \lambda), \label{etadec}
\end{equation}
where we introduced the spinor $1$-forms $\xi$ and $\lambda$ satisfying
\begin{eqnarray}
D\xi &=&  \ii \Gamma_a \Psi \wedge V^a -\frac 1{10} \Gamma_{ab}\Psi \wedge B^{ab}\,,\label{dxinew}\\
D\lambda &=&-\frac{C}{5 E_3}\left( \frac {\ii}2 \Gamma_a\Psi \wedge V^a + \frac 14 \Gamma_{ab}\Psi \wedge B^{ab} + \frac{\ii}{2(5!)}\Gamma_{a_1\ldots a_5}\Psi \wedge B^{a_1\ldots a_5}\right)= \nonumber\\
&=&-\frac{C}{5 E_3}D\eta_{SB}\,. \label{dlambda}
\end{eqnarray}
From equation (\ref{dlambda}), we can now see that $\lambda$ can be chosen as proportional to the spinor $1$-form $\eta_{SB}$ introduced in (\ref{castella}) as a Lorentz-valued central extension of the RSB-superalgebra: 
\begin{equation}
\lambda=- \frac{C}{5 E_3}\eta_{SB}.
\end{equation}

Then, equations (\ref{deco}) and (\ref{etadec}) allow to decompose also $A^{(3)}(\sigma)$ into two pieces, namely
\begin{equation}\label{a3d}
A^{(3)}(\sigma) = A^{(3)}_{(0)} + \alpha A^{(3)}_{(e)},
\end{equation}
where
\begin{equation}\label{a30}
A^{(3)}_{(0)}= \frac{1}{6}\left(B_{ab} \wedge V^a \wedge V^b - \frac{1}{15} B^{ab}\wedge B_{bc} \wedge B^c_{\;\; a} - \frac{5\ii}2\bar \Psi \wedge \Gamma_a \xi \wedge V^a + \frac{1}{4} \bar \Psi \wedge \Gamma_{ab} \xi \wedge B^{ab}\right)\,,
\end{equation}
while
\begin{equation}\label{a3e}
A^{(3)}_{(e)}= \frac 1 e H^{(3)} +2 \bar \eta_{SB}\wedge D \eta_{SB}\,.
\end{equation}
We now recognize, in the first term of (\ref{a3e}), the $OSp(1|32)$-invariant $3$-form $\frac 1e H^{(3)}$ introduced in (\ref{3cocycle2}), which is finite in the $e\to 0$ limit, but looses its character of being a $3$-cocycle (namely, a closed form), becoming just a $3$-cochain of the $M$-algebra. Explicitly, we have:
\begin{align}
\frac 1e H^{(3)}&=
 \Bigl( B_{ab} \wedge V^a \wedge V^b + \frac 13 B_{a b}\wedge B^{b} _{\;c}\wedge B^{c a}+\frac 1{4!}B^{b_1 b_2} \wedge  B_{b_1 a_1...a_4}\wedge B_{b_2}^{\; \ a_1...a_4}+ \nonumber\\
& + \frac 1{(5!)^2}  \epsilon_{a_1...a_5 b_1...b_5 m}B^{a_1...a_5}\wedge B^{b_1...b_5}\wedge V^m + \nonumber\\
&- \frac 13 \frac 1{[2!\cdot (3!)^2 \cdot 5!]}  \epsilon_{m_1...m_6 n_1...n_5}B^{m_1m_2m_3p_1p_2}\wedge B^{m_4m_5m_6p_1p_2}\wedge B^{n_1...n_5}\Bigr)\,, \label{3cocyclelim}
\end{align}
and, by a straightforward differentiation using the Maurer-Cartan equations of DF-algebra (see Section \ref{11D} of Chapter \ref{chapter 2}), we can easily verify that $d\left(\frac 1e H^{(3)}\right)_{e\rightarrow0}= 0$, while
\begin{eqnarray}
dA^{(3)}_{(0)}&=&\frac 12 \bar \Psi \wedge \Gamma_{ab} \Psi \wedge V^a \wedge V^b\label{da30},\\
dA^{(3)}_{(e)}&=&0\label{da3e}\,.
\end{eqnarray}

Now, let us observe that $A^{(3)}_{(0)}$ only depends on the restricted set of $1$-forms $\{V^a,\Psi,B^{ab},\xi\}$, which does not include the $1$-form $B^{a_1\ldots a_5}$, through an expression, that is (\ref{a30}), which does not contain any free parameter. The term $A^{(3)}_{(0)}$ is, however, the only one contributing to the (vacuum) $4$-form cohomology in superspace (see equation (\ref{da30})), $A^{(3)}_{(e)}$ being instead a closed $3$-form in the vacuum.\footnote{Surprisingly, it corresponds to one of the solutions found in the original paper of R. D'Auria and P. Fr\'{e} \cite{D'AuriaFre}.}

On the other hand, we can see that the one-parameter family of solutions to the DF-algebra, whose presence was clarified in \cite{Bandos:2004xw}, actually only depends on the contribution $A^{(3)}_{(e)}$, which appears as a trivial deformation of $A^{(3)}_{(0)}$ in $A^{(3)}$, since it does not contribute to the vacuum $4$-form cohomology (\ref{solito}). However, $A^{(3)}_{(e)}$ is invariant not only under the DF-algebra, but also under the SB-algebra, even at finite $e$. Instead, the other term, $A^{(3)}_{(0)}$, explicitly breaks the invariance under the SB-algebra.

In Figure \ref{figure11d} the reader can find a map which schematically collects and summarizes the relations among the superalgebras we have analyzed.

\begin{figure}
\centering
\pgfdeclareimage[height=8cm]{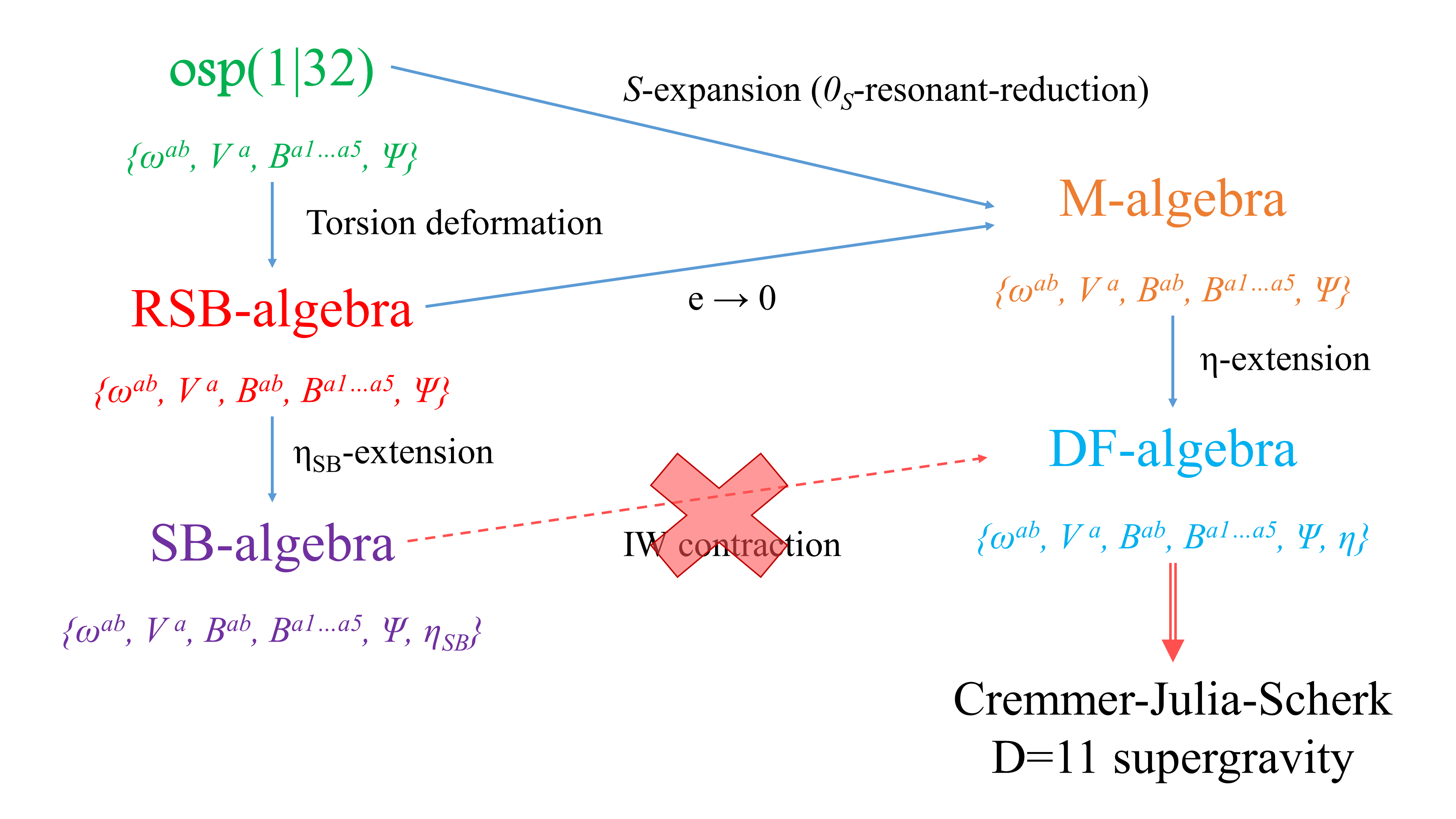}{map11d}
\pgfuseimage{map11d}
\caption[Map among superalgebras related to $D=11$ supergravity]{\textit{Map among superalgebras related to $D=11$ supergravity.} The RSB-algebra can be viewed as a torsion deformation of $\mathfrak{osp}(1|32)$, the latter being related to the $M$-algebra by an $S$-expansion (actually, an $S$-expansion with $0_S$-resonant-reduction, as we will explicitly show in Chapter \ref{chapter 6}). The $\eta$-extension (spinor central extension) of the $M$-algebra leads to the DF-algebra, which is the superalgebra underlying the FDA description of the Cremmer-Julia-Scherk supergravity theory in eleven dimensions. The SB-algebra, which includes the extra spinor $1$-form $\eta_{SB}$, cannot be directly related to the DF-algebra by an In\"{o}n\"{u}-Wigner contraction $e\rightarrow 0$, due to the fact that the structure constants of the the two superalgebras are different.}
\label{figure11d}
\end{figure}

\subsection{More on the role of the extra spinor $1$-forms}

Let us conclude our analysis by spending some words to discuss the role of the spinor $1$-forms $\xi$ and $\lambda$ introduced in the decomposition (\ref{etadec}) of $\eta$ and appearing in $A^{(3)}$ through equation (\ref{a3d}), following the lines of what we have done in \cite{Hidden} and previously recalled in this chapter.

The spinor $\xi$ appears in $A^{(3)}_{(0)}$ and its role is that of allowing for $dA^{(3)}_{(0)}$ to be a closed $4$-form on \textit{ordinary superspace}; it behaves as a cohomological ghost, since its supersymmetry and gauge transformations exactly cancel the non-physical contributions coming from the tensor field $B^{ab}$. In other words, the group manifold generated by the set of $1$-forms $\{\sigma^\Lambda\}$ including $\xi$ presents a fiber bundle structure with ordinary superspace as base space.

The role of the second spinor, $\lambda \propto \eta_{SB}$, appearing, instead, in the $\osp(1|32)$-invariant term $A^{(3)}_{(e)}$, at first sight appear less clear, since $dA^{(3)}_{(e)}=0$ in the FDA where the vacuum relation  (\ref{solito}) holds. It plays, however, a role that is analogous to the one of $\xi$: Indeed, in its absence, $A^{(3)}_{(e)}$ would reduce to the bosonic $3$-form $\frac 1 e H^{(3)}$, which is a closed $3$-form for $e\neq 0$, while this property is lost in the limit $e\to 0$. In the same limit, $ \frac 1 e d H^{(3)}$ is, instead, a $4$-form polynomial of all the $\sigma^\Lambda$'s, that is a cochain of the superspace enlarged to include $B^{ab}$ and $B^{a_1\cdots a_5}$. The role of $\eta_{SB}$ is then crucial in order to restore, also for $\alpha \neq 0$,
the correct $4$-form cohomology (\ref{solito}) on the vacuum superspace for $dA^{(3)}$, by allowing $dA^{(3)}_{(e)}=0$. On the other hand, in the semisimple case $e\neq 0$, $H^{(3)}$ is a closed $3$-form, and $\eta_{SB}$ looses its cohomological role.

Let us remark that when considering the interacting theory out of the vacuum, one should introduce a $4$-form super field-strength $G^{(4)}$ in superspace, namely
\begin{equation}\label{g4}
 G^{(4)}\equiv dA^{(3)}-\frac12 \bar\psi\wedge\Gamma_{ab}\psi\wedge V^a\wedge V^b\,.
\end{equation}
In this case, one would expect that the superspace $4$-form cohomology could also get non-trivial contributions from $dA^{(3)}_{(e)}$.

We can finally summarize our results as follows: We have found that, despite of the fact that the $M$-algebra is a In\"on\"u-Wigner contraction of the torsion deformation of $\osp(1|32)$ that we have called RSB-algebra, still the DF-algebra \textit{cannot} be directly obtained as an In\"on\"u-Wigner contraction from the SB-algebra, the latter being a Lorentz-valued, central extension of the RSB-algebra.
Correspondingly, the $D=11$ supergravity  theory is \textit{not} left invariant by $\osp(1|32)$, while being invariant under the DF-algebra. This is due to the fact that the spinor $1$-form $\eta$ of the DF-algebra (spinor central extension of the $M$-algebra) contributes to the DF-algebra with structure constants different from  the ones of the SB-algebra. In particular, referring to equation (\ref{etadec}), we can see that $\eta$ differs from ${\eta_{SB}}\propto \lambda$ by the extra $1$-form generator $\xi$. This has a counterpart in the expression of $A^{(3)}= A^{(3)}(\sigma^\Lambda)$, which trivializes the vacuum $4$-form cohomology in superspace in terms of DF-algebra $1$-forms $\sigma^\Lambda$'s. As we can see by looking at the decomposition (\ref{a3d}), $A^{(3)}(\sigma^\Lambda)$ is \textit{not} invariant under $\osp(1|32)$ (neither under its torsion deformation), and this is due to the contribution $A^{(3)}_{(0)}$, which explicitly breaks this symmetry. However, this latter term is the only one contributing to the vacuum $4$-form cohomology in superspace, because of the presence in the DF-algebra of the two spinors $\xi$ and $\eta_{SB}$ into which the cohomological spinor $\eta$ can be decomposed.

\section{Comments on the FDAs of $D=4$ theories}

Let us mention, now, that in Ref. \cite{SM4}, in collaboration with D. M. Pe\~{n}afiel, we have considered a new minimal super-Maxwell like algebra in $D=4$ dimensions, containing, besides the Poincaré and the supersymmetry generators, also Maxwell-like bosonic generators and an extra fermionic generator, and we have written the Maurer-Cartan equations dual to the superalgebra. Then, we have added a $4$-form field-strength to the theory and performed a study on the FDA in $D=4$ thus obtained, on the same lines of \cite{Hidden, Malg}. 

The minimal Maxwell superalgebras are minimal super-extension of the Maxwell algebra, which, in turn, is a non-central extension of Poincaré algebra involving an \textit{abelian} generator (along the lines of non-commutative geometry), and, as we have already mentioned in Chapter \ref{chapter 4}, it arises when one considers symmetries of systems evolving in flat Minkowski space filled in by a \textit{constant electromagnetic background} \cite{bacry, schrader}.\footnote{Indeed, if one constructs an action for a massive particle, invariant under the Maxwell symmetries, one obtains that ``it satisfies the equations of motion of a charged particle interacting with a constant electromagnetic field via the Lorentz force'' \cite{gibbons}.} 

In particular, the $\mathcal{N}=1$, $D=4$ supersymmetrization of the Maxwell algebra introduced in \cite{gomis2} seems to be specially appealing, since it describes the supersymmetries of generalized $\mathcal{N}=1$, $D=4$ superspace in the presence of a constant, abelian supersymmetric field-strength background.
Subsequently, this superalgebra and its generalizations have been
obtained as $S$-expansions of the $AdS$ superalgebra \cite{Concha2}. This family of superalgebras containing the Maxwell algebras type as bosonic subalgebras may be viewed as a generalization of the D'Auria-Fr\'{e} superalgebra (DF-algebra), introduced in \cite{D'AuriaFre} and recalled before in this thesis, and of the so called Green algebras \cite{green}.

The main reason for having chosen a Maxwell-like superalgebra in four dimensions (whose basis is given by a set of generators $\lbrace J_{ab},P_a,Z_{ab},\tilde{Z}_{ab},Q_\alpha ,\Sigma_\alpha \rbrace$) as a starting point of our analysis has been the fact that the Maxwell-like generators $Z_{ab}$ and $\tilde{Z}_{ab}$ can be related to dual bosonic $1$-forms associated with an antisymmetric $3$-form $A^{(3)}$ on superspace, appropriately introduced in the context of FDAs, whose field-strength is given by $F^{(4)}= dA^{(3)}$ (modulo gravitino $1$-form bilinears, when it has support on superspace). In $D=4$, $dA^{(3)}$ can be viewed as a trivial boundary term of an hypothetical Lagrangian (in \cite{SM4} we did not discuss the dynamics of the theory, while we concentrated, instead, on the pure FDA structure of the model).

Now, it is well known that $D=11$ supergravity admits spontaneous compactification to $D=4$ and that through the Freund-Rubin ansatz one ends up with a $\mathcal{M}_4 \times \mathcal{M}_7$ ground state (see \cite{Duff:1986hr} for exhaustive details). The two manifolds then correspond to (either) $AdS_4 \times S^7$ (or $AdS_7 \times S^4$). In this set up, one can see that, even if the $3$-form $A^{(3)}$ does not give any dynamical contribution to the theory in four dimensions, however, its field-strength (which is proportional to the volume element in four dimensions) can be related to the presence of \textit{fluxes} (see, for example, \cite{Trigiante2}), that are background quantities which can be switched on in a toroidal compactification.
In particular, the Freund-Rubin solution $AdS_4 \times S^7$ is characterized by the vacuum expectation value along the four non-compact space-time directions, $\langle F_{\mu \nu \rho \sigma}\rangle = m \epsilon_{\mu \nu \rho \sigma}$, of the $4$-form field-strength $F^{(4)}= dA^{(3)}$ of the eleven-dimensional theory, and it describes the full back-reaction of this flux on the space-time
geometry. Then, $dA^{(3)}$ can be written as
\begin{equation}
dA^{(3)} \propto e \Omega,
\end{equation}
where $\Omega \propto \epsilon_{abcd} V^a \wedge V^b \wedge V^c \wedge V^d$ is the volume element in four dimensions; $dA^{(3)}$ can thus be associated with a flux with charge $e$, being $e$ a constant parameter. 
%Background fluxes typically
%induce, in the lower-dimensional effective theory, minimal couplings as well as mass terms and a scalar potential. The former, involving only the electric vector fields, manifestly
%break the electric-magnetic duality symmetry of the original ungauged theory. 

With this motivations, we considered of some interest to study a semisimple Maxwell-like superalgebra in $D=4$ in the context of FDAs (involving a $3$-form $A^{(3)}$), by writing the deformation to the $4$-form $F^{(4)}$ induced by the presence a scaling parameter $e$. Since $dA^{(3)}$ is a boundary contribution in four dimensions, we expected a topological form of $dA^{(3)}$ to lie in the parametrization of $A^{(3)}$.
Indeed, on the same lines of what we have seen in the previous section, given the Maurer-Cartan equations $d \sigma^A + \frac{1}{2} C^A_{\;\;BC}\, \sigma^B \wedge \sigma^C =0$ for a semisimple Lie algebra, one can write (see the CE-theorem \ref{CEteobello} in Chapter \ref{chapter 2}):
\begin{equation}
A^{(3)} = C_{A BC} \, \sigma^A \wedge \sigma^B \wedge \sigma^C =-2 h_{AB} \sigma^A \wedge d \sigma^B= A^{(3)}(\sigma),
\end{equation}
satisfying $dA^{(3)}= 0$ and providing $dA^{(3)}$ with a topological hidden structure.
One can then request the parametrization $A^{(3)}(\sigma)$ in terms of $1$-forms to satisfy $dA^{(3)}\equiv 0$, which is allowed, as well as expected, in $D=4$ when $e\rightarrow 0$. However, since, as we have already mentioned, one can endow $dA^{(3)}$ with a support on superspace, we can alternatively request the parametrization $A^{(3)}(\sigma)$ to satisfy $dA^{(3)}(\sigma)= \frac{1}{2}\bar{\psi}\wedge \gamma_{ab}\psi \wedge V^a \wedge V^b$ in the superspace vacuum of the four-dimensional theory.\footnote{Indeed, \textit{in superspace} the super field-strength in the vacuum is given by $F^{(4)}=dA^{(3)}-\frac{1}{2} \bar{\psi}\wedge \gamma_{ab}\psi \wedge V^a \wedge V^b=0$.}

In \cite{SM4}, we first requested $dA^{(3)}$ to have a topological structure, by selecting a particular ansatz for the $3$-form parametrization, involving the scaling parameter $e$. Then, we checked that $dA^{(3)}(\sigma)= \frac{1}{2}\bar{\psi}\wedge \gamma_{ab}\psi \wedge V^a \wedge V^b$ in the limit $e\rightarrow 0$. We have obtained that this happens for a particular non-semisimple contraction of the $D=4$ Maxwell-like superalgebra we considered, which is also an extension (involving the parameter $e$) of the super-Poincar\'e algebra underlying supergravity in four dimensions. 
The superalgebra we got in this way could be considered in future works for the construction of a Lagrangian and for the study of the dynamics of the theory, also in the presence of a non-trivial space-time boundary. 

Let us mention that the extra fermionic generator $\Sigma$ of the $D=4$ Maxwell-like superalgebra is still nilpotent, but it does not appear as a ``central spinor'' extension of some other algebra.

It would be interesting to study the hidden parametrization of the $3$-form $A^{(3)}$ and, consequently, the hidden structure associated to the \textit{complete} extended $D=4$ supergravity theory (in both the $\mathcal{N}=1$ and the $\mathcal{N}=2$ cases) which also includes gauge fields and scalars, some of which can be dualized to antisymmetric tensors and can introduce non-trivial conditions in the vacuum of the $D=4$ theory \cite{Trigiante2}.
One could also add to this set up a $2$-form (on the lines of \cite{galalg}), which, contrary to the case of the $3$-form, in four-dimensional theories exhibits a dynamical role. The same could be done, for example, in the case of the $AdS$-Lorentz superalgebra in $D=4$ studied and analyzed in Chapter \ref{chapter 4}; we conjecture that, in that case, the $k^{ab}$ field appearing in the theory could also be viewed as an extra superspace direction, leading to a superspace which is enlarged with respect to the ordinary one. 

Another future analysis would be try to better understand the possible relations among the extra bosonic fields appearing in the above-mentioned theories in $D=4$ and the extra bosonic $1$-forms appearing in the hidden structure underlying $D=11$ (and $D=7$) supergravities. In this context, the study of the dimensional reduction from eleven (or directly seven) to four dimensions would certainly be clarifying (work in progress on this topic).

\chapter{New results on $S$-expansion} \label{chapter 6}

% **************************** Define Graphics Path **************************
\ifpdf
    \graphicspath{{Chapter6/Figs/}{Chapter6/Figs/PDF/}{Chapter6/Figs/}}
\else
    \graphicspath{{Chapter6/Figs/Vector/}{Chapter6/Figs/}}
\fi

In Physics, there is a great interest in studying the relations among different Lie (super)algebras related to the symmetries of different theories, since this can disclose connections among the aforementioned theories. 

There are many different methods to obtain new Lie (super)algebras from given ones, for example deformations, extensions, expansions, and contractions.

In $2006$, a particular expansion approach, which goes under the name of \textit{$S$-expansion}, was developed \cite{Iza1}. In performing the $S$-expansion method, one combines the structure constants of an initial Lie (super)algebra $\mathfrak{g}$ with the inner multiplication law of a discrete set $S$, endowed with the structure of a semigroup, in such a way to define the Lie bracket of a new, larger, $S$-expanded (super)algebra, commonly written as $\mathfrak{g}_S= S \times \mathfrak{g}$. 

Several (super)gravity theories have been extensively analyzed in the context of expansions and contractions, enabling numerous results (among which, for example, those presented in Refs. \cite{Azca1, Azca2, Azca3, Iza4, Fierro2, Salgado, Fierro1, Concha:2013uhq, Concha1, Concha2, CR2, Concha:2014zsa, Concha:2016hbt, Concha:2016tms}).

This is the reason why, on the pure algebraic and group theoretical side of my PhD research, in collaborations with some colleagues and friends from Chile (M.C. Ipinza and D. M. Pe\~{n}afiel) and from the Polytechnic of Turin (F. Lingua), I have moved towards the $S$-expansion (and the In\"{o}n\"{u}-Wigner contraction) of Lie (super)algebras (for a review of $S$-expansion and In\"{o}n\"{u}-Wigner contraction, see Chapter \ref{chapter 3}). 

A fundamental task to accomplish when performing the $S$-expansion method is to find the appropriate semigroup connecting two different Lie (super)algebras. This involved, until more or less one year ago, a kind of ``trial and error'' process. 

With this in  mind, in the work \cite{Analytic} we have developed an \textit{analytic method} to find the semigroup(s) $S$ (there can also be more than one) linking two different (super)algebras, once certain particular conditions on the subspace decomposition of the starting and target (super)algebras and on the partition of the set(s) involved in the procedure are met. The details will be given in this chapter, where we will describe this analytic method and also give an interesting example of application, concerning the superalgebra $\mathfrak{osp}(1 \vert 32)$ and the $M$-algebra. 

The $S$-expansion is valid regardless of the structure of the original Lie (super)algebra; however, when something about the structure
of the starting (super)algebra is known and when certain particular conditions are met, the $S$-expansion is able not only to lead to diverse expanded (super)algebras, but also to reproduce the effects of the so called \textit{standard} as well as the \textit{generalized} In\"on\"u-Wigner contraction (see Chapter \ref{chapter 3} for definitions).

In \cite{GenIW}, we have developed a new prescription for $S$-expansion which involves an \textit{infinite abelian semigroup} and the subsequent \textit{removal of an infinite subalgebra}. 
We have shown that the ideal subtraction corresponds to a \textit{reduction} (in the sense intended in Chapter \ref{chapter 3}) and, in particular, it can be viewed as a (generalization of the) \textit{$0_S$-reduction of $S$-expanded algebras}.

The ``infinite $S$-expansion'' is an extension and generalization of the finite case and, with the subtraction of an infinite ideal subalgebra from an infinite resonant subalgebra of the infinitely $S$-expanded (super)algebra, it also offers an alternative view of the generalized In\"on\"u-Wigner contraction. Indeed, in \cite{GenIW} we have explicitly shown how to reproduce a generalized In\"on\"u-Wigner contraction within our scheme. 
The In\"on\"u-Wigner contraction does not change the dimension of the original (super)algebra. Thus, the subtraction of the infinite ideal subalgebra here is crucial, since it allows to end up with Lie (super)algebras with a finite number of generators.

We have also given a theorem for writing the invariant tensors for the (super)algebras obtained by applying our method of infinite $S$-expansion with ideal subtraction. Indeed, since the ideal subtraction  can be viewed as a $0_S$-reduction, one can then apply Theorem VII.2 of Ref. \cite{Iza1}, ending up with the invariant tensors for the $0_S$-reduced (super)algebras. 
This is very useful, since it allows to develop the dynamics and construct the Lagrangian of several physical theories, starting from their algebraic structure (in particular, in this context the construction of Chern-Simons forms becomes more accessible).
In the current chapter we will also recall this new prescription for $S$-expansion, following what we have done in \cite{GenIW}.

\section{An analytic method for $S$-expansion}\label{Method}

In this section, we describe the analytic method we have developed in \cite{Analytic} for linking different (super)algebras in the context of $S$-expansion.

To this aim, let us first of all consider a finite Lie (super)algebra $\mathfrak{g}$, with basis $\lbrace T_A \rbrace$, which can be decomposed into $n$ subspaces $V_{p}$, with $p=0,1,\ldots,n-1$, namely $\mathfrak{g}=\bigoplus _{p=0}^{n-1} V_{p}$. 

Then, let us consider another Lie (super)algebra $\tilde{\mathfrak{g}}$ (our target). Here and in the following, the quantities having a ``tilde'' symbol above will be quantities pertaining to the target (super)algebra. 

By definition, the (graded) Jacobi identities are clearly satisfied for both $\mathfrak{g}$ and $\tilde{\mathfrak{g}}$.

Let us also define a discrete, finite, and \textit{abelian magma}\footnote{A \textit{magma} (or also \textit{groupoid}) is a basic algebraic structure, consisting of a set equipped with a single, binary operation that is closed by definition, without any other requirement.} $\tilde{S}$ with $P$ elements as follows:
\begin{equation}
\tilde{S} = \lbrace \lambda_0, \lambda_1 ,\ldots, \lambda_{P-1} \rbrace , \quad \lambda_\alpha \lambda_\beta = \lambda_\gamma, \quad \lambda_\alpha \lambda_\beta = \lambda_\beta \lambda_\alpha , \; \forall \lambda_\alpha , \, \lambda_\beta , \, \lambda_\gamma \in \tilde{S} .
\end{equation}
Now, let us decompose $\tilde{S}$ into $n$ subsets $S_{\Delta_p}$, $p=0,1,\ldots,n-1$, such that we can write the following partition:
\begin{equation}\label{decompinizio}
\tilde{S}= \sqcup _{\Delta_p } S_{\Delta_p } ,
\end{equation}
where with the symbol $\sqcup $ we denote the disjoint union of (sub)sets.
The composed index $\Delta_p$ labeling the subsets $S_{\Delta_p}$, $p=0,1,\ldots,n-1$, takes into account the cardinality (number of elements) of each subsets through the capital Greek letter $\Delta$.

We now give the conditions under which our analytic method can be applied, that is we require:
\begin{enumerate}
\item The (target) Lie (super)algebra to be decomposed into $n$ subspaces $\tilde{V}_{p}$, $p=0,1,\ldots,n-1$, that is $\tilde{\mathfrak{g}}=\bigoplus _{p=0}^{n-1} \tilde{V}_{p}$. This means that we request that the number of subspaces in the decomposition of $\tilde{\mathfrak{g}}$ is equal to the number of subspaces in the decomposition of the initial (super)algebra $\mathfrak{g}=\bigoplus _{p=0}^{n-1} V_{p}$.
\item The dimensions of the subspaces of the Lie (super)algebra $\tilde{\mathfrak{g}}$ to be multiples of the dimensions of the subspaces of $\mathfrak{g}$.
\item $V_0$ and $\tilde{V}_0$ to be the vector spaces of the subalgebras $\mathfrak{h}_0$ and $\tilde{\mathfrak{h}}_0$ of $\mathfrak{g}$ and $\tilde{\mathfrak{g}}$, respectively.
\item The subspace decomposition of $\mathfrak{g}$ to satisfy the following Weimar-Woods conditions \cite{WW1, WW2} (see equation (\ref{conditionsWW}) in Chapter \ref{chapter 3}):
\begin{equation}\label{WWperg}
[V_p , V_q]\subset \bigoplus_{s\leq p + q} V_s,
\end{equation}
$p,q=0,1,\ldots,n-1$. Analogously, we require the subspace decomposition of $\tilde{\mathfrak{g}}$ to satisfy
\begin{equation}\label{WWpergtilde}
[\tilde{V}_p , \tilde{V}_q]\subset \bigoplus_{s\leq p + q} \tilde{V}_s,
\end{equation}
$p,q=0,1,\ldots,n-1$.
\item The partition (\ref{decompinizio}) to be \textit{in resonance} (see Definition \ref{defresonant} in Chapter \ref{chapter 3}) with the decomposition of $\mathfrak{g}$ into subspaces  $\mathfrak{g}=\bigoplus _{p=0}^{n-1} V _{p}$, namely we require:
\begin{equation}\label{prrrrrrrrrrrrrrrruff}
S_{\Delta_p}\cdot S_{\Delta_q} \subset \bigcap_{s \leq p + q} S_{\Delta_s},
\end{equation}
where the product $S_{\Delta_p}\cdot S_{\Delta_q} $ is defined as
\begin{equation}\label{prodottellobello}
S_{\Delta_p}\cdot S_{\Delta_q}  = \lbrace \lambda_\gamma \mid \lambda_\gamma= \lambda_{\alpha_p}\lambda_{\alpha_q}, \; \text{with} \; \lambda_{\alpha_p}\in S_{\Delta_p}, \lambda_{\alpha_q}\in S_{\Delta_q} \rbrace \subset \tilde{S}.
\end{equation}
Under resonance condition, the association between the subsets $S_{\Delta_p}$'s of $\tilde{S}$ and the subspaces $V_p$'s of $\mathfrak{g}$ is uniquely determined ($S_{\Delta_p} \leftrightarrow V_p$).
\end{enumerate}

Thus, our method will be limited to these choices on the subspace decomposition of $\mathfrak{g}$ and $\tilde{\mathfrak{g}}$ and on the subsets partition of $\tilde{S}$; before applying our analytic method, one should always fulfill these requirements. On the other hand, the above assumptions can also be seen as a convenient (even if restrictive) criterion for choosing the decomposition of Lie (super)algebras and a proper resonant partition for the semigroup involved in the procedure in the context of $S$-expansion.

Now, let us suppose that $\tilde{S}$ also fulfills \textit{associativity} (that will be checked at the final step of the method). Thus, $\tilde{S}$ is now supposed to be a semigroup, and we will denote it with $S$, according with the notation of \cite{Iza1}, where the $S$-expansion (that is semigroup expansion) procedure was developed for the first time.

Now, let us define $\mathfrak{g}_S \equiv S \times \mathfrak{g}$. In this way, $\mathfrak{g}_S$ results to be an $S$-expanded (super)algebra obtained by $S$-expanding $\mathfrak{g}$ (see Chapter \ref{chapter 3}). Then, according to Theorem \ref{Tres} recalled in Chapter \ref{chapter 3}, we have that $\tilde{\mathfrak{g}}=\bigoplus_{p=0}^{n-1} W_p$, where $W_p=S_{\Delta_p} \times V_p$, is a \textit{resonant subalgebra} of the $S$-expanded (super)algebra $\mathfrak{g}_S$. Furthermore, by construction, we have that the subspace structure of $\tilde{\mathfrak{g}}$ is, in turn, in resonance with the partition (\ref{prrrrrrrrrrrrrrrruff}) of $S$, and we can write:
\begin{equation}\label{resredG}
\tilde{\mathfrak{g}} =V_0 \oplus V_1 \oplus \cdots \oplus V_{n-1}= \left(S_{\Delta_0}\times V_0\right)\oplus \left(S_{\Delta_1}\times V_1 \right) \oplus \cdots \oplus  \left(S_{\Delta_{n-1}}\times V_{n-1} \right).
\end{equation}
Thus, the following system of equations naturally arises:
\begin{equation}\label{systemres}
\left\{
\begin{aligned}
& \dim \left( \tilde{V}_{0}\right) =\dim \left( V_{0}\right) \left( 
\Delta_{0}\right),   \\
& \dim \left( \tilde{V}_{1}\right) =\dim \left( V_{1}\right) \left( 
\Delta_{1}\right),    \\
& \;\;\;\;\;\;\;\;\;\;\;\;\;\;\;\;\;\;\;\; \vdots  \\
& \dim \left( \tilde{V}_{n-1}\right) =\dim \left( V_{N-1}\right) \left( 
\Delta_{n-1}\right),  \\
& \tilde{P} =\sum_{p=0}^{n-1}\Delta _{p},
\end{aligned} \right.
\end{equation}
where $\tilde{P} \geq P$ (let us recall that $P$ is the total number of (non-zero) elements of $S$).

Let us mention that, by taking into account the possible presence of a (unique) \textit{zero element} $\lambda_{0_S} \in S$ (which can always be factorized out of the commutation relations) such that
\begin{equation}
\lambda_{0_S} \lambda_\alpha = \lambda_\alpha \lambda_{0_S} \equiv \lambda_{0_S}, \; \forall \lambda_\alpha \in S,
\end{equation}
the system (\ref{systemres}) acquires the following form:\footnote{And the whole procedure goes under the name of \textit{$0_S$-resonant-reduction}, since, imposing $\lambda_{0_S} T_A=\mathbf{0}$ (being $\lbrace T_A \rbrace$ the set of generators of $\mathfrak{g}$), we end up with a $0_S$-reduced algebra of the resonant subalgebra $\tilde{\mathfrak{g}}$ (see Chapter \ref{chapter 3} for details).}
\begin{equation}\label{system}
\left\{
\begin{aligned}
& \dim \left( \tilde{V}_{0}\right)  =\dim \left( V_{0}\right) \left( 
\Delta_{0}\right), \\
& \dim \left( \tilde{V}_{1}\right)  =\dim \left( V_{1}\right) \left( 
\Delta_{1}\right),    \\
& \;\;\;\;\;\;\;\;\;\;\;\;\;\;\;\;\;\;\;\; \vdots  \\
& \dim \left( \tilde{V}_{n-1}\right) =\dim \left( V_{n-1}\right) \left( 
\Delta_{n-1}\right),  \\
& \tilde{P} =\sum_{p=0}^{n-1}\Delta _{p}+1 ,
\end{aligned} \right.
\end{equation}
where, again, $\tilde{P} \geq P$, and where the $+1$ contribution now appearing in $\tilde{P} =\sum_{p=0}^{n-1}\Delta _{p}+1$ is due to the presence of the zero element $\lambda_{0_S}$.

Now, solving the system (\ref{systemres}) (or, equivalently, (\ref{system}), when $S$ is endowed with a zero element) we will immediately know the cardinality of each of the subsets $S_{\Delta_p}$. 
Indeed, both the systems (\ref{systemres}) and (\ref{system}) can be solved with respect to the variables $\tilde{P}$ and $\Delta_p$, and the unique (by construction) solution admits only values in $\mathbb{N}^*$ (the value zero is obviously excluded). 

In this way, we are left with the knowledge of the cardinality of each subset of $S$. Furthermore, by construction, we already know the partition structure of $S$ (which is in resonance with the subspace decomposition of $\mathfrak{g}$ (and $\tilde{\mathfrak{g}}$)).

Let us now complete this first part with the following theorem, which we have written and proven in \cite{Analytic}:
\begin{theorem}\label{T}
In the $S$-expansion procedure, when the commutator of two generators in the original Lie (super)algebra falls into a linear combination involving more than one generator, all the terms appearing in this resultant linear combination of generators must share the same element of $S$.
\end{theorem} 
\begin{proof}
The proof of Theorem \ref{T} can be treated as a \textit{reductio ad absurdum}. Indeed, if the linear combination of generators were coupled with different elements of $S$, this would mean having different two-selectors associated with the same resulting element, and, according to the definition of two-selector given in (\ref{kseldef}) of Chapter \ref{chapter 3}, this would break the uniqueness of the internal composition law of $S$. 
But this cannot be true, since the composition law associates each couple of elements $\lambda_\alpha$ and $\lambda_\beta$ with a \textit{unique} element $\lambda_\gamma$ (see again (\ref{kseldef}) in Chapter \ref{chapter 3}). 
We can thus conclude that when the commutator of two generators in the original Lie (super)algebra falls into a linear combination involving more than one generator, the terms appearing in this resultant linear combination of generators must be multiplied by the same element.
\end{proof}
Theorem \ref{T} also reflects on the subspace structure of the original Lie (super)algebra and on the subsets partition of $S$.

We have thus exhausted the information coming from the starting (super)algebra $\mathfrak{g}$. We can now exploit the information coming from the target (super)algebra $\tilde{\mathfrak{g}}$ to fix some detail on the multiplication rules and to build up the whole multiplication table of $S$.\footnote{Let us recall that, at the end of the whole procedure, one should check the associativity of $S$, since, till now, we have just \textit{hypothesized} that $S$ is a semigroup.} This step is based on the following \textit{identification criterion}.  

\subsection{Identification criterion between generators}\label{IdCr}

We can now write (according with what we have recalled in Chapter \ref{chapter 3}) the following identification between generators: 
\begin{equation} \label{identification}
\tilde{T}_{a_p} = T_{{a_p},\alpha} \equiv \lambda_\alpha T_{a_p},
\end{equation}
being $\lbrace T_{a_p} \rbrace$ a basis of $V_p$, while $\lbrace \tilde{T}_{a_p}\rbrace$ is a basis of $\tilde{V}_p$; $\lambda_\alpha \in S$ denotes a general element of $S$.
We have to perform the identification (\ref{identification}) for each element of $S$, that is we have to associate each element of each subset with the generators in the subspace related to the considered subset. 
This means
\begin{equation}
\tilde{T}_{a_p} = \lambda_{(\alpha,\Delta_p )}T_{a_p},
\end{equation}
where $\lambda_{(\alpha, \Delta_p)}\equiv \lambda_{\alpha} \in S_{\Delta_p}$.

Now, observe that for $\tilde{\mathfrak{g}}$ we can write the commutation relations
\begin{equation}\label{commtarg}
\left[ \tilde{T}_{a_p}, \tilde{T}_{b_q} \right] =\tilde{C}_{a_p b_q}^{\;\;\;\;\;\;\;\;c_r} \tilde{T}_{c_r},
\end{equation}
where we have denoted by $\tilde{C}_{a_p b_q}^{\;\;\;\;\;\;\;\;c_r}$ the structure constants of $\tilde{\mathfrak{g}}$; due to the identification (\ref{identification}), they read $\tilde{C}_{a_p b_q}^{\;\;\;\;\;\;\;\;c_r} \equiv C_{(a_p,\alpha)(b_q,\beta)}^{\;\;\;\;\;\;\;\;\;\;\;\;\;\;\;\;\;\;\;\;\;\;\;(c_r,\gamma)}$, where $ C_{a_p b_q}^{\;\;\;\;\;\;\;\;c_r}$ are the structure constants of $\mathfrak{g}$.

Then, since for $\mathfrak{g}$ we can write
\begin{equation}\label{commstart}
\left[T_{a_p}, T_{b_q} \right]=  C_{a_p b_q}^{\;\;\;\;\;\;\;\;c_r} \; T_{c_r} ,
\end{equation}
we are now able to write the structure constants of $\tilde{\mathfrak{g}}$ in terms of the two-selector and of the structure constants of $\mathfrak{g}$ (see Chapter \ref{chapter 3}), namely
\begin{equation}
\tilde{C}_{a_p b_q}^{\;\;\;\;\;\;\;\;c_r} \equiv C_{(a_p,\alpha)(b_q,\beta)}^{\;\;\;\;\;\;\;\;\;\;\;\;\;\;\;\;\;\;\;\;\;\;\;(c_r,\gamma)}= K_{\alpha \beta}^{\;\;\;\gamma}C_{a_p b_q}^{\;\;\;\;\;\;\;\;c_r} .
\end{equation}

Subsequently, we can exploit the identification (\ref{identification}) in order to write the commutation relations of $\tilde{\mathfrak{g}}$, namely (\ref{commtarg}), in terms of the commutation relations between the expanded generators of $\mathfrak{g}$ (factorizing the elements of $S$ out of the commutators). Thus, we end up with the following relations:
\begin{equation}\label{commrelafterid}
\left[\lambda_\alpha T_{a_p}, \lambda_\beta T_{b_q} \right] = K_{\alpha \beta}^{\;\;\;\; \gamma}  C_{a_p b_q}^{\;\;\;\;\;\;\;\;c_r} \lambda_\gamma T_{c_r}, \;\;\; \rightarrow \;\;\; \lambda_\alpha \lambda_\beta \left[T_{a_p}, T_{b_q} \right] =K_{\alpha \beta}^{\;\;\;\; \gamma}  C_{a_p b_q}^{\;\;\;\;\;\;\;\;c_r} \lambda_\gamma T_{c_r}.
\end{equation}
If we now compare the commutation relations (\ref{commrelafterid}) with the ones in (\ref{commstart}), we can deduce further information on the multiplication rules between the elements of $S$, that is to say:
\begin{equation}
\lambda_\alpha \lambda_ \beta = \lambda_\gamma .
\end{equation}
Let us stress that we should repeat this procedure for all the commutation rules of $\tilde{\mathfrak{g}}$, in order to get all the multiplication rules between the elements of $S$. 

During this process, the possible existence of the (unique) zero element $\lambda_{0_S} \in S$ can play a crucial role. Indeed, in the case in which the commutation relations of the target (super)algebra $\tilde{\mathfrak{g}}$ read
\begin{equation}
\left[ \tilde{T}_A, \tilde{T}_B \right] = 0,
\end{equation}
and, at the same time, from the initial (super)algebra $\mathfrak{g}$ we have
\begin{equation}\label{nozero}
\left[ T_A, T_B \right] \neq 0,
\end{equation}
considering the relations
\begin{equation}
\left[\lambda_\alpha T_A, \lambda_\beta T_B \right] = \lambda_\alpha \lambda_\beta \left[ T_A,  T_B \right]  =  C_{AB}^{\;\;\;\;C} \lambda_\gamma T_C =0
\end{equation}
together with (\ref{nozero}), we conclude that
\begin{equation}
\lambda_\alpha \lambda_ \beta = \lambda_{0_S}.
\end{equation}
If we then impose $\lambda_{0_S} T_A = \mathbf{0}$, we end up with the \textit{$0_S$-reduced} algebra of the resonant subalgebra $\tilde{\mathfrak{g}}$ of the $S$-expanded (super)algebra $\mathfrak{g}_S = S \times \mathfrak{g}$.

At the end of the whole procedure, in any case, we are left with the complete multiplication table(s) describing the abelian semigroup(s) $S$ for moving from the initial Lie (super)algebra $\mathfrak{g}$ to the target one $\tilde{\mathfrak{g}}$. One can end up with more than one abelian semigroup linking the two (super)algebras; this results in a degeneracy of the multiplication table. When performing our analytic method, however, one just supposes the abelian magma $\tilde{S}$ to fulfill associativity and, thus, to be a semigroup $S$. The final check for associativity can remove (completely or in part) the degeneracy appearing in the multiplication table.

Let us mention here that, in the cases in which the (graded) Jacobi identities of the initial (super)algebra $\mathfrak{g}$ are trivially satisfied (that is, each term of the Jacobi identities is equal to zero), the abelian magma $\tilde{S}$ does not necessarily need to be a semigroup (namely, to fulfill associativity) in order end up with a consistent result after having applied our analytic method. Indeed, in those cases, after having performed the identification (\ref{identification}), we can write the (graded) Jacobi identities of the $S$-expanded (super)algebra $\mathfrak{g}_S= \tilde{S} \times \mathfrak{g}$ in terms of the generators of $\mathfrak{g}$, factorizing (thanks to the fact that the magma $\tilde{S}$ is abelian) the elements of $\tilde{S}$, and this will not give any constraints on the multiplication rules among the elements of $\tilde{S}$; in particular, $\tilde{S}$ does not have to fulfill associativity. This means that, when the (graded) Jacobi identities of the initial (super)algebra $\mathfrak{g}$ are trivially satisfied, $\tilde{\mathfrak{g}} = \bigoplus_{p=0}^{n-1} W_p $, with $W_p= S_{\Delta_p} \times V_p$, is a resonant subalgebra of $\mathfrak{g}_S=\tilde{S} \times \mathfrak{g}$, where $\tilde{S}$ is just an abelian magma. We can then perform all the above procedure and end up with the multiplication table(s) associated with $\tilde{S}$ (also in this case, we can end up with more than one abelian magma).

\subsection{A simple algorithm to check associativity}\label{Associativity}

The final step of our method consist in checking that $\tilde{S}$ is indeed an abelian semigroup, $S$.\footnote{Which, as we have already said, is not required only in the cases in which the (graded) Jacobi identities of the initial (super)algebra $\mathfrak{g}$ are trivially satisfied.}

The check for associativity can be rather tedious if performed by hand, but, fortunately, it can be implemented by means of a simple computational algorithm. 
Indeed, mapping all the elements $\lambda_\alpha$'s of $\tilde{S}$ to the set of the integer numbers $\lambda_\alpha \leftrightarrow \alpha \in\mathbb{N}$, the multiplication table of $S$ can then be stored as a matrix $M$, such that
\begin{equation}
\lambda_\beta \lambda_\gamma =\lambda_\alpha \leftrightarrow M(\beta,\gamma)= \alpha ,
\end{equation}
where $\alpha$ is the index associated with the element $\lambda_\alpha$. Associativity can now be easily tested by checking that, for any $\alpha$, $\beta$, and $\gamma$, the following relation holds:
\begin{equation}\label{eq_ass_test}
M(M(\alpha,\beta),\gamma)=M(\alpha, M(\beta, \gamma)).
\end{equation}

\subsection{Example of application: From $\osp(1|32)$ to the $M$-algebra} \label{Examples}

In our paper \cite{Analytic}, the reader can find many examples of application of our analytic method.
In the following, we will present an example involving the Lie superalgebra $\mathfrak{osp}(1 \vert 32)$ and the $M$-algebra (which are related to $D=11$ supergravity and have already been presented and discussed in Chapter \ref{chapter 5}), reproducing the result presented in \cite{Iza1}.

The authors of \cite{Iza1} also considered a ``D'Auria-Fr\'{e}-like'' superalgebra and another $D=11$ superalgebra different from but resembling in some aspects to both the $M$-algebra and the DF-algebra, always in the context of $S$-expansion. 
The D'Auria-Fr\'{e}-like superalgebra of \cite{Iza1} presents the same structure (the same number and type of generators, with commutators valued on the same subspaces) as the one introduced by R. D'Auria and P. Fré in \cite{D'AuriaFre} (DF-algebra), but some details are different, so that it \textit{cannot} really correspond the DF-algebra. 
Indeed, according to the results presented in Chapter \ref{chapter 5} of this thesis, the DF-algebra \textit{cannot} be \textit{directly} related through $S$-expansion to $\mathfrak{osp}(1|32)$, while this is allowed when dealing with the D'Auria-Fr\'{e}-like superalgebra considered in \cite{Iza1}.

Let us also observe that, instead, the $M$-algebra can be obtained from the DF-algebra with a suitable truncation of the nilpotent fermionic generator $Q'$. Indeed, due to the presence of the spinor $1$-form $\eta$ (dual to the generator $Q'$) in the DF-algebra, the latter can be viewed as a spinor central extension of the $M$-algebra.
One can thus move from $\mathfrak{osp}(1|32)$ to the DF-algebra by performing two subsequent steps: First, one has to go from $\mathfrak{osp}(1|32)$ to the $M$-algebra (we will see in a while the procedure one can adopt to perform this step); then, one has to centrally extend the $M$-algebra with an extra (nilpotent) fermionic generator.

We now want to find the semigroup leading from $\mathfrak{osp}(1 \vert 32)$ to the $M$-algebra through our analytic method.

Let us first collect the useful information coming from the starting algebra $\mathfrak{osp}(1 \vert 32)$.\footnote{The detailed (anti)commutation relations for $\mathfrak{osp}(1 \vert 32)$ and for the $M$-algebra can be found in Ref. \cite{Iza1}.} 

The generators of $\mathfrak{osp}(1 \vert 32)$ are given, with respect to the Lorentz subgroup $SO(1,10)\subset OSp(1 \vert 32)$, by the following set:
\begin{equation}
\lbrace{P_a,J_{ab},Z_{a_1 \ldots a_5},Q_\alpha \rbrace},
\end{equation}
where $J_{ab}$, $P_a$, $Q_\alpha$ can be respectively interpreted as the Lorentz, translations and supersymmetry generators; $Z_{a_1 \ldots a_5}$ is a $5$-indexes skew-symmetric generator.
Let us perform the following subspace decomposition on the $\mathfrak{osp}(1 \vert 32)$ algebra:
\begin{equation}
\mathfrak{osp}(1 \vert 32) = V_0 \oplus V_1 \oplus V_2,
\end{equation}
\begin{eqnarray}
\left[ V_{0},V_{0}\right] &\subset &V_{0} ,\\
\left[ V_{0},V_{1}\right] &\subset &V_{1} ,\\
\left[ V_{0},V_{2}\right] &\subset &V_{2} ,\\
\left[ V_{1},V_{1}\right] &\subset &V_{0}\oplus V_{2}, \\
\left[ V_{1},V_{2}\right] &\subset &V_{1} ,\\
\left[ V_{2},V_{2}\right] &\subset &V_{0}\oplus V_{2},
\end{eqnarray}
where we have set $V_0= \lbrace J_{ab} \rbrace$, $V_1=\lbrace Q_\alpha \rbrace$, and $V_2=\lbrace P_a, Z_{a_1 \ldots a_5} \rbrace$. 
Thus, the dimensions of the internal decomposition of $\mathfrak{osp}(1 \vert 32)$ read
\begin{eqnarray}
\dim \left( V_{0}\right) &=&\underset{J_{ab}}{\underbrace{55}} ,\\
\dim \left( V_{1}\right) &=&\underset{Q_{\alpha }}{\underbrace{32}}, \\
\dim \left( V_{2}\right) &=&\underset{P_{a}}{\underbrace{11}}+\underset{%
Z_{a_1 \ldots a_5}}{\underbrace{462}}=473.
\end{eqnarray}

Now, we do an analogous analysis for the $M$-algebra (our target). Its generators (which we denote with an upper ``tilde'' symbol) are given by the set
\begin{equation}
\lbrace \tilde{P}_a, \tilde{J}_{ab}, \tilde{Z}_{ab},\tilde{Z}_{a_1 \ldots a_5},\tilde{Q}_\alpha \rbrace .
\end{equation}
We can thus proceed by performing the following subspace decomposition on the $M$-algebra:
\begin{equation}
\text{$M$-algebra} = V_0 \oplus V_1 \oplus V_2,
\end{equation}
\begin{eqnarray}
\left[ \tilde{V}_{0},\tilde{V}_{0}\right] &\subset &\tilde{V}_{0} ,\\
\left[ \tilde{V}_{0},\tilde{V}_{1}\right] &\subset & \tilde{V}_{1} ,\\
\left[ \tilde{V}_{0},\tilde{V}_{2}\right] &\subset & \tilde{V}_{2} ,\\
\left[ \tilde{V}_{1},\tilde{V}_{1}\right] &\subset &\tilde{V}_{0}\oplus \tilde{V}_{2}, \\
\left[ \tilde{V}_{1},\tilde{V}_{2}\right] &\subset & \emptyset ,\\
\left[ \tilde{V}_{2},\tilde{V}_{2}\right] &\subset & \emptyset ,
\end{eqnarray}
where we have set $\tilde{V}_0=\lbrace \tilde{J}_{ab},\tilde{Z}_{ab}\rbrace$, $\tilde{V}_1=\lbrace \tilde{Q}\rbrace$, and $\tilde{V}_2=\lbrace \tilde{P}_a, \tilde{Z}_{a_1 \ldots a_5} \rbrace$ (as usual, each subspace also contains the empty set $\emptyset$).
We can see that
\begin{eqnarray}
\dim \left( \tilde{V}_{0}\right) &=&\underset{\tilde{J}_{ab},\; \tilde{Z}_{ab}}{\underbrace{110}} ,\\
\dim \left( \tilde{V}_{1}\right) &=&\underset{Q_{\alpha }}{\underbrace{32}}, \\
\dim \left( \tilde{V}_{2}\right) &=&\underset{\tilde{P}_{a}}{\underbrace{11}}+\underset{%
\tilde{Z}_{a_1 \ldots a_5}}{\underbrace{462}}=473.
\end{eqnarray}

Both the superalgebra $\mathfrak{osp}(1 \vert 32)$ and the $M$-algebra satisfy the Weimar-Woods conditions (\ref{WWperg}) and (\ref{WWpergtilde}), and they also fulfill the other requirements allowing to apply our method.

We can now proceed with the study of the system (\ref{system}), which, in this case, reads
\begin{equation}
\left\{
\begin{aligned}
& 110 =55 \Delta_0, \\
& 32 =32 \Delta_1, \\
& 473 =473 \Delta_2, \\
& \tilde{P} = \Delta_0 + \Delta_1 + \Delta_2 +1,
\end{aligned} \right.
\end{equation}
where $\Delta_0$, $\Delta_1$, $\Delta_2$ denote the cardinality of the subsets related to the subspaces $V_0$, $V_1$, and $V_2$, respectively, and we have taken into account the possible existence of the zero element of the set $\tilde{S}$ involved in the procedure.
This system admits the unique solution
\begin{equation}
\tilde{P}=5, \;\;\; \Delta_0 = 2, \;\;\; \Delta_1 = 1, \;\;\; \Delta_2 = 1.
\end{equation}
Thus, we can now write the following (resonant) subset partition of $\tilde{S}$:
\begin{eqnarray}
S_{2_0} &=& \lbrace \lambda_\alpha, \lambda_\beta \rbrace , \\
S_{1_1} &=& \lbrace \lambda_\gamma \rbrace , \\
S_{1_2} &=& \lbrace \lambda_\delta  \rbrace . 
\end{eqnarray}
Then, taking into account the resonance condition, we can write the multiplication rules
\begin{eqnarray}
\lambda_{\alpha,\beta} \lambda_{\alpha,\beta} &=& \lambda_{\alpha,\beta,0_S}, \\
\lambda_{\alpha,\beta} \lambda_{\gamma} &=& \lambda_{\gamma,0_S}, \\
\lambda_{\alpha,\beta} \lambda_{\delta} &=& \lambda_{\delta,0_S}, \\
\lambda_{\gamma} \lambda_{\gamma} &=& \lambda_{\alpha,\beta,\delta,0_S}, \\
\lambda_{\gamma} \lambda_{\delta} &=& \lambda_{\delta,0_S}, \\
\lambda_{\delta} \lambda_{\delta} &=& \lambda_{\alpha,\beta,\delta,0_S},
\end{eqnarray}
where we have already taken into account the fact that 
\begin{eqnarray}
\lambda_{0_S}\lambda_{0_S} & = &\lambda_{0_S}, \\
\lambda_{0_S}\lambda_{\alpha,\beta,\gamma,\delta} & = &\lambda_{0_S},
\end{eqnarray}
by definition of zero element.

We can now fix the degeneracy appearing in the above multiplication rules by analyzing the information coming from the target superalgebra, that is, in this case, the $M$-algebra. To this aim, we first perform the identification
\begin{equation}
\lambda_\alpha J_{ab} = \tilde{J}_{ab}, \;\; \lambda_\beta J_{ab} = \tilde{Z}_{ab}, \;\; \lambda_\gamma Q = \tilde{Q}, \;\; \lambda_\delta P_a = \tilde{P}_a, \;\; \lambda_\delta Z_{a_1 \ldots a_5} = \tilde{Z}_{a_1 \ldots a_5}.
\end{equation}
Then, we have to write the commutation relations between the generators of the $M$-algebra in terms of the commutation relations between the generators of the $S$-expanded $\mathfrak{osp}(1 \vert 32)$. The details of this calculation can be found in Section \ref{osp} of Appendix \ref{appanex}, while here we just report and discuss the results we end up with.

The whole procedure fixes the degeneracy of the multiplication rules (and, in particular, $\lambda_\delta = \lambda_\beta$) between the elements of the subsets of $\tilde{S}$, and we are finally able to write the complete multiplication table of $\tilde{S}$, which reads
\begin{equation}
\begin{array}{c|cccc}
 & \lambda _{\alpha} & \lambda _{\beta} & \lambda _{\gamma} & \lambda _{0_S}  \\ 
\hline
\lambda _{\alpha} & \lambda _{\alpha} & \lambda _{\beta} & \lambda _{\gamma} & \lambda _{0_S} 
\\ 
\lambda _{\beta} & \lambda _{\beta} & \lambda _{0_S} & \lambda _{0_S} & \lambda _{0_S}
\\ 
\lambda _{\gamma} & \lambda _{\gamma} & \lambda _{0_S} & \lambda _{\beta} & \lambda _{0_S}
\\ 
\lambda _{0_S} & \lambda _{0_S} & \lambda _{0_S} & \lambda _{0_S} & \lambda _{0_S}
\end{array}
\end{equation}
Then, after having performed the identification
\begin{equation}
\alpha \leftrightarrow  0, \;\;\; \beta \leftrightarrow  2, \;\;\; \gamma \leftrightarrow  1, \;\;\;
0_S \leftrightarrow  3,
\end{equation} 
we can reorganize the multiplication table above as follows:
\begin{equation}
\begin{array}{c|cccc}
 & \lambda _{0} & \lambda _{1} & \lambda _{2} & \lambda _{3}  \\ 
\hline
\lambda _{0} & \lambda _{0} & \lambda _{1} & \lambda _{2} & \lambda _{3} 
\\ 
\lambda _{1} & \lambda _{1} & \lambda _{2} & \lambda _{3} & \lambda _{3}
\\ 
\lambda _{2} & \lambda _{2} & \lambda _{3} & \lambda _{3} & \lambda _{3}
\\ 
\lambda _{3} & \lambda _{3} & \lambda _{3} & \lambda _{3} & \lambda _{3}
\end{array}
\end{equation}
This is exactly the multiplication table of the \textit{semigroup} commonly denoted by $S^{(2)}_E$, which satisfies the following multiplication rules:
\begin{equation}
\lambda_\alpha \lambda_\beta = \left\{ \begin{aligned} &
\lambda_{\alpha+\beta} , \;\;\;\;\; \text{when} \; \alpha+\beta \leqslant 3,
\\ & \lambda_3 , \;\;\;\;\;\;\;\;\; \text{when} \; \alpha + \beta > 3 , \end{aligned} 
\right. \quad  \forall \alpha, \; \beta \in S^{(2)}_E.
\end{equation}

We can thus come to the same conclusions given in Ref. \cite{Iza1}, namely that $S^{(2)}_E$ is the semigroup leading, through an $S$-expansion procedure (actually, a $0_S$-resonant-reduction, due to the presence of the zero element $\lambda_{0_S}$), from $\mathfrak{osp}(1 \vert 32)$ to the $M$-algebra. Our prescription immediately allows to recover this result, without resort to any ``trial and error'' process. 

We have thus shown that our method is reliable and it can also be applied to rather complicated superalgebras in higher dimensions. This becomes particularly interesting when the superalgebras are associated with higher-dimensional supergravity theories (as in the case discussed in this example). One of the possible future developments could consist in developing extensions and generalizations of our analytic method.

\section{Infinite $S$-expansion with ideal subtraction} \label{Generalized}

In this section, we describe our prescription for infinite $S$-expansion (involving an infinite abelian semigroup $S^{(\infty)}$) with subsequent subtraction of an infinite ideal subalgebra, and we show how to reproduce a generalized In\"on\"u-Wigner contraction in this context. This method is both a new prescription for $S$-expansion and an alternative way of seeing the (generalized) In\"on\"u-Wigner contraction.
We also give a theorem for writing the invariant tensors of (super)algebras obtained through infinite $S$-expansion with ideal subtraction.
The discussion we present here is based on the work \cite{GenIW}, in collaboration with D. M. Pe\~{n}afiel.

As we will discuss in the following, the prescription for infinite $S$-expansion with subtraction of an infinite ideal subalgebra developed in \cite{GenIW} leads to reduced algebras (in the sense intended in Chapter \ref{chapter 3}). In particular, the subtraction of the infinite ideal can be viewed as a $0_S$-reduction involving an infinite number of elements which play the role of ``generating zeros". 

The subtraction of the ideal allows to obtain Lie (super)algebras with a finite number of generators, after having infinitely expanded the original Lie (super)algebras. 

\subsection{General formulation of our prescription}

In the $S$-expansion procedure, if the finite semigroup is generalized to the case of an
\textit{infinite semigroup}, then the $S$-expanded algebra will be an \textit{infinite-dimensional algebra} \cite{Caroca}.
We can thus generate an infinitely $S$-expanded algebra as a ``loop-like'' Lie algebra \cite{Caroca}, where the semigroup elements can be represented by the set $(\mathbb{N},+)$,\footnote{The loop algebra of \cite{Caroca} was constructed by considering the semigroup $(\mathbb{Z},+)$ (which is an abelian group with the sum operation). Here we restrict to the case of $(\mathbb{N},+)$, following \cite{GenIW}; this is the reason why we have written ``loop-like''.} that presents the same multiplication rules (extended to an infinite set) of the general semigroup $S^{(N)}_E = \lbrace{ \lambda_\alpha\rbrace}_{\alpha=0}^{N+1}$, that is to say: $\lambda_\alpha \lambda_\beta = \lambda_{\alpha+\beta}$ if $\alpha+\beta \leq N+1$, and $\lambda_\alpha \lambda_\beta = \lambda_{N+1}$ if $\alpha+\beta > N+1$. 

We now write the following definition from Ref. \cite{GenIW}:
\begin{definition}\label{definftysemigr}
Let $\left\{\lambda_\alpha\right\}_{\alpha=0}^{\infty}=\left\{\lambda_0,\lambda_1,\lambda_2, \ldots, \lambda_\infty \right\}$ be an infinite discrete set of elements. Then, the infinite set $\left\{\lambda_\alpha\right\}_{\alpha=0}^{\infty}$ satisfying commutation rules like the ones of the set $(\mathbb{N},+)$ (that is to say, of $S^{(N)}_E$), namely
\begin{equation}\label{prodsemigrinfinito}
\lambda_\alpha \lambda_\beta = \lambda_{\alpha + \beta},
\end{equation}
where 
\begin{equation}\label{infelementprod}
\lambda_\alpha \lambda_{\infty} = \lambda_\infty , \quad \forall \lambda_\alpha \in \left\{\lambda_\alpha\right\}_{\alpha=0}^{\infty}, \quad \text{and} \; \lambda_\infty \lambda_\infty = \lambda_\infty ,
\end{equation}
is an infinite abelian semigoup symbolized by $S^{(\infty)}$.
\end{definition}
Notice that, since the multiplication rules in (\ref{infelementprod}) hold, the element $\lambda_\infty \in S^{(\infty)}$ can be regarded as an ``ideal element'' of $S^{(\infty)}$.

Now, let $\mathfrak{g} = \bigoplus_{p \in I} V_p$ be a subspace decomposition of $\mathfrak{g}$.
We now perform an infinite $S$-expansion on $\mathfrak{g}$ using the semigroup $S^{(\infty)}$. 

The infinite $S$-expanded algebra can be rewritten as:
\begin{align}
\mathfrak{g}_S^\infty &=\left\{\lambda_\alpha\right\}^\infty_{\alpha=0}\times\mathfrak{g}= \nonumber\\
&=\left\{\lambda_\alpha\right\}^\infty_{\alpha=0}\times \left[ \bigoplus_{p \in I} V_p \right] \label{suma}.
\end{align}
Observe that the Jacobi identity is fulfilled for the infinite $S$-expanded algebra $\mathfrak{g}_S^\infty$. This is due to the fact that the starting algebra satisfies the Jacobi identity and the semigroup $S^{(\infty)}$ is abelian (and associative by definition of semigroup): These are the requirements that the starting algebra and the semigroup involved in the procedure must satisfy so that the $S$-expanded algebra satisfies the Jacobi identity when performing an $S$-expansion process \cite{Iza1}.

We can now split the infinite semigroup in subsets in such a way to be able to properly extract a resonant subalgebra from the infinitely $S$-expanded one, and define partitions on these subsets such that one can isolate an ideal structure from the resonant subalgebra. In this way, we will reproduce a reduction, ending up with a finitely generated algebra (see Chapter \ref{chapter 3} for details). 

Now, in order proceed with the extraction of the infinite resonant subalgebra, we have to define a resonant subset decomposition of the infinite semigroup $S^{(\infty)}$
\begin{equation}\label{subsetdecinf1}
S^{(\infty)} = \bigcup_{p \in I} S_p ,
\end{equation}
under the product
\begin{equation}\label{prodello}
S_p \cdot S_q = \lbrace \lambda_\gamma \mid \lambda_\gamma= \lambda_{\alpha_p}\lambda_{\alpha_q}, \; \text{with} \; \lambda_{\alpha_p}\in S_p, \lambda_{\alpha_q}\in S_q \rbrace \subset S^{(\infty)} ,
\end{equation}
that is to say, a decomposition such that equation (\ref{groupdecomposition}) of Chapter \ref{chapter 3} is fulfilled. In this case, the $S_p$'s are infinite subsets.

Once such a resonant subset decomposition has been found, the direct sum
\begin{equation}
\mathfrak{g}^\infty_R = \bigoplus_{p \in I} W_p ,
\end{equation}
with 
\begin{equation}
W_p = S_p \times V_p , \quad p \in I ,
\end{equation}
is a \textit{resonant subalgebra} of $\mathfrak{g}^\infty_S$ (we have used Theorem \ref{Tres} recalled in Chapter \ref{chapter 3}). 

In particular, we observe that $\mathfrak{g}^\infty_R $ is the direct sum of a finite number of infinite subspaces $W_p$, due to the fact that the subsets $S_p$'s contains an infinite amount of semigroup elements.

In \cite{GenIW}, we have then presented the following theorem:
\begin{theorem}\label{TeoRedIdeal}
Let $\mathfrak{g}$ be a Lie (super)algebra and let $\mathfrak{g}^\infty_S = S^{(\infty)}\times \mathfrak{g}$ be the infinite $S$-expanded (super)algebra obtained using the infinite abelian semigroup $S^{(\infty)}$. 
Let $\mathfrak{g}^\infty_R$ be an infinite resonant subalgebra of $\mathfrak{g}^\infty_S$ and let $\mathcal{I}$ be an infinite ideal subalgebra of $\mathfrak{g}^\infty_R$.
Then, the (super)algebra
\begin{equation}
\check{\mathfrak{g}}_R = \mathfrak{g}^\infty_R \ominus \mathcal{I}
\end{equation}
is a reduced (super)algebra.
\end{theorem}

\begin{proof}
After having performed the infinite $S$-expansion on $\mathfrak{g}$, obtaining $\mathfrak{g}^\infty_S = S^{(\infty)}\times \mathfrak{g}$, and after having extracted a resonant subalgebra $\mathfrak{g}^\infty_R $ from $\mathfrak{g}^\infty_S$, we can write an $S_p$ partition $S_p = \hat{S}_p \cup \check{S}_p$ (where the $\check{S}_p$'s are finite subsets, while the $\hat{S}_p$'s are infinite ones) satisfying the conditions (\ref{intzero}) and (\ref{ressemigroup}) of Chapter \ref{chapter 3}, which we also report here for completeness:
\begin{equation}\label{intzeronewella}
\hat{S}_p \cap \check{S}_p = \emptyset ,
\end{equation}
\begin{equation}\label{ressemigroupnewella}
\check{S}_p \cdot \hat{S}_q \subset \bigcap_{r \in i_{(p,q)}}\hat{S}_r.
\end{equation}
Once such a partition has been found, it induces, according with Theorem \ref{teoiza} recalled in Chapter \ref{chapter 3}, the following decomposition on the resonant subalgebra $\mathfrak{g}^\infty_R$:
\begin{equation}
\mathfrak{g}^\infty_R = \check{\mathfrak{g}}_R \oplus \hat{\mathfrak{g}}^\infty_R ,
\end{equation}
where
\begin{equation}
\check{\mathfrak{g}}_R = \bigoplus_{p\in I} \check{S}_p \times V_p , 
\end{equation}
\begin{equation}
\hat{\mathfrak{g}}^\infty_R = \bigoplus_{p \in I} \hat{S}_p \times V_p.
\end{equation}
Then, we have
\begin{equation}\label{reducedcond1}
\left[\check{\mathfrak{g}}_R , \hat{\mathfrak{g}}^\infty_R\right] \subset  \hat{\mathfrak{g}}^\infty_R ,
\end{equation}
and, therefore, $\vert \check{\mathfrak{g}}_R \vert$ correspond to a reduced (super)algebra of $\mathfrak{g}^\infty_S$.
Moreover, in the case in which
\begin{equation}
\left[\hat{\mathfrak{g}}^\infty_R , \hat{\mathfrak{g}}^\infty_R \right] \subset \hat{\mathfrak{g}}^\infty_R ,
\end{equation} 
$\hat{\mathfrak{g}}^\infty_R$ is, in particular, an \textit{infinite ideal subalgebra}, due to the fact that it also satisfies (\ref{reducedcond1}).

We can thus write
\begin{equation}
\check{\mathfrak{g}}_R = \mathfrak{g}^\infty_R \ominus \mathcal{I},
\end{equation}
where we have denoted by $\mathcal{I}$ the infinite ideal subalgebra, $\mathcal{I}\equiv \hat{\mathfrak{g}}^\infty_R$, and where $\check{\mathfrak{g}}_R$ corresponds to the reduced algebra we obtain at the end of the procedure.
\end{proof}

Notice that $\check{\mathfrak{g}}_R$ is \textit{finite}, since it is the direct sum of products between finite subsets and finite subspaces, while $\hat{\mathfrak{g}}^\infty_R$ is \textit{infinite}, due to the fact that the $\hat{S}_p$'s are infinite subsets.

Thus, we have explicitly shown that the subtraction of an infinite ideal subalgebra from an infinite resonant subalgebra of an infinite $S$-expanded (super)algebra corresponds to a \textit{reduction} and leads to a \textit{finite, reduced algebra} (in the sense intended in Chapter \ref{chapter 3}).

A reduced algebra, in general, does not correspond to a subalgebra; this means that the (super)algebra we end up with after having performed the ideal subtraction does not correspond, in general, to a subalgebra.

Let us finally observe that \textit{the ideal subtraction can be viewed as a (generalization of the) $0_S$-reduction}, in the sense that all the elements of the infinite ideal are mapped to zero after the ideal subtraction; this has the same effect which is usually produced by the zero element $\lambda_{0_S}$ of a semigroup, namely
\begin{equation}
\lambda_{0_S} T_A = \bf{0} .
\end{equation}
We are now conferring the role of ``generating zeros'' to a particular infinite set of generators: The ones belonging to the infinite ideal. The reduced algebra $\check{\mathfrak{g}}_R$ can be viewed, in this sense, as a $0_S$-reduced algebra.

One of the fundamental step in performing our procedure consists, after having found a resonant subset decomposition of $S^{(\infty)}$, in choosing properly the $S_p$ partition $S_p = \check{S}_p \cup \hat{S}_p$, in order to be able to extract an infinite ideal subalgebra $\mathcal{I}$ from the infinite resonant subalgebra $\mathfrak{g}^\infty_R$.

\subsection{How to reproduce a generalized In\"{o}n\"{u}-Wigner contraction}\label{GENSE}

In order to see how to reproduce a generalized In\"{o}n\"{u}-Wigner contraction by following the above prescription, let us now apply our method to the case in which the original Lie algebra $\mathfrak{g}$ can be decomposed into $n+1$ subspaces
\begin{equation}
\mathfrak{g} = V_0 \oplus V_1 \oplus \ldots \oplus V_n 
\end{equation}
satisfying the following Weimar-Woods conditions \cite{WW1, WW2}:
\begin{equation}
\left[V_p,V_q\right]\subset\bigoplus_{s\leq p+q}V_s,\quad p,q=0,1,\ldots,n .
\end{equation}

Now we can properly choose the subset partition of $S^{(\infty)}$ and apply our method of infinite $S$-expansion with ideal subtraction in order to show that the generalized In\"{o}n\"{u}-Wigner contraction fits within our scheme. 

First, we perform the infinite $S$-expansion on $\mathfrak{g}$, obtaining the $S$-expanded algebra
\begin{equation}\label{sumanew}
\mathfrak{g}^\infty_S = S^{(\infty)} \times \mathfrak{g} =\left( \lbrace \lambda_\alpha \rbrace _0 ^\infty \times V_0 \right) \oplus \left( \lbrace \lambda_\alpha \rbrace_0 ^\infty \times V_1 \right) \oplus \ldots \oplus \left( \lbrace \lambda_\alpha \rbrace_0^\infty \times V_n \right).
\end{equation}
Then, in order to be able to perform the extraction of the infinite resonant subalgebra, we have to split the semigroup $S^{(\infty)}$ into $n+1$ infinite subsets $S_p$ such that, when $\mathfrak{g}$ satisfies the Weimar-Woods conditions (that is our case), the following condition
\begin{equation}\label{rescond}
S_p \cdot S_q \subset \bigcap_{r\leq p+q} S_r
\end{equation}
is fulfilled. Here, $S_p \cdot S_q$ denotes the set of all products of all elements of $S_p$ and all elements of $S_q$. Thus, we must define such a decomposition for the infinite semigroup $S^{(\infty)}$. 

Now, let
\begin{equation}\label{subsetdecinf}
S^{(\infty)} = \bigcup_{p=0}^n S_p
\end{equation}
be a subset decomposition of $S^{(\infty)}$, where the subsets $S_p \subset S^{(\infty)}$ are defined by
\begin{equation}
S_p = \lbrace \lambda_{\alpha_p}, \; \alpha_p = p, \ldots, \infty  \rbrace , \quad p = 0, \ldots n.
\end{equation}
The subset decomposition (\ref{subsetdecinf}) is a \textit{resonant} one under the semigroup product (\ref{prodello}), since it satisfies (\ref{rescond}).
Then, according to Theorem \ref{Tres} recalled in Chapter \ref{chapter 3}, the direct sum
\begin{equation}
\mathfrak{g}^\infty_R = \bigoplus_{p=0}^n W_p ,
\end{equation}
with 
\begin{equation}
W_p = S_p \times V_p ,
\end{equation}
is a resonant subalgebra of $\mathfrak{g}^\infty_S$. 

Let us now consider $\mathfrak{g}^\infty_R $ and write the following $S_p$ partition: $S_p = \hat{S}_p \cup \check{S}_p $, where
\begin{equation}\label{partinf1}
\check{S}_p = \lbrace \lambda_{\alpha_p} , \; \alpha_p = p \rbrace  \equiv \lbrace \lambda_p \rbrace ,
\end{equation}
\begin{equation}\label{partinf2}
\hat{S}_p = \lbrace \lambda_{\alpha_p} , \; \alpha_p= p+1,\ldots , \infty  \rbrace .
\end{equation}
The $S_p$ partition just defined satisfies
\begin{equation}\label{firstcond}
\hat{S}_p \cap \check{S}_p = \emptyset ,
\end{equation}
which is exactly the condition (\ref{intzeronewella}).
The second condition that must be fulfilled in order to have the chance of extracting a reduced algebra from the resonant subalgebra $\mathfrak{g}^\infty_S$, when the original algebra $\mathfrak{g}$ satisfies the Weimar-Woods conditions, reads
\begin{equation}\label{secondcond}
\check{S}_p \cdot \hat{S}_q \subset \bigcap _{r \leq p+q} \hat{S}_r .
\end{equation} 
In the present case, $\check{S}_p$ and $\hat{S}_q$ are given by
\begin{equation}
\check{S}_p = \lbrace \lambda_{p} \rbrace , 
\end{equation}
\begin{equation}
\hat{S}_q = \lbrace \lambda_{\alpha_q}, \; \alpha_q = q+1, \ldots , \infty \rbrace ,
\end{equation}
respectively.
Thus, the condition (\ref{secondcond}) is fulfilled, since
\begin{equation}
\bigcap _{r=0}^{p+q} \hat{S}_r= \hat{S}_{p+q},
\end{equation}
where $\hat{S}_{p+q} = \lbrace \lambda_{p+q+m}, m =1, \ldots, \infty \rbrace$, and
\begin{equation}
\check{S}_p \cdot \hat{S}_q = \hat{S}_{p+q} ,
\end{equation}
where we have used (\ref{prodsemigrinfinito}). 

We can then conclude that the $S_p$ partition we have chosen satisfies the reduction condition, and we can now extract a reduced algebra from the resonant subalgebra $\mathfrak{g}^\infty_S$.
Indeed, what we have done induces, according to Theorem \ref{teoiza} recalled in Chapter \ref{chapter 3}, the following decomposition on the resonant subalgebra:
\begin{equation}
\mathfrak{g}^\infty_R = \check{\mathfrak{g}}_R \oplus \hat{\mathfrak{g}}^\infty_R ,
\end{equation}
where
\begin{equation}
\check{\mathfrak{g}}_R = \bigoplus_{p=0}^n \check{S}_p \times V_p , 
\end{equation}
\begin{equation}
\hat{\mathfrak{g}}^\infty_R = \bigoplus_{p=0}^n \hat{S}_p \times V_p.
\end{equation}
We can now write
\begin{equation}
\hat{\mathfrak{g}}^\infty_R = \bigoplus_{p=0}^n W' _p ,
\end{equation}
with
\begin{equation}
W' _p = \hat{S}_p \times V_p = S_{p+1} \times V_p ,
\end{equation}
where $S_{p+1}= \lbrace \lambda_{p+m}, \; m=1, \ldots , \infty \rbrace$. One can now easily prove that, by construction, we have
\begin{equation}
\left[\hat{\mathfrak{g}}^\infty_R , \hat{\mathfrak{g}}^\infty_R \right] \subset \hat{\mathfrak{g}}^\infty_R .
\end{equation}
This means that $\hat{\mathfrak{g}}^\infty_R$ is an infinite subalgebra of $\mathfrak{g}^\infty_R$ and, consequently, an infinite subalgebra of $\mathfrak{g}^\infty_S$. In particular, it is an ideal subalgebra, since it also satisfies 
\begin{equation}
\left[\check{\mathfrak{g}}_R , \hat{\mathfrak{g}}^\infty_R \right] \subset \hat{\mathfrak{g}}^\infty_R .
\end{equation}
Thus, we can finally write:
\begin{equation}
\check{\mathfrak{g}}_R = \mathfrak{g}^\infty_R \ominus \mathcal{I},
\end{equation}
and, applying Theorem \ref{TeoRedIdeal}, we can prove that the algebra $\check{\mathfrak{g}}_R$ corresponds to a \textit{reduced algebra}. Now, since the generalized In\"{o}n\"{u}-Wigner contraction corresponds to the reduction of a resonant subalgebra of the $S$-expanded one \cite{Iza1}, we can conclude that the generalized In\"{o}n\"{u}-Wigner fits within our scheme.

\subsection{Invariant tensors and infinite $S$-expansion}  
          
The authors of \cite{Iza1} developed a theorem for writing the components of the invariant tensor of a target algebra obtained through a finite $S$-expansion in terms of those of the initial algebra (see Theorem VII.1 of Ref. \cite{Iza1}). Furthermore, in Theorem VII.2 of the same paper, they gave an expression for the invariant tensor of a $0_S$-reduced algebra. 

Starting from their results, in \cite{GenIW} we presented the following theorem: 

\begin{theorem}\label{teolalala}
Let $\mathfrak{g}$ be a Lie (super)algebra of basis $\lbrace T_A \rbrace$ and let $\langle T_{A_0}\ldots T_{A_N}\rangle$ be an invariant tensor for $\mathfrak{g}$. 
Let $\mathfrak{g}^\infty_S = S^{(\infty)}\times \mathfrak{g}$ be the infinite $S$-expanded (super)algebra obtained using the infinite abelian semigroup $S^{(\infty)}$.
Let $\mathfrak{g}^\infty_R$ be an infinite resonant subalgebra of $\mathfrak{g}^\infty_S$ and let $\mathcal{I}$ be an infinite ideal subalgebra of $\mathfrak{g}^\infty_R$.
Then, 

\begin{equation}\label{topologicalInvariant}
\langle T^{\alpha_{p_0}}_{A_{p_0}}\ldots T^{\alpha_{p_N}}_{A{p_N}}\rangle =  \alpha^{m}\delta^{\alpha_{p_0} + \alpha_{p_1}+\alpha_{p_2}+\ldots + \alpha_{p_N}}_m \langle T_{A_0}\ldots T_{A_N}\rangle ,
\end{equation}
where the $\alpha^{m}$'s are arbitrary constants, corresponds to an invariant tensor for the finite (super)algebra

\begin{equation}
\check{\mathfrak{g}}_R = \mathfrak{g}^\infty_R \ominus \mathcal{I},
\end{equation}
having denoted the generators of $\check{\mathfrak{g}}_R$ by $\lambda_{\alpha_{p_i}} T_{A_{p_i}}\equiv T^{\alpha_{p_i}}_{A_{p_i}}$, $i=0, \ldots, N $, and where the set $\lbrace \lambda_{\alpha_p} \rbrace$ is finite. 
\end{theorem}

\begin{proof}
The proof of Theorem \ref{teolalala} can be developed by applying Theorem \ref{TeoRedIdeal}. Indeed, as stated in Theorem \ref{TeoRedIdeal}, we have that the subtraction of an infinite ideal subalgebra from an infinite resonant subalgebra of an infinitely $S$-expanded (super)algebra (using the semigroup $S^{(\infty)}$ on the original (super)algebra) corresponds to a reduction. In particular, we have seen that it reproduces the same result of a $0_S$-reduction. 

In this way, one can write the invariant tensor of the (super)algebra obtained with our method of infinite $S$-expansion with ideal subtraction by applying Theorem VII.2 of Ref. \cite{Iza1}, which, indeed, gives an expression for the invariant tensor for a $0_S$-reduced algebra.

Thus, it is straightforward to show that the invariant tensor for the (super)algebra $\check{\mathfrak{g}}_R=\mathfrak{g}_R^\infty\ominus\mathcal{I}$ can be written in the form
\begin{equation}
\langle T^{\alpha_{p_0}}_{A_{p_0}}\ldots T^{\alpha_{p_N}}_{A{p_N}}\rangle =  \alpha^{m}\delta^{\alpha_{p_0} + \alpha_{p_1}+\alpha_{p_2}+\ldots + \alpha_{p_N}}_m \langle T_{A_0}\ldots T_{A_N}\rangle ,
\end{equation}
being $\alpha^{m}$ arbitrary constants, where we have denoted the generators of $\check{\mathfrak{g}}_R$ by $\lambda_{\alpha_{p_i}} T_{A_{p_i}}\equiv T^{\alpha_{p_i}}_{A_{p_i}}$, with $i=0, \ldots, N$, and where the set $\lbrace \lambda_{\alpha_p} \rbrace$ is finite. 
\end{proof}

Let us conclude by saying that, in Ref. \cite{GenIW}, the interested reader can also find examples of application of our prescription for infinite $S$-expansion with ideal subtraction, in which we reproduced some results already known from the literature and also gave some new features.

\chapter{Conclusions and future developments} \label{chapter 7}

% **************************** Define Graphics Path **************************
\ifpdf
    \graphicspath{{Chapter7/Figs/}{Chapter7/Figs/PDF/}{Chapter7/Figs/}}
\else
    \graphicspath{{Chapter7/Figs/Vector/}{Chapter7/Figs/}}
\fi

In this concluding chapter, we summarize the original results obtained during my PhD research activity (and collected, reorganized, and clarified in this thesis), and discuss some possible future developments. 

\section{Original results concerning supergravity theories}

% chapter 4 conclusions

In \cite{Gauss} (see Chapter \ref{chapter 4}), we have presented the explicit construction of the $\mathcal{N} = 1$, $D = 4$ $AdS$-Lorentz supergravity bulk Lagragian in the rheonomic framework, showing an alternative way to introduce a generalized supersymmetric cosmological term to supergravity. Subsequently, we have studied the supersymmetry invariance of the Lagrangian in the presence of a non-trivial space-time boundary; we have found that the supersymmetric
extension of a Gauss-Bonnet like term is required in order to restore the supersymmetry invariance of the full Lagrangian (bulk plus boundary). The Lagrangian we have finally obtained can be recast in a suggestive MacDowell-Mansouri like form \cite{MM}.

The results we have presented in \cite{Gauss} and reviewed in this thesis could be useful to study higher-dimensional supergravity theories in the presence of a non-trivial boundary using the rheonomic (geometric) approach. Furthermore, it would be interesting to analyze the possible role played by the bosonic field $k^{ab}$ appearing in our model in the context of the $AdS$/CFT correspondence, in particular in the holographic renormalization language.  

% chapter 5 conclusions

%%%% Hidden

The core of my PhD research activity is concentrated in \cite{Hidden} and \cite{Malg} (see Chapter \ref{chapter 5}), both in collaboration with L. Andrianopoli and R. D'Auria. In these papers, in particular, we have deeply analyzed and discussed diverse superalgebras in eleven dimensions, which are somehow ``hot topics'' of the supergravity research field since the action of $D=11$ supergravity was first constructed in \cite{Cremmer}.

In particular, in \cite{Hidden} we have reconsidered the hidden superalgebra structure underlying supergravity theories in  space-time dimensions $D>5$ (and, in general, supersymmetric theories necessarily involving $p$-form gauge fields with $p>1$), first introduced in \cite{D'AuriaFre} in the context of the $D=11$ supergravity theory (we have called the hidden superalgebra underlying $D=11$ supergravity the DF-algebra). It generalizes the supersymmetry algebra to include the set of almost-central charges (carrying Lorentz indexes) which are currently associated with $(p-1)$-brane charges.

We have focused on the role played by the nilpotent spinor charges naturally appearing in the hidden superalgebra, showing that such extra charges, besides they are required for the equivalence of the hidden superalgebra to the FDA, are also necessary to project out of the physical superspace the non-physical degrees of freedom, decoupling them from the physical spectrum of the theory. In this sense, the extra spinors behave like cohomological BRST ghosts.

Thus, analyzing in detail the $D=11$ case, we have clarified the physical interpretation of the spinor $1$-form field dual to the nilpotent spinor charge: It is not a physical field in superspace, its differential being parametrized in an enlarged superspace which also includes the almost-central charges as bosonic tangent space generators, besides the supervielbein $\{V^a,\psi^\alpha\}$. Because of this feature, it guarantees that, instead, the $1$-forms dual to the almost-central charges are abelian gauge fields whose generators, together with the nilpotent fermionic generators, close an abelian ideal of the supergroup.

As the generators of the hidden Lie superalgebra span the tangent space of a supergroup manifold, then, in our geometrical approach, the fields are naturally defined in an enlarged manifold corresponding to the supergroup manifold, where all the invariances of the FDA are diffeomorphisms, generated by Lie derivatives. The extra spinor $1$-form $\eta$ (dual to the nilpotent fermionic generator $Q'$) appearing in the hidden superalgebra underlying $D=11$ supergravity allows, in a dynamical way, the diffeomorphisms in the directions spanned by the almost-central charges to be in particular gauge transformations, so that one obtains the ordinary superspace as the quotient of the supergroup over the fiber subgroup of gauge transformations.

We have further considered a lower-dimensional case,
with the aim of investigating a possible enlargement of the hidden supergroup structure found in $D=11$, focusing, in particular, on the minimal $D=7$ FDA. In that case, we have been able to parametrize in terms of $1$-forms the couple of mutually non-local forms $B^{(2)}$ and $B^{(3)}$. An analogous investigation in $D=11$ would have required the knowledge of the explicit parametrization of $B^{(6)}$, that is mutually non-local with $A^{(3)}$, but which at the moment has not been worked out yet.

In the $D=7$ case, we have found that \textit{two} nilpotent spinor charges are required in order to find the most general hidden Lie superalgebra equivalent to the FDA in superspace. In this case, two subalgebras exist, where only one spinor, parametrizing only one of the two mutually non-local $p$-forms, is present. We have called them \textit{Lagrangian subalgebras}, since they should correspond to the expected symmetries of a Lagrangian description of the theory in terms of $1$-forms, or, for the corresponding FDA, to the presence of either $B^{(2)}$ or $B^{(3)}$ in the Lagrangian.

Actually, as we will further discuss later on, also the $D=11$ case admits the presence of (at least) two nilpotent fermionic generators, in the sense that the extra spinor $1$-form $\eta$ appearing in the DF-algebra can be parted into two contributions, whose integrability conditions close separately (see the work \cite{Malg}, recalled in the second part of Chapter \ref{chapter 5} of this thesis).

%%%% $M$-algebra

In particular, in \cite{Malg} we have shown that, despite the $M$-algebra is a In\"on\"u-Wigner contraction of the $\osp(1|32)$ algebra (more precisely, of its torsion deformation, namely of the superalgebra we have called the RSB-algebra), still the DF-algebra cannot be directly obtained as an In\"on\"u-Wigner contraction from the SB-algebra, the latter being a (Lorentz-valued) central extension of the RSB-algebra.
Correspondingly, $D=11$ supergravity is not left invariant by the $\osp(1|32)$ algebra (not even in its torsion deformed formulation RSB), while being invariant under the DF-algebra.
This is due to the fact that the spinor $1$-form $\eta$ of the DF-algebra (that is a  spinor central extension of the $M$-algebra) contributes to the DF-algebra with structure constants different from the ones of the SB-algebra (which is related to the $\osp(1|32)$ algebra). 

More precisely, we have seen that $\eta$ differs from ${\eta_{SB}}\propto \lambda$ by the extra $1$-form generator $\xi$.
This has a counterpart in the expression of $A^{(3)}= A^{(3)}(\sigma^\Lambda)$, which trivializes the vacuum $4$-form cohomology in superspace in terms of DF-algebra $1$-form generators $\sigma^\Lambda$. $A^{(3)}(\sigma^\Lambda)$ is not invariant under the $\osp(1|32)$ algebra (neither under its torsion deformation RSB) because of the contribution $A^{(3)}_{(0)}$ explicitly breaking this symmetry; however, such term is the only one contributing to the vacuum $4$-form cohomology in superspace, due to the presence in the DF-algebra of the two spinors $\xi$ and $\eta_{SB}$ into which the cohomological spinor $\eta$ can be decomposed.

The decomposition of $A^{(3)}(\sigma^\Lambda)= A^{(3)}_{(0)}+ \alpha  A^{(3)}_{(e)}$ in superspace, where we have disclosed different contributions to the $4$-form cohomology on superspace from the two terms $dA^{(3)}_{(0)}(\sigma^\Lambda)$ and $dA^{(3)}_{(e)}(\sigma^\Lambda)$, suggests that such contributions could be possibly related to the general analysis done in \cite{Witten:1996md, Diaconescu:2000wy, Diaconescu:2000altro, Diaconescu:2003bm}, where the $4$-form cohomology of the $M$-theory on a spin manifold $Y$ is shown to be shifted, with respect to the integral cohomology class, by the canonical integral class of the spin bundle of $Y$. Referring to our discussion, it appears reasonable to conjecture that one could rephrase the above statement into the following one, in terms of the super field-strength $G^{(4)}$ in superspace: $G^{(4)}$ has integral periods in superspace, while the periods of $dA^{(3)}$ are shifted by the contribution (possibly fractional) of the spin bundle. Since our analysis refers to the FDA describing the vacuum in \textit{superspace}, we should consider as spin manifold $Y$ flat superspace, where the \textit{integral} cohomology class is trivial.    This corresponds, in our formulation, to the trivial contribution from the RSB-invariant term $A^{(3)}_{(e)}(\sigma^\Lambda)$, the only non-trivial contribution to the $4$-form cohomology on flat superspace coming from $dA^{(3)}_{(0)}(\sigma^\Lambda)$, which accounts for the contribution from the spin bundle.  

A deeper analysis of the correspondence between the two approaches, for the vacuum theory and for the dynamical theory out of the vacuum, is currently under investigation and left to future works. In particular, it is still to be explicitly shown that the contribution to the $4$-form cohomology in superspace from $dA^{(3)}_{(0)}(\sigma^\Lambda)$ could assume both integer and half-integer values. In this direction, the techniques developed in \cite{Castellani:2016ibp}, where a formulation of supergravity in superspace with \textit{integral forms} was introduced, could be particular useful.

Some future works could consist on extending the study of the hidden gauge structure of eleven-dimensional supergravity to the complete FDA (including the $6$-form), considering the full dynamical content of the theory out of the vacuum and performing an analysis for the $6$-form $B^{(6)}$ of $D=11$ supergravity similar to the one we have done in \cite{Malg} for the $3$-form $A^{(3)}$. We expect that, in the study of the $6$-form, a cohomological $1$-form spinor different from $\eta$ should play a crucial role. The decomposition of the spinor $\eta$ into a linear combination of $1$-form spinors, $\xi$ and $\eta_{SB}$, suggests that possibly the relevant spinor in the case of $B^{(6)}$ could correspond to a different linear combination of $\xi$ and $\eta_{SB}$. Such analysis should preliminarily require the knowledge of the parametrization of $B^{(6)}$ in terms of $1$-forms, which, as we have already mentioned, is not available yet.

On the other side, it appears that the extra spinor $1$-form $\eta$ could be an important addition towards the construction of a possible \textit{off-shell} theory underlying $D=11$ supergravity.
In \cite{Hassaine:2003vq}, a supersymmetric $D=11$ Lagrangian invariant under the $M$-algebra and closing off-shell without requiring auxiliary fields was constructed as a Chern-Simons form. It would be very intriguing to investigate the possible connections between our formulation and the approach adopted in \cite{Hassaine:2003vq}.

It might also be worth analyzing the connection between our approach and the theories of \textit{generalized geometry}.
In particular, the approach presented in \cite{Hidden}, where all the invariances of the FDA are expressed as Lie derivatives of the $p$-forms in the hidden supergroup manifold, could be an appropriate framework to discuss theories defined in enlarged versions of superspace recently considered in the literature, such as \textit{Double Field Theory} (DFT) and \textit{Exceptional Field Theory} (EFT) (see, for example, \cite{Hohm:2013pua, Hohm:2013uia, Hohm:2014qga} and references therein). This conjecture is based on the fact that we have recognized that the presence of extra bosonic $1$-forms in the Lie superalgebras appears to be quite analogous to the presence of extra coordinate directions in the formulation of DFT and EFT. 
We expect that our approach, where the gauge and supersymmetry constraints are dynamically implemented by the presence of the nilpotent fermionic generators, could be appropriate to formulate the constraints on which the consistency of DFT and EFT are based. In particular, the $1$-form fields $\sigma^\Lambda$ of the DF-algebra should give an alternative description of EFT, where the \textit{section constraints}, required in that theory to project the field equations on ordinary superspace, should be dynamically implemented through the presence of the cohomological spinor $\eta$. Some work is in progress on this topic.

In this context, referring to the concluding comments in Chapter (\ref{chapter 5}) done when discussing a $D=4$ case considered in \cite{SM4} and to the fact that the description of supergravity in eleven dimensions in terms of its hidden DF-algebra could be useful in the analysis of its compactification to lower dimensions, it might be attractive to better understand the possible relations between the extra bosonic fields appearing in different $D=4$ theories (such as those related to the $AdS$-Lorentz and Maxwell-like superalgebras considered in \cite{Gauss} and \cite{SM4}, respectively) and the extra bosonic $1$-forms appearing in the hidden structures underlying $D=11$ (and $D=7$) supergravities; the study of the dimensional reduction from eleven (or directly seven) to four dimensions would certainly be clarifying.

% chapter 6 conclusions

\section{New results in the context of $S$-expansion}

%%%% Analytic

On the pure group theoretical and algebraic side of my research, driven by the fact that connecting different Lie (super)algebras can give birth to new links among physical theories (and, sometimes, also to new physical theories), in the work \cite{Analytic} (see Chapter {\ref{chapter 6}}) we have developed an \textit{analytic method} (in the context of $S$-expansion) to find the semigroup(s) $S$ (we could also find more than one semigroup) linking two different (super)algebras, once certain particular conditions on the subspace decomposition of the starting and target (super)algebras and on the partition of the set(s) involved in the procedure are met. 

In the cases in which the (graded) Jacobi identities of the initial (super)algebra $\mathfrak{g}$ are trivially satisfied (each term of the Jacobi identities is equal to zero, separately), the abelian magma(s) $\tilde{S}$ involved in the procedure does not necessarily be a semigroup $S$, since associativity, in those particular cases, is not a necessary condition for the consistency of the method.

We have then given an interesting example of application
involving the Lie superalgebra $\mathfrak{osp}(1 \vert 32)$ and the $M$-algebra, reproducing the result presented in \cite{Iza1}, namely obtaining that $S^{(2)}_E$ is the semigroup leading from $\mathfrak{osp}(1 \vert 32)$ to the $M$-algebra. Our analytic method immediately allowed to recover this result, without resorting to any ``trial and error'' process; it is reliable and can also be adopted in more complicated cases.

Let us mention here that one can move from $\mathfrak{osp}(1|32)$ to the DF-algebra by performing two subsequent steps: First, one has to go from $\mathfrak{osp}(1|32)$ to the $M$-algebra (through, for example, $S$-expansion and $0_S$-resonant-reduction, with the abelian semigroup $S^{(2)}_E$); then, one has to ``centrally'' extend the $M$-algebra with an extra (nilpotent) fermionic generator.

A possible future development of the results presented in this context consists in extensions and generalizations of our method (for example, trying to release some of the initial assumptions).

%%%% GenIW

Subsequently, in \cite{GenIW} (see the second part of Chapter \ref{chapter 6}), we have given a new prescription for $S$-expansion, using an \textit{infinite} abelian semigroup $S^{(\infty)}$ and performing the subtraction of an infinite ideal subalgebra from an infinite resonant subalgebra of the infinitely $S$-expanded one. We have explicitly shown that the subtraction of the infinite ideal subalgebra corresponds to a reduction, leading to a reduced (super)algebra.
In particular, it can be viewed as a (generalization of the) $0_S$-reduction. This method also offers an alternative view of the generalized In\"on\"u-Wigner contraction. Indeed, an infinite $S$-expansion with ideal subtraction allows to reproduce the standard as well as the generalized In\"on\"u-Wigner contraction. The removal of the infinite ideal is crucial, since it allows to end up with finite-dimensional Lie (super)algebras.

We have then given a theorem for writing the invariant tensors for the (super)algebras obtained by applying our method of infinite $S$-expansion with ideal subtraction. Indeed, since the ideal subtraction  can be viewed as a $0_S$-reduction, one can then write the invariant tensors for the $0_S$-reduced (super)algebras in terms of those of the starting ones. 
This procedure allows to develop the dynamics and construct the Lagrangians of physical theories. 
In particular, in this context the construction of Chern-Simons forms becomes more accessible.

By performing our method, one can get diverse (super)algebras from the original one (depending on different choices for the resonant subspace partitions and subset decomposition of the starting algebra and of the semigroup, respectively, and, consequently, on the subtraction of different infinite ideal subalgebras), obtaining, in this way, an exhaustive overview on the possible reduced (super)algebras associated with the starting one.

In \cite{GenIW}, we have restricted our study to the case of an infinite semigroup $S^{(\infty)}$ related to the set $(\mathbb{N},+)$. We leave a possible upgrade to the set $(\mathbb{Z},+)$ to future works. 

Another possible development concerning $S$-expansion would consists in extending the procedures recalled in this thesis to include algebraic structures which link different (super)algebras by also involving Grassmann-like variables. Some work is in progress on this topic.

\begin{flushright}
``\textit{Learn from yesterday, live for today, hope for tomorrow. \\ The important thing is not to stop questioning}.''
\\
Albert Einstein
\end{flushright}

% ********************************** Back Matter *******************************
% Backmatter should be commented out, if you are using appendices after References
%\backmatter

% ********************************** Bibliography ******************************
\begin{spacing}{0.9}

% To use the conventional natbib style referencing
% Bibliography style previews: http://nodonn.tipido.net/bibstyle.php
% Reference styles: http://sites.stat.psu.edu/~surajit/present/bib.htm

%\bibliographystyle{apalike}
%%\bibliographystyle{unsrt} % Use for unsorted references  
%\bibliographystyle{plainnat} % use this to have URLs listed in References
%%\cleardoublepage
%%\bibliography{References/references} % Path to your References.bib file

% If you would like to use BibLaTeX for your references, pass `custombib' as
% an option in the document class. The location of 'reference.bib' should be
% specified in the preamble.tex file in the custombib section.
% Comment out the lines related to natbib above and uncomment the following line.

%\printbibliography[heading=bibintoc, title={References}]

\end{spacing}

% ********************************** Appendices ********************************

\begin{appendices} % Using appendices environment for more functunality

% ******************************* Thesis Appendix A ********************************

\chapter{The vielbein basis}\label{rgv}

The geometry of linear spaces as well as that of a general Riemannian manifold can be studied using the (orthonormal) moving frame\footnote{Namely, a frame of reference which moves together with the observer along a trajectory.} $\lbrace \overrightarrow{e}_i \rbrace$ and the so called dual vielbein (co)frame $\lbrace V^i \rbrace$. Let us discuss what we mean, following the same lines of Ref. \cite{Libro1}. 

From now on, we use Greek indexes to denote the coordinate indexes (also called holonomic indexes, world-indexes, or curved indexes), while the Latin indexes (called anholonomic, tangent space indexes, flat indexes, or intrinsic indexes) label the moving frame $\lbrace \overrightarrow{e}_i \rbrace$ and the new basis of $1$-forms $\lbrace V^i \rbrace$.

\section{Geometry of linear spaces in the vielbein basis}

Consider curvilinear coordinates $\lbrace x^\mu \rbrace$ on $\mathbb{R}^n$ ($n$-dimensional linear space); the tangent vectors at $P$ to the lines $x^\mu=\text{constant}$ span the so called \textit{natural basis}. The vectors of the \textit{natural frame} are given by
\begin{equation}\label{natbasismiau}
\overrightarrow{e}_\mu = \frac{\partial}{\partial x^\mu} \overrightarrow{P},
\end{equation}
where we have used the symbol $\overrightarrow{P}$ to denote the position vector of $P$ referred to some origin in $\mathbb{R}^n$.
Each vector at $\overrightarrow{P}$ can be expressed in terms of its local components. In particular, the displacement vector $d\overrightarrow{P}$ can be written as
\begin{equation}\label{vettorinodispl}
d \overrightarrow{P}= dx^\mu \frac{\partial}{\partial x^\mu}\overrightarrow{P}.
\end{equation}

Besides the natural basis (\ref{natbasismiau}), any other frame could be a suitable one. In particular, we can introduce a set of vectors $\lbrace \overrightarrow{e}_i \rbrace$ which are \textit{orthonormal} with respect to the $n$-dimensional Minkowski metric $\eta_{ij}=\left(1,-1,\ldots,-1 \right)$:\footnote{The choice of the signature $(+,-,-,\ldots,-)$, which actually corresponds to \textit{pseudo}-Riemannian, rather than Riemannian, geometries, is motivated by the fact that our aim is that of describing a theory of \textit{gravitation}. In the sequel, we will omit all the time the term ``pseudo'' and we will use Riemannian for pseudo-Riemannian, according with the convention of \cite{Libro1}.} 
\begin{equation}\label{orthobella}
\overrightarrow{e}_i \cdot \overrightarrow{e}_j = \eta_{ij}.
\end{equation}
The frame $\lbrace \overrightarrow{e}_i \rbrace$ is called the \textit{moving frame} and it is related to the natural basis (\ref{natbasismiau}) by a non-singular matrix $V^\mu_i$:
\begin{equation}\label{124a}
\overrightarrow{e}_i = V^\mu_ {\;\;i} \overrightarrow{e}_\mu , \quad \overrightarrow{e}_\mu =V^i_{\;\;\mu} \overrightarrow{e}_i,
\end{equation} 
\begin{equation}\label{124b}
V^\mu_{\;\;i}V^i_{\;\;\nu}=\delta^\mu_{\;\; \nu}, \quad V^\mu_{\;\;i} V^j_{\;\;\mu}=\delta_i^{\;\;j}.
\end{equation}
Then, introducing the \textit{differential $1$-forms} (antisymmetric tensors)
\begin{equation}
V^i = V^i_{\;\; \mu}dx^\mu ,
\end{equation}
equation (\ref{vettorinodispl}) becomes:
\begin{equation}\label{newvectordispl}
\overrightarrow{dP}=dx^\mu \left(V^i_{\;\;\mu}V^\nu_{\;\;i} \right)\frac{\partial}{\partial x^\nu} \overrightarrow{P}=V^i e_i (\overrightarrow{P}).
\end{equation}
The set of $1$-forms $\lbrace V^i \rbrace$ is the so called \textit{vielbein frame}, which is \textit{dual to the moving frame} $\lbrace e_i \rbrace$. Indeed:
\begin{equation}
V^i (\overrightarrow{e}_j)= V^i_{\;\;\mu}V^\nu_{\;\;j}dx^\mu (\overrightarrow{\partial}_\nu)=\delta^i_{\;\;j}.
\end{equation}
The relation occurring between two infinitesimally close frames $\lbrace \overrightarrow{e}_i \rbrace$ and $\lbrace \overrightarrow{e}_i + d\overrightarrow{e}_i \rbrace$ is
\begin{equation}
d\overrightarrow{e}_i = \frac{\partial \overrightarrow{e}_i}{\partial x^j}dx^j
\end{equation}
and, since $d \overrightarrow{e}_i$ is a vectorial $1$-form, we find:
\begin{equation}\label{aiutooooo}
d \overrightarrow{e}_i = - \overrightarrow{e}_j \omega^j_{\;\;i},
\end{equation}
where $\omega^j_{\;\;i}$ is an infinitesimal matrix of $1$-forms:
\begin{equation}
\omega^j_{\;\;i}= \omega^j_{\;\;i \vert \mu}dx^\mu .
\end{equation}
Differentiating the orthonormality relation (\ref{orthobella}) and using (\ref{aiutooooo}), one can show that
\begin{equation}
d (\overrightarrow{e}_i \cdot \overrightarrow{e}_j) = - (\omega_{ij} + \omega_{ji})=0 \quad \rightarrow \quad \omega_{ij}=-\omega_{ji} .
\end{equation}
Therefore, $\omega^i_{\;\;j}$ is an infinitesimal ``rotation'' matrix of the Lorentzian group $SO(1,n-1)$ and it is called the \textit{spin connection}.

\subsection{Torsion and curvature in linear spaces}

We now apply the $d$-operator to both sides of (\ref{newvectordispl}) and (\ref{aiutooooo}); the integrability condition $d^2=0$ gives the following equations:
\begin{align}
& R^i \equiv dV^i - \omega^i_{\;\;j} \wedge V^j =0, \\
& R^{i}_{\;\;j} \equiv d \omega^i_{\;\;k} \wedge \omega^k_{\;\; j}=0,
\end{align}
being ``$\wedge$'' the wedge product between differential forms.
The left-hand sides of these equations are called the \textit{torsion} and the \textit{curvature} $2$-forms, respectively.
In the $\mathbb{R}^n$ case, they are identically zero (the spin connection is a \textit{pure gauge}).

\subsection{Covariant derivatives}

Let us now consider a vector field $\overrightarrow{v}_i$ defined over a region of $\mathbb{R}^n$. Referring to the moving frame, we have
\begin{equation}
\overrightarrow{v}=v^i \overrightarrow{e}_i.
\end{equation}
Using (\ref{aiutooooo}), we can evaluate the $d\overrightarrow{v}$ due to an infinitesimal displacement:
\begin{equation}
d \overrightarrow{v}= dv^j \overrightarrow{e}_j - v^i \omega^j_{\;\;i} \overrightarrow{e}_j = (dv^i - \omega^i_{\;\;j} v^j)\overrightarrow{e}_i,
\end{equation}
where
\begin{equation}
dv^i - \omega^i_{\;\; j}v^j \equiv Dv^i
\end{equation}
is called the \textit{covariant derivative} of $v^i$.

The whole procedure can then extended to the case of $n$-dimensional smooth Riemannian manifolds $\mathcal{M}_n$. Let us see how. 

\section{Riemannian manifolds geometry in the vielbein basis}

Consider a $n$-dimensional manifold $\mathcal{M}_n$ on which a metric $g_{\mu \nu}$ has been defined. Then, $\mathcal{M}_n$ is, by definition, a (smooth) Riemannian manifold (namely, a smooth manifold with a Riemannian metric, see, for example, Ref. \cite{Libro1} for details).\footnote{In our case, the signature of the metric is actually that of pseudo-Riemannian geometry, as we have already mentioned in the previous section.}

Now, at each point $P$ of $\mathcal{M}_n$ we can set up an orthonormal local reference frame $\lbrace \overrightarrow{e}_i \rbrace$ spanning a basis of the tangent space $T_P(\mathcal{M})$ at $P$:
\begin{equation}\label{maheur}
\overrightarrow{e}_i \cdot \overrightarrow{e}_j \equiv \eta_{ij},
\end{equation}
where $\eta_{ij}$ is the Minkowskian metric on the tangent space.

Let us mention that we insist to consider orthonormal frames since one would also introduce \textit{spinor} fields on $\mathcal{M}_n$, which are $SO(1,n-1)$ representations. We are therefore forced to restrict the set of affine frames at $P$, related to each other by elements of $GL(n, \mathbb{R})$, to the subset of orthonormal frames related to each other by elements of $SO(1,n-1)$.
In particular, \textit{spinors cannot be described in the natural frame} $\lbrace \overrightarrow{\partial}_\mu \rbrace$. Indeed, under a coordinate transformation the vectors $\overrightarrow{\partial}_\mu$ transform as
\begin{equation}
\frac{\partial}{\partial x'^\mu} = \frac{\partial x^\nu}{\partial x'^\mu} \frac{\partial}{\partial x^\nu},
\end{equation}
where the Jacobian matrix $\left( \frac{\partial x^\nu}{\partial x'^\mu} \right)_P$ is, in general, an element of $GL(n, \mathbb{R})$.

The relation between the moving (orthonormal) frame and the natural one is (as in the Euclidean case):
\begin{align}
& \overrightarrow{e}_i = V^\mu _{\;\;i} \frac{\partial}{\partial x^\mu}, \\
& \frac{\partial}{\partial x^\mu}= V^i_{\;\; \mu}\overrightarrow{e}_i,
\end{align}
where $V^\mu _{\;\;i}$ is a non-singular matrix satisfying
\begin{equation}
V^\mu _{\;\;i}V^i_{\;\; \nu} = \delta^\mu_{\;\; \nu}, \quad V^\mu _{\;\;i}V^j_{\;\; \mu} = \delta^j_{\;\; i}
\end{equation}
($V^\mu _{\;\;i}$ is the inverse matrix of $V^i_{\;\; \nu}$).

The reader can find the relation with the usual tensor formulation, which utilizes the natural frame, in Ref. \cite{Libro1}.

We then express an infinitesimal displacement $\overrightarrow{d P}$ in terms of the moving frame at $T_P(\mathcal{M})$:
\begin{equation}\label{diffdispl}
\overrightarrow{d P} = V^i \overrightarrow{e}_i,
\end{equation} 
where $V^i$ are the \textit{vielbein fields dual to the moving frame} defined by
\begin{equation}
V^i (\overrightarrow{e}_j)= \delta^i_{\; \; j},
\end{equation}
that is
\begin{equation}
V^i = V^i_{\;\;\mu} dx^\mu .
\end{equation}
They are a basis for the $1$-forms on the cotangent plane at $P$.
In other words, being $\lbrace \overrightarrow{e}_i \rbrace$ the orthonormal moving frame, the corresponding orthonormal dual frame of covectors in $T^\star_P(\mathcal{M})$ is the vielbein frame.
Let us mention that, in this frame, a \textit{canonical oriented volume element} is given by the $n$-form (or \textit{volume form})
\begin{equation}
\Omega ^{(n)}=V^1 \wedge V^2 \wedge \ldots \wedge V^n ,
\end{equation}
where we have denoted by ``$\wedge$'' the wedge product between differential forms.

Actually, one can observe that the notation $\overrightarrow{d P}$ for the infinitesimal displacement is a little misleading, due to the fact that, in Riemannian geometry, (\ref{diffdispl}) is not, in general, an exact differential, since $P$ is not a function of the coordinates (contrary to what happens in the case of Euclidean geometry). With an abuse of notation, however, we continue to use the symbol ``$d$''. The same remark applies to the evaluation of the change of the moving frame under an infinitesimal translation $\overrightarrow{P} \rightarrow \overrightarrow{P} + \overrightarrow{d P}$:
\begin{equation}\label{diffdisplmoving}
d \overrightarrow{e}_i = - \overrightarrow{e}_j \omega^j_{\;\; i},
\end{equation}
where 
\begin{equation}
\omega^j_{\;\; i}= \omega^j_{\;\; i \vert \mu} dx^\mu 
\end{equation}
is called the \textit{connection}.
Then, applying the $d$-operator to both sides of (\ref{maheur}), one can (heuristically) prove\footnote{Following \cite{Libro1}, we have added the term ``heuristically'' because, at this point, in the differentiation we are actually using a differential operator that is not an exact one, and thus cannot be identified with what is commonly referred to as the $d$-operator.} that the infinitesimal matrix $\omega^j_{\;\; i}$ is antisymmetric
\begin{equation}\label{veraspinconn}
\omega_{ij}= - \omega_{ji} \, ,
\end{equation}
and therefore it is a ``rotation'' matrix that belongs to the Lie algebra of $SO(1,n-1)$ (as one would obtain in Euclidean geometry). In the sequel, we assume the validity of (\ref{veraspinconn}). In this case, $\omega^j_{\;\; i}$ is called the \textit{spin connection}.\footnote{Equation (\ref{veraspinconn}) can be referred to as the ``metric postulate'', since it strictly depends on the signature of the metric (see Ref. \cite{Libro1} for details).}

\subsection{Torsion and curvature in Riemannian manifolds geometry}

On any manifold $\mathcal{M}_n$ one can introduce the \textit{torsion} and the \textit{curvature} $2$-forms by means of the following definitions, respectively:
\begin{align}
& R^i \equiv dV^i - \omega^i_{\;\;j}\wedge V^i , \label{torsionnnnnn} \\
& R^{ij} \equiv d \omega^{ij} - \omega^i_{\;\; k }\wedge \omega^{kj}, \label{curvatureeeeeeeeee}
\end{align}
where $\omega^{ij}\equiv \omega^i_{\;\;k}\eta^{kj}$.
We will also refer to both $R^i$ and $R^{ij}$ together as the curvatures.

In general, $R^i$ and $R^{ij}$ have non-vanishing values (in Riemannian geometry). Equations (\ref{diffdispl}), (\ref{diffdisplmoving}), (\ref{torsionnnnnn}), and (\ref{curvatureeeeeeeeee}) are called the \textit{structure equations}. Let us mention that the structure equations (\ref{torsionnnnnn}) and (\ref{curvatureeeeeeeeee}) could also be (again heuristically) retrieved by taking the exterior derivative of both sides of equations (\ref{diffdispl}) and (\ref{diffdisplmoving}).

The \textit{metric tensor} on $\mathcal{M}_n$ can be written as
\begin{equation}
g_{\mu \nu}= V^i_{\;\;\mu}V^j_{\;\; \nu}\eta_{ij}.
\end{equation}
Then, differentiating both sides of equations (\ref{torsionnnnnn}) and (\ref{curvatureeeeeeeeee}), and using $d^2 =0$, we get the following integrability conditions:
\begin{align}
& dR^i + \omega^i_{\;\;j} \wedge R^j + R^{i}_{\;\;j} \wedge V^j =0, \label{b1bella} \\
& dR^i_{\;\;j} - R^i_{\;\;k} \wedge \omega^k_{\;\;j} + \omega^i_{\;\;k}\wedge R^k_{\;\;j}=0 . \label{b2bella}
\end{align}
Equations (\ref{b1bella}) and (\ref{b2bella}) are referred to as the \textit{Bianchi identities} obeyed by $R^i$ and $R^{i}_{\;\;j}$, respectively.

Now, let us explicitly observe that all the equations introduced so far are exterior equations and, as such, they are scalars under diffeomorphisms\footnote{Strictly speaking, diffeomorphisms are isomorphisms of smooth manifolds.} on $\mathcal{M}_n$. Latin indexes are inert under diffeomorphisms, being indexes of the local gauge group $SO(1,n-1)$. The same is true if we expand $\omega^i_{\;\;j}$, $R^i$, and $R^i_{\;\;j}$ in a local cotangent basis $\lbrace V^i \rbrace$:
\begin{align}
& \omega^i_{\;\;j} = \omega^i_{\;\;j \vert k} V^k, \label{expomlatin} \\
& R^i = R^i_{kl} V^k \wedge V^l , \label{exptorlatin} \\
& R^i_{\;\;j} = R^i_{\;\;j \vert k l} V^k \wedge V^l . \label{expcurvlatin}
\end{align}
Indeed, the component fields $\omega^i_{\;\;j \vert k}$, $R^i_{kl}$, and $R^i_{\;\;j \vert k l}$ carry indexes of the Latin type and are hence inter under diffeomorphisms. $R^i_{\;\;j \vert k l}$ is called the \textit{intrinsic curvature tensor}.

Summarizing: Our starting point was a Riemann manifold $\mathcal{M}_n$ endowed with a local (orthonormal) moving frame and its dual local vielbein frame $\lbrace V^i \rbrace$ in the cotangent plane. The frame $\lbrace V^i \rbrace$ is acted on by the local gauge group $SO(1,n-1)$. 
We have also introduced a local connection $1$-form $\omega^i_{\;\;j}$ and, postulating $\omega_{ij}=-\omega_{ji}$, we have identified it with an infinitesimal ``rotation'' matrix of $SO(1,n-1)$, called the spin connection. Then, we have defined the torsion and the curvature $2$-forms, and we have subsequently derived the Bianchi identities. 

If one further assumes
\begin{equation}\label{zerella}
R^i=0,
\end{equation}
then $\mathcal{M}_n$ is said to be a (Riemannian) manifold with a \textit{Riemannian connection}. 
In this case, one can express the spin connection in terms of (the space-time derivatives of) the vielbein field (see \cite{Libro1} for details).

\subsection{Lorentz covariant derivatives}

Let us now explore the \textit{gauge invariance under $SO(1,n-1)$} and define the \textit{covariant derivatives} in the case of Riemannian manifolds geometry.

Suppose we perform an $SO(1,n-1)$ gauge transformation on the local frames:
\begin{equation}\label{trasformella}
\overrightarrow{e}'_i = \overrightarrow{e}_j \Lambda^j_{\;\; i}, \quad \Lambda \in SO(1,n-1).
\end{equation}
From 
\begin{equation}
d \overrightarrow{P}= \overrightarrow{e}_i V^i = \overrightarrow{e}' _i V'^i
\end{equation}
(remember that, actually, ``$d$'' is not an exact differential operator) we obtain
\begin{equation}
V'^i = \left(\Lambda^{-1}\right)^i_{\;\;j}V^j .
\end{equation}
Then, from
\begin{equation}
d \overrightarrow{e}' = - \overrightarrow{e}' \omega'
\end{equation}
(where we have adopted a matrix notation), using (\ref{diffdisplmoving}) and (\ref{trasformella}), we have
\begin{equation}
- \overrightarrow{e} \omega \Lambda + \overrightarrow{e} d \Lambda = - \overrightarrow{e}\Lambda \omega',
\end{equation}
and therefore we can write
\begin{equation}
\omega ' = 	\Lambda^{-1} \omega \Lambda - \Lambda^{-1}  d \Lambda \;\;\; \Rightarrow \;\;\; \omega'^i_{\;\; j}= \left( \Lambda^{-1}\right)^i_{\;\; k} \omega^k_{\;\; l}\Lambda^l_{\;\; j} - \left( \Lambda^{-1} \right)^i_{\;\; k}\left( d\Lambda \right)^k_{\;\;j}.
\end{equation}
The result is that the spin connection $\omega^i_{\;\;j}$ undergoes an $SO(1,n-1)$ gauge transformation. 

Then, one can find that the torsion and the curvature $2$-forms transform in the vector and in the adjoint representations of $SO(1,n-1)$, respectively:
\begin{align}
& R'^i = \left( \Lambda^{-1}\right)^i_{\;\;j}R^j, \\
& R'^i_{\;\;j} = \left( \Lambda^{-1}\right)^i_{\;\;k} R^k_{\;\; l}\Lambda^l_{\;\;j}.
\end{align}
After that, computing the change of a vector 
\begin{equation}\label{piccolovett}
\overrightarrow{v}=v^i \overrightarrow{e}_i
\end{equation}
under an infinitesimal displacement, differentiating both sides of (\ref{piccolovett}) and using (\ref{diffdisplmoving}), one finds:
\begin{equation}
d \overrightarrow{v}= \overrightarrow{e}_i \left( d v^i - \omega^i_{\;\; j}v^j \right).
\end{equation}

Hence, we define the $SO(1,n-1)$ \textit{covariant exterior derivative} of $v^i$ by:
\begin{equation}
D v^i \equiv dv^i - \omega^i_{\;\; j} v^j.
\end{equation}
It is referred to as the \textit{Lorentz covariant derivative}.

One can also introduce $p$-form fields which are in the spinor representations of the gauge group $SO(1,n-1)$.
Let $\sigma$ be one such field in the lowest spinor representation, and let 
\begin{equation}
\Gamma_{ij}=\frac{1}{2}[\Gamma_i , \Gamma_j]
\end{equation}
be the Lorentz generators in the spinor representation, where $\Gamma^i$ are Dirac gamma matrices for $SO(1,n-1)$. Then, one can show that
\begin{equation}
D \sigma = d \sigma - \frac{1}{4} \omega_{ij} \wedge \Gamma^{ij} \sigma
\end{equation}
is the \textit{covariant derivative of the spinor $p$-form $\sigma$}.

Then, using the Lorentz covariant derivative, the torsion $2$-form can be rewritten as follows:
\begin{equation}
R^i = DV^i ,
\end{equation}
and the Bianchi identities (\ref{b1bella}) and (\ref{b2bella}) become, respectively:
\begin{align}
& D R^i + R^i_{\;\; j}\wedge V^i =0 , \label{b1newcov} \\
& D R^{i}_{\;\; j} =0. \label{b2newcov}
\end{align}

\section{Curvature tensor, Ricci tensor, and curvature scalar}

Let us finally make the symmetries of the intrinsic curvature tensor $R^i_{\;\;j \vert k l}$ explicit. Indeed, from equation (\ref{expcurvlatin}) one immediately gets
\begin{equation}
R^i_{\;\;j \vert kl}= - R^{i}_{\;\;j \vert lk} ,
\end{equation} 
and from the metric postulate (\ref{veraspinconn}):
\begin{equation}
R_{ij \vert kl}=- R_{ji \vert kl}.
\end{equation} 
Furthermore, when $\omega^i_{\;\;j}$ is a Riemannian connection, that is when equation (\ref{zerella}) holds, we get
\begin{equation}\label{dimentica}
R^i_{\;\; j}\wedge V^j =0.
\end{equation}
Expanding (\ref{dimentica}) along the vielbein basis, we find
\begin{equation}
R^i_{\;\; j \vert kl} V^j \wedge V^k \wedge V^l =0,
\end{equation}
which gives the cyclic identity
\begin{equation}
R^{i}_{\;\; j \vert kl} + R^{i}_{\;\; k \vert lj}+ R^{i}_{l \vert jk}=0.
\end{equation}
Then, one can show that
\begin{equation}
R_{ij \vert kl}=R_{kl \vert ij}.
\end{equation}
From $R^i_{\;\; j \vert kl}$ one may construct the \textit{Ricci tensor}
\begin{equation}\label{ricciolotensor}
R^i_{\;\; j \vert ik} \equiv R_{jk},
\end{equation}
which turns out to be symmetric in the indexes $j,k$, and the \textit{curvature scalar}
\begin{equation}\label{scalarettocurvatura}
\eta^{ij}R_{ij}\equiv R .
\end{equation}
Because of the aforementioned symmetry properties, any other contraction possibility gives, at most, a change of sign with respect to the definitions (\ref{ricciolotensor}) and (\ref{scalarettocurvatura}).

% ******************************* Thesis Appendix B ********************************

\chapter{Technical details on the hidden structure of FDAs}\label{apphidden}

In this appendix, we collect the notation and conventions adopted in Chapters \ref{chapter 2} and \ref{chapter 5}, together with some technical details.

\section{Fierz identities and irreducible representations}\label{fierz}

In this section, we give the $3$-gravitinos irreducible representations and the Fierz identities in $D=11$ and $D=7$ space-time dimensions.

\subsection{$3$-gravitinos irreducible representations in $D=11$}

The gravitino $\Psi_\alpha $ $(\alpha =1,    \ldots     , 32)$ of $D=11$ supergravity is a spinor $1$-form belonging to the spinor representation of $ SO(1,10) \simeq Spin(32)$.
The symmetric product
$(\alpha , \beta , \gamma)\equiv\Psi_{(\alpha} \wedge \Psi_\beta \wedge \Psi_{\gamma )}$, of dimension $\mathbf{5984}$, belongs to the three-times symmetric reducible representation  of $Spin(32)$:
The Fierz identities amount to decompose the representation $(\alpha , \beta , \gamma)$ into irreducible representations of $Spin(32)$. In this way, we obtain
\begin{equation}
\mathbf{5984} \to \mathbf{32}+\mathbf{320}+\mathbf{1408}+\mathbf{4224} .
\end{equation}
We denote the corresponding irreducible spinor representations of the Lorentz group $SO(1,10)$ as follows:

\begin{equation}
\Xi^{(32)} \in \mathbf{32} \,,\quad \Xi^{(320)}_a \in \mathbf{320}\,,\quad \Xi^{(1408)}_{a_1a_2}\in \mathbf{1408}\,,\quad \Xi^{(4224)}_{a_1    \ldots    a_5}\in \mathbf{4224}\,,
\end{equation}
where the indexes $a_1 \ldots a_n$ are antisymmetrized, and each of them satisfies 
\begin{equation}
\Gamma^a \Xi_{ab_1    \ldots     b_n}=0.
\end{equation}
Now, one can easily compute the coefficients of the explicit decomposition into the irreducible basis, obtaining (see Refs. \cite{Libro2, D'AuriaFre} for details):
\begin{eqnarray}
\Psi \wedge \bar{\Psi} \wedge \Gamma_a \Psi & = & \Xi^{(320)} _a+ \frac{1}{11}\Gamma_a \Xi^{(32)},  \\
\Psi \wedge \bar{\Psi} \Gamma_{a_1 a_2}\Psi & = & \Xi^{(1408)}_{a_1a_2}-\frac{2}{9}\Gamma_{[a_2}\Xi^{(320)}_{a_2]}+\frac{1}{11}\Gamma_{a_1 a_2}\Xi^{(32)},  \\
\Psi \wedge \bar{\Psi}\wedge \Gamma_{a_1    \ldots    a_5}\Psi & = & \Xi^{(4224)}_{a_1    \ldots    a_5}+2 \Gamma_{[a_1 a_2 a_3}\Xi^{(1408)}_{a_4a_5]}+ \frac{5}{9}\Gamma_{[a_1    \ldots    a_4}\Xi^{(320)}_{a_5]}- \frac{1}{77}\Gamma_{a_1    \ldots    a_5}\Xi^{(32)} .
\end{eqnarray}

\subsection{Irreducible representations in $D=7$}

In $D=7$, an analogous decomposition leads to:

\begin{eqnarray}
\psi_C \wedge \bar{\psi}^C \wedge \psi_A & = & \Xi _A,  \\
\psi_A \wedge \bar{\psi}^C \wedge \Gamma^{ab}\psi_C & = & \Xi^{ab} _A - \frac{2}{5}\Gamma^{[a}\Xi^{b]}_A + \frac{2}{7}\Gamma^{ab}\Xi_A ,  \\
\psi_A \wedge \bar{\psi}^C \wedge \Gamma^{a}\psi_C & = & \Xi^a_A + \frac{2}{7}\Gamma^a \Xi_A ,  \\
\psi_{(A} \wedge \bar{\psi}_B \wedge \psi_{C)} & = & \Xi_{(ABC)},  \\
\psi_C \wedge \bar{\psi}^C \wedge \Gamma^{abc}\psi_A & = & \frac{3}{2}\Gamma^{[a}\Xi^{bc]}_A + \frac{9}{10}\Gamma^{[ab}\Xi^{c]}_A - \frac{1}{7}\Gamma^{abc}\Xi_A ,  \\
\psi^C \wedge \bar{\psi}^A \wedge \Gamma^{abc}\psi^B & = & \Xi^{(ABC)\vert abc} + \frac{1}{5}\Gamma^{abc}\Xi^{(ABC)} + \nonumber \\
& & -\frac{2}{3}\epsilon^{C(A}\left(\frac{3}{2}\Gamma^{[a}\Xi^{bc]\vert B)} + \frac{9}{10}\Gamma^{[ab}\Xi^{c]\vert B)} - \frac{1}{7}\Gamma^{abc}\Xi^{\vert B)} \right) ,  \\
\psi^C \wedge \bar{\psi}^A \wedge \psi^B & = & \Xi^{(ABC)}- \frac{2}{3}\epsilon^{C(A}\Xi^{B)}.
\end{eqnarray}

\section{Some useful formulas in $D=7$}

\begin{equation}
\begin{aligned}
& \sigma^{x\vert B}_{\;\;\;\;\;A} \sigma^{x\vert D}_{\;\;\;\;\;C}= - \delta^B_{\;\;\;A} \delta^D_{\;\;\;C} + 2 \delta^B_{\;\;\;C} \delta^D_{\;\;\;A} , \\ 
& \sigma^{x \vert C}_{\;\;\;\;\;B} \sigma^{y\vert B}_{\;\;\;\;\;A}= \delta^{xy}\delta^C_{\;\;\;A} + \ii \epsilon_{xyz}\sigma^{z \vert C}_{\;\;\;\;\;A}.
\end{aligned}
\end{equation}

\section{Explicit solution for the $3$-form in $D=11$}\label{coeff11D}

In $D=11$ supergravity, for the consistency of the parametrization of the $3$-form $A^{(3)}$, given by equation (\ref{a3par}) of Chapter \ref{chapter 2}, the following set of equations must be satisfied:
\begin{equation} \label{cond11}
\left\{
\begin{array}{l}
T_0-2 S_1 E_1-1=0 , \cr
T_0-2 S_1 E_2 -2 S_2 E_1=0 , \cr
3 T_1-8 S_2 E_2=0 , \cr
T_2+10 S_2 E_3+10 S_3 E_2=0 , \cr
120 T_3-S_3 E_1-S_1 E_3=0 , \cr
T_2+1200 S_3 E_3 =0 ,  \cr
T_3-2S_3 E_3=0 , \cr
9T_4+10 S_3 E_3=0 , \cr
S_1+10 S_2-720 S_3=0,
\end{array}
\right.
\end{equation}
while the integrability condition $D^2 \eta=0$ further implies
\begin{equation}
E_1+10 E_2-720 E_3 = 0 
\end{equation}
(here we have corrected some misprints, which were in part already recognized in \cite{Bandos:2004xw}, appearing in \cite{D'AuriaFre}). This system is solved by the relations (\ref{11dsol}) written in Chapter \ref{chapter 2}.

In \cite{D'AuriaFre}, the coefficient $T_0$ was arbitrarily fixed to $T_0=1$, leading to two distinct solutions; then, if we now fix the normalization $T_0=1$ in our system, we see that we get two distinct solutions, depending on the parameter $E_2$ (which just fixes the normalization of $\eta$):
\begin{eqnarray}
& T_0 = 1 , \;\;\; T_1 =  \frac{4}{15} , \;\;\; T_2 = -\frac{5}{144}, \;\;\; T_3 = \frac{1}{17280}, \;\;\; T_4 = -\frac{1}{31104} ,  \nonumber \\
& S_1 = \begin{pmatrix}
0 \\
\frac{1}{2E_2}
\end{pmatrix}  , \;\;\; S_2 = \frac{1}{10 E_2}, \;\;\; S_3= \begin{pmatrix}
\frac{1}{720 E_2}  \\
\frac{1}{480 E_2}
\end{pmatrix} , \;\;\; E_1 = \begin{pmatrix}
5E_2 \\
0
\end{pmatrix} , \;\;\; E_3 = \begin{pmatrix}
\frac{E_2}{48} \\
\frac{E_2}{72}
\end{pmatrix} . \label{t01}
\end{eqnarray}

\section{Dimensional reduction of the gamma matrices}\label{gammamat}

Here we write the dimensional reduction of the gamma matrices from eleven to seven dimensions. We first decompose the gamma matrices in $D=11$ (hatted ones) as follows:
\begin{equation}
\hat{\Gamma}_{\hat{a}} \rightarrow  \left\{
\begin{aligned}
& D=4 \;\;\; \Gamma_i , \\
& D=7 \;\;\; \Gamma_a ,
\end{aligned} \right.
\end{equation}
where $\hat{a}=0,    \ldots    ,10$, $a=0,    \ldots    ,6$, and $i=7,8,9,10$.
Then, we can write the following decomposition:
\begin{equation}
\Gamma _i = \mathbf{1} _{4} \otimes \gamma _i , \;\;\; \Gamma_a= \Gamma_a \otimes \gamma_5 ,
\end{equation}
with
\begin{equation}
\gamma_5 = \begin{pmatrix}
\delta_{A}^{\; B} & 0 \\
0 & -\delta_{A'}^{\; B'}
\end{pmatrix} , \;\;\;\;\; \gamma^5 = \mathbf{1} _4,
\end{equation}
and
\begin{equation}
\gamma_i = \begin{pmatrix}
0 & (\gamma_i)_{A}^{\; A'} \\
(\gamma_i)_{A'}^{\; A} & 0
\end{pmatrix} , \;\;\;\;\; \lbrace \gamma_i , \gamma_j \rbrace = 2 \eta_{ij}= - 2 \delta_{ij},
\end{equation}
where $i,j,    \ldots    $ are the internal indexes running from $7$ to $10$. Let us mention that we are using a \textit{mostly minus} Minkowski metric.
Thus, we can finally write
\begin{equation}
\Gamma_a = \begin{pmatrix}
(\Gamma_a)_\alpha^{\; \beta}\delta_A^{\; B} & \mathbf{0} \\
\mathbf{0} & -(\Gamma_a)_\alpha^{\; \beta}\delta_{A'}^{\; B'}
\end{pmatrix} , \;\;\; \Gamma_i = \begin{pmatrix}
\mathbf{0} & (\gamma_i)_A^{\; A'}\delta_\alpha^{\; \beta} \\
(\gamma_i)_{A'}^{\; A}\delta_{\alpha}^{\; \beta} & \mathbf{0}
\end{pmatrix} .
\end{equation}

\section{Properties of the `t Hooft matrices}\label{tooooft}

In the following, we write the properties of the `t Hooft matrices.
The self-dual and antiself-dual `t Hooft matrices satisfy the quaternionic algebra:
\begin{align}
& J^{\pm \vert x}J^{\pm \vert y} = - \delta^{xy}\mathbf{1} _{4\times 4}+ \epsilon^{xyz} J^{\pm \vert z}, \\
& J^{\pm\vert x}_{ab} = \pm \frac{1}{2}\epsilon_{abcd}J^{\pm \vert x}_{cd}, \\
& [J^{+\vert x},J^{-\vert y}]=0, \;\;\; \forall \; x,\;y.
\end{align}
From the above relations, it follows:
\begin{equation}
\text{Tr}(J^{x}_{rs}J^{y}_{st}J^{z}_{tr})= \text{Tr}(\epsilon^{xyz'}J^{z'}J^{z})= \text{Tr}(-\epsilon^{xyz'}\delta^{z z'}\mathbf{1} _4 )= -4 \epsilon^{xyz}.
\end{equation}

% ******************************* Thesis Appendix C ********************************

\chapter{Detailed calculations concerning $S$-expansion}\label{appanex}

This appendix contains the detailed calculations referring to an example of application of the analytic method presented in \cite{Analytic} and recalled in Chapter \ref{chapter 6} of this thesis.

\section{From $\osp(1|32)$ to the $M$-algebra}\label{osp}

The (anti)commutation relations for $\mathfrak{osp}(1 \vert 32)$ and for the $M$-algebra can be found in Ref. \cite{Iza1}.
For simplicity, in the following we will just consider the structure of the (anti)commutation relations, since the explicit values of the coefficients are not relevant to our analysis. We also neglect the Lorentz indexes of the generators, labeling, in particular, by $\tilde{Z}_2$, $\tilde{Z}_5$, and $Z_5$ the generators $\tilde{Z}_{ab}$, $\tilde{Z}_{a_1 \ldots a_5}$, and $Z_{a_1 \ldots a_5}$, respectively (the former refer to the $M$-algebra, the latter to $\mathfrak{osp}(1 \vert 32)$).

We can write the (anti)commutation relations between the generators of the target $M$-algebra in terms of the (anti)commutation relations between the generators of $\mathfrak{osp}(1 \vert 32)$ (obtaining, in this way, the multiplication rules between the elements of $\tilde{S}$), schematically, as follows:

\begin{align}
& \left[\tilde{J}, \tilde{J} \right] = \left[\lambda_\alpha J_, \lambda_\alpha J \right]= \lambda_\alpha \lambda_\alpha \left[J,J \right]\; \text{and} \; \left[\tilde{J}, \tilde{J} \right] \propto \lambda_\alpha  J \;\;\; \Rightarrow \;\;\; \lambda_\alpha  \lambda_\alpha = \lambda_\alpha , \\
& \left[\tilde{J}, \tilde{P} \right] = \left[\lambda_\alpha J, \lambda_\delta P \right]= \lambda_\alpha \lambda_\delta \left[J, P \right]\; \text{and} \; \left[\tilde{J}, \tilde{P} \right] \propto \lambda_\delta P \;\;\; \Rightarrow \;\;\; \lambda_\alpha \lambda_\delta = \lambda_\delta , \\
& \left[\tilde{J}, \tilde{Z}_2 \right] = \left[\lambda_\alpha J, \lambda_\beta J \right]= \lambda_\alpha \lambda_\beta \left[J, J \right] \; \text{and} \; \left[\tilde{J}, \tilde{Z}_2 \right] \propto \lambda_\beta J \;\;\; \Rightarrow \;\;\; \lambda_\alpha \lambda_\beta = \lambda_\beta , \\
& \left[\tilde{J}, \tilde{Z}_5 \right] = \left[\lambda_\alpha J, \lambda_\delta Z_5 \right]= \lambda_\alpha \lambda_\delta \left[J,Z_5 \right] \; \text{and} \; \left[\tilde{J}, \tilde{Z}_5 \right] \propto \lambda_\delta  Z_5 \;\;\; \Rightarrow \;\;\; \lambda_\alpha \lambda_\delta = \lambda_\delta , \\
& \left[\tilde{J}, \tilde{Q} \right] = \left[\lambda_\alpha J, \lambda_\gamma Q \right]=\lambda_\alpha \lambda_\gamma \left[J ,Q \right]\; \text{and} \; \left[\tilde{J}, \tilde{Q} \right]\propto \lambda_\gamma Q \;\;\; \Rightarrow \;\;\; \lambda_\alpha \lambda_\gamma = \lambda_\gamma , \\
& \left[\tilde{P}, \tilde{P} \right] = \left[\lambda_\delta P, \lambda_\delta P \right]=\lambda_\delta \lambda_\delta \left[P,P \right] \; \text{and} \; \left[\tilde{P}, \tilde{P} \right]= 0, \; \text{while} \; \left[P,P \right] \neq 0 \;\;\; \Rightarrow \;\;\; \lambda_\delta \lambda_\delta = \lambda_{0_S} , \\
& \left[\tilde{P}, \tilde{Z}_2 \right] = \left[\lambda_\delta P, \lambda_\beta J \right]= \lambda_\delta \lambda_\beta \left[P, J \right] \; \text{and} \; \left[\tilde{P}, \tilde{Z}_2 \right]=0, \; \text{while} \; \left[P, J \right] \neq 0  \;\;\; \Rightarrow \;\;\; \lambda_\delta \lambda_\beta = \lambda_{0_S} , \\
& \left[\tilde{P}, \tilde{Z}_5 \right] = \left[\lambda_\delta P, \lambda_\delta Z_5 \right]= \lambda_\delta \lambda_\delta \left[P, Z_5 \right]  \; \left[\tilde{P}, \tilde{Z}_5 \right]=0, \; \text{while} \; \left[P, Z_5 \right] \neq 0 \;\;\; \Rightarrow \;\;\; \lambda_\delta \lambda_\delta = \lambda_{0_S} , \\
& \left[\tilde{P}, \tilde{Q} \right] = \left[\lambda_\delta P, \lambda_\gamma Q \right]=\lambda_\delta \lambda_\gamma \left[P,Q \right]\; \text{and} \; \left[\tilde{P}, \tilde{Q} \right] =0, \; \text{while} \; \left[P,Q \right] \neq 0 \;\;\; \Rightarrow \;\;\; \lambda_\delta \lambda_\gamma = \lambda_{0_S} , \\
& \left[\tilde{Z}_{2}, \tilde{Z}_{2} \right] = \left[\lambda_\beta J, \lambda_\beta J \right]= \lambda_\beta \lambda_\beta \left[J,J \right]\; \text{and} \;\left[\tilde{Z}_{2}, \tilde{Z}_{2} \right] =0, \; \text{while} \; \left[J,J \right] \neq 0 \;\;\; \Rightarrow \;\;\; \lambda_\beta \lambda_\beta = \lambda_{0_S} , \\
& \left[\tilde{Z}_{2}, \tilde{Z}_{5} \right] = \left[\lambda_\beta J, \lambda_\delta Z_5 \right]= \lambda_\beta \lambda_\delta \left[J,Z_5 \right] \; \text{and} \; \left[\tilde{Z}_{2}, \tilde{Z}_{5} \right]=0, \; \text{while} \left[J,Z_5 \right] \neq 0 \Rightarrow \lambda_\beta \lambda_\delta = \lambda_{0_S} , \\
& \left[\tilde{Z}_{2}, \tilde{Q} \right] = \left[\lambda_\beta J, \lambda_\gamma Q \right]=\lambda_\beta \lambda_\gamma \left[J,Q \right]\; \text{and}\; \left[\tilde{Z}_{2}, \tilde{Q} \right] =0, \; \text{while} \; \left[J,Q \right] \neq 0 \;\;\; \Rightarrow \;\;\; \lambda_\beta \lambda_\gamma = \lambda_{0_S} , \\
& \left[\tilde{Z}_{5}, \tilde{Z}_{5} \right] = \left[\lambda_\delta Z_5, \lambda_\delta Z_5 \right]= \lambda_\delta \lambda_\delta \left[Z_5,Z_5 \right], \; \left[\tilde{Z}_{5}, \tilde{Z}_{5} \right]=0, \; \text{while} \; \left[Z_5,Z_5 \right] \neq 0 \; \Rightarrow \; \lambda_\delta \lambda_\delta = \lambda_{0_S} , \\
& \left[\tilde{Z}_{5}, \tilde{Q} \right] = \left[\lambda_\delta Z_5, \lambda_\gamma Q \right]=\lambda_\delta \lambda_\gamma \left[Z_5,Q \right]\; \text{and}\; \left[\tilde{Z}_{5}, \tilde{Q} \right]=0, \; \text{while} \; \left[Z_5,Q \right] \neq 0 \; \Rightarrow \; \lambda_\delta \lambda_\gamma = \lambda_{0_S} , \\
& \lbrace \tilde{Q}, \tilde{Q} \rbrace = \lbrace\lambda_\gamma Q, \lambda_\gamma Q \rbrace = \lambda_\gamma \lambda_\gamma \lbrace Q,Q \rbrace \; \text{and} \; \lbrace \tilde{Q}, \tilde{Q} \rbrace \propto \lambda_\delta P + \lambda_\beta J +\lambda_\delta Z_{5} \; \Rightarrow  \label{lastcomm} \\
& \Rightarrow \; \lambda_\gamma \lambda_\gamma = \lambda_\beta , \;  \text{and} \; \lambda_\delta = \lambda_\beta.  \nonumber
\end{align}
Note that, in equation (\ref{lastcomm}), we must set 
\begin{equation}\label{equal}
\lambda_\beta=\lambda_\delta
\end{equation}
in order to get consistent relations without breaking the uniqueness of the internal composition law of $\tilde{S}$.
For performing this identification with consistency, we have exploited Theorem \ref{T} and the resulting statements.

This procedure fixes the degeneracy of the multiplication rules between the elements of the subsets of $\tilde{S}$; indeed, we are left with
\begin{align}
& \lambda_\alpha \lambda_\alpha = \lambda_\alpha , \\
& \lambda_\alpha \lambda_\beta = \lambda_\beta \lambda_\alpha = \lambda_\beta , \\
& \lambda_\alpha \lambda_\gamma = \lambda_\gamma \lambda_\alpha = \lambda_\gamma , \\
& \lambda_\beta \lambda_\beta = \lambda_{0_S}, \\
& \lambda_\beta \lambda_\gamma = \lambda_\gamma \lambda_\beta = \lambda_{0_S}, \\
& \lambda_\gamma \lambda_\gamma = \lambda_\beta .
\end{align}
We are then able to write the complete multiplication table of $\tilde{S}$.

\end{appendices}

%%% Bibliography

\newpage

\pagestyle{empty}

\begin{center}
``Tesi discussa per il conseguimento del titolo di dottore di ricerca in Fisica, svolta presso il corso di dottorato in Fisica (ciclo 30) del
Politecnico di Torino''.
\end{center}

\begin{center}
``Thesis discussed for the Ph.D title achievement in Physics, carried out in the Politecnico di Torino Ph.D program in Physics (cycle 30th).''
\end{center}

% *************************************** Index ********************************
\printthesisindex % If index is present

%\newpage
%
%\begin{figure}
%\centering
%\pgfdeclareimage[height=20cm]{Cecia}{Cecia}
%\pgfuseimage{Cecia}
%\caption[Cecia's (Space-)Time]{Cecia's Awakening.} \label{Ceciaaaa}
%\end{figure}

\end{document}